\documentclass[11pt]{article}

\usepackage[top=2cm, bottom=2cm, left=2.5cm, right=2.5cm]{geometry}
\usepackage{amsmath,amssymb,amsthm}
\usepackage{mathtools}
\usepackage{color}
\usepackage{graphicx}
\usepackage{tikz}
\usepackage{xspace}
\usepackage[backend=biber,style=alphabetic]{biblatex}
\usepackage[hidelinks]{hyperref}
\usepackage[nameinlink,capitalise]{cleveref}
\usepackage{stackengine}

\usetikzlibrary{arrows.meta,calc,patterns}

\newcommand{\diff}{\mathop{}\!\mathrm{d}}								
\newcommand{\Diff}{\ensuremath{\mathcal{D}}}							
\newcommand{\Dplus}{\ensuremath{\mathcal{D}^+}}							
\newcommand{\eps}{{\varepsilon}}										
\newcommand{\clip}{{\text{clip}}}										
\newcommand{\Lip}[1]{\text{Lip}_{#1}}									
\newcommand{\diam}{{\text{diam}}}										
\DeclareMathOperator{\polydeg}{deg}									
\DeclareMathOperator{\rem}{rem}										
\DeclareMathOperator{\wind}{wind}									
\DeclareMathOperator{\length}{length}								
\renewcommand{\Re}[1]{\operatorname{Re}\left(#1\right)}				
\DeclareMathOperator{\domop}{dom}										
\newcommand{\dom}[1]{\domop\left(#1\right)}								
\newcommand{\supp}[1]{\text{supp}\left(#1\right)}						
\newcommand{\indic}[1]{\text{Ind}\left(#1\right)}						
\newcommand{\proj}[2]{\text{proj}_{#2}\left(#1\right)}					
\newcommand{\aproj}[3]{\widehat{\text{proj}}^{#3}_{#2}\left(#1\right)}	
\newcommand{\conv}[1]{\operatorname{conv}\left(#1\right)}				
\newcommand{\Reg}[1]{\operatorname{Reg}\left(#1\right)}					
\newcommand{\BReg}[1]{\operatorname{BReg}\left(#1\right)}				
\newcommand{\SReg}[1]{\operatorname{SReg}\left(#1\right)}				
\newcommand{\transp}{\ensuremath{^{\top}}}								
\newcommand{\nrm}[1]{\left\lVert#1\right\rVert}							
\newcommand{\abs}[1]{\left\lvert#1\right\rvert}							
\newcommand{\inter}[1]{\operatorname{int}\left(#1\right)}				
\newcommand{\msubset}{%
	\mathrel{\raisebox{0.3ex}{\stackunder[1pt]{$\subset$}{$\scriptscriptstyle\sim$}}}%
}																		
\newcommand{\scal}[1]{\left\langle#1\right\rangle}						
\newcommand{\E}[2][{}]{\mathbb E_{#1}\left[#2\right]}					
\newcommand{\PP}[1]{\mathbb P\left[#1\right]}							
\newcommand{\bigo}[1]{\mathcal O\left(#1\right)}						
\newcommand{\low}{{\text{low}}}											
\newcommand{\up}{{\text{up}}}											
\DeclareMathAlphabet{\mathbbold}{U}{dsrom}{m}{n}						
\DeclareMathAlphabet{\mathbboldzero}{U}{bbold}{m}{n}					
\newcommand{\one}{\mathbbold{1}}										
\newcommand{\zero}{\mathbboldzero{0}}									
\newcommand{\ie}{\unskip, i.\,e.,\xspace}								
\newcommand{\eg}{\unskip, e.\,g.,\xspace}								
\newcommand{\sut}{\text{s.\,t.\,}}										

\let\emptyset\varnothing
\newcommand{\N}{{\mathbb{N}}}											
\newcommand{\Q}{{\mathbb{Q}}}											
\newcommand{\R}{{\mathbb{R}}}											
\newcommand{\spc}{{\,\,}}												

\title{Some Essential Constructive Foundations for Systems and Control}
\author{
	Pavel Osinenko
	\thanks{
	Central University, 7 Gasheka St., bld. 1, Moscow, 123056, Russia;
	Center for Engineering Systems and Sciences, Skoltech, 30 Bolshoy Boulevard, bld. 1, Moscow, 121205, Russia;
	Sirius University of Science and Technology, 1 Olimpiyskiy Avenue, Sirius, 354340, Russia;
	Email: p.osinenko@gmail.com	
	}
	}
\date{2026}

\newtheoremstyle{definitionname}
	{\topsep}
	{\topsep}
	{\normalfont}
	{}
	{\bfseries}
	{}
	{.5em}
	{\thmname{#1}\thmnumber{ #2}.\thmnote{ (#3)}}

\theoremstyle{definitionname}
\newtheorem{definition}{Definition}[section]
\newtheorem{notation}{Notation}[section]

\newtheoremstyle{plainname}
	{\topsep}
	{\topsep}
	{\itshape}
	{}
	{\bfseries}
	{}
	{.5em}
	{\thmname{#1}\thmnumber{ #2}.\thmnote{ (#3)}}

\theoremstyle{plainname}
\newtheorem{theorem}{Theorem}[section]
\newtheorem{proposition}{Proposition}[section]
\newtheorem{lemma}{Lemma}[section]
\newtheorem{corollary}{Corollary}[section]

\theoremstyle{remark}
\newtheorem{remark}{Remark}[section]
\newtheorem{example}{Example}[section]

\begin{document}

\maketitle

\begin{abstract}
	This work develops several constructive foundations for systems and control within Bishop-style constructive mathematics.
For an engineer, the guiding principle is that an object claimed to exist, such as a trajectory, an optimal control law, a selector, or a viable solution, should come with finite data and an operation computing approximations to any prescribed precision.
The style remains close to classical analysis, but existential statements are organized so that their computational content is visible.
The paper begins with elementary geometric data in finite-dimensional Euclidean spaces: blocks, multiblocks, representable sets, regular functions, and certified integrals.
This set-first integration route is meant to complement, rather than replace, abstract constructive integration theories such as Daniell-type or integration-space approaches.
The developed apparatus is then applied to a constructive functional extremum-value theorem, selector extraction for multifunctions, Filippov-type and viable solutions of differential inclusions, regular probability densities, controlled Markov chains, and empirical density certificates.
A short account of resolvent projectors and linear stability is included for completeness.

\end{abstract}

\setcounter{tocdepth}{3} 
\tableofcontents         
\clearpage               

\let\standardsection\section
\renewcommand{\section}{\clearpage\standardsection}

\section{Introduction}

\subsection{Bishop's Program}

Bishop's program of constructive mathematics is a form of analysis in which existence statements are expected to carry computational meaning \cite{Bishop1967FoundationsCon,Bishop1985ConstructiveAn,Bridges2007TechniquesCons}.
In this reading, saying that a real number exists means having a procedure which computes rational approximations with a Cauchy certificate.
Saying that a function is continuous means having a modulus which tells how accurately the input must be known in order to obtain a prescribed output accuracy.
Saying that a trajectory, a selector, or an optimizer exists means having finite data from which approximations can be produced to any requested precision.
This is why constructive mathematics is attractive for systems and control.
Control theory is full of objects that are useful only after they become numerical or algorithmic: solutions of differential inclusions, stabilizing feedbacks, optimal policies, invariant sets, probability densities, and spectral decompositions.
Bishop-style analysis asks for the finite content of such objects at the level of the theorem statement.

This viewpoint should not be read as a rejection of classical analysis.
Most of the definitions and proofs below are intentionally close to the classical ones.
The difference is that several classical shortcuts are avoided when they hide computational information.
For instance, a proof by contradiction that an optimizer exists does not by itself explain how to approximate an optimizer.
A compactness argument that extracts a subsequence does not by itself explain which finite stage should be used.
A measurable-selection theorem may identify a selector while leaving no usable description of the selector for integration or simulation.
Constructive mathematics replaces such steps by witnesses, moduli, finite nets, explicit exception sets, and certified limiting procedures.

The present paper uses a deliberately elementary layer of data.
At the lowest level we use rational blocks, finite unions of blocks, finite refinements, moduli, and explicitly bounded regular functions.
This choice is not meant to exhaust constructive analysis.
It is meant to keep the mathematical objects recognizable to an engineer who wants to know what can be computed and which finite certificates have to be stored.
In particular, the theory starts from subsets of Euclidean spaces rather than from abstract measure spaces.
This makes the later constructions of densities, Markov kernels, selectors, and approximate solutions depend on geometric data that can be inspected and refined.

\subsection{Motivation and Goal}

The motivation follows the numerical uncertainty viewpoint developed in \cite{Osinenko2021Towardsconstru}.
Classical control theory often proves the existence of exact objects which are then replaced by approximate computations in practice.
This mismatch is harmless in some routine situations, but it becomes visible in safety-critical, certified, or formalized control workflows.
Here are some examples: no stable computation of an optimal control policy might be possible due to non-uniqueness;
a switching surface may be hard to locate exactly in simulation; a Filippov right-hand side may require a selector which may fail to be computable; a probability transition may be described by a density whose products and marginals must still be valid densities after computation; an eigenvector-based decomposition may be numerically unstable near repeated eigenvalues.
Robust control addresses perturbations of plants and signals, but it does not automatically address the finite meaning of the mathematical objects used in a proof.

The earlier constructive-control outline \cite{Osinenko2021Towardsconstru} emphasized trajectories, selectors, extremum-value arguments, and linear spectral questions.
It did not develop a full measure and integration layer.
The present work fills this gap.
The central objects are representable sets and regular functions.
Representable sets encode geometric supports by multiblock cores, exteriors, and controlled exceptions.
Regular functions are finite sums of complemented continuous functions over representable supports, together with effective bounds when integration is required.
From these data we build integrals, probability densities, Markov kernels, multifunction selectors, Filippov approximants, viability certificates, empirical density tests, and a resolvent-based linear systems appendix.

The goal is not to replace existing constructive integration theories by a more general one.
It is narrower and more applied.
We develop an elementary Euclidean calculus which is sufficiently strong for a set of recurring control-theoretic tasks and sufficiently explicit that each theorem records the finite data it needs.
This explains the recurring emphasis on blocks, multiblocks, flattening, effective boundedness, tail certificates, and support geometry.

\subsection{Related Work}

The constructive foundations used here go back to Bishop and to the later development by Bishop--Bridges, Bridges--Richman, Troelstra--van Dalen, Beeson, and others \cite{Bishop1967FoundationsCon,Bishop1985ConstructiveAn,Bridges1987VarietiesConst,Bridges2007TechniquesCons,Troelstra2014Constructivism,Beeson1980FoundationsCon}.
For constructive algebra we rely on the finite algebraic viewpoint represented by Mines--Richman--Ruitenburg and Richman \cite{Mines1988CourseConstruc,Richman2000fundamentalthe}.
The strict-finitist perspective of Ye is also relevant to the style of finite witnesses used below \cite{Ye2011StrictFinitism}.
The integration-space viewpoint of Chan gives a broad constructive framework for integration \cite{Chan2019FoundationsCon}.
Our approach is different in emphasis: it starts from elementary Euclidean set descriptions and builds only the integration machinery needed for regular functions, kernels, and finite products.
This makes the resulting theory less abstract and more directly tied to geometric certificates.
It is not a claim that Daniell-type or integration-space constructions are unnecessary; rather, they answer a more general foundational question.

There is also a large literature on computable analysis and computability in dynamical systems \cite{Weihrauch2012ComputableAnal,Collins2009computabletype,Buescu2011Computabilityd,Graca2011Computabilityp}.
That literature is close in spirit but different in technical emphasis.
Computable analysis usually studies whether an object is computable relative to a representation.
The present paper instead keeps a Bishop-style proof discipline in which theorem statements name the concrete moduli, bounds, exception multiblocks, and tail certificates used by the construction.
This difference is one of emphasis rather than opposition.

Formal methods and verified control provide another point of contact.
Sound abstraction, deductive stability proofs, differential dynamic logic, proof assistants, and numerically robust verification all push control theory toward explicit certificates \cite{Shoukry2015Soundcomplete,Bessa2016Formalnonfrag,Cohen2017formalproofCo,Rouhling2018formalproofCo,GalloisWong2018Coqformalizati,Tan2020DeductiveStabi,Platzer2008Differentialdy,Platzer2008Keymaerahybrid,Gao2019Numericallyrob,Yin2024FormalSynthesisControllers}.
Algorithmic control perspectives make a similar demand from a design viewpoint \cite{Tsiotras2017algorithmiccon}.
The numerical uncertainty examples in \cite{Vasile2017Optimisingresi,Sutherland2019closedloopLya} are representative of the practical relevance of this direction.
Recent work on learning-enabled control makes this relevance even more articulate.
Reachability methods for neural-network-controlled systems combine set propagation, Taylor models, zonotopes, Lipschitz optimization, and inner--outer approximations to obtain safety certificates for increasingly complex closed loops \cite{Schilling2022VerificationNeuralNetworkControl,Goubault2022RINO,Zhang2023ReachabilityNNCS,ArjomandBigdeli2024ContinuousTimeNNCS}.
Safety-filter methods give a complementary control-theoretic certificate layer for autonomous systems \cite{Hsu2024SafetyFilter}.
These developments provide certified analyses of concrete systems or controller classes.
The distinct feature of the present paper is its lower-level and foundational character: it asks which Euclidean sets, functions, integrals, densities, selectors, and solution concepts can be represented with finite witnesses before a particular controller or verification algorithm is chosen.

The present manuscript continues a line of constructive work by the author on control-oriented extremum-value theorems, selectors, Caratheodory-type solutions, stability, and spectral calculations \cite{Osinenko2018constructiveve,Osinenko2018Analysisextrem,Osinenko2021constructiveex,Osinenko2018AnalysisCarath,Osinenko2018Practicalsampl,Osinenko2018Practicalstabi,Osinenko2020Constructivean,Osinenko2021Towardsconstru}.
The new contribution here is the common measure-theoretic and functional backbone supporting these themes.
Differential inclusions and viability theory provide the classical control background for the later solution concepts \cite{Filippov1962certainquestio,Aubin2012DifferentialIn,Aubin2011ViabilityTheor,Clarke2011Lyapunovfuncti}.
The linear systems section uses standard finite-dimensional resolvent ideas in the spirit of Kato \cite{Kato1995Perturbation}, but phrases them through constructive algebra and contour certificates rather than through exact eigenvector extraction.

\subsection{Structure of Work}

Section by section, the work is summarized as follows.

\begin{itemize}
	\item \Cref{sec_sets} introduces blocks, multiblocks, locatedness, total boundedness, representable sets, and measure containment.
	The main point is that sets are described by finite multiblock data with controlled exceptional parts.
	This is the geometric language used throughout the manuscript.
	\item \Cref{sec_functions} develops continuous, complemented, regular, locally regular, and effectively bounded regular functions.
	It proves measurability and Riemann integrability of bounded regular functions, records closure properties, and establishes the Fubini-type tools needed later for kernels and products.
	\item \Cref{sec_optimal_functions} proves constructive functional extremum-value theorems.
	The first part extends the earlier supremum-norm result to functions with values in $\R^m$ by a mollifier argument.
	The second part uses the integration theory to obtain $d_1$ and $d_2$ variants for regular functions whose supports have finite measure nets or cone-collar data.
	\item \Cref{sec_multifunctions} develops multifunctions as set-valued analogues of regular functions.
	It separates domain conditions from value conditions and proves selector theorems for block-full, marginalized, and Hausdorff-Lipschitz convex values.
	Steiner points are used as a canonical Lipschitz selection mechanism.
	\item \Cref{sec_filippov_solutions} applies the selector results to differential inclusions.
	It defines exact and certified approximate Filippov solutions, sample-and-hold certificates, Picard-type approximants, and viable solutions.
	The emphasis is on finite certificates rather than on abstract existence alone.
	\item \Cref{sec_probability_densities} introduces regular densities, events, expectations, conditional densities, kernels, and finite-horizon controlled Markov chains.
	All probabilistic objects are ordinary integrals of regular functions over Euclidean spaces.
	No abstract sample space is used.
	\item \Cref{sec_empirical_density_certificates} records a finite substitute for randomness.
	It defines empirical density certificates for finite sample bundles and shows how such certificates propagate through regular Markov links.
	The statements are test-relative and deliberately avoid algorithmic randomness or pseudorandomness assumptions.
	\item \Cref{sec_linear_systems} is included for completeness.
	It adapts finite-dimensional resolvent projectors and constructive polynomial algebra to obtain spectral-cluster projectors and a basic linear stability criterion.
	This addresses approximate spectral decomposition without relying on exact eigenvector extraction.
\end{itemize}

The material is deliberately organized in designated environments: definitions, notations, examples, remarks, lemmas, propositions, corollaries, and theorems.
This makes the dependency structure easier to audit.

\section{Basic Notions}

\begin{remark}
	This section fixes the elementary language used throughout the paper.
	Real numbers, functions, and sequences are read through operations producing finite data, and the notation introduced here is used without repeatedly naming the underlying witnesses.
\end{remark}

\begin{definition}[Elementary data]
	\label{dfn_elem_data}
	\emph{Elementary data} are a finite collection of natural numbers and possibly zero.
\end{definition}

\begin{remark}
	Elementary data are to be understood contextually.
	For instance, elementary data may encode an integer as a tuple of a natural number and a bit indicating the sign.
	A rational can be encoded as pair of a natural and an integer number.
	A rational-valued vector can be encoded as a finite collection of rationals etc.
\end{remark}

\begin{definition}[Operation]
	\label{dfn_operation}
	An \emph{operation} is a computational routine of finitely many iterations producing elementary data from elementary data.
\end{definition}

\begin{remark}
	\label{rem_reals_funcs}
	An example of an operation is a routine that computes rational approximations of a given real number satisfying the Cauchy property:
	\begin{equation}
		\label{eqn_real}
		x \in \R \triangleq q_x: \Q_{>0} \rightarrow \Q \quad \sut \forall \eps, \eps' \in \Q_{>0} \; \abs{q_x(\eps) - q_x(\eps')} \le \eps + \eps'.
	\end{equation}
	Such an operation $q_x$ is called a \textit{witness} for $x \in \R$.
	For brevity, we do not explicitly mention the underlying operations and just say, for instance, ``consider a real number $x$''.
	Notice the contrast to the classical case, where a real is an \textit{equivalence} class of Cauchy sequences.
	Another example is an operation that computes the values of a given function as rational approximations.
\end{remark}

\begin{definition}[Sequence]
	\label{dfn_sequence}
	A \emph{sequence} is an operation producing data from natural numbers.
	In other words, for every $i \in \N$, the sequence, seen as an operation, produces some elementary data $d_i$.
\end{definition}

\begin{remark}
	A finite sequence is simply a finite collection of pieces of elementary data.
	In constructive mathematics, we do not speak of infinite collections of arbitrary nature.
	Hence, any sequence must come with an operation computing the members of this sequence.
	Sets should not be confused with ``infinite collections of arbitrary nature''.
	When we speak of \eg real numbers $\R$, we treat $\R$ as a type of data rather than a collection.
	To specify a set in constructive mathematics, we only need two ingredients:
	\begin{itemize}
		\item to specify what it means to be a member of the given set,
		\item to specify what it means for two members of the set to be equal.
	\end{itemize} 
	For $\R$, we already addressed the item membership in \Cref{rem_reals_funcs}.
	Next, if $x, y \in \R$ (with witnesses $q_x, q_y$) and $x = y$, then it holds that
	\begin{equation}
		\label{eqn_reals_equal}
		\forall \eps \in \Q_{>0} \; \abs{q_x(\eps) - q_y(\eps)} \le 2\eps.
	\end{equation}	
\end{remark}

\begin{remark}
	Sequences may be nested or, in other words, parametrized.
	For instance, a sequence of real numbers is an operation $q_x(\eps \mid i)$ that is parametrized by the member index $i$.
	Again, we often omit the underlying operations and just say ``a sequence of reals''.
\end{remark}

\begin{remark}
	The following basic things are used.
	\begin{itemize}
		\item We refer to the sets classically called ``non-empty'' as inhabited (it is a ``positive'' notion, liked in constructive mathematics).
		\item We will usually refer to an index $k$ as running in a finite range, whereas $i$ will run in an infinite range.
		For instance, we will sometimes just write $\sum_k$ meaning the sum is finite.
		In contrast, $\sum_i$ is an infinite sum.
	\end{itemize}
\end{remark}

\section{Sets}
\label{sec_sets}

\begin{remark}
	This section develops the finite geometric language used by the rest of the paper.
	Blocks and multiblocks are the elementary bricks, while representable sets encode a set by core, exterior, and exception data whose measure cost can be made small.
	The later definitions of regular functions, events, supports, and exception layers all rely on this set-theoretic layer.
\end{remark}

\subsection{Common Notions}

\begin{notation}
	\label{ntn_interior}
	The interior of a set $X$ is denoted by $\inter{X}$.
\end{notation}

\begin{notation}
	\label{ntn_convex_hull}
	For $X\subset\R^n$, its convex hull is denoted by
	\[
		\conv{X}.
	\]
	For finitely many points $x_1,\ldots,x_N$, we write
	\[
		\conv{\{x_1,\ldots,x_N\}}
	\]
	for the set of all finite convex combinations of these points.
\end{notation}

\begin{definition}[Disjoint sets]
	\label{dfn_disjoint_sets}
	Two sets $X,Y\subseteq\R^n$ are called \emph{disjoint} if
	\[
		X\cap Y=\emptyset .
	\]
	A finite family $X_1,\ldots,X_N$ is called \emph{pairwise disjoint} if $X_k\cap X_l=\emptyset$ for all $k\ne l$.
\end{definition}

\subsection{Locatedness and Total Boundedness}
\label{subsec_locatedness_total_boundedness}

\begin{definition}[Located set]
	\label{dfn_located_set}
	An inhabited set $X\subset\R^n$ is called \emph{located} if it is equipped with an operation which, for every $x\in\R^n$, computes the real number
	\[
		d(x,X):=\inf_{y\in X}\nrm{x-y}.
	\]
	Equivalently, the distance from any point to $X$ is available as constructive real data.
	Upper estimates are understood as part of the infimum data: a proof that $d(x,X)<r$ supplies $y\in X$ with $\nrm{x-y}<r$.
\end{definition}

\begin{definition}[Finite net and total boundedness]
	\label{dfn_finite_net}
	Let $(\mathbb{X},d)$ be a metric space and let $\eps>0$.
	A finite sequence $x_1,\ldots,x_N$ in $\mathbb{X}$ is called an \emph{$\eps$-net} for $\mathbb{X}$ if
	\[
		\forall x\in\mathbb{X}\spc \exists k\in\{1,\ldots,N\}\spc d(x,x_k)\le\eps.
	\]
	The space $\mathbb{X}$ is called \emph{totally bounded} if it admits an $\eps$-net for every $\eps>0$.
\end{definition}

\begin{lemma}[Finite nets locate subsets of Euclidean space]
	\label{lem_totally_bounded_located}
	Every inhabited totally bounded subset $X\subset\R^n$ is located.
\end{lemma}

\begin{proof}
	This is the standard constructive finite-net proof of locatedness for totally bounded metric subspaces; compare Bishop--Bridges \cite[Chapter~4]{Bishop1985ConstructiveAn}.
	Fix $z\in\R^n$.
	For a prescribed rational $\eps>0$, take an $\eps$-net
	\[
		x_1,\ldots,x_N\in X
	\]
	and compute
	\[
		d_N(z):=\min_{k=1,\ldots,N}\nrm{z-x_k}.
	\]
	Since the net points belong to $X$,
	\[
		d(z,X)\le d_N(z).
	\]
	Conversely, for any $x\in X$ choose $x_k$ with $\nrm{x-x_k}\le\eps$.
	Then
	\[
		d_N(z)\le\nrm{z-x_k}\le\nrm{z-x}+\eps.
	\]
	Taking the infimum over $x\in X$ gives
	\[
		d_N(z)\le d(z,X)+\eps.
	\]
	Thus the computable numbers $d_N(z)$ approximate $d(z,X)$ with arbitrarily small error, which is the required locatedness operation.
\end{proof}

\subsection{Blocks and Multiblocks}

\begin{definition}[Block]
	\label{dfn_block}
	A \emph{block} is a closed hyperrectangle in $\R^n$ with rational vertices.
\end{definition}

\begin{remark}
	A block is an example of a piece of elementary data.
\end{remark}

\begin{remark}
	\label{rem_degenerate_block}
	A block can be empty or degenerate, meaning, one of its side lengths is zero.
\end{remark}

\begin{definition}[Multiblock]
	\label{dfn_multiblock}
	A \emph{multiblock} is a sequence of blocks $\mathcal M = \{B_i\}_i$.
	Its support is defined as
	\begin{equation}
		\label{eqn_multiblock_support}
		\bar{\mathcal M} := \bigcup_i B_i.
	\end{equation}
	Accordingly, a \textit{finite} multiblock is a finite sequence of blocks.
\end{definition}

\begin{remark}
	\label{rem_empty_multiblock}
	The empty finite sequence is allowed.
	It is called the empty multiblock, its support is empty.
\end{remark}

\begin{remark}
	\label{rem_support_multiblock}
	For any union of finitely many blocks, there is a corresponding finite multiblock.
\end{remark}

\begin{notation}[Concatenation of finite multiblocks]
	\label{ntn_multiblock_concatenation}
	If
	\[
	\mathcal M_1=\{B_{1,k}\}_{k=1}^{N_1},\ldots,
	\mathcal M_r=\{B_{r,k}\}_{k=1}^{N_r}
	\]
	are finite multiblocks, their concatenation is the finite multiblock
	\[
	\mathcal M_1\sqcup\cdots\sqcup\mathcal M_r
	:=
	\{B_{1,1},\ldots,B_{1,N_1},\ldots,B_{r,1},\ldots,B_{r,N_r}\}.
	\]
	The support of the concatenation is
	\[
	\overline{\mathcal M_1\sqcup\cdots\sqcup\mathcal M_r}
	=
	\bar{\mathcal M}_1\cup\cdots\cup\bar{\mathcal M}_r.
	\]
\end{notation}

\begin{remark}
	\label{rem_being_in_union}
	Membership in the support of a multiblock is understood positively.
	Thus, a proof that $x\in\bar{\mathcal M}=\bigcup_i B_i$ consists of an index $j\in\N$ and a proof that $x\in B_j$.
	This convention does not assert that, for an arbitrary point $x$, membership in $\bar{\mathcal M}$ is decidable or that such an index can be computed without further data.
\end{remark}

\begin{notation}
	\label{ntn_boundary}
	The boundary of a block $B$ is denoted $\partial B$.
	By analogy, we denote the boundary of the support of a multiblock $\bar{\mathcal M}$ as $\partial \bar{\mathcal M}$.
	For a multiblock $\mathcal M$, we abuse the notation and also write $\partial \mathcal M$ meaning the respective multiblock of the boundaries of the constituents.
	Notice, the dimensions of the constituents of the boundary multiblock decreases relative to the original multiblock.
\end{notation}

\begin{definition}[Full block]
	\label{dfn_full_block}
	A block \(B=\prod\limits_{k=1}^n [a_k,b_k]\) is called \textit{full}, if
	\begin{equation}
		\label{eqn_full_block}
		a_k < b_k
	\end{equation}
	for all \(k=1,\ldots,n\).
	A multiblock is full if all its constituents are full.
	Analogously, its support is also called full.
\end{definition}

\begin{definition}[Well-connected multiblock]
	\label{dfn_connected_block}
	A full multiblock $\mathcal M = \{B_i\}_i$ is called \textit{well-connected} if, for any two constituent blocks $B_i,B_j\in\mathcal M$, there exists a finite chain
	\[
	B_i = C_0,\ C_1,\ldots,\ C_r = B_j
	\]
	of constituent blocks of $\mathcal M$ such that $C_{q-1}\cap C_q$ is full for every $q=1,\ldots,r$.
\end{definition}

\begin{remark}
	The support of a well-connected multiblock does not have ``gaps'' and degenerate ``interchanges''.
	In other words, there exists a tiny block which one could move through the whole support without ever leaving the boundary.
\end{remark}

\begin{definition}[Block operations]
	\label{dfn_block_operations}
	Define an operation $\oplus$ on a block \(B=\prod\limits_{k=1}^n [a_k,b_k]\) for any $\lambda>0$ as follows:
	\begin{equation}
		\label{eqn_block_oplus}
		B \oplus \lambda = \prod\limits_{k=1}^n [a_k - \lambda, b_k + \lambda]
	\end{equation}
	By analogy, define $\ominus$ with the exception that if the respective side length is shrunk into negative, it is rendered zero.
\end{definition}

\begin{remark}
	\label{rem_block_operations}
	The operations $\oplus$ and $\ominus$ can be extended to multiblocks and their supports straightforwardly.
\end{remark}

\subsection{Grid Refinement}
\label{sub_grid_refinement}

\begin{definition}[Coordinate refinement]
	\label{dfn_coordinate_refinement}
	Let $B=\prod_{l=1}^n[a_l,b_l]$ be a full block, and let $\mathcal A$ be a finite multiblock.
	For $l=1,\ldots,n$, let $T_l(B,\mathcal A)$ be the finite set consisting of $a_l,b_l$ and all $l$-th coordinate endpoints of all constituent blocks of $\mathcal A$.
	Write the elements of $T_l(B,\mathcal A)\cap[a_l,b_l]$ as
	\[
	t_{l,0}<t_{l,1}<\cdots<t_{l,r_l}.
	\]
	The \emph{coordinate refinement} of $B$ generated by $\mathcal A$ is the finite full multiblock
	\[
	\mathcal R(B;\mathcal A)
	:=
	\left\{
	R_q=
	\prod_{l=1}^n[t_{l,q_l},t_{l,q_l+1}]
	:
	q=(q_1,\ldots,q_n),\quad 0\le q_l<r_l
	\right\},
	\]
	after omitting empty and degenerate products.
	Its constituent blocks have pairwise disjoint interiors and are contained in $B$.
	Moreover, for every $R\in\mathcal R(B;\mathcal A)$ and every constituent block $A$ of $\mathcal A$, either $R\subseteq A$ or $\inter{R}\cap\inter{A}=\emptyset$.
	This is a finite rational block relation, decidable from the coordinate endpoints.

	If $\rho>0$ is smaller than half the distance between any two distinct elements of any $T_l(B,\mathcal A)$, define the coordinate-boundary multiblock
	\[
	\mathcal H_\rho(B;\mathcal A)
	:=
	\left\{
	B\cap\{x\in\R^n:\abs{x_l-t}\le\rho\}:
	l=1,\ldots,n,\quad t\in T_l(B,\mathcal A)
	\right\},
	\]
	after omitting empty and degenerate blocks, and define the trimmed coordinate refinement
	\[
	\mathcal R_\rho(B;\mathcal A)
	:=
	\left\{
	R_{q,\rho}=
	\prod_{l=1}^n[t_{l,q_l}+\rho,t_{l,q_l+1}-\rho]:
	q=(q_1,\ldots,q_n),\quad 0\le q_l<r_l
	\right\},
	\]
	again after omitting empty and degenerate products.
	For every $R\in\mathcal R_\rho(B;\mathcal A)$ and every constituent block $A$ of $\mathcal A$, either $R\Subset A$ or $R\cap A=\emptyset$.
\end{definition}

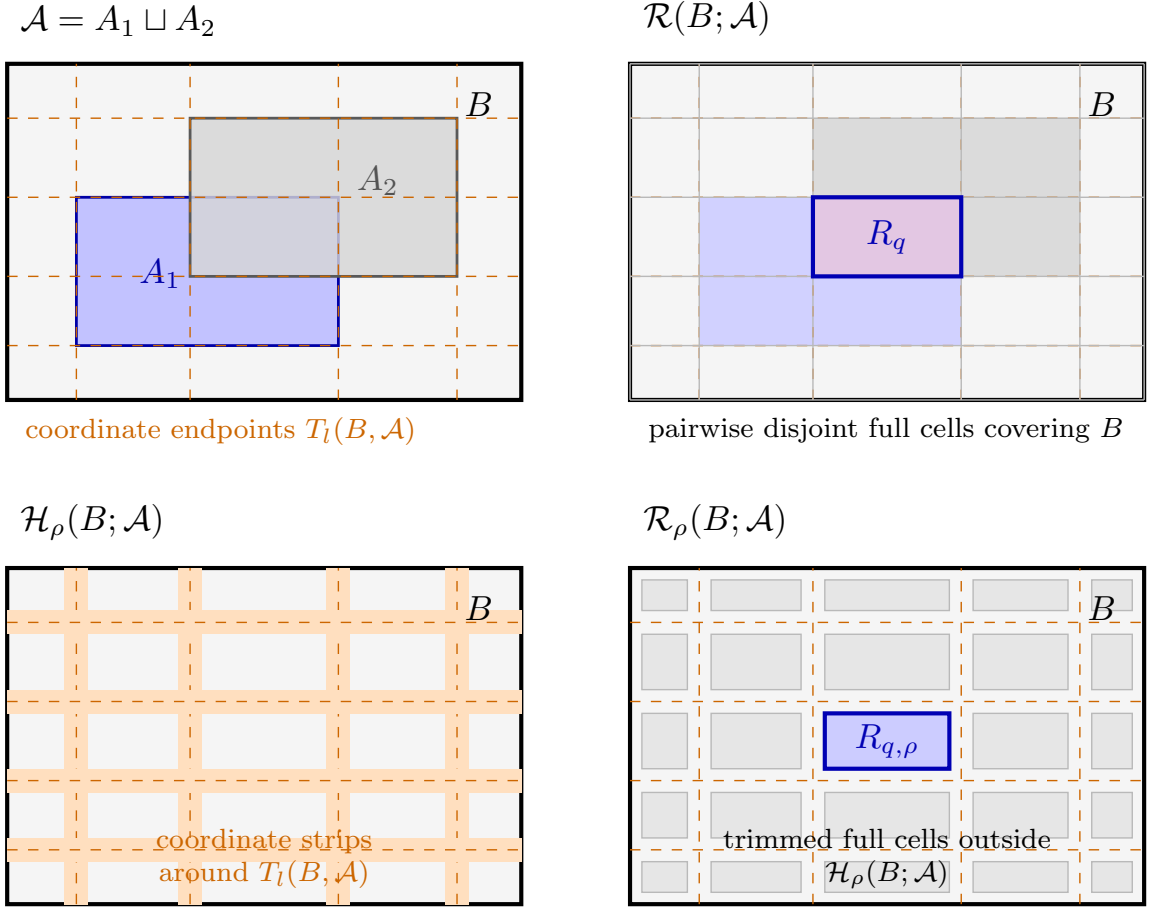
\begin{figure}[ht]
	\centering
	\resizebox{0.92\linewidth}{!}{\begin{tikzpicture}[x=1cm,y=1cm,font=\small]
	\begin{scope}
		\fill[gray!8] (0,0) rectangle (5.2,3.4);
		\draw[very thick] (0,0) rectangle (5.2,3.4);
		
		\fill[blue!24] (0.7,0.55) rectangle (3.35,2.05);
		\draw[blue!65!black,thick] (0.7,0.55) rectangle (3.35,2.05);
		\node[blue!65!black] at (1.55,1.28) {$A_1$};
		
		\fill[gray!32,opacity=0.82] (1.85,1.25) rectangle (4.55,2.85);
		\draw[gray!70!black,thick] (1.85,1.25) rectangle (4.55,2.85);
		\node[gray!65!black] at (3.75,2.22) {$A_2$};
		
		\foreach \x in {0.7,1.85,3.35,4.55} {
			\draw[orange!80!black,dashed] (\x,0) -- (\x,3.4);
		}
		\foreach \y in {0.55,1.25,2.05,2.85} {
			\draw[orange!80!black,dashed] (0,\y) -- (5.2,\y);
		}
		
		\node[anchor=north east] at (5.05,3.25) {$B$};
		
		\node[anchor=south west] at (0,3.55) {$\mathcal A=A_1\sqcup A_2$};
		\node[orange!80!black,anchor=north west,font=\scriptsize, yshift=0.1cm] at (0.05,-0.14) {coordinate endpoints $T_l(B,\mathcal A)$};
	\end{scope}
	
	\begin{scope}[xshift=6.3cm]
		\fill[gray!8] (0,0) rectangle (5.2,3.4);
		\draw[very thick] (0,0) rectangle (5.2,3.4);
		
		\fill[blue!18] (0.7,0.55) rectangle (1.85,1.25);
		\fill[blue!18] (1.85,0.55) rectangle (3.35,1.25);
		\fill[blue!18] (0.7,1.25) rectangle (1.85,2.05);
		\fill[violet!22] (1.85,1.25) rectangle (3.35,2.05);
		\fill[gray!28] (3.35,1.25) rectangle (4.55,2.05);
		\fill[gray!28] (1.85,2.05) rectangle (3.35,2.85);
		\fill[gray!28] (3.35,2.05) rectangle (4.55,2.85);
		
		\foreach \x in {0.7,1.85,3.35,4.55} {
			\draw[orange!80!black,dashed] (\x,0) -- (\x,3.4);
		}
		\foreach \y in {0.55,1.25,2.05,2.85} {
			\draw[orange!80!black,dashed] (0,\y) -- (5.2,\y);
		}
		\foreach \x in {0,0.7,1.85,3.35,4.55,5.2} {
			\draw[thin,gray!55] (\x,0) -- (\x,3.4);
		}
		\foreach \y in {0,0.55,1.25,2.05,2.85,3.4} {
			\draw[thin,gray!55] (0,\y) -- (5.2,\y);
		}
		
		\draw[very thick,blue!70!black] (1.85,1.25) rectangle (3.35,2.05);
		\node[blue!70!black] at (2.6,1.65) {$R_q$};
		
		\node[anchor=north east] at (5.05,3.25) {$B$};
		
		\node[anchor=south west] at (0,3.55) {$\mathcal R(B;\mathcal A)$};
		\node[anchor=north west,font=\scriptsize, yshift=0.1cm] at (0.05,-0.14) {pairwise disjoint full cells covering $B$};
	\end{scope}
	
	\begin{scope}[yshift=-5.1cm]
		\fill[gray!8] (0,0) rectangle (5.2,3.4);
		\draw[very thick] (0,0) rectangle (5.2,3.4);
		
		\foreach \x in {0.7,1.85,3.35,4.55} {
			\fill[orange!25] (\x-0.12,0) rectangle (\x+0.12,3.4);
			\draw[orange!80!black,dashed] (\x,0) -- (\x,3.4);
		}
		\foreach \y in {0.55,1.25,2.05,2.85} {
			\fill[orange!25] (0,\y-0.12) rectangle (5.2,\y+0.12);
			\draw[orange!80!black,dashed] (0,\y) -- (5.2,\y);
		}
		
		\node[anchor=north east] at (5.05,3.25) {$B$};
		
		\node[anchor=south west] at (0,3.55) {$\mathcal H_\rho(B;\mathcal A)$};
		\node[orange!80!black,font=\scriptsize,align=center,anchor=north,yshift=0.9cm] at (2.6,0) {coordinate strips\\around $T_l(B,\mathcal A)$};
	\end{scope}
	
	\begin{scope}[xshift=6.3cm,yshift=-5.1cm]
		\fill[gray!8] (0,0) rectangle (5.2,3.4);
		\draw[very thick] (0,0) rectangle (5.2,3.4);
		
		\foreach \x in {0.7,1.85,3.35,4.55} {
			\draw[orange!80!black,dashed] (\x,0) -- (\x,3.4);
		}
		\foreach \y in {0.55,1.25,2.05,2.85} {
			\draw[orange!80!black,dashed] (0,\y) -- (5.2,\y);
		}
		
		\foreach \xa/\xb in {0.12/0.58,0.82/1.73,1.97/3.23,3.47/4.43,4.67/5.08} {
			\foreach \ya/\yb in {0.12/0.43,0.67/1.13,1.37/1.93,2.17/2.73,2.97/3.28} {
				\fill[gray!20] (\xa,\ya) rectangle (\xb,\yb);
				\draw[thin,gray!55] (\xa,\ya) rectangle (\xb,\yb);
			}
		}
		
		\fill[blue!20] (1.97,1.37) rectangle (3.23,1.93);
		\draw[very thick,blue!70!black] (1.97,1.37) rectangle (3.23,1.93);
		\node[blue!70!black] at (2.6,1.65) {$R_{q,\rho}$};
		
		\node[anchor=north east] at (5.05,3.25) {$B$};
		
		\node[anchor=south west] at (0,3.55) {$\mathcal R_\rho(B;\mathcal A)$};
		\node[font=\scriptsize,align=center,anchor=north,yshift=0.9cm] at (2.6,0) {trimmed full cells outside\\$\mathcal H_\rho(B;\mathcal A)$};
	\end{scope}
\end{tikzpicture}}
	\caption{Coordinate refinement from \Cref{dfn_coordinate_refinement}. The concatenated multiblock $\mathcal A=A_1\sqcup A_2$ determines the coordinate endpoint sets $T_l(B,\mathcal A)$. The exact refinement $\mathcal R(B;\mathcal A)$ consists of the full cells $R_q$ cut out by those endpoints. The thickened endpoint strips form $\mathcal H_\rho(B;\mathcal A)$, and the remaining trimmed full cells form $\mathcal R_\rho(B;\mathcal A)$.}
	\label{fig_coordinate_refinement}
\end{figure}

\subsection{Block Measure, Cost and Distance}
\label{sub_block_measure_etc}

\begin{definition}
	\label{dfn_block_volume}
	Let $\mu(B)$ denote the standard Lebesgue volume of a block $B$.
\end{definition}

\begin{remark}
	\label{rem_block_volume}
	Empty and degenerate blocks have volume zero.
\end{remark}

\begin{definition}[Volume of a finite multiblock]
	\label{dfn_finite_multiblock_volume}
	Let $\mathcal M=\{B_k\}_{k=1}^N$ be a finite multiblock.
	Its volume is the Lebesgue volume of its support,
	\[
	\mu(\bar{\mathcal M})
	:=
	\mu\left(\bigcup_{k=1}^N B_k\right).
	\]
\end{definition}

\begin{remark}
	\label{rem_finite_multiblock_volume}
	Let $\mathcal M=\{B_k\}_{k=1}^N$ be a finite multiblock, and let $B$ be any full block containing $\bar{\mathcal M}$.
	Apply \Cref{dfn_coordinate_refinement} to $B$ and $\mathcal M$.
	For $R\in\mathcal R(B;\mathcal M)$, the defining property of the refinement gives either $R\subseteq B_k$ or $\inter{R}\cap\inter{B_k}=\emptyset$ for each $k$.
	Hence the finite set
	\[
		\mathcal R_{\mathcal M}
		:=
		\{R\in\mathcal R(B;\mathcal M):\exists k,\ R\subseteq B_k\}
	\]
	is computable from the rational endpoints, and
	\[
		\mu(\bar{\mathcal M})
		=
		\sum_{R\in\mathcal R_{\mathcal M}}\mu(R).
	\]
	The equality is independent of the auxiliary block $B$: the omitted part is contained in the union of finitely many coordinate faces of the blocks $B_k$, hence has volume zero.
	See \Cref{fig_coordinate_refinement}.
\end{remark}

\begin{lemma}[Monotonicity and finite additivity of multiblock volume]
	\label{lem_finite_multiblock_volume_monotone_additive}
	Let $\mathcal M$ and $\mathcal N$ be finite multiblocks.
	If
	\[
		\bar{\mathcal M}\subseteq\bar{\mathcal N},
	\]
	then
	\[
		\mu(\bar{\mathcal M})\le\mu(\bar{\mathcal N}).
	\]
	If $\mathcal M_1,\ldots,\mathcal M_r$ are finite multiblocks whose supports have pairwise disjoint interiors, then
	\[
		\mu\left(\overline{\mathcal M_1\sqcup\cdots\sqcup\mathcal M_r}\right)
		=
		\sum_{s=1}^r\mu(\bar{\mathcal M}_s).
	\]
\end{lemma}

\begin{proof}
	Choose a full block $B$ containing all supports under consideration and apply \Cref{dfn_coordinate_refinement} to $B$ and to the concatenation of the relevant multiblocks.
	The volume of each support is computed by summing the volumes of the refined blocks contained in that support, as in \Cref{rem_finite_multiblock_volume}.
	If $\bar{\mathcal M}\subseteq\bar{\mathcal N}$, every refined block counted for $\bar{\mathcal M}$ is counted for $\bar{\mathcal N}$.
	This gives monotonicity.
	If the interiors of the supports $\bar{\mathcal M}_s$ are pairwise disjoint, no full refined block can be counted for two different supports.
	The refined blocks counted for the concatenation are therefore the disjoint union of the refined blocks counted for the separate supports, which gives the displayed additivity formula.
\end{proof}

\begin{definition}[Measure containment]
	\label{dfn_measure_containment}
	Let $S,T\subset\R^n$ be supports of finite multiblocks.
	We write
	\[
	S \msubset T
	\]
	if
	\[
		S\subseteq T
	\]
	and
	\[
		\mu(T)-\mu(S)=0.
	\]
	Equivalently, $T$ contains no positive-volume part outside $S$.
	The same notation is used for finite multiblocks by applying the definition to their supports.
\end{definition}

\begin{remark}
	\label{rem_measure_containment_no_localization}
	The relation $\msubset$ is stronger than pointwise inclusion and weaker than equality.
	It records that the larger set may differ from the smaller one only by a zero-volume part.
	The inclusion in \Cref{dfn_measure_containment} is an ordinary implication between finite block supports, usually verified by finite coordinate refinements.
	For finite multiblocks with rational vertices, such inclusions reduce to finitely many rational endpoint checks after passing to a common coordinate refinement.
	It does not provide a procedure localizing an arbitrary real point of the larger set into one constituent block of the smaller multiblock.
\end{remark}

\begin{lemma}[Measure containment for coordinate refinements]
	\label{lem_coordinate_refinement_measure_containment}
	Let $B$ be a full block, let $\mathcal A$ be a finite multiblock, and use the notation of \Cref{dfn_coordinate_refinement}.
	Then
	\[
		\overline{\mathcal R(B;\mathcal A)}\msubset B.
	\]
	If $\rho>0$ is admissible in \Cref{dfn_coordinate_refinement}, then
	\[
		\overline{
			\mathcal H_\rho(B;\mathcal A)
			\sqcup
			\mathcal R_\rho(B;\mathcal A)
		}
		\msubset
		B.
	\]
\end{lemma}

\begin{proof}
	The exact coordinate refinement omits only empty or degenerate coordinate products.
	The omitted products lie in finitely many coordinate hyperplanes and have volume zero.
	Thus $\overline{\mathcal R(B;\mathcal A)}\subseteq B$ and
	\[
		\mu(B)-\mu(\overline{\mathcal R(B;\mathcal A)})=0.
	\]
	For the trimmed refinement, each interval
	\[
		[t_{l,q_l},t_{l,q_l+1}]
	\]
	is the union, up to endpoints, of
	\[
		[t_{l,q_l},t_{l,q_l}+\rho],
		\qquad
		[t_{l,q_l}+\rho,t_{l,q_l+1}-\rho],
		\qquad
		[t_{l,q_l+1}-\rho,t_{l,q_l+1}].
	\]
	Taking products over $l=1,\ldots,n$ shows that every point of $B$ which is not in a trimmed cell of $\mathcal R_\rho(B;\mathcal A)$ is contained in $\bar{\mathcal H}_\rho(B;\mathcal A)$.
	Since both $\mathcal H_\rho(B;\mathcal A)$ and $\mathcal R_\rho(B;\mathcal A)$ are contained in $B$, this gives the stated measure containment.
\end{proof}

\begin{definition}[Block cost]
	\label{dfn_block_cost}
	For a multiblock $\mathcal M = \{B_i\}_i$, define its \textit{cost} by
	\begin{equation}
		\label{eqn_block_cost}
		\gamma(\mathcal M) := \sum_i \mu(B_i)
	\end{equation}
	if it converges.
\end{definition}

\begin{remark}
	\label{rem_block_cost_not_measure}
	The functional $\gamma$ is not a measure, as overlaps between constituent blocks are counted with multiplicity.
	For every finite multiblock,
	\[
	\mu(\bar{\mathcal M})\le\gamma(\mathcal M).
	\]
	In the definition of representability, $\gamma$ is used only as an upper bound for the size of exception multiblocks.
	If $\mathcal M$ is finite, $\gamma$ always exists \ie is finite.
	If $\mathcal M$ is infinite, $\gamma$ may not exist even if the support of $\mathcal M$ has finite volume $\mu$.
\end{remark}

\begin{definition}[Distance to a finite multiblock]
	\label{dfn_distance_multiblock}
	For a finite multiblock $\mathcal M=\{B_k\}_{k=1}^N$, set
	\[
	d(x,\bar{\mathcal M}) := \min_{k=1,\ldots,N} d(x,B_k),
	\]
	where $d(x,B_k)=\inf_{y\in B_k}\nrm{x-y}$.
	For a block $B=\prod_k[a_k,b_k]$, the number $d(x,B)$ is computed coordinatewise by clamping $x$ to $B$.
\end{definition}

\subsection{Apartness}
\label{sub_apartness}

\begin{definition}[Set apartness]
	\label{dfn_set_apartness}
	We say $X$ is \textit{apart} from $Y$ if
	\begin{equation}
		\label{eqn_apart}
		\forall x \in X \; \forall y \in Y \; \exists \delta > 0 \; \lVert x - y \rVert \ge \delta,
	\end{equation}
	Where $\lVert x - y \rVert$ is the Euclidean distance.
	In other words, apartness is a relation between sets witnessed by $\delta$ on pairs of elements as above.
\end{definition}

\begin{notation}
	\label{ntn_apartness_complement}
	According to \Cref{dfn_set_apartness}, denote
	\begin{equation}
		\label{eqn_setminus}
		X \setminus Y \triangleq \{ x \in X: \forall y \in Y \; \exists \delta > 0 \; \lVert x - y \rVert \ge \delta \}.
	\end{equation}
	Factually, $X \setminus Y$ is a short-hand for an apartness complement $X^c_Y$ which comes as a part of a complemented set $(Y, X^c_Y)$.
	It implies that whenever we write $X \setminus Y$, an ambient complemented set with apartness relation is meant.
	We denote the (global apartness) set complement as $X^c \triangleq \R^n \setminus X$.
\end{notation}

\begin{remark}
	\label{rem_disjointness_negative}
	Disjointness is a negative no-overlap condition.
	It does not imply the positive apartness information needed to evaluate complemented branches.
	In particular, from $X\cap Y=\emptyset$ one cannot infer constructively that $X\subseteq Y^c$, where $Y^c$ denotes the apartness complement.
\end{remark}

\begin{definition}[Mutual apartness]
	\label{dfn_mutual_apartness}
	A finite family of complemented sets $X_1,\ldots,X_N\subseteq\R^n$ is called \emph{mutually apart} if, for every $k\ne l$,
	\[
		X_l\subseteq X_k^c,
	\]
	where $X_k^c$ is the apartness complement supplied with $X_k$.
\end{definition}

\begin{remark}
	\label{rem_set_apartness}
	The described notion of apartness is ``positive'' which is common in constructive mathematics.
	Instead of plain negation, we require witnesses for two points to be ``different'' in terms of the Euclidean distance.
	Evidently, if $X \subset B$, then $B = X \cup \{B \setminus X\}$ does not hold constructively.
\end{remark}

\begin{remark}
	\label{rem_disjoint_apart_use}
	Disjointness and mutual apartness are used for different purposes.
	Grid refinements, multiblock volume additivity, and measure additivity require only disjoint interiors or pairwise disjoint supports.
	Pointwise branch arguments require mutual apartness, because they must turn a witness $x\in X_l$ into exterior witnesses $x\in X_k^c$ for the other branches.
	If $x\in(\bigcup_{k=1}^N X_k)^c$, then $x\in X_k^c$ for every $k$ by the definition of the apartness complement of the union.
\end{remark}

\begin{remark}
	\label{rem_mutual_apart_disjoint}
	Mutual apartness implies pairwise disjointness.
	Indeed, if $x\in X_l\cap X_k$ with $k\ne l$, then $x\in X_k$ and the witness $x\in X_k^c$ supplied by $X_l\subseteq X_k^c$ gives a positive distance from $x$ to itself, which is impossible.
\end{remark}

\begin{definition}[Uniform apartness]
	\label{dfn_unif_apartness}
	We say $X$ is uniformly \textit{apart} from $Y$ if
	\begin{equation}
		\label{eqn_uniform_apartness}
		\exists \delta > 0 \; \forall x \in X \; \forall y \in Y \;  \; \lVert x - y \rVert \ge \delta.
	\end{equation}
\end{definition}

\begin{remark}
	Notice the sequence of quantifies in the ordinary and uniform apartness notions.
\end{remark}

\begin{notation}
	\label{ntn_Subset}
	According to \Cref{dfn_unif_apartness}, write
	\begin{equation}
		\label{eqn_well_contained}
		X \Subset Y
	\end{equation}
	if there exists $\delta>0$ such that
	\[
	\forall x\in X \spc \forall y\in Y^c \spc \nrm{x-y}\ge\delta .
	\]
	Hence, $\Subset$ is a \textit{well-containment} relation.
\end{notation}

\begin{remark}
	If a set $X$ is bounded \ie there exists a block $B$ such that $X \subseteq B$, then there exists another block $B'$ such that $X \Subset B'$.
\end{remark}

\begin{remark}
	For every bounded set $X$, the complement $X^c$ is clearly inhabited.
\end{remark}

\subsection{Representable Sets}
\label{subsec_representable_sets}

\begin{definition}[Representable set]
	\label{dfn_representable_set}
	A set $X \subset \R^n$ is called \emph{representable} if:

	\begin{enumerate}
		\item $X$ is bounded, i.e., there exists a block $B_X$ such that
		\begin{equation}
			\label{eqn_representable_set_bounded}
			X \Subset B_X,
		\end{equation}

		\item for every $\eps > 0$, there exist finite multiblocks $\mathcal E_X, \mathcal M_X, \mathcal M_X^c$ such that the following conditions hold:

		\begin{enumerate}
			\item $\mathcal E_X, \mathcal M_X^c$ are full, and $\gamma(\mathcal E_X) \le \eps$,
			\item the sets
				\begin{equation}
					\label{eqn_representable_set_disjoint_sets}
					\bar{\mathcal M}_X \setminus \partial \bar{\mathcal M}_X,\quad
					\bar{\mathcal E}_X \setminus \partial \bar{\mathcal E}_X,\quad
					\overline{\mathcal M_X^c} \setminus \partial \overline{\mathcal M_X^c}
				\end{equation}
				are pairwise disjoint,
				\item
				\begin{equation}
					\label{eqn_representable_set_decomposition}
					\overline{
						\mathcal M_X
						\sqcup
						\mathcal E_X
						\sqcup
						\mathcal M_X^c
					}
					\msubset
					B_X,
				\end{equation}
			\item
			\begin{equation}
				\label{eqn_representable_set_interior}
				\bar{\mathcal M}_X \Subset X,
			\end{equation}
			\item
			\begin{equation}
				\label{eqn_representable_set_exterior}
				\overline{\mathcal M_X^c} \Subset B_X \setminus X.
			\end{equation}
		\end{enumerate}
	\end{enumerate}
	Accordingly, an unbounded set $X$ is called \textit{locally representable} if $B\cap X$ is representable for every block $B$.
\end{definition}

\begin{remark}
	\label{rem_representable_cover}
	The relation in \Cref{eqn_representable_set_decomposition} is measure containment in the sense of \Cref{dfn_measure_containment}.
	It says that the support of the concatenated witness multiblock is contained in $B_X$ and that the part of $B_X$ not covered by this support has volume zero.
	It does not assert a procedure deciding, for an arbitrary real point, which of the three supports contains it.
	The well-contained information is only the core condition \Cref{eqn_representable_set_interior} and the exterior condition \Cref{eqn_representable_set_exterior}.
\end{remark}

\begin{notation}
	\label{ntn_representability_witness_parts}
	For a fixed representability witness
	\[
	B_X,\qquad \mathcal E_X,\qquad \mathcal M_X,\qquad \mathcal M_X^c
	\]
	of a set $X$, we call $B_X$ the \textit{bounding block}, $\mathcal M_X$ the \textit{core multiblock}, $\mathcal M_X^c$ the \textit{exterior multiblock}, and $\mathcal E_X$ the \textit{exception multiblock}.
	These multiblocks always refer to a chosen accuracy and to the corresponding witness.
\end{notation}

\begin{remark}
	\label{rem_boundary_layer}
	The multiblock $\mathcal E_X$ represents a boundary layer of arbitrarily small cost separating the multiblocks $\mathcal M_X$ and $\mathcal M_X^c$.
	The core may be empty.
	Representability records measurability data with no guarantee that the set contains a full block.
\end{remark}

\begin{definition}[Full representable set]
	\label{dfn_full_representable_set}
	A representable set $X\subset\R^n$ is called \emph{full representable} if, for every $\eps>0$, its representability witness can be chosen so that the core multiblock $\mathcal M_X$ is full.
\end{definition}

\begin{remark}
	\label{rem_empty_set_representable}
	The empty set is representable.
	Indeed, choose any bounding block $B$ and use the empty multiblock as the core, an arbitrarily thin exception multiblock, and the remaining full blocks of a rational subdivision of $B$ as the exterior.
	No decision that a refined cell contains a point of the empty set is involved.
\end{remark}

\begin{remark}
	\label{rem_constructive_complement}
	The exterior condition in \Cref{eqn_representable_set_exterior} is a constructive complement condition.
	It says not merely that the exterior multiblock is outside $X$, but that it is well contained in the apartness complement of $X$ relative to the bounding block.
\end{remark}

\begin{remark}
	\label{rem_local_repr}
	If $X$ is locally representable and some full block $B$ is used as the ambient localization block, one may first replace $B$ by $B\ominus\lambda$ for a prescribed geometric margin $\lambda>0$.
	This keeps the boundary of $B$ from being counted as part of the exception multiblock of $X$ itself.
\end{remark}

\begin{definition}[Well-connected representable set]
	\label{dfn_well_connected_representable_set}
	A representable set $X \subset \R^n$ is called \emph{well-connected} if it is full representable and, under the conditions of \Cref{dfn_representable_set}, for every $\eps>0$ the witness can be chosen so that $\mathcal M_X$ is well-connected.
\end{definition}

\begin{remark}
	In this work, \Cref{dfn_well_connected_representable_set} is a constructive counterpart of a connected set in classical mathematics.
\end{remark}

\begin{definition}[Measure of a representable set]
	\label{dfn_representable_set_measure}
	Let $X$ be representable.
	For every $\eps>0$, choose a representability witness
	\[
	B_X,\qquad \mathcal E_{X,\eps},\qquad \mathcal M_{X,\eps},\qquad \mathcal M^c_{X,\eps}
	\]
	with the same bounding block $B_X$ and with $\gamma(\mathcal E_{X,\eps})\le\eps$.
	Set
	\[
	\mu^\low_\eps(X)
	:=
	\mu(\bar{\mathcal M}_{X,\eps}),
	\qquad
	\mu^\up_\eps(X)
	:=
	\mu(B_X)-\mu(\overline{\mathcal M^c_{X,\eps}}).
	\]
	For rational $\eps>0$, choose such a witness and set
	\[
	q_{\mu(X)}(\eps):=\mu^\low_\eps(X).
	\]
	By \Cref{rem_representable_measure_gap}, this operation satisfies the Cauchy condition in \Cref{eqn_real}.
	The represented real number is denoted by $\mu(X)$.
	The quantities $\mu^\up_\eps(X)$ are certified upper bounds for the same real number.
\end{definition}

\begin{definition}[Measure containment of representable sets]
	\label{dfn_measure_containment_representable_sets}
	Let $X$ and $Y$ be representable sets.
	Write
	\[
		X\msubset Y
	\]
	when
	\[
		X\subseteq Y
		\qquad\text{and}\qquad
		\mu(Y)-\mu(X)=0.
	\]
\end{definition}

\begin{remark}[Existence of the measure]
	\label{rem_representable_measure_gap}
	The construction in \Cref{dfn_representable_set_measure} is independent of the chosen witnesses.
	For a witness at accuracy $\eps$, write
	\[
	l_\eps:=\mu(\bar{\mathcal M}_{X,\eps}),
	\qquad
	u_\eps:=\mu(B_X)-\mu(\overline{\mathcal M^c_{X,\eps}}).
	\]
	By \Cref{eqn_representable_set_decomposition},
	\[
	\overline{
		\mathcal M_{X,\eps}
		\sqcup
		\mathcal E_{X,\eps}
		\sqcup
		\mathcal M^c_{X,\eps}
	}
	\msubset
	B_X.
	\]
	The three supports have pairwise disjoint interiors.
	By measure containment and \Cref{lem_finite_multiblock_volume_monotone_additive},
	\[
	\mu(B_X)
	=
	\mu(\bar{\mathcal M}_{X,\eps})
	+
	\mu(\bar{\mathcal E}_{X,\eps})
	+
	\mu(\overline{\mathcal M^c_{X,\eps}}).
	\]
	Hence
	\[
	0
	\le
	u_\eps-l_\eps
	=
	\mu(\bar{\mathcal E}_{X,\eps})
	\le
	\gamma(\mathcal E_{X,\eps})
	\le
	\eps.
	\]
	Thus each witness gives an interval of length at most $\eps$.

	Now take two witnesses with accuracies $\eps$ and $\delta$.
	The supports $\bar{\mathcal M}_{X,\eps}$ and $\overline{\mathcal M^c_{X,\delta}}$ are disjoint: the first is well contained in $X$, and the second is well contained in $B_X\setminus X$.
	By \Cref{lem_finite_multiblock_volume_monotone_additive},
	\[
	l_\eps+\mu(\overline{\mathcal M^c_{X,\delta}})
	=
	\mu\left(
	\bar{\mathcal M}_{X,\eps}
	\cup
	\overline{\mathcal M^c_{X,\delta}}
	\right)
	\le
	\mu(B_X).
	\]
	Therefore
	\[
	l_\eps
	\le
	\mu(B_X)-\mu(\overline{\mathcal M^c_{X,\delta}})
	=
	u_\delta
	\le
	l_\delta+\delta.
	\]
	Interchanging $\eps$ and $\delta$ gives $l_\delta\le l_\eps+\eps$.
	Thus
	\[
	\abs{l_\eps-l_\delta}\le\eps+\delta,
	\]
	so the lower approximations satisfy the Cauchy condition from \Cref{eqn_real}.
	Since $0\le u_\eps-l_\eps\le\eps$, the upper approximations determine the same real number.
\end{remark}

\subsubsection{Closure Properties}
\label{subsubsec_representable_set_closure}

\begin{lemma}[Finite unions of representable sets]
	\label{lem_representable_finite_union}
	If $X_1,\ldots,X_N\subset\R^n$ are representable, then
	\[
	X:=\bigcup_{k=1}^N X_k
	\]
	is representable.
\end{lemma}

\begin{proof}
	Choose bounding blocks $B_k$ for $X_k$, and then choose one full block $B$ such that $B_k\Subset B$ for every $k$.
	Representability witnesses may be taken with this common bounding block: decompose $B\setminus\inter{B_k}$ into finitely many full blocks and append these blocks to the exterior multiblock of the witness for $X_k$.
	Fix $\eps>0$.
	For each $k$, take a witness
	\[
	B,\qquad \mathcal E_k,\qquad \mathcal M_k,\qquad \mathcal M_k^c
	\]
	for $X_k$ with $\gamma(\mathcal E_k)\le\eps/(2N)$.
	Let
	\[
	\mathcal A
	:=
	\bigsqcup_{k=1}^N
	(\mathcal E_k\sqcup\mathcal M_k\sqcup\mathcal M_k^c).
	\]
	Use the notation of \Cref{dfn_coordinate_refinement}, and write $T_l:=T_l(B,\mathcal A)$.
	Choose a rational $\rho>0$ smaller than half the distance between any two distinct elements of any $T_l$ and so small that the finite full multiblock $\mathcal H$ consisting of the blocks
	\[
		B\cap\{x\in\R^n:\abs{x_l-t}\le\rho\},
		\qquad
		l=1,\ldots,n,\quad t\in T_l,
	\]
	after omitting empty and degenerate blocks, satisfies
	\[
		\gamma(\mathcal H)\le\eps/2.
	\]
	Equivalently, $\mathcal H=\mathcal H_\rho(B;\mathcal A)$.
	Let $\mathcal C_\rho:=\mathcal R_\rho(B;\mathcal A)$ be the corresponding trimmed coordinate refinement.
	By \Cref{lem_coordinate_refinement_measure_containment},
	\[
		\overline{\mathcal H\sqcup\mathcal C_\rho}
		\msubset
		B.
	\]
	Moreover, for every $C\in\mathcal C_\rho$ and every constituent block $A$ of $\mathcal A$, either $C\Subset A$ or $C\cap A=\varnothing$.
	The multiblock $\mathcal H$ contains a positive neighborhood of every coordinate boundary used by the old witnesses; hence a block $C\in\mathcal C_\rho$ cannot cross from one constituent block of a witness to another.
	Since the concatenated witness support for $X_k$ is measure-contained in $B$, a positive-volume block $C\in\mathcal C_\rho$ cannot be disjoint from all three supports.
	The coordinate refinement property and the disjointness of interiors imply that at least one and at most one of
	\[
	C\Subset\bar{\mathcal M}_k,\qquad
	C\Subset\bar{\mathcal E}_k,\qquad
	C\Subset\overline{\mathcal M_k^c}
	\]
	holds for each pair $(C,k)$.
	Set
	\[
	\mathcal C_M
	:=
	\{C\in\mathcal C_\rho:\exists k,\ C\Subset\bar{\mathcal M}_k\},
	\]
	\[
	\mathcal C_c
	:=
	\{C\in\mathcal C_\rho:\forall k,\ C\Subset\overline{\mathcal M_k^c}\},
	\]
	and
	\[
	\mathcal C_E
	:=
	\mathcal C_\rho\setminus(\mathcal C_M\cup\mathcal C_c).
	\]
	Let $\mathcal M$ and $\mathcal M^c$ be the finite full multiblocks with constituent blocks $\mathcal C_M$ and $\mathcal C_c$, respectively.
	Let $\mathcal E$ be obtained by concatenating $\mathcal H$ with the finite full multiblock whose constituent blocks are $\mathcal C_E$.
	Thus $\mathcal M,\mathcal E,\mathcal M^c$ are finite full multiblocks.
	Their interiors are pairwise disjoint and
	\[
		\overline{\mathcal M\sqcup\mathcal E\sqcup\mathcal M^c}
		\msubset
		B,
	\]
	as required by \Cref{eqn_representable_set_decomposition}.
	Every block in $\mathcal C_E$ is well contained in $\bar{\mathcal E}_k$ for at least one $k$.
	Choose one such index $k(C)$ for each $C\in\mathcal C_E$.
	Because the blocks in $\mathcal C_\rho$ have pairwise disjoint interiors,
	\[
	\sum_{\substack{C\in\mathcal C_E\\ k(C)=k}}\mu(C)
	\le
	\gamma(\mathcal E_k),
	\qquad
	k=1,\ldots,N.
	\]
	Consequently,
	\[
	\gamma(\mathcal E)
	\le
	\gamma(\mathcal H)+\sum_{k=1}^N\gamma(\mathcal E_k)
	\le
	\eps.
	\]
	If $C$ is a constituent block of $\mathcal M$, then $C\Subset\bar{\mathcal M}_k$ for some $k$, hence
	\[
	C\Subset X_k\subset X.
	\]
	If $C$ is a constituent block of $\mathcal M^c$, then $C\Subset\overline{\mathcal M_k^c}$ for every $k$, hence
	\[
	C\Subset B\setminus X_k
	\qquad
	k=1,\ldots,N.
	\]
	Since the multiblock is finite,
	\[
	C
	\Subset
	\bigcap_{k=1}^N(B\setminus X_k).
	\]
	Indeed, if $\delta_k>0$ witnesses $C\Subset B\setminus X_k$, then $\min_k\delta_k>0$ witnesses apartness from $\bigcup_k X_k$.
	Thus
	\[
	C
	\Subset
	B\setminus X,
	\qquad
	B\setminus X
	=
	\bigcap_{k=1}^N(B\setminus X_k).
	\]
	Thus the constructed multiblocks form a representability witness for $X$.
\end{proof}

\begin{figure}[ht]
	\centering
	\resizebox{0.75\linewidth}{!}{\begin{tikzpicture}[x=1cm,y=1cm,font=\small]
	\fill[gray!10] (0,0) rectangle (8,4);
	\draw[very thick] (0,0) rectangle (8,4);
	\node[anchor=north east] at (7.9,3.9) {$B$};
	
	\fill[orange!25] (1.35,0) rectangle (1.65,4);
	\fill[orange!25] (3.85,0) rectangle (4.15,4);
	\fill[orange!25] (6.25,0) rectangle (6.55,4);
	\fill[orange!25] (0,1.25) rectangle (8,1.55);
	\fill[orange!25] (0,2.65) rectangle (8,2.95);
	
	\foreach \x in {1.5,4.0,6.4} {
		\draw[dashed,orange!80!black] (\x,0) -- (\x,4);
	}
	\foreach \y in {1.4,2.8} {
		\draw[dashed,orange!80!black] (0,\y) -- (8,\y);
	}
	
	\fill[blue!25] (0.15,0.15) rectangle (1.35,1.25);
	\fill[blue!25] (1.65,1.55) rectangle (3.85,2.65);
	\fill[blue!25] (4.15,2.95) rectangle (6.25,3.85);
	\fill[orange!25] (6.55,1.55) rectangle (7.85,2.65);
	
	\fill[gray!28] (1.65,0.15) rectangle (3.85,1.25);
	\fill[gray!28] (4.15,0.15) rectangle (6.25,1.25);
	\fill[gray!28] (0.15,2.95) rectangle (1.35,3.85);
	\fill[gray!28] (6.55,2.95) rectangle (7.85,3.85);
	\fill[gray!28] (0.15,1.55) rectangle (1.35,2.65);
	\fill[gray!28] (1.65,2.95) rectangle (3.85,3.85);
	\fill[gray!28] (4.15,1.55) rectangle (6.25,2.65);
	\fill[gray!28] (6.55,0.15) rectangle (7.85,1.25);
	
	\foreach \x/\w in {0/1.5,1.5/2.5,4.0/2.4,6.4/1.6} {
		\draw[thin,gray!55] (\x,0) -- (\x,4);
	}
	\foreach \y in {1.4,2.8} {
		\draw[thin,gray!55] (0,\y) -- (8,\y);
	}
	
	\node[blue!70!black] at (2.75,2.1) {$\mathcal C_M$};
	\node[gray!65!black] at (5.15,1.0) {$\mathcal C_c$};
	\node[orange!80!black] at (4.0,3.35) {$\mathcal H$};
	
	\node[align=left,font=\scriptsize] at (0.75,0.72) {core for\\some $X_k$};
	\node[align=center,font=\scriptsize] at (5.15,0.46) {exterior for\\all $X_k$};
	\node[align=center,font=\scriptsize,orange!80!black] at (7.25,2.0) {$\mathcal C_E$};
\end{tikzpicture}}
	\caption{Subdivision used in \Cref{lem_representable_finite_union}. Inside the bounding block $B$, the thin orange coordinate strips form $\mathcal H$ and contain neighborhoods of the old coordinate boundaries. Outside $\mathcal H$, the full blocks form $\mathcal C_\rho$: blue blocks belong to $\mathcal C_M$, gray blocks belong to $\mathcal C_c$, and the orange block labelled $\mathcal C_E$ is one whole block of $\mathcal C_\rho$ left after these tests. Both $\mathcal H$ and the blocks in $\mathcal C_E$ are appended to the exception multiblock.}
	\label{fig_lem31_union_subdivision}
\end{figure}
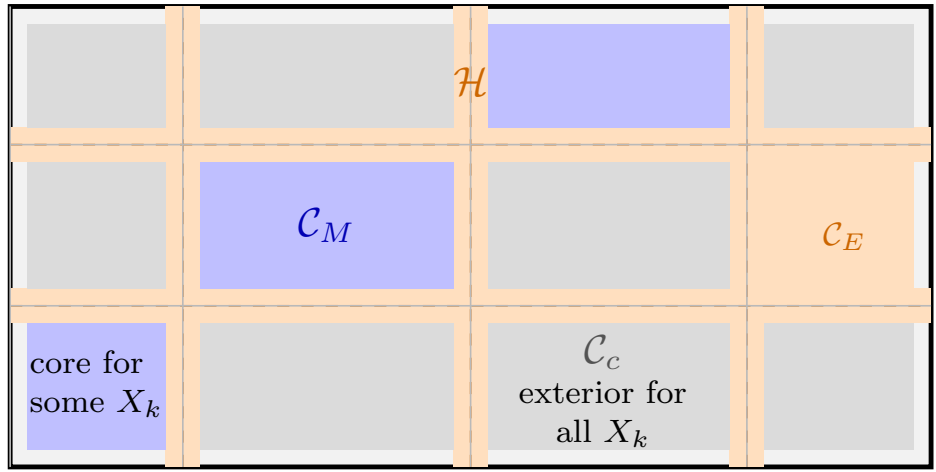

\begin{definition}[Locally finite union]
	\label{dfn_effectively_locally_finite_union}
	A sequence of sets $(X_i)_{i\in\N}$ is called \emph{effectively locally finite} if, for every block $B$, one can compute a finite set of indices $I(B)\subset\N$ such that
	\[
	B\cap\bigcup_{i=1}^{\infty}X_i
	=
	B\cap\bigcup_{i\in I(B)}X_i.
	\]
\end{definition}

\begin{corollary}[Locally finite unions]
	\label{cor_locally_finite_union_representable}
	If $(X_i)_{i\in\N}$ is an effectively locally finite sequence of locally representable sets, then
	\[
	X:=\bigcup_{i=1}^{\infty}X_i
	\]
	is locally representable.
\end{corollary}

\begin{proof}
	Fix a block $B$.
	By effective local finiteness, $B\cap X$ is the finite union of the sets $B\cap X_i$ with $i\in I(B)$.
	Each of these sets is representable by local representability of $X_i$.
	The claim follows from \Cref{lem_representable_finite_union}.
\end{proof}

\begin{lemma}[Products of representable sets]
	\label{lem_representable_product}
	If $X\subset\R^n$ and $Y\subset\R^m$ are representable, then $X\times Y\subset\R^{n+m}$ is representable.
	Moreover,
	\[
		\mu(X\times Y)=\mu(X)\mu(Y).
	\]
\end{lemma}

\begin{proof}
	Choose bounding blocks $B_X$ and $B_Y$.
	Fix $\eps>0$.
	Take witnesses
	\[
		B_X,\quad \mathcal E_X,\quad \mathcal M_X,\quad \mathcal M_X^c,
		\qquad
		B_Y,\quad \mathcal E_Y,\quad \mathcal M_Y,\quad \mathcal M_Y^c
	\]
	with exception costs so small that
	\[
		\gamma(\mathcal E_X)\mu(B_Y)
		+
		\gamma(\mathcal E_Y)\mu(B_X)
		\le\eps.
	\]
	Before forming products, replace each witness by the common coordinate refinement of all constituent endpoints in that witness.
	Equivalently, each of
	\[
		\mathcal M_X,\mathcal E_X,\mathcal M_X^c,
		\qquad
		\mathcal M_Y,\mathcal E_Y,\mathcal M_Y^c
	\]
	is represented by the finite list of refinement cells contained in its support.
	This normalization does not change the three supports, preserves the core and exterior containment conditions, and gives pairwise disjoint interiors for all cells in each witness decomposition.
	In particular, the total cost of the three normalized multiblocks for $X$ is at most $\mu(B_X)$, and similarly for $Y$.
	Products of constituent blocks form finite full multiblocks in $\R^{n+m}$.
	Set
	\[
		\mathcal M_{X\times Y}:=\mathcal M_X\times\mathcal M_Y.
	\]
	Let $\mathcal E_{X\times Y}$ be the concatenation of all products in which at least one factor is an exception multiblock and the other factor is one of the three witness multiblocks for the other set.
	Let $\mathcal M_{X\times Y}^c$ be the concatenation of the remaining products in which at least one factor is an exterior multiblock.
	The interiors of these three product multiblocks are pairwise disjoint, and their concatenation is measure-contained in $B_X\times B_Y$ by applying the one-dimensional endpoint refinement to the two original witness decompositions.
	The core condition follows from
	\[
		\bar{\mathcal M}_X\times\bar{\mathcal M}_Y\Subset X\times Y.
	\]
	If a product cell has an exterior $X$-factor, it is well contained in $(B_X\setminus X)\times B_Y$ and hence in $(B_X\times B_Y)\setminus(X\times Y)$.
	The same argument applies to an exterior $Y$-factor.
	The exception cost is bounded by
	\[
		\gamma(\mathcal E_{X\times Y})
		\le
		\gamma(\mathcal E_X)\mu(B_Y)
		+
		\gamma(\mathcal E_Y)\mu(B_X)
		\le\eps.
	\]
	Thus $X\times Y$ is representable.
	For the measure identity, use the lower approximations
	\[
		\mu(\bar{\mathcal M}_{X,\delta})\mu(\bar{\mathcal M}_{Y,\delta})
		=
		\mu(\bar{\mathcal M}_{X,\delta}\times\bar{\mathcal M}_{Y,\delta})
	\]
	and let $\delta\to0$.
	The corresponding upper approximations differ from these lower approximations by the product-exception estimate above, so they determine the same real number.
\end{proof}

\begin{lemma}[Relative apartness complements]
	\label{lem_representable_relative_complement}
	Let $X$ be representable with bounding block $B_X$.
	Then the relative apartness complement $B_X\setminus X$ has a representability witness relative to the same bounding block.
	After enlarging $B_X$ by an arbitrarily thin rational shell, it is representable in the sense of \Cref{dfn_representable_set}.
\end{lemma}

\begin{proof}
	For a witness
	\[
	B_X,\qquad \mathcal E_X,\qquad \mathcal M_X,\qquad \mathcal M_X^c
	\]
	of $X$, use
	\[
	\mathcal M_{B_X\setminus X}:=\mathcal M_X^c,
	\qquad
	\mathcal E_{B_X\setminus X}:=\mathcal E_X,
	\qquad
	\mathcal M^c_{B_X\setminus X}:=\mathcal M_X.
	\]
	The exception cost, the disjointness condition, and the measure-containment condition are unchanged.
	The two well-containment conditions are exactly \Cref{eqn_representable_set_interior,eqn_representable_set_exterior} with the core and exterior interchanged.
	For the formal boundedness condition, choose a full block $B'\supset B_X$ with $B_X\Subset B'$.
	A rational shell around $\partial B_X$ of arbitrarily small cost is appended to the exception multiblock; the remaining part of $B'\setminus B_X$ is appended to the exterior multiblock.
\end{proof}

\begin{lemma}[Finite intersections of representable sets]
	\label{lem_representable_finite_intersection}
	Let $X_1,\ldots,X_N\subset\R^n$ be representable.
	Then
	\[
	X:=\bigcap_{k=1}^N X_k
	\]
	is representable.
\end{lemma}

\begin{proof}
	Choose bounding blocks $B_k$ for $X_k$, and then choose one full block $B$ such that $B_k\Subset B$ for every $k$.
	As in \Cref{lem_representable_finite_union}, all witnesses may be taken relative to this common bounding block.
	Fix $\eps>0$.
	For each $k$, take a witness
	\[
	B,\qquad \mathcal E_k,\qquad \mathcal M_k,\qquad \mathcal M_k^c
	\]
	with $\gamma(\mathcal E_k)\le\eps/(2N)$.
	Let
	\[
	\mathcal A
	:=
	\bigsqcup_{k=1}^N
	(\mathcal E_k\sqcup\mathcal M_k\sqcup\mathcal M_k^c).
	\]
	Choose $\rho>0$ as in \Cref{dfn_coordinate_refinement}, small enough that
	\[
	\gamma(\mathcal H_\rho(B;\mathcal A))\le\eps/2.
	\]
	Set
	\[
	\mathcal H:=\mathcal H_\rho(B;\mathcal A),
	\qquad
	\mathcal C_\rho:=\mathcal R_\rho(B;\mathcal A).
	\]
	For every $C\in\mathcal C_\rho$ and every $k$, the same finite rational endpoint test as in \Cref{lem_representable_finite_union} puts $C$ in exactly one of the three relations
	\[
	C\Subset\bar{\mathcal M}_k,\qquad
	C\Subset\bar{\mathcal E}_k,\qquad
	C\Subset\overline{\mathcal M_k^c}.
	\]
	Define
	\[
	\mathcal C_M
	:=
	\{C\in\mathcal C_\rho:\forall k,\ C\Subset\bar{\mathcal M}_k\},
	\]
	\[
	\mathcal C_c
	:=
	\{C\in\mathcal C_\rho:\exists k,\ C\Subset\overline{\mathcal M_k^c}\},
	\]
	and
	\[
	\mathcal C_E
	:=
	\mathcal C_\rho\setminus(\mathcal C_M\cup\mathcal C_c).
	\]
	Let $\mathcal M$ and $\mathcal M^c$ be the finite full multiblocks with constituent blocks $\mathcal C_M$ and $\mathcal C_c$, respectively.
	Let $\mathcal E$ be obtained by concatenating $\mathcal H$ with the finite full multiblock whose constituent blocks are $\mathcal C_E$.
	By \Cref{lem_coordinate_refinement_measure_containment},
	\[
		\overline{\mathcal M\sqcup\mathcal E\sqcup\mathcal M^c}
		\msubset
		B.
	\]
	Every block in $\mathcal C_E$ is well contained in $\bar{\mathcal E}_k$ for at least one $k$.
	Indeed, such a block is not exterior for any $k$ and not core for all $k$, while it has exactly one of the three relations above for each fixed $k$.
	Choose one such index $k(C)$ for each $C\in\mathcal C_E$.
	Because the blocks in $\mathcal C_\rho$ have pairwise disjoint interiors,
	\[
	\sum_{\substack{C\in\mathcal C_E\\ k(C)=k}}\mu(C)
	\le
	\gamma(\mathcal E_k),
	\qquad
	k=1,\ldots,N.
	\]
	Therefore
	\[
	\gamma(\mathcal E)
	\le
	\gamma(\mathcal H)+\sum_{k=1}^N\gamma(\mathcal E_k)
	\le
	\eps.
	\]
	If $C$ is a constituent block of $\mathcal M$, then $C\Subset\bar{\mathcal M}_k$ for every $k$.
	Hence $C\Subset X_k$ for every $k$, and finiteness gives
	\[
	C\Subset\bigcap_{k=1}^N X_k.
	\]
	If $C$ is a constituent block of $\mathcal M^c$, then $C\Subset\overline{\mathcal M_k^c}$ for some $k$, hence
	\[
		C\Subset B\setminus X_k\subset B\setminus\bigcap_{l=1}^N X_l.
	\]
	Thus the constructed multiblocks form a representability witness for $X$.
	The core multiblock may be empty; this is the intended degenerate case and requires no further decision.
\end{proof}

\begin{remark}[Mechanical intersections]
	\label{rem_representable_intersection_mechanical}
	\Cref{lem_representable_finite_intersection} is a finite rational block construction.
	It does not decide whether the intersection contains a full block.
	Empty and degenerate cells are simply absent from the core and are absorbed by the exterior or the exception layer.
\end{remark}

\begin{definition}[Flattening]
	\label{dfn_representable_set_flattening}
	Let $B$ be a full block and let $X_1,\ldots,X_N\Subset B$ be representable sets.
	For each subset $I\subseteq\{1,\ldots,N\}$ set
	\[
		Y_I
		:=
		\left(\bigcap_{i\in I}X_i\right)
		\cap
		\left(\bigcap_{i\notin I}(B\setminus X_i)\right),
	\]
	where $B\setminus X_i$ is the relative apartness complement inside $B$.
	The finite family $\{Y_I\}_I$ is called the \emph{flattening} of $X_1,\ldots,X_N$ over $B$.
	The full finite family is retained as refinement data.
	No cell is removed because it might be empty or degenerate.
\end{definition}

\begin{lemma}[Flattening]
	\label{lem_representable_set_flattening}
	The flattening of a finite family of representable sets well contained in a full block consists of pairwise disjoint representable sets, possibly with empty cores.
	The construction is finite and mechanical.
\end{lemma}

\begin{proof}
	For each support $X_i\Subset B$, the relative apartness complement $B\setminus X_i$ is representable relative to $B$ by \Cref{lem_representable_relative_complement}.
	After enlarging the ambient block by an arbitrarily thin shell if the formal boundedness clause is needed, it is representable in the sense of \Cref{dfn_representable_set}.
	Finite intersections are representable by \Cref{lem_representable_finite_intersection}, so every $Y_I$ is representable.
	If $I\ne J$, choose $i$ belonging to exactly one of $I,J$.
	Then $Y_I$ is contained in $X_i$ and $Y_J$ is contained in $B\setminus X_i$, or conversely.
	Hence $Y_I\cap Y_J=\emptyset$.
	The core of some $Y_I$ may be empty; no emptiness test is performed.
\end{proof}

\begin{remark}
	\label{rem_representable_generation}
	Thus complicated representable sets may be generated from simple pieces.
	For instance, finitely many obstacles, locally finite obstacle arrays, polyhedral corners, and piecewise graph boundaries need not be treated as separate primitive cases once their pieces are representable.
\end{remark}

\subsubsection{Representable Set Examples and Counter-Examples}
\label{subsubsec_representable_set_examples}

\begin{definition}[Effective Lipschitz graph set]
	\label{dfn_effective_lipschitz_graph_set}
	Let $A\subset\R^{n-1}$ be a full block and let $c<d$ be rational.
	A set $X\subset A\times[c,d]$ is called an \emph{effective Lipschitz graph set} if, after possibly interchanging coordinates,
	\[
	X=\{(u,t)\in A\times[c,d]:t\le\varphi(u)\}
	\]
	for a function $\varphi:A\to[c,d]$ equipped with a rational Lipschitz constant $L>0$, meaning
	\[
	\abs{\varphi(u)-\varphi(v)}
	\le
	L\nrm{u-v}
	\qquad
	(u,v\in A).
	\]
\end{definition}

\begin{remark}
	The function-theoretic terminology, including Lipschitz continuity, used here is developed in \Cref{subsec_continuous_functions} in detail; at this point it is only used as geometric data describing the boundary graph.
\end{remark}

\begin{figure}[ht]
	\centering
	\resizebox{0.78\linewidth}{!}{\begin{tikzpicture}[x=1cm,y=1cm,font=\small]
	\fill[gray!12] (0,0) rectangle (6,4);

	\fill[blue!18]
		(0,0) -- (6,0)
		-- (6,2.95)
		.. controls (5.1,2.7) and (4.6,2.0) .. (3.8,2.1)
		.. controls (2.8,2.2) and (2.5,3.0) .. (1.5,2.8)
		.. controls (0.8,2.65) and (0.35,2.1) .. (0,2.2)
		-- cycle;

	\draw[very thick] (0,0) rectangle (6,4);
	\draw[very thick,blue!70!black]
		(0,2.2)
		.. controls (0.35,2.1) and (0.8,2.65) .. (1.5,2.8)
		.. controls (2.5,3.0) and (2.8,2.2) .. (3.8,2.1)
		.. controls (4.6,2.0) and (5.1,2.7) .. (6,2.95);

	\draw[->] (-0.25,0) -- (6.35,0);
	\draw[->] (0,-0.25) -- (0,4.35);
	\node[below] at (3,0) {$u\in A$};
	\node[left] at (0,0.15) {$c$};
	\node[left] at (0,4) {$d$};
	\node[right,blue!70!black] at (6,2.95) {$t=\varphi(u)$};
	\node[blue!70!black] at (3,1.15) {$X=\{(u,t):t\le\varphi(u)\}$};
	\node[gray!65!black] at (3,3.55) {relative exterior};
\end{tikzpicture}}
	\caption{Effective Lipschitz graph set from \Cref{dfn_effective_lipschitz_graph_set}. The set is the part of the chart below the graph of the computable Lipschitz function.}
	\label{fig_dfn314_lipschitz_graph_set}
\end{figure}
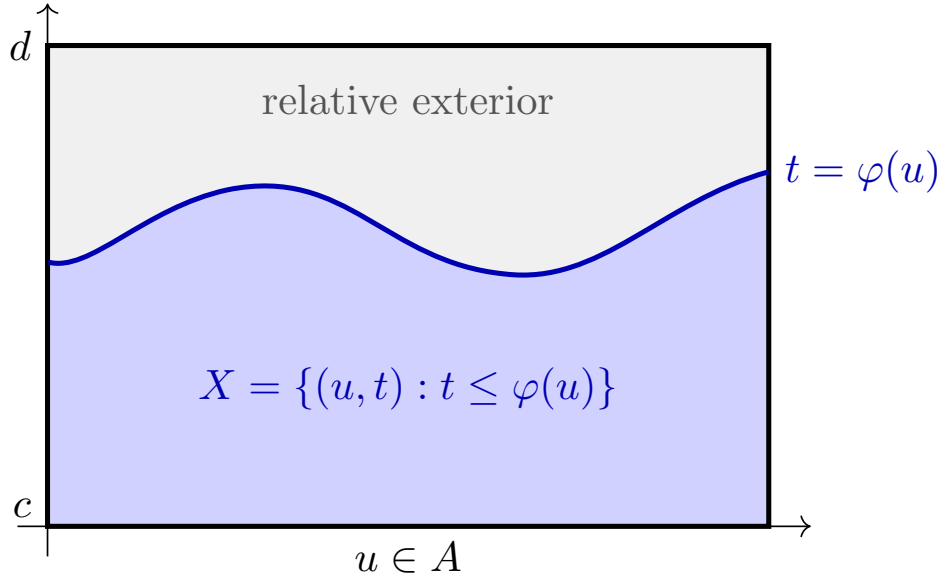

\begin{lemma}[Lipschitz graph sets are representable]
	\label{lem_lipschitz_graph_sets_representable}
	Every effective Lipschitz graph set is representable.
\end{lemma}

\begin{proof}
	Let $X=\{(u,t)\in A\times[c,d]:t\le\varphi(u)\}$ and let $B_X$ be a slightly larger block containing $A\times[c,d]$ well inside.
	Fix $\eps>0$.
	Choose $\delta_0>0$ such that
	\[
	A\times[c,d]\oplus\delta_0\Subset B_X.
	\]
	Choose $\lambda>0$ with $\lambda<\delta_0/4$ so small that the collar multiblocks introduced below have total cost below $\eps/2$ and
	\[
		4\lambda\mu(A)<\eps/2.
	\]
	Write
	\[
		A=\prod_{r=1}^{n-1}[a_r,b_r]
	\]
	and put
	\[
		A^\lambda
		:=
		\prod_{r=1}^{n-1}[a_r+\lambda,b_r-\lambda],
	\]
	omitting this block if it is empty or degenerate.
	The set
	\[
		(A\times[c,d])\setminus(A^\lambda\times[c+\lambda,d])
	\]
	is covered by a finite rational collar multiblock $\mathcal T_\lambda$ with arbitrarily small cost as $\lambda\downarrow0$.
	Subdivide $A$ into finitely many full blocks $A_k$ with pairwise disjoint interiors and with diameter at most $\lambda/(2L)$.
	Let
	\[
		A_k^\lambda:=A_k\cap A^\lambda,
	\]
	again omitting empty or degenerate blocks.
	For each $A_k$, compute rational numbers $\alpha_k<\beta_k$ such that
	\[
	\varphi(A_k)\subset[\alpha_k,\beta_k],
	\qquad
	[c,d]\cap[\alpha_k,\beta_k]\neq\varnothing,
	\qquad
	\beta_k-\alpha_k\le 2\lambda.
	\]
	This is possible because the Lipschitz estimate bounds the variation on $A_k$ by $\lambda/2$, while $\varphi$ can be approximated at a rational tag of $A_k$ within $\lambda/2$.
	Since $\varphi(A_k)\subset[c,d]$, the numbers $\alpha_k,\beta_k$ may be chosen so that
	\[
	[\alpha_k-\lambda,\beta_k+\lambda]\subset[c-2\lambda,d+2\lambda]\subset[c-\delta_0,d+\delta_0].
	\]
	Define
	\[
	\begin{aligned}
		C_k&:=A_k^\lambda\times[c+\lambda,\alpha_k-\lambda],\\
		E_k&:=A_k^\lambda\times[\alpha_k-\lambda,\beta_k+\lambda],\\
		D_k&:=A_k^\lambda\times[\beta_k+\lambda,d].
	\end{aligned}
	\]
	After omitting empty or degenerate blocks, the $C_k$ form a core multiblock, the $E_k$ form the graph exception multiblock, and the $D_k$ form the exterior multiblock over $A^\lambda\times[c+\lambda,d]$.
	The collar multiblock $\mathcal T_\lambda$ is appended to the exception multiblock.
	For each $k$,
	\[
	A_k^\lambda\times[c+\lambda,d]\subseteq C_k\cup E_k\cup D_k
	\]
	as an interval-cover statement in the $t$-coordinate.
	Together with $\mathcal T_\lambda$, these multiblocks cover $A\times[c,d]$ without requiring a decidable comparison of an arbitrary real $t$ with the two rational levels.
	Their interiors are pairwise disjoint because the interiors of the $A_k$ are pairwise disjoint and, for fixed $k$, the three vertical intervals have pairwise disjoint interiors.
	The part of $B_X$ outside $A\times[c,d]$ is decomposed into finitely many full blocks.
	Choose this decomposition so that the multiblock
	\[
	\mathcal S_\lambda
	:=
	\{S:S\subset B_X\setminus(A\times[c,d]),\ d(S,A\times[c,d])<\lambda\}
	\]
	satisfies
	\[
		\gamma(\mathcal T_\lambda)+\gamma(\mathcal S_\lambda)<\eps/2.
	\]
	The blocks of $B_X\setminus(A\times[c,d])$ at distance at least $\lambda$ from $A\times[c,d]$ are appended to the exterior multiblock, and the blocks in $\mathcal S_\lambda$ are appended to the exception multiblock.
	Let the resulting multiblocks be denoted by
	\[
	\mathcal M_X,\qquad \mathcal E_X,\qquad \mathcal M_X^c.
	\]
	The interval-cover statement on $A\times[c,d]$, together with the finite decomposition of $B_X\setminus(A\times[c,d])$, gives
	\[
		\overline{\mathcal M_X\sqcup\mathcal E_X\sqcup\mathcal M_X^c}
		\msubset
		B_X.
	\]
	For each $k$,
	\[
	C_k\Subset X,
	\qquad
	D_k\Subset B_X\setminus X,
	\]
	Indeed, $C_k$ has vertical distance at least $\lambda$ from the graph, distance at least $\lambda$ from the lower face $t=c$, and base distance at least $\lambda$ from $\partial A$.
	For $D_k$, if $(u,t)\in D_k$ and $(v,\tau)\in X$, then either $\nrm{u-v}\ge\lambda/(2L)$, or the Lipschitz estimate gives
	\[
		t-\tau
		\ge
		\lambda-L\nrm{u-v}
		\ge
		\lambda/2.
	\]
	Thus $D_k$ is positively separated from $X$.
	Moreover,
	\[
	\sum_k\mu(E_k)
	\le
	4\lambda\sum_k\mu(A_k)
	=
	4\lambda\mu(A)
	<
	\eps/2.
	\]
	Hence
	\[
	\gamma(\mathcal E_X)
	\le
	\sum_k\mu(E_k)+\gamma(\mathcal T_\lambda)+\gamma(\mathcal S_\lambda)
	<
	\eps.
	\]
	The tuple
	\[
	B_X,\qquad \mathcal E_X,\qquad \mathcal M_X,\qquad \mathcal M_X^c
	\]
	is a representability witness for $X$.
\end{proof}

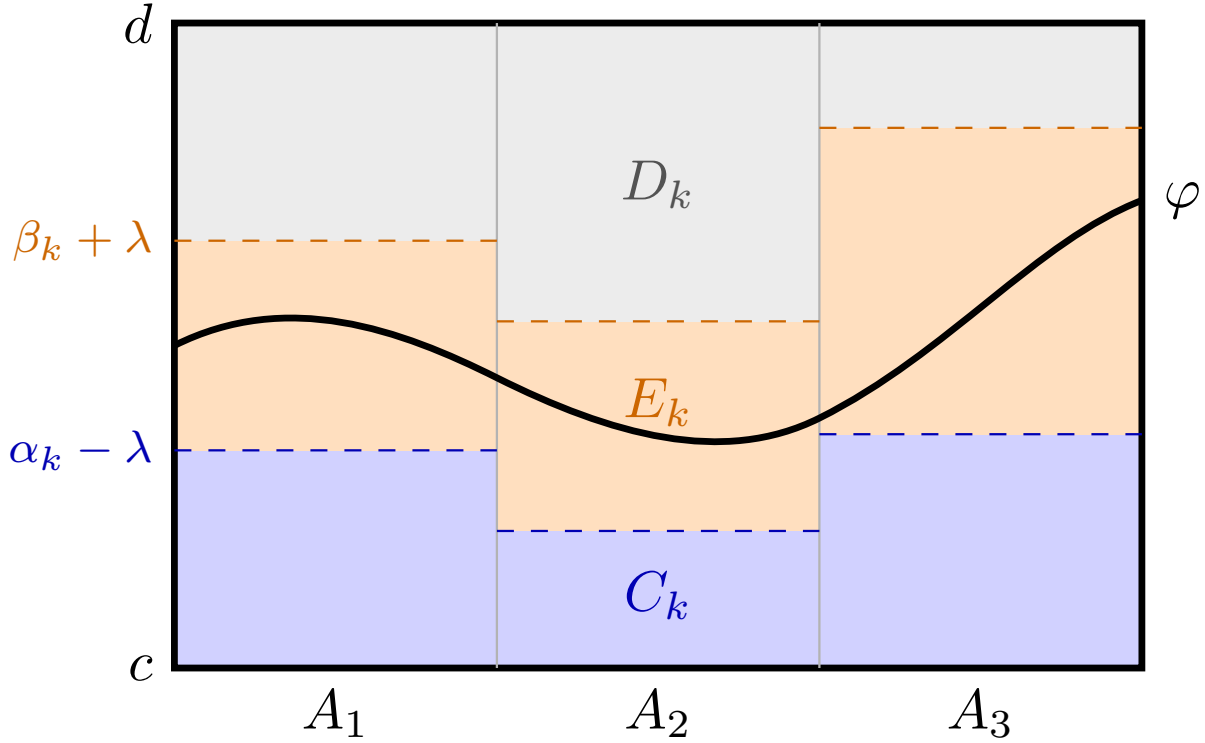
\begin{figure}[ht]
	\centering
	\resizebox{\linewidth}{!}{\begin{tikzpicture}[x=1cm,y=1cm,font=\small]
	\fill[blue!18] (0,0) rectangle (2,1.35);
	\fill[orange!25] (0,1.35) rectangle (2,2.65);
	\fill[gray!15] (0,2.65) rectangle (2,4);

	\fill[blue!18] (2,0) rectangle (4,0.85);
	\fill[orange!25] (2,0.85) rectangle (4,2.15);
	\fill[gray!15] (2,2.15) rectangle (4,4);

	\fill[blue!18] (4,0) rectangle (6,1.45);
	\fill[orange!25] (4,1.45) rectangle (6,3.35);
	\fill[gray!15] (4,3.35) rectangle (6,4);

	\draw[very thick] (0,0) rectangle (6,4);
	\draw[thin,gray!60] (2,0) -- (2,4);
	\draw[thin,gray!60] (4,0) -- (4,4);

	\foreach \xa/\xb/\ya/\yb in {0/2/1.35/2.65,2/4/0.85/2.15,4/6/1.45/3.35} {
		\draw[dashed,blue!70!black] (\xa,\ya) -- (\xb,\ya);
		\draw[dashed,orange!80!black] (\xa,\yb) -- (\xb,\yb);
	}

	\draw[very thick]
		(0,2.0) .. controls (0.7,2.35) and (1.4,2.1) .. (2,1.8)
		.. controls (2.7,1.45) and (3.4,1.25) .. (4,1.55)
		.. controls (4.8,1.95) and (5.35,2.65) .. (6,2.9);
	\node[right] at (6,2.9) {$\varphi$};

	\node[blue!70!black] at (3,0.45) {$C_k$};
	\node[orange!80!black] at (3,1.65) {$E_k$};
	\node[gray!65!black] at (3,3.0) {$D_k$};

	\node[below] at (1,0) {$A_1$};
	\node[below] at (3,0) {$A_2$};
	\node[below] at (5,0) {$A_3$};
	\node[left] at (0,0) {$c$};
	\node[left] at (0,4) {$d$};
	\node[font=\scriptsize,blue!70!black,left] at (0,1.35) {$\alpha_k-\lambda$};
	\node[font=\scriptsize,orange!80!black,left] at (0,2.65) {$\beta_k+\lambda$};
\end{tikzpicture}}
	\caption{Witness construction in \Cref{lem_lipschitz_graph_sets_representable}. Each base block contributes a core column, an exception column containing the graph of $\varphi$, and an exterior column.}
	\label{fig_lem34_lipschitz_witness}
\end{figure}

\begin{example}[Rational polytopes]
	\label{ex_representable_sets_positive}
	Let $B_X$ be a full block.
	Let
	\[
	X=\bigcup_{r=1}^R P_r,
	\qquad
	P_r=\{x\in B_X:a_{rq}\cdot x\le b_{rq},\ q=1,\ldots,Q_r\},
	\]
	with rational vectors $a_{rq}$ and rational numbers $b_{rq}$.
	Here, $a_{rq}\cdot x$ means the dot product of vectors.
	Assume that no displayed inequality has $a_{rq}=0$.
	Indeed, if $a_{rq}=0$ and $b_{rq}\ge0$, the inequality is vacuous, while if $a_{rq}=0$ and $b_{rq}<0$, the corresponding polytope is empty and is omitted.
	Assume also that $X\Subset B_X$.
	We show that $X$ is representable.
	Fix $\eps>0$.
	For rational $\lambda>0$, define the inner and outer polyhedral approximations
	\[
	\begin{aligned}
		P_r^-&:=\{x\in B_X:a_{rq}\cdot x\le b_{rq}-\lambda,\ q=1,\ldots,Q_r\},\\
		P_r^+&:=\{x\in B_X:a_{rq}\cdot x\le b_{rq}+\lambda,\ q=1,\ldots,Q_r\}.
	\end{aligned}
	\]
	Then
	\[
	\bigcup_r P_r^-\Subset X,
	\qquad
	B_X\setminus\bigcup_r P_r^+\Subset B_X\setminus X.
	\]
	The transition set
	\[
	\left(\bigcup_r P_r^+\right)\setminus\left(\bigcup_r P_r^-\right)
	\]
	is contained in the finite union of rational slabs
	\[
	\{x\in B_X:\abs{a_{rq}\cdot x-b_{rq}}\le\lambda\},
	\qquad r=1,\ldots,R,\quad q=1,\ldots,Q_r.
	\]
	Indeed, if
	\[
	x\in
	\left(\bigcup_r P_r^+\right)\setminus\left(\bigcup_r P_r^-\right),
	\]
	then $x\in P_r^+$ for some $r$, while $x\notin P_r^-$.
	Hence, for some $q$,
	\[
		b_{rq}-\lambda
		<
		a_{rq}\cdot x
		\le
		b_{rq}+\lambda,
	\]
	so $x$ belongs to the displayed slab for this pair $(r,q)$.
	Here a rational slab means the intersection of $B_X$ with the strip between two parallel rational hyperplanes.
	We do not need to regard the slab itself as a multiblock.
	For each nonzero rational vector $a$, there is a computable constant $C(B_X,a)>0$ such that, for every $\eta>0$, the slab
	\[
		\{x\in B_X:\abs{a\cdot x-b}\le\lambda\}
	\]
	has a finite rational block cover of cost at most $C(B_X,a)\lambda+\eta$.
	Indeed, choosing a coordinate $j$ with $a_j\ne0$ reduces the estimate to a one-dimensional interval of length $2\lambda/\abs{a_j}$ over the remaining coordinates.
	Hence $\lambda$ may be chosen so that the finite union of these slabs has a rational block cover $\mathcal E_0$ with
	\[
		\gamma(\mathcal E_0)<\eps/2.
	\]
	After $\mathcal E_0$ is fixed, subdivide $B_X$ into sufficiently small full blocks.
	Let $\mathcal P$ be the resulting finite full multiblock.
	Define
	\[
	\mathcal C_M
	:=
	\{C\in\mathcal P:C\Subset\bigcup_rP_r^-\},
	\]
	\[
	\mathcal C_c
	:=
	\{C\in\mathcal P:C\Subset B_X\setminus\bigcup_rP_r^+\},
	\]
	and
	\[
	\mathcal C_E
	:=
	\mathcal P\setminus(\mathcal C_M\cup\mathcal C_c).
	\]
	The subdivision is chosen so that the multiblock with constituent blocks $\mathcal C_E$ has cost below $\eps/2$ after the blocks already covered by $\mathcal E_0$ are omitted.
	Let
	\[
	\mathcal M_X:=\mathcal C_M,\qquad
	\mathcal M_X^c:=\mathcal C_c,\qquad
	\mathcal E_X:=\mathcal E_0\sqcup\mathcal C_E .
	\]
	Here $\mathcal C_M,\mathcal C_c,\mathcal C_E$ are regarded as finite multiblocks with the indicated constituent blocks.
	Then
	\[
	\bar{\mathcal M}_X\Subset X,
	\qquad
	\overline{\mathcal M_X^c}\Subset B_X\setminus X,
	\qquad
	\gamma(\mathcal E_X)<\eps.
	\]
	These three finite multiblocks form a representability witness.
	The same argument gives local representability for effectively locally finite unions of rational polytopes.
\end{example}

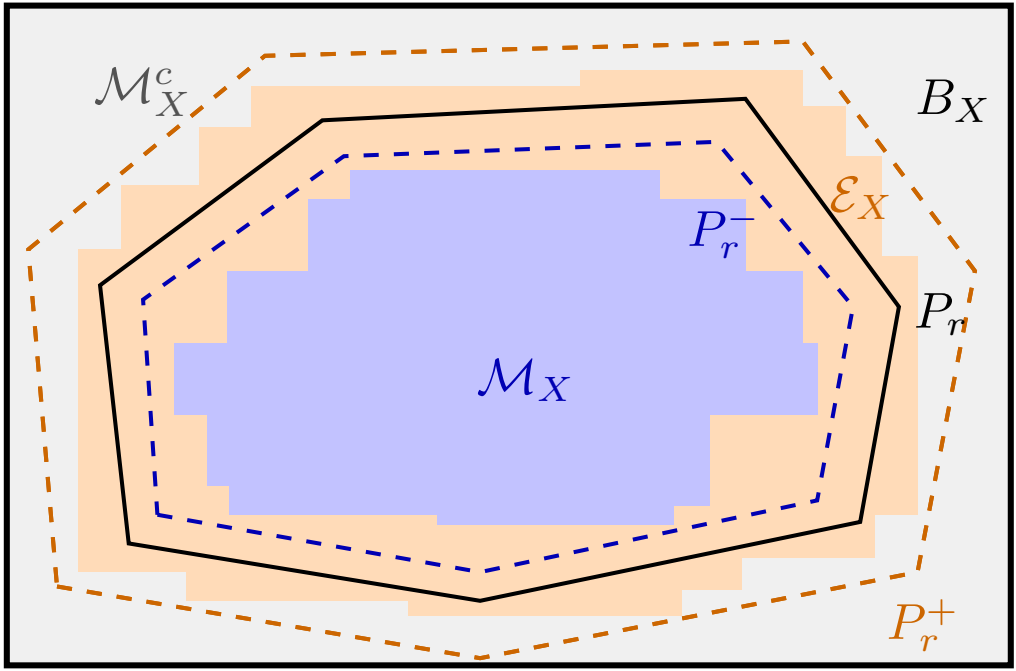
\begin{figure}[ht]
	\centering
	\resizebox{0.82\linewidth}{!}{\begin{tikzpicture}[x=1cm,y=1cm,font=\small]
	\fill[gray!12] (0,0) rectangle (7,4.6);
	\draw[very thick] (0,0) rectangle (7,4.6);

	\draw[dashed,thick,orange!80!black]
		(0.35,0.55) -- (3.3,0.05) -- (6.35,0.65)
		-- (6.75,2.75) -- (5.55,4.35) -- (1.8,4.25) -- (0.15,2.9) -- cycle;
	\draw[thick]
		(0.85,0.85) -- (3.3,0.45) -- (5.95,1.0)
		-- (6.22,2.5) -- (5.15,3.95) -- (2.2,3.8) -- (0.65,2.65) -- cycle;
	\draw[dashed,thick,blue!70!black]
		(1.05,1.05) -- (3.3,0.65) -- (5.65,1.15)
		-- (5.9,2.5) -- (4.95,3.65) -- (2.35,3.55) -- (0.95,2.55) -- cycle;

	\fill[orange!28] (0.50,0.65) rectangle (1.4,2.9);
	\fill[orange!28] (1.34,2.55) rectangle (2.40,3.75);
	\fill[orange!28] (1.40,2.17) rectangle (2.46,3.7);
	\fill[orange!28] (1.7,3.45) rectangle (4.0,4.04);
	\fill[orange!28] (4.0,3.45) rectangle (5.55,4.15);
	\fill[orange!28] (4.9,2.65) rectangle (5.85,3.9);
	\fill[orange!28] (5.55,1.05) rectangle (6.35,2.85);
	\fill[orange!28] (4.9,1.35) rectangle (5.55,2.65);
	\fill[orange!28] (4.65,0.75) rectangle (6.05,1.35);
	\fill[orange!28] (2.80,0.35) rectangle (4.71,1.25);
	\fill[orange!28] (1.25,0.45) rectangle (3.0,1.4);
	\fill[orange!28] (0.80,2.78) rectangle (1.57,3.35);		
	\fill[orange!28] (4.0,2.75) rectangle (5.55,3.45);

	\fill[blue!24] (1.4,1.25) rectangle (3.0,1.75);
	\fill[blue!24] (3.0,1.11) rectangle (4.9,1.75);
	\fill[blue!24] (1.55,1.05) rectangle (3.0,1.25);
	\fill[blue!24] (3.0,0.98) rectangle (4.65,1.25);
	\fill[blue!24] (1.17,1.75) rectangle (3.0,2.25);
	\fill[blue!24] (3.0,1.75) rectangle (5.65,2.25);
	\fill[blue!24] (1.54,2.25) rectangle (3.0,2.75);
	\fill[blue!24] (3.0,2.25) rectangle (5.55,2.75);
	\fill[blue!24] (2.1,2.75) rectangle (5.15,3.25);
	\fill[blue!24] (2.39,3.25) rectangle (4.55,3.45);

	\fill[orange!28] (0.75,0.75) rectangle (1.25,1.25);
	\fill[orange!28] (3.0,0.53) rectangle (5.12,0.90);
	\fill[orange!28] (4.9,0.75) rectangle (5.55,1.05);
	\fill[orange!28] (5.55,2.85) rectangle (6.10,3.55);

	\draw[thick]
		(0.85,0.85) -- (3.3,0.45) -- (5.95,1.0)
		-- (6.22,2.5) -- (5.15,3.95) -- (2.2,3.8) -- (0.65,2.65) -- cycle;
	\draw[dashed,thick,blue!70!black]
		(1.05,1.05) -- (3.3,0.65) -- (5.65,1.15)
		-- (5.9,2.5) -- (4.95,3.65) -- (2.35,3.55) -- (0.95,2.55) -- cycle;
	\draw[dashed,thick,orange!80!black]
		(0.35,0.55) -- (3.3,0.05) -- (6.35,0.65)
		-- (6.75,2.75) -- (5.55,4.35) -- (1.8,4.25) -- (0.15,2.9) -- cycle;

	\node[anchor=north west] at (6.2,4.22) {$B_X$};
	\node[anchor=north west,blue!70!black] at (3.14,2.27) {$\mathcal M_X$};
	\node[anchor=north west,orange!80!black] at (5.6,3.54) {$\mathcal E_X$};
	\node[anchor=north west,gray!65!black] at (0.47,4.3) {$\mathcal M_X^c$};
	\node[anchor=north west] at (6.19,2.73) {$P_r$};
	\node[anchor=north west,blue!70!black] at (4.61,3.34) {$P_r^-$};
	\node[anchor=north west,orange!80!black] at (6.00,0.60) {$P_r^+$};
\end{tikzpicture}}
	\caption{Rational polytope cell from \Cref{ex_representable_sets_positive}. The black polygon is $P_r$, the blue dashed polygon is $P_r^-$, and the orange dashed polygon is $P_r^+$. The displayed witness uses rectangular blocks: $\mathcal M_X$ lies inside $P_r^-$, $\mathcal M_X^c$ lies in $B_X\setminus P_r^+$, and the orange block cover $\mathcal E_X$ lies between these two regions and covers $\partial P_r$.}
	\label{fig_ex31_rational_polytope}
\end{figure}

\begin{definition}[Finite Lipschitz chart presentation]
	\label{dfn_finite_lipschitz_chart_presentation}
	Let $X\Subset B$.
	A \emph{finite Lipschitz chart presentation} of $X$ consists of the following data.
	First, there are finitely many pairwise disjoint full blocks
	\[
	Q_l\Subset B,\qquad l=1,\ldots,L,
	\]
	and, for each $l$, a coordinate permutation $\pi_l$ such that
	\[
	\pi_l(Q_l)=A_l\times[c_l,d_l],
	\]
	where $A_l\subset\R^{n-1}$ is a full block and $c_l<d_l$ are rational.
	Second, for each $l$ there is an effective Lipschitz function $\varphi_l:A_l\to[c_l,d_l]$ and a sign $\sigma_l\in\{+,-\}$.
	The chart contribution is
	\[
	X_l
	=
	\begin{cases}
		\pi_l^{-1}\{(u,t)\in A_l\times[c_l,d_l]:t\le\varphi_l(u)\}, & \sigma_l=+,\\[2mm]
		Q_l\setminus
		\pi_l^{-1}\{(u,t)\in A_l\times[c_l,d_l]:t<\varphi_l(u)\}, & \sigma_l=-.
	\end{cases}
	\]
	Finally, there is a finite union $P_0$ of rational polytopes, called the \emph{polyhedral remainder}, such that
	\[
	X=P_0\cup\bigcup_{l=1}^L X_l.
	\]
	This is the post-refinement presentation.
	If an initial finite chart list has overlapping blocks, one first takes the finite Boolean subdivision generated by those blocks and replaces the original charts by the resulting full subblocks on which the same formula for $X_l$ remains valid.
	The portions of $X$ not covered by those subblocks are included in $P_0$.
\end{definition}

\begin{remark}
	A coordinate permutation $\pi_l$ only reorders coordinates.
	Thus $\pi_l^{-1}$ in the formula for $X_l$ is the inverse reordering, applied to the block $A_l\times[c_l,d_l]$ and to points $(u,t)$ in that block.
\end{remark}

\begin{definition}[Finite cone and rolling-ball data]
	\label{dfn_finite_cone_rolling_ball_data}
	Let $X\Subset B$.
	A \emph{finite cone presentation} consists of finitely many pairwise disjoint full blocks
	\[
	Q_l\Subset B,\qquad l=1,\ldots,L,
	\]
	coordinate permutations $\pi_l$ with
	\[
	\pi_l(Q_l)=A_l\times[c_l,d_l],
	\]
	effective boundary functions $\varphi_l:A_l\to[c_l,d_l]$, signs $\sigma_l\in\{+,-\}$, and a rational cone slope $L_\theta>0$.
	The boundary portion in $Q_l$ is
	\[
	\partial X\cap Q_l
	=
	\pi_l^{-1}\{(u,t)\in A_l\times[c_l,d_l]:t=\varphi_l(u)\},
	\]
	and the side $\sigma_l$ determines which of the two cones
	\[
	\{(v,t):t\le\varphi_l(u)-L_\theta\nrm{v-u}\},
	\qquad
	\{(v,t):t\ge\varphi_l(u)+L_\theta\nrm{v-u}\},
	\]
	is contained in the corresponding side of $X$ for every $u\in A_l$.
	Geometrically, $L_\theta=\cot\theta$ when $\theta$ denotes the cone half-aperture.
	Thus the cone condition is finite in the following sense: there are only the listed charts, and the same rational number $L_\theta$ works throughout every listed chart.

	A \emph{finite rolling-ball presentation} consists of finitely many such blocks $Q_l$, coordinate permutations $\pi_l$, effective boundary functions $\varphi_l$, signs $\sigma_l$, a rational radius $r>0$, and effective assignments which, for each chart boundary point, give an interior ball and an exterior ball of radius $r$ tangent at that point and contained in the corresponding side.
	The chart radii are part of the data and are chosen smaller than $r/4$.
\end{definition}

\begin{example}[Effective Lipschitz sets]
	\label{ex_effective_lipschitz_sets}
	Let $X\Subset B$ have a finite Lipschitz chart presentation in the sense of \Cref{dfn_finite_lipschitz_chart_presentation}.
	Thus
	\[
	X=P_0\cup\bigcup_{l=1}^L X_l,
	\]
	where each $X_l$ has the form
	\[
	X_l
	=
	\begin{cases}
		\pi_l^{-1}\{(u,t)\in A_l\times[c_l,d_l]:t\le\varphi_l(u)\}, & \sigma_l=+,\\[2mm]
		Q_l\setminus
		\pi_l^{-1}\{(u,t)\in A_l\times[c_l,d_l]:t<\varphi_l(u)\}, & \sigma_l=-.
	\end{cases}
	\]
	For every chart with $\sigma_l=+$, the chart piece $X_l$ is representable by \Cref{lem_lipschitz_graph_sets_representable}.
	For every chart with $\sigma_l=-$, the chart piece is the relative apartness complement, inside $Q_l$, of a Lipschitz graph subgraph; hence it is representable by \Cref{lem_lipschitz_graph_sets_representable,lem_representable_relative_complement}.
	The polyhedral remainder $P_0$ is representable by \Cref{ex_representable_sets_positive}.
	For a prescribed $\eps>0$, choose witnesses
	\[
	B,\mathcal E_0,\mathcal M_0,\mathcal M_0^c
	\]
	for $P_0$ and
	\[
	B,\mathcal E_l,\mathcal M_l,\mathcal M_l^c,
	\qquad l=1,\ldots,L,
	\]
	for the chart pieces, all with the common bounding block $B$ and with exception costs small enough that their sum is below $\eps/2$.
	Applying \Cref{lem_representable_finite_union} to
	\[
	P_0,\quad X_1,\ldots,X_L
	\]
	gives a witness
	\[
	B,\qquad \mathcal E_X,\qquad \mathcal M_X,\qquad \mathcal M_X^c
	\]
	for $X$, with
	\[
	\gamma(\mathcal E_X)
	\le
	\gamma(\mathcal H)+\gamma(\mathcal E_0)+\sum_{l=1}^L\gamma(\mathcal E_l)
	<
	\eps,
	\]
	where $\mathcal H$ is the coordinate-boundary multiblock produced in the finite-union construction.
	The construction is local: the exception multiblock is formed from the coordinate-boundary multiblock, the thin graph neighborhoods from the chart witnesses, and the polyhedral face neighborhoods from $P_0$.
	In \Cref{fig_ex32_lipschitz_domain}, the part of $X$ between the displayed chart boxes is included in this polyhedral remainder.
	In particular, bounded polygonal sets, finite unions of regions under computable Lipschitz graphs, and sets with finitely many explicitly charted Lipschitz corners are representable.
\end{example}

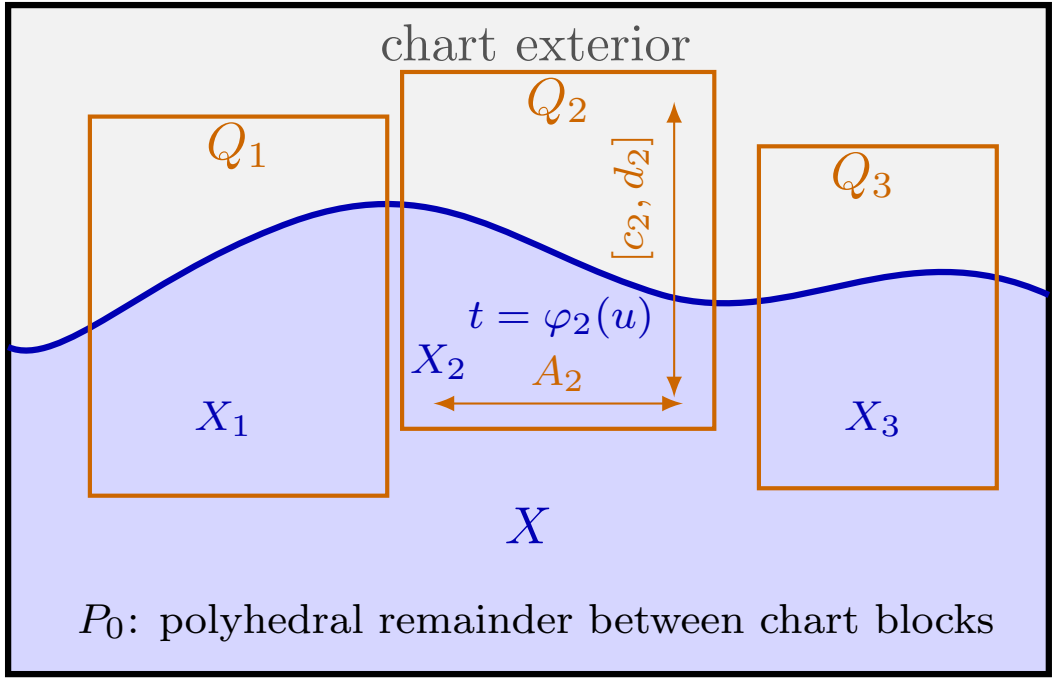
\begin{figure}[ht]
	\centering
	\resizebox{0.85\linewidth}{!}{\begin{tikzpicture}[x=1cm,y=1cm,font=\small]
	\fill[gray!10] (0,0) rectangle (7,4.5);

	\fill[blue!16]
		(0,0) -- (7,0) -- (7,2.55)
		.. controls (6.0,3.0) and (5.3,2.3) .. (4.4,2.55)
		.. controls (3.4,2.85) and (2.9,3.45) .. (1.8,3.0)
		.. controls (0.9,2.65) and (0.35,2.05) .. (0,2.2)
		-- cycle;

	\draw[very thick] (0,0) rectangle (7,4.5);
	\draw[very thick,blue!70!black]
		(0,2.2)
		.. controls (0.35,2.05) and (0.9,2.65) .. (1.8,3.0)
		.. controls (2.9,3.45) and (3.4,2.85) .. (4.4,2.55)
		.. controls (5.3,2.3) and (6.0,3.0) .. (7,2.55);

	\draw[thick,orange!80!black] (0.55,1.2) rectangle (2.55,3.75);
	\draw[thick,orange!80!black] (2.65,1.65) rectangle (4.75,4.05);
	\draw[thick,orange!80!black] (5.05,1.25) rectangle (6.65,3.55);

	\draw[orange!80!black,latex-latex] (2.86,1.82) -- (4.54,1.82);
	\node[orange!80!black,font=\scriptsize] at (3.70,2.02) {$A_2$};
	\draw[orange!80!black,latex-latex] (4.48,1.86) -- (4.48,3.85);
	\node[orange!80!black,font=\scriptsize,rotate=90] at (4.2,3.2) {$[c_2,d_2]$};
	\node[blue!70!black,font=\scriptsize] at (3.72,2.4) {$t=\varphi_2(u)$};

	\node[orange!80!black] at (1.55,3.55) {$Q_1$};
	\node[orange!80!black] at (3.70,3.85) {$Q_2$};
	\node[orange!80!black] at (5.75,3.35) {$Q_3$};
	\node[blue!70!black,font=\scriptsize] at (1.45,1.72) {$X_1$};
	\node[blue!70!black,font=\scriptsize] at (2.9,2.1) {$X_2$};
	\node[blue!70!black,font=\scriptsize] at (5.82,1.72) {$X_3$};
	\node[blue!70!black] at (3.5,1.0) {$X$};
	\node[gray!65!black] at (3.55,4.25) {chart exterior};
	\node[font=\scriptsize,align=center] at (3.55,0.35) {$P_0$: polyhedral remainder between chart blocks};
\end{tikzpicture}}
	\caption{Effective Lipschitz set from \Cref{ex_effective_lipschitz_sets}. The chart block $Q_l$ is $\pi_l^{-1}(A_l\times[c_l,d_l])$, and the portion labelled $X_l$ is described in those coordinates by the graph of $\varphi_l$. The part of $X$ between the chart boxes is not an additional chart: it is included in the polyhedral remainder $P_0$.}
	\label{fig_ex32_lipschitz_domain}
\end{figure}

\begin{lemma}[Graphs of uniformly continuous curves are representable]
	\label{lem_uniformly_continuous_curve_graph_representable}
	Let $I=[a,b]\subset\R$ be a full rational interval and let $x:I\to\R^n$ be uniformly continuous with a modulus.
	Then the graph
	\[
		\Gamma_x:=\{(t,x(t)):t\in I\}\subset\R^{1+n}
	\]
	is representable.
	It has witnesses with empty core.
\end{lemma}

\begin{proof}
	The set $\Gamma_x$ is bounded because $x$ is uniformly continuous on a full interval and hence bounded by a finite-net estimate.
	Fix $\eps>0$.
	Choose a rational $\eta>0$ so small that the volume of an $\eta$-tube around finitely many vertical boxes over a partition of $I$ is at most $\eps$.
	Using the modulus of uniform continuity, choose a rational partition
	\[
		a=t_0<t_1<\cdots<t_N=b
	\]
	so fine that
	\[
		\nrm{x(t)-x(t_k)}_\infty\le\eta
	\]
	for all $t\in[t_k,t_{k+1}]$.
	For each $k$, choose a rational box $V_k\subset\R^n$ containing the $\eta$-neighborhood of $x(t_k)$ and with side lengths prescribed by the preceding choice of $\eta$.
	The finite multiblock
	\[
		\mathcal E
		:=
		\{[t_k,t_{k+1}]\times V_k:k=0,\ldots,N-1\}
	\]
	contains $\Gamma_x$ and has cost at most $\eps$.
	Take the core multiblock empty.
	In a rational bounding block for $\Gamma_x$, refine by the endpoints of $\mathcal E$ and put every full refined block disjoint from $\bar{\mathcal E}$ into the exterior multiblock.
	This exterior is well contained in the apartness complement of $\Gamma_x$ after the boxes $V_k$ are chosen with a positive rational margin around the graph.
	The remaining refined blocks are appended to the exception multiblock.
	This gives a representability witness with empty core and exception cost at most $\eps$.
\end{proof}

\begin{remark}
	\label{rem_curve_image_vs_graph}
	\Cref{lem_uniformly_continuous_curve_graph_representable} concerns the graph in the product space $\R^{1+n}$.
	The image $x(I)\subset\R^n$ is a different set.
	A Lipschitz image has finite length and is representable with empty core in dimensions $n\ge2$, while in dimension $n=1$ the image is an interval and is represented as such.
	For a merely uniformly continuous curve, the image need not have zero volume in higher dimension.
\end{remark}

\begin{example}[Cone and rolling-ball conditions]
	\label{ex_cone_rolling_ball_representable}
	Let $X\Subset B$ be supplied with finite cone data in the sense of \Cref{dfn_finite_cone_rolling_ball_data}.
	Fix a chart and write it, after applying the given coordinate permutation, as $A_l\times[c_l,d_l]$.
	Assume, for definiteness, that the interior side in this chart is the subgraph side.
	At a boundary point $p=(u,\varphi_l(u))$, the corresponding interior cone contains the set
	\[
	K_\theta(p):=
	\{(v,t):t\le\varphi_l(u)-L_\theta\nrm{v-u}\}.
	\]
	No other boundary point can lie in this cone.
	Hence, if $q=(v,\varphi_l(v))$ is another boundary point and $\varphi_l(v)\le\varphi_l(u)$, then
	\[
	\varphi_l(u)-\varphi_l(v)
	\le
	L_\theta\nrm{v-u}.
	\]
	Interchanging $p$ and $q$ gives the opposite inequality, and therefore
	\[
	|\varphi_l(u)-\varphi_l(v)|
	\le
	L_\theta\nrm{u-v},
	\qquad u,v\in A_l,
	\]
	with the side data specifying whether $X\cap Q_l$ is the subgraph or the relative complement.
	Thus the cone presentation supplies a finite Lipschitz chart presentation, and \Cref{ex_effective_lipschitz_sets} applies.

	If $X$ is instead supplied with finite rolling-ball data, then the interior and exterior balls imply, after shrinking the listed charts if needed, an interior cone condition and an exterior cone condition whose aperture depends only on the ratio of the chart radius to $r$.
	Equivalently, each local graph has a Lipschitz constant bounded by a rational number computed from the supplied rolling-ball data.
	A boundary layer of thickness $\lambda$ in each chart has cost bounded by $C_l\lambda$, so the total boundary-layer cost is made arbitrarily small.
	The regions at distance at least $\lambda$ from the boundary are assigned to the core or exterior according to the supplied side data.
\end{example}

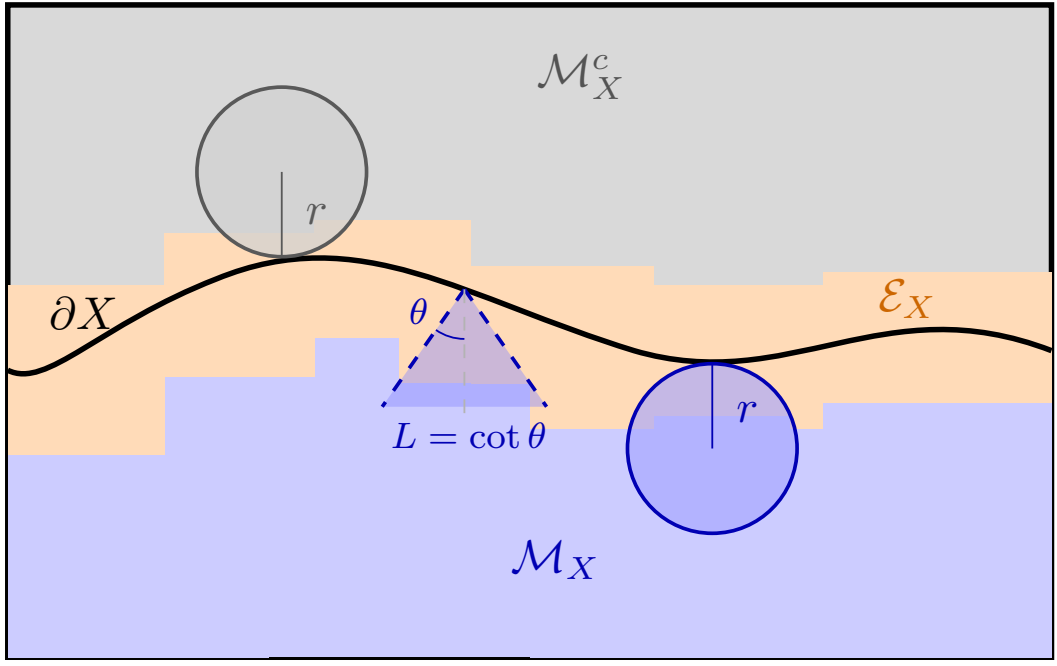
\begin{figure}[ht]
	\centering
	\resizebox{0.85\linewidth}{!}{\begin{tikzpicture}[x=1cm,y=1cm,font=\small]
	\fill[gray!30] (0,0) rectangle (8,5);
	\draw[very thick] (0,0) rectangle (8,5);

	\fill[blue!20] (0,0) rectangle (2.0,1.75);
	\fill[blue!20] (2.0,0) rectangle (4.0,2.10);
	\fill[blue!20] (4.0,0) rectangle (6.0,1.75);
	\fill[blue!20] (6.0,0) rectangle (8,1.95);
	\fill[blue!20] (1.2,1.75) rectangle (3.0,2.45);
	\fill[blue!20] (4.9,1.75) rectangle (7.2,2.15);

	\fill[orange!28] (0,1.55) rectangle (1.20,2.85);
	\fill[orange!28] (1.20,2.15) rectangle (2.35,3.25);
	\fill[orange!28] (2.35,2.45) rectangle (3.55,3.35);
	\fill[orange!28] (3.0,2.10) rectangle (3.55,2.45);
	\fill[orange!28] (3.55,2.10) rectangle (4.95,3.00);
	\fill[orange!28] (4.0,1.75) rectangle (4.95,2.10);
	\fill[orange!28] (4.95,1.85) rectangle (6.25,2.85);
	\fill[orange!28] (6.0,1.75) rectangle (6.25,1.95);
	\fill[orange!28] (6.25,1.95) rectangle (8,2.95);

	\draw[very thick]
		(0,2.2)
		.. controls (0.3,2.05) and (0.7,2.55) .. (1.6,2.9)
		.. controls (2.7,3.35) and (3.5,2.75) .. (4.8,2.35)
		.. controls (6.0,2.0) and (6.7,2.85) .. (8,2.35);

	\fill[blue!45,opacity=0.45]
		(3.50,2.82) -- (2.87,1.92) -- (4.13,1.92) -- cycle;
	\draw[blue!70!black,thick,dashed] (3.50,2.82) -- (2.87,1.92);
	\draw[blue!70!black,thick,dashed] (3.50,2.82) -- (4.13,1.92);
	\draw[gray!60,dashed,thin] (3.50,2.82) -- (3.50,1.85);
	\draw[blue!70!black,thin] (3.50,2.82) ++(235:0.38) arc[start angle=235,end angle=270,radius=0.38];
	\node[blue!70!black,font=\scriptsize] at (3.15,2.65) {$\theta$};
	\node[blue!70!black,font=\scriptsize,anchor=west] at (2.8,1.7) {$L=\cot\theta$};

	\fill[blue!35,opacity=0.65] (5.4,1.6) circle (0.65);
	\draw[blue!70!black,thick] (5.4,1.6) circle (0.65);
	\draw[blue!70!black] (5.4,1.6) -- (5.4,2.25);
	\node[blue!70!black,right] at (5.45,1.88) {$r$};

	\fill[gray!35,opacity=0.75] (2.1,3.72) circle (0.65);
	\draw[gray!70!black,thick] (2.1,3.72) circle (0.65);
	\draw[gray!70!black] (2.1,3.72) -- (2.1,3.07);
	\node[gray!70!black,right] at (2.15,3.40) {$r$};

	\node[blue!70!black] at (4.2,0.75) {$\mathcal M_X$};
	\node[gray!65!black] at (4.4,4.45) {$\mathcal M_X^c$};
	\node[orange!80!black] at (6.9,2.75) {$\mathcal E_X$};
	\node at (0.58,2.64) {$\partial X$};
\end{tikzpicture}}
	\caption{Cone and rolling-ball data from \Cref{ex_cone_rolling_ball_representable}. The displayed cone records the local Lipschitz bound $L=\cot\theta$. The gray ball lies in the exterior side and the blue ball lies in the interior side; the representability witness is still formed by rectangular core, exception, and exterior blocks.}
	\label{fig_ex33_rolling_ball}
\end{figure}

\begin{example}[Non-representable sets]
	\label{ex_non_representable_sets}
	The set $\Q\cap[0,1]$ is not representable.
	No nonempty full interval is well contained in $\Q\cap[0,1]$, because every such interval contains irrational points.
	No nonempty full interval is well contained in the relative apartness complement either, because every such interval contains rational points.
	Thus the core and exterior multiblocks in any witness are empty after degenerate constituents are removed.
	The exception multiblock would have to cover $[0,1]$, so its cost is at least $1$.
	This contradicts the requirement that the exception cost be below an arbitrary $\eps>0$.
	There are also constructive failures: if
	\[
	X_P:=\{x\in[0,1]:x\le 1/2\ \vee\ P\}
	\]
	for a proposition $P$ for which no proof or refutation is available, then a witness with exception cost below $1/4$ cannot put all of $[3/4,1]$ into the exception multiblock.
	Some nonempty full subinterval of $[3/4,1]$ must therefore lie in the core or in the exterior.
	The first case gives $P$; the second gives a refutation of $P$.
	Thus such a witness cannot be constructed without deciding $P$.
\end{example}

\begin{figure}[ht]
	\centering
	\resizebox{\linewidth}{!}{\begin{tikzpicture}[x=1cm,y=1cm,font=\small]
	\begin{scope}
		\draw[very thick] (0,0) -- (4,0);
		\foreach \x in {0.15,0.45,0.8,1.05,1.35,1.7,2.05,2.35,2.7,3.0,3.35,3.7} {
			\fill (\x,0) circle (0.045);
		}
		\foreach \x in {0.3,0.62,0.92,1.22,1.55,1.88,2.18,2.52,2.85,3.18,3.52,3.88} {
			\draw (\x,0) circle (0.05);
		}
		\draw[<->,red!70!black,thick] (1.05,-0.45) -- (2.05,-0.45);
		\node[red!70!black,below,font=\scriptsize] at (1.55,-0.45) {any full interval meets both};
		\node[above] at (2,0.50) {$\Q\cap[0,1]$};
		\node[below] at (0,-0.05) {$0$};
		\node[below] at (4,-0.05) {$1$};
		\fill (0.35,-1.02) circle (0.045);
		\node[right,font=\scriptsize] at (0.48,-1.02) {rational};
		\draw (2.05,-1.02) circle (0.05);
		\node[right,font=\scriptsize] at (2.18,-1.02) {irrational};
	\end{scope}

	\begin{scope}[xshift=5.6cm]
		\fill[blue!18] (0,-0.15) rectangle (2,0.15);
		\fill[gray!22] (2,-0.15) rectangle (4,0.15);
		\draw[very thick] (0,0) -- (4,0);
		\draw[dashed] (2,-0.3) -- (2,0.3);
		\draw[red!70!black,very thick] (3,-0.18) rectangle (4,0.18);
		\node[above] at (2,0.72) {$X_P$};
		\node[blue!70!black,font=\scriptsize] at (1,0.43) {always in};
		\node[gray!65!black,font=\scriptsize] at (3,0.43) {depends on $P$};
		\node[below] at (0,-0.05) {$0$};
		\node[below] at (2,-0.05) {$1/2$};
		\node[below] at (4,-0.05) {$1$};
		\node[red!70!black,below,font=\scriptsize] at (3.5,-0.25) {$[3/4,1]$};
		\node[red!70!black,align=center,font=\scriptsize] at (3.5,-0.95) {core gives $P$\\exterior gives $\neg P$};
	\end{scope}
\end{tikzpicture}}
	\caption{Two obstructions from \Cref{ex_non_representable_sets}. Dense interleaving prevents full core and exterior intervals; the set $X_P$ would force a decision of $P$.}
	\label{fig_ex34_nonrepresentable}
\end{figure}
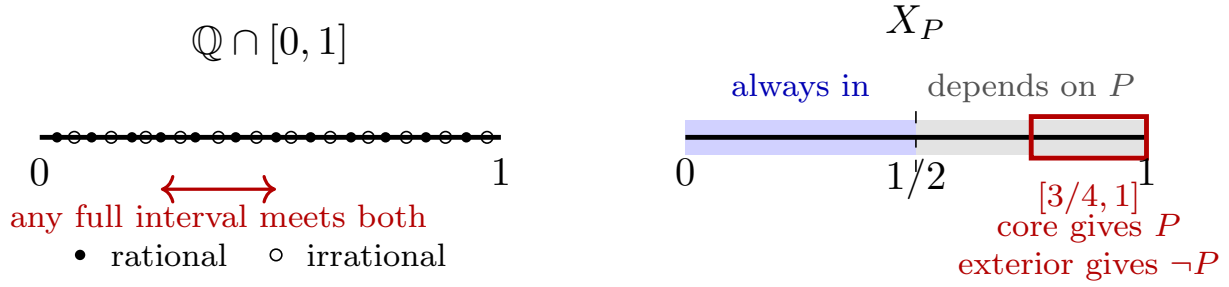

\begin{example}[Located but not representable]
	\label{ex_located_not_representable}
	Let $S\subset[0,1]$ be the effective Smith--Volterra--Cantor set obtained by removing, at stage $r$, middle open intervals of total length $2^{-r-1}$ from the intervals remaining after stage $r-1$.
	Then the total removed length is
	\[
	\sum_{r=1}^{\infty}2^{-r-1}=\frac12,
	\]
	so $S$ has length $1/2$.
	It is located.
	Indeed, at every finite stage one has a finite union $S_r$ of rational closed intervals, hence $d(x,S_r)$ is computable.
	The sets $S_r$ decrease to $S$, and every point of $S_r$ lies within the maximal length of a stage-$r$ interval from a point of $S$.
	Since these maximal lengths tend effectively to zero, the numbers $d(x,S_r)$ converge effectively to $d(x,S)$.
	In particular, $S$ is totally bounded: at a prescribed precision one takes the finite stage $S_r$ with sufficiently small maximal interval length and uses the endpoints of the finitely many remaining closed intervals.
	These endpoints belong to $S$ and form a finite net for $S$.

	The set $S$ is closed and has empty interior, hence $\partial S=S$.
	Thus its boundary is located as well.
	Nevertheless, $S$ is not representable.
	If $S$ were representable, no nonempty full constituent block of the core multiblock could occur, because every positive-length interval contains points outside $S$.
	Thus the core is empty after degenerate constituents are removed.
	The exterior multiblock is well contained in the apartness complement of $S$, so it cannot contain points of $S$.
	Hence the exception multiblock would have to cover $S$ up to the zero-volume remainder allowed by measure containment.
	Any finite interval cover of $S$ up to such a remainder has total length at least the outer length of $S$, namely at least $1/2$.
	The exception cost therefore cannot be made arbitrarily small.
\end{example}

\begin{example}[Representable but not located]
	\label{ex_representable_not_located}
	Let $P$ be a proposition for which neither a proof nor a refutation is available.
	In $\R^2$, put
	\[
	Q:=[0,1]^2,\qquad a:=(0,3),\qquad b:=(1,3),
	\]
	and define
	\[
	Y_P:=Q\cup\{a\}\cup\{b:P\}.
	\]
	The set $Y_P$ is representable without deciding $P$.
	For instance, take
	\[
	B_Y=[-1,2]\times[-1,4],
	\qquad
	L=[0,1]\times\{3\}.
	\]
	For a prescribed $\eps>0$, choose the core to be a rational square well contained in $Q$.
	Choose the exception multiblock to consist of a thin rectangular shell around $\partial Q$ and a thin rational rectangle around $L$, with total cost below $\eps$.
	The exterior multiblock consists of the remaining rational blocks of a sufficiently fine subdivision of $B_Y$ which are well contained in
	\[
	B_Y\setminus(Q\cup L).
	\]
	This exterior is well contained in $B_Y\setminus Y_P$ for either truth value of $P$, since all possible additional points are contained in $L$ and hence in the exception layer.
	Thus the construction gives a representability witness uniformly in $P$.
	However, this example is not known to be totally bounded.
	Indeed, a finite $\frac12$-net for $Y_P$ would have to contain a point within distance $\frac12$ of $b$ whenever $P$ holds, and no point of $Q\cup\{a\}$ has this property.
	Such net data would therefore decide whether the optional point $b$ is present.

	A locatedness operation for $Y_P$ would decide $P$.
	At the point $b$ one has
	\[
	d(b,Y_P)=0 \quad\text{if } P,
	\qquad
	d(b,Y_P)=1 \quad\text{if } \neg P,
	\]
	because $d(b,Q)=2$ and $d(b,a)=1$.
	Computing $d(b,Y_P)$ with error below $1/4$ separates these two alternatives.
	In the first alternative, the upper distance bound gives a point of $Y_P$ within distance $1/2$ of $b$, which must be $b$ itself and therefore yields $P$.
	In the second alternative, $P$ is refuted.
	Hence no locatedness operation for $Y_P$ is constructively available unless $P$ is decidable.
\end{example}


\section{Functions}
\label{sec_functions}

\begin{remark}
	This section passes from representable supports to regular functions.
	A regular function is built from finitely many complemented continuous pieces and becomes integrable once an effective bound is supplied.
	The section also records closure, convergence, differentiation, integration, and Fubini tools used later for optimization, multifunctions, probability densities, and solution certificates.
\end{remark}

\subsection{Continuous Functions}
\label{subsec_continuous_functions}

\begin{definition}
	\label{dfn_func}
	A map $f:X\to\R^m$ with $X\subseteq\R^n$ is called a \textit{function} if it is equipped with an operation to compute $f(x)$ as a rational approximation up to any precision for any $x\in X$.
\end{definition}

\begin{remark}
	\label{rem_fnc_rat_approx}
		Computing function values as rational approximations up to any precision means a Cauchy witness in the spirit of \Cref{rem_reals_funcs}.
\end{remark}

\begin{notation}
	\label{ntn_dom}
	The set, in the above case $X$, on which the function $f$ is defined is called the \textit{domain} of $f$ and denoted $\dom{f}$.
\end{notation}

\begin{definition}[Moduli of continuity]
	\label{dfn_modulus_continuity}
	Let $X\subset\R^n$ and let $f:X\to\R^m$.
	A map $\omega_f:\Q_{>0}\to\Q_{>0}$ is called a \emph{modulus of continuity} of $f$ if, for every $\eps>0$ and all $x,y\in X$,
	\begin{equation}
		\label{eqn_modulus_continuity}
		\nrm{x-y}\le\omega_f(\eps)
		\quad\Longrightarrow\quad
		\nrm{f(x)-f(y)}\le\eps.
	\end{equation}
	A rule assigning a positive rational $\omega_f(\eps\mid B)$ to each $\eps>0$ and each block $B$ is called a \emph{modulus of local continuity} of $f$ if, for every $\eps>0$, every block $B$, and all $x,y\in X\cap B$,
	\begin{equation}
		\label{eqn_modulus_local_continuity}
		\nrm{x-y}\le\omega_f(\eps\mid B)
		\quad\Longrightarrow\quad
		\nrm{f(x)-f(y)}\le\eps.
	\end{equation}
\end{definition}

\begin{remark}
	\label{rem_func_bounded_dom}
	Let $X \subset \R^n$ be bounded.
	Let $f: X \to \R^m$ be a function.
	Then, there exists a block $B_X$ such that $X \subseteq B_X$.
	If $f$ is locally continuous, the local modulus on this fixed block gives a modulus on $X$.
	In this bounded situation we write $\omega_f(\eps)$ when the ambient block has been fixed.
\end{remark}


\begin{remark}
	\label{rem_modulus_pairwise}
	Continuity in \Cref{dfn_continuous_function} is a pairwise property on the domain.
	In particular, no connectedness of $X$ is assumed.
	If two components of $X$ are separated by a positive distance, pairs of points belonging to different components impose no additional small-scale restriction.
\end{remark}

\begin{definition}[Continuous function]
	\label{dfn_continuous_function}
	Let $X\subset\R^n$.
	A function $f:X\to\R^m$ is called \emph{continuous} on $X$ if it is given together with a modulus of continuity depending only on the requested accuracy.
	It is called \emph{locally continuous} on $X$ if it is given together with a modulus of local continuity, allowed to depend on the ambient block.
\end{definition}

\begin{remark}
	\label{rem_constructive_continuity}
	The modulus is part of the constructive data of a continuous function.
	Thus, a continuity claim is not merely the assertion that nearby inputs have nearby outputs, but also provides an explicit rule for how close the inputs must be for a prescribed output accuracy.
	For locally continuous functions this rule may also depend on a block on which the comparison is made.
\end{remark}

\begin{remark}
	\label{rem_uniform_conty}
	Constructively, a continuous function on a block is equivalent to uniformly continuous function on that block in classical analysis.
\end{remark}

\begin{definition}[Lipschitz constant]
	\label{dfn_lipschitz_constant}
	Let $X\subseteq\R^n$ and let $f:X\to\R^m$.
	The function $f$ is called \emph{$L$-Lipschitz} on $X$ if
	\[
		\forall x,y\in X\spc \nrm{f(x)-f(y)}\le L\nrm{x-y}.
	\]
	We write $\Lip{X}(f)\le L$ for this estimate.
	If the least such number exists, it is denoted by $\Lip{X}(f)$.
	An $L$-Lipschitz function is continuous with modulus $\omega_f(\eps)=\eps/L$ when $L>0$; for $L=0$ any positive modulus may be used.
\end{definition}

\begin{definition}[Complemented function]
	\label{dfn_complemented_function}
	Let $X\subset\R^n$.
	A \emph{complemented function} over $X$ is a pair
	\begin{equation}
		\label{eqn_complemented_function_branches}
		f=(f^+,f^-),
		\qquad
		f^+:X\to\R^m,
		\qquad
		f^-:X^c\to\R^m,
	\end{equation}
	where both branches are continuous on their respective domains, each with its own modulus.
	The expression $f(x)$ is understood branchwise: it means $f^+(x)$ when a witness $x\in X$ is given, and $f^-(x)$ when a witness $x\in X^c$ is given.
\end{definition}

\begin{remark}
	\label{rem_complemented_domain_not_total}
	The notation in \Cref{dfn_complemented_function} does not assert that $\R^n=X\cup X^c$ which does not hold constructively.
	The two branches live on the positive complemented domain determined by the available membership witness.
\end{remark}

\begin{notation}
	\label{ntn_zero_one_functions}
	The symbols $\zero$ and $\one$ denote the identically zero and identically one functions on the domain under discussion.
\end{notation}

\begin{remark}
	\label{rem_complemented_function_requires_plain_continuity}
	The branches in \Cref{dfn_complemented_function} are continuous in the sense of \Cref{dfn_continuous_function}, not merely locally continuous.
	For instance, the pair
	\[
		\left\{\frac1x,\zero\right\}
	\]
	on the complemented set
	\[
		(0,\infty),\qquad (-\infty,0)
	\]
	is not a complemented function in this sense.
	The positive branch has no modulus depending only on the requested accuracy on all of $(0,\infty)$, although it is locally continuous on blocks separated from $0$.
\end{remark}

\begin{definition}[Bounded-support complemented function]
	\label{dfn_bounded_support}
	A complemented function $f=(f^+,f^-)$ over $X$ is said to have \emph{bounded support} if $X$ is bounded and
	\begin{equation}
		\label{eqn_bounded_support}
		\forall x \in X^c \spc f^-(x)=0.
	\end{equation}
	Such a set $X$ is called a \textit{support} for $f$.
\end{definition}

\begin{notation}
	\label{ntn_function_support}
	For a bounded-support complemented function $f$, we denote the respective support set by $\supp{f}$.
\end{notation}

\begin{remark}
	\label{rem_bounded_support_positive_complement}
	The complement $X^c$ in \Cref{dfn_bounded_support} is understood in the apartness sense introduced in \Cref{ntn_apartness_complement}.
	Thus, the condition asserts vanishing where membership outside $X$ is witnessed positively.
\end{remark}

\begin{definition}[Simple complemented function]
	\label{dfn_simple_complemented_function}
	A bounded-support complemented function $f=(f^+,f^-)$ over $X$ is called \emph{simple} if there exists $c\in\R^m$ such that
	\[
		\forall x\in X \spc f^+(x)=c,
		\qquad
		\forall x\in X^c \spc f^-(x)=0.
	\]
	If $m=1$ and $c=1$, it is called the \emph{indicator} of $X$.
\end{definition}

\begin{notation}[Indicator]
	\label{ntn_indicator}
	We denote indicators as $\indic{X}$.
\end{notation}

\subsection{Measurable Functions}

\begin{definition}[Measurable function on a block]
	\label{dfn_measurable_function}
	Let $B$ be a block.
	Let $f$ be a function with $\dom f\subseteq\R^n$ and values in $\R^m$.
	The function $f$ is called \emph{measurable on $B$} if, for every $\eps > 0$, there exist a finite multiblock $\mathcal E_{\eps\mid B}\Subset B$ and a continuous function $g_{\eps\mid B}:B \to \R^m$ such that
	\begin{equation}
		\label{eqn_measurable_function_exceptional_block}
		\gamma(\mathcal E_{\eps\mid B}) \le \eps
	\end{equation}
	and
	\begin{equation}
		\label{eqn_measurable_function_continuous_approximation}
		\forall x \in B\cap\dom f\cap\bar{\mathcal E}_{\eps\mid B}^c \spc \nrm{f(x)-g_{\eps\mid B}(x)} \le \eps.
	\end{equation}
	The multiblock $\mathcal E_{\eps\mid B}$ is called an exception multiblock for the accuracy $\eps$ on $B$.
\end{definition}

\begin{definition}[Locally measurable function]
	\label{dfn_locally_measurable_function}
	A function $f$ with $\dom f\subseteq\R^n$ and values in $\R^m$ is called \emph{locally measurable}, or simply \emph{measurable}, if it is measurable on every block $B$.
\end{definition}

\subsubsection{Convergence in Measure}
\label{subsubsec_convergence_in_measure}

\begin{definition}[Convergence almost uniformly]
	\label{dfn_convergence_almost_uniformly}
	Let $A\subseteq\R^n$, let $B$ be a block, and let $f_i:A\to\R^m$, $i\in\N$, and $f:A\to\R^m$ be measurable functions.
	The sequence $(f_i)_i$ is said to \emph{converge almost uniformly} to $f$ on $A\cap B$ if, for every $\eta>0$, there exists a finite multiblock $\mathcal E\Subset B$ with
	\[
		\gamma(\mathcal E)\le\eta
	\]
	such that, for every $\eps>0$, there exists $N\in\N$ satisfying
	\[
		i\ge N,\quad x\in A\cap B\cap\bar{\mathcal E}^c
		\quad\Longrightarrow\quad
		\nrm{f_i(x)-f(x)}\le\eps.
	\]
\end{definition}

\begin{definition}[Pointwise convergence]
	\label{dfn_pointwise_convergence}
	Let $A\subseteq\R^n$, and let $f_i:A\to\R^m$, $i\in\N$, and $f:A\to\R^m$ be functions.
	The sequence $(f_i)_i$ is said to \emph{converge pointwise} to $f$ on $A$ if, for every $x\in A$ and every $\eps>0$, there exists $N\in\N$ such that
	\[
		i\ge N
		\quad\Longrightarrow\quad
		\nrm{f_i(x)-f(x)}\le\eps.
	\]
\end{definition}

\begin{definition}[Convergence in measure]
	\label{dfn_convergence_in_measure}
	Let $A\subseteq\R^n$, let $B$ be a block, and let $f_i:A\to\R^m$, $i\in\N$, and $f:A\to\R^m$ be measurable functions.
	The sequence $(f_i)_i$ is said to \emph{converge in measure} to $f$ on $A\cap B$ if, for every $\eps>0$, there exists $N\in\N$ such that, for every $i\ge N$, there exists a finite multiblock $\mathcal E_i\Subset B$ with
	\[
		\gamma(\mathcal E_i)\le\eps
	\]
	and
	\[
		x\in A\cap B\cap\bar{\mathcal E}_i^c
		\quad\Longrightarrow\quad
		\nrm{f_i(x)-f(x)}\le\eps.
	\]
\end{definition}

\begin{definition}[Cauchy in measure]
	\label{dfn_cauchy_in_measure}
	Let $A\subseteq\R^n$, let $B$ be a block, and let $f_i:A\to\R^m$, $i\in\N$, be measurable functions.
	The sequence $(f_i)_i$ is called \emph{Cauchy in measure} on $A\cap B$ if, for every $\eps>0$, there exists $N\in\N$ such that, for every $i,j\ge N$, there exists a finite multiblock $\mathcal E_{ij}\Subset B$ with
	\[
		\gamma(\mathcal E_{ij})\le\eps
	\]
	and
	\[
		x\in A\cap B\cap\bar{\mathcal E}_{ij}^c
		\quad\Longrightarrow\quad
		\nrm{f_i(x)-f_j(x)}\le\eps.
	\]
\end{definition}

\begin{definition}[Full subset relative to a block]
	\label{dfn_full_subset_relative_block}
	Let $A\subseteq B\subset\R^n$ with $B$ a block.
	A subset $A_\ast\subseteq A$ is called \emph{full in $A$ relative to $B$} if, for every $\eta>0$, there exists a finite full multiblock $\mathcal E\Subset B$ such that
	\[
		\gamma(\mathcal E)\le\eta
	\]
	and
	\[
		A\cap\bar{\mathcal E}^c\subseteq A_\ast.
	\]
\end{definition}

\begin{remark}
	\label{rem_convergence_in_measure_bb}
	\Cref{dfn_convergence_almost_uniformly,dfn_convergence_in_measure,dfn_cauchy_in_measure} are blockwise versions of Bishop--Bridges convergence notions for measurable functions \cite[Chapter~6, Section~8]{Bishop1985ConstructiveAn}.
	The following theorem is the corresponding completion result for measurable-function convergence.
\end{remark}

\begin{theorem}[Cauchy-in-measure completion]
	\label{thm_cauchy_in_measure_completion}
	Let $A\subseteq B\subset\R^n$ with $B$ a block.
	Let $(f_i)_i$ be a sequence of measurable functions $f_i:A\to\R^m$ which is Cauchy in measure on $A$.
	Then there exists a measurable function $f:A\to\R^m$, a subsequence $(f_{i_k})_k$, and a set $A_\ast\subseteq A$ full in $A$ relative to $B$ such that $(f_i)_i$ converges to $f$ in measure on $A$, $(f_{i_k})_k$ converges almost uniformly to $f$ on $A$, and $(f_{i_k})_k$ converges pointwise to $f$ on $A_\ast$.
\end{theorem}

\begin{proof}
	This is Bishop--Bridges' Cauchy-in-measure completion theorem \cite[Chapter~6, Theorem~8.16]{Bishop1985ConstructiveAn}, applied on the finite block $B$ and written in the notation of \Cref{dfn_convergence_almost_uniformly,dfn_pointwise_convergence,dfn_convergence_in_measure,dfn_cauchy_in_measure,dfn_full_subset_relative_block}.
\end{proof}

\subsection{Regular Functions}
\label{subsec_regular_functions}

\begin{definition}[Regular function]
	\label{dfn_regular_function}
	A \emph{regular function} is a finite formal sum
	\[
		f=\sum_{k=1}^N f_k
	\]
	such that:
	\begin{enumerate}
		\item each $f_k$ is a bounded-support complemented function with support $X_k$,
		\item the supports $X_1,\ldots,X_N$ are pairwise disjoint and representable,
		\item the effective support is
	\[
		\supp f:=\bigcup_{k=1}^N X_k.
	\]
	\end{enumerate}
	The domain of this representation is
	\[
		\dom f
		:=
		\supp f\cup(\supp f)^c,
	\]
	where the complement is the apartness complement.
	The empty sum is allowed; then $\supp f=\emptyset$ and the associated bounded-support function is zero on every ambient block.
	Its value is determined at a point precisely when the representation supplies the branch witnesses needed by all summands.
	In particular, if $x$ is in the apartness complement of $\supp f$, then all summands take their exterior value and $f(x)=0$.
\end{definition}

\begin{remark}
	\label{rem_regular_function_piecewise}
	A regular function is a finite sum of complemented continuous pieces with representable bounded supports.
	Requiring the displayed supports to be disjoint loses no generality: any finite family of supports can be flattened by \Cref{lem_representable_set_flattening}.
	This disjointness is the no-overlap condition of \Cref{dfn_disjoint_sets}, not mutual apartness.
	Thus a pointwise evaluation of the formal sum still requires the branch witnesses specified in \Cref{dfn_regular_function}; no exterior witness for one branch is inferred merely from membership in another disjoint support.
\end{remark}

\begin{remark}
	\label{rem_regular_function_measurable}
	Definition \Cref{dfn_regular_function} is not meant in the sense that $f=\sum_k f_k$ is necessarily continuous on the whole union of the supports.
	Near boundaries of the supports, neighboring pieces need not fit together continuously.
	Measurability follows by replacing each piece, outside a small exception multiblock, by a continuous comparison function and then summing these comparisons.
\end{remark}

\begin{remark}
	\label{rem_empty_regular_function_measurable}
	A regular function with empty support is the empty sum and hence the zero function.
	It is measurable and integrable with value zero by taking the continuous comparison function $0$ and an exception multiblock of arbitrarily small cost.
\end{remark}

\begin{definition}[Bound witness]
	\label{dfn_regular_function_bound_witness}
	Let $B$ be a block.
	A regular function $f:B\to\R^m$ is called \emph{effectively bounded on $B$} if it is equipped with a rational number $M>0$ such that
	\[
		\nrm{f(x)}\le M,
		\qquad
		x\in B\cap\dom f.
	\]
	The number $M$ is called a \emph{bound witness} for $f$ on $B$.
\end{definition}

\begin{remark}
	\label{rem_regular_bound_witness}
	Regularity by itself is a finite piecewise description, but it does not automatically provide a numerical bound in the data.
	Whenever integration or a uniform estimate uses boundedness, the bound witness in \Cref{dfn_regular_function_bound_witness} is part of the hypothesis.
	This does not affect measurability: the comparison functions are built on finite multiblocks, where the needed bounds come from finitely many sampled branch values.
	It does affect Riemann-integrability and norm estimates, where the exceptional contribution is bounded by a term of the form $M\gamma(\mathcal E)$.
	This is the same role played by bounded test functions and domination data in Bishop--Bridges integration theory \cite[Chapter~6]{Bishop1985ConstructiveAn}, and by bounded uniformly continuous test functions in Chan's integration spaces \cite{Chan2019FoundationsCon}.
	In Ye-style finitistic arguments, such numerical bounds are likewise part of the data needed to keep recursive estimates finite \cite{Ye2011StrictFinitism}.
\end{remark}

\begin{definition}[Simple regular function]
	\label{dfn_simple_function}
	A regular function $f=\sum_k f_k$ is called \emph{simple} if each complemented summand $f_k$ is simple.
\end{definition}

\begin{lemma}[Approximation by simple regular functions]
	\label{lem_regular_simple_approximation}
	Let $B$ be a full block and let $f:B\to\R^m$ be regular with supports well contained in $B$.
	For every $\eps>0$ and $\eta>0$, there exist a simple regular function $s:B\to\R^m$ and an exception multiblock $\mathcal E$ such that
	\[
		\gamma(\mathcal E)\le\eta
	\]
	and
	\[
		\forall x\in B\cap\bar{\mathcal E}^c
		\quad
		\nrm{f(x)-s(x)}\le\eps.
	\]
\end{lemma}

\begin{proof}
	Write $f=\sum_{k=1}^N f_k$ with representable supports $X_k$.
	If $N=0$, take $s=0$ and take the empty exception multiblock; the claim is immediate.
	Hence assume $N>0$.
	Choose representability witnesses for the $X_k$ with exception costs at most $\eta/(3N)$.
	For each $k$, apply \Cref{dfn_coordinate_refinement} to $B$ and to
	\[
		\mathcal A_k
		:=
		\mathcal M_{X_k}
		\sqcup
		\mathcal E_{X_k}
		\sqcup
		\mathcal M^c_{X_k}.
	\]
	Choose the corresponding coordinate-boundary multiblock $\mathcal H_k$ so thin that
	\[
		\gamma(\mathcal H_k)\le\eta/(3N).
	\]
	On the core multiblock $\mathcal M_{X_k}$, the branch $f_k^+$ is continuous and hence uniformly continuous.
	Subdivide the constituent blocks of $\mathcal M_{X_k}$ into finitely many smaller full blocks on which the variation of $f_k^+$ is at most $\eps/N$.
	Move the boundaries of this subdivision into an additional exception multiblock $\mathcal D_k$ of cost at most $\eta/(3N)$.
	On each remaining subblock choose a tag $p_{k,j}$ and put the constant value $f_k^+(p_{k,j})$.
	These subblocks are representable supports, so the resulting finite sum is a simple regular function.
	Let $\mathcal E$ be the concatenation of all multiblocks
	\[
		\mathcal E_{X_k},\qquad
		\mathcal H_k,\qquad
		\mathcal D_k,
		\qquad k=1,\ldots,N.
	\]
	Then $\gamma(\mathcal E)\le\eta$.
	On the exterior multiblock $\mathcal M_{X_k}^c$ both the corresponding approximant and $f_k$ vanish.
	Outside the original exception multiblock, the coordinate-boundary multiblock $\mathcal H_k$, and the subdivision-boundary exception, every point of $B$ lies in a trimmed cell belonging to either the core or the exterior multiblock of the $k$th witness.
	Indeed, a full trimmed cell disjoint from all three witness supports would be a positive-volume part of $B$ outside the measure-contained support of that witness.
	Outside $\bar{\mathcal E}$, each summand is therefore approximated within $\eps/N$.
	Summing over $k$ gives the stated estimate.
\end{proof}

\begin{definition}[Full regular function]
	\label{dfn_full_regular_function}
	A regular function is called \emph{full regular} if it admits a representation
	\[
		f=\sum_{k=1}^N f_k
	\]
	with supports $X_1,\ldots,X_N$ such that:
	\begin{enumerate}
		\item at least one support $X_k$ is full representable,
		\item for every $\eps>0$, the supports admit representability witnesses
		such that the finite multiblock obtained by concatenating all constituent blocks of
		\[
			\mathcal M_{X_1},\mathcal E_{X_1},
			\ldots,
			\mathcal M_{X_N},\mathcal E_{X_N}
		\]
		is well-connected.
	\end{enumerate}
\end{definition}

\begin{remark}
	The supports in \Cref{dfn_full_regular_function} are already pairwise disjoint by \Cref{dfn_regular_function}.
	The concatenated multiblock records the pieces where the supports are present or not yet separated from their boundaries.
	Requiring it to be well-connected allows supports to touch through exception multiblocks, but rules out degenerate junctions where the covering is connected only through lower-dimensional contacts.
\end{remark}

\begin{definition}[$\mathcal G$-full regular representation]
	\label{dfn_G_full_regular_representation}
	Let $\mathcal G$ be a finite full multiblock.
	A representation of a regular function
	\[
		f=\sum_{k=1}^N f_k
	\]
	with supports $X_1,\ldots,X_N$ is called \emph{$\mathcal G$-full} if, for every $\eps>0$, the supports admit representability witnesses such that
	\begin{equation}
		\label{eqn_G_full_regular_representation}
		\bar{\mathcal G}
		\Subset
		\bigcup_{k=1}^N
		\left(
			\bar{\mathcal M}_{X_k}
			\cup
			\bar{\mathcal E}_{X_k}
		\right).
	\end{equation}
	Here the union on the right-hand side is the support of the finite multiblock obtained by concatenating the indicated core and exception multiblocks.
\end{definition}

\begin{remark}
	\label{rem_G_full_regular_representation}
	The condition in \Cref{eqn_G_full_regular_representation} is stronger than measurement containment.
	It records finite multiblock coverage of $\bar{\mathcal G}$ by the core and exception layers of the chosen support witnesses.
	It does not by itself provide locatedness, total boundedness of the supports as metric spaces, or a supremum norm for arbitrary comparisons on $\bar{\mathcal G}$.
\end{remark}

\begin{definition}[Locally regular function]
	\label{dfn_locally_regular_function}
	A function $f$ is called \emph{locally regular} if, for every full block $B$, there is a regular function $\hat f_B$ such that $f$ and $\hat f_B$ agree on their common determined part in $B$.
	The local domain on $B$ is taken to be
	\[
		B\cap\dom{\hat f_B}.
	\]
\end{definition}

\begin{remark}
	\label{rem_regular_function_flattened_by_default}
	When a formula first produces overlapping supports, the regular representation is obtained by finite flattening.
	Empty, degenerate, and zero-valued pieces may remain in the finite data.
	They are harmless redundant summands and no emptiness decision is used to remove them.
\end{remark}

\begin{definition}[Lipschitz-regular function]
	\label{dfn_lipschitz_regular_function}
	A regular function
	\[
		f=\sum_{k=1}^N f_k
	\]
	is called \emph{$L$-Lipschitz-regular} if every positive branch
	\[
		f_k^+:X_k\to\R^m
	\]
	is $L$-Lipschitz on its support $X_k$.
\end{definition}

\begin{lemma}[Finite multiblock interpolation]
	\label{lem_finite_multiblock_interpolation}
	Let $B$ be a block, let $\mathcal A$ be a finite full multiblock with $\bar{\mathcal A}\subseteq B$, and let $h:\bar{\mathcal A}\to\R^m$ be continuous.
	For every $\eta>0$, there exists a continuous function $H:B\to\R^m$ such that
	\[
		\forall x\in\bar{\mathcal A}\spc \nrm{H(x)-h(x)}\le\eta.
	\]
\end{lemma}

\begin{proof}
	If $\mathcal A$ is empty, take $H=0$.
	Otherwise, choose $\delta>0$ such that
	\[
		x,y\in\bar{\mathcal A},\quad \nrm{x-y}\le 3\delta
		\quad\Longrightarrow\quad
		\nrm{h(x)-h(y)}\le\eta .
	\]
	Take a finite $\delta$-net $p_1,\ldots,p_S$ in $\bar{\mathcal A}$ by meshing the finitely many constituent blocks of $\mathcal A$.
	Set
	\[
		\phi_s(x):=\max\{0,2\delta-\nrm{x-p_s}\},
		\qquad
		\psi(x):=\max\{0,d(x,\bar{\mathcal A})-\delta/2\}.
	\]
	The denominator $\psi(x)+\sum_s\phi_s(x)$ is positive on $B$:
	if $d(x,\bar{\mathcal A})>\delta/2$, then $\psi(x)>0$, while if $d(x,\bar{\mathcal A})\le\delta/2$, a net point $p_s$ satisfies $\nrm{x-p_s}\le3\delta/2$ and hence $\phi_s(x)>0$.
	Define
	\[
		H(x):=
		\frac{\sum_s\phi_s(x)h(p_s)}
		{\psi(x)+\sum_s\phi_s(x)}.
	\]
	This function is continuous on $B$.
	If $x\in\bar{\mathcal A}$, then $\psi(x)=0$ and every index with $\phi_s(x)>0$ satisfies $\nrm{x-p_s}<2\delta$.
	Thus every value $h(p_s)$ participating in the convex combination defining $H(x)$ is within $\eta$ of $h(x)$, which gives the claim.
\end{proof}

\begin{lemma}[Regular functions are measurable]
	\label{lem_regular_functions_measurable}
	Every regular function is measurable in the sense of \Cref{dfn_locally_measurable_function}.
\end{lemma}

\begin{proof}
	Let $f=\sum_{k=1}^N f_k$ be a regular function with representable supports $X_1,\ldots,X_N$.
	Fix an ambient block $B$ be a block and $\eps>0$.
	If $N=0$, take $g_{\eps\mid B}=0$ and take the empty exception multiblock.
	Then the measurability conditions are immediate.
	Hence assume $N>0$.
	For each $k$, take a representability witness for $X_k$ with exception cost at most $\eps/(2N)$ and with bounding block containing $B$.
	This is obtained, if necessary, by enlarging the bounding block and adding the newly created annulus to the exterior multiblock.
	Write the witness as
	\[
		B_{X_k},\qquad
		\mathcal E_{X_k}, \quad \mathcal M_{X_k}, \quad \mathcal M_{X_k}^c .
	\]
	For each $k$, apply \Cref{dfn_coordinate_refinement} to $B$ and to the finite multiblock
	\[
		\mathcal A_k
		:=
		\mathcal M_{X_k}
		\sqcup
		\mathcal E_{X_k}
		\sqcup
		\mathcal M_{X_k}^c.
	\]
	Choose an admissible $\rho_k>0$ so small that
	\[
		\gamma(\mathcal H_{\rho_k}(B;\mathcal A_k))
		\le
		\frac{\eps}{2N}.
	\]
	Set
	\[
		\mathcal H_k:=\mathcal H_{\rho_k}(B;\mathcal A_k),
		\qquad
		\mathcal C_k:=\mathcal R_{\rho_k}(B;\mathcal A_k).
	\]
	Define the exception multiblock on $B$ by concatenating the exception and coordinate-boundary multiblocks:
	\begin{equation}
		\label{eqn_regular_measurable_exception_multiblock}
		\mathcal E_{\eps\mid B}
		:=
		\bigsqcup_{k=1}^N
		(\mathcal E_{X_k}\sqcup\mathcal H_k).
	\end{equation}
	Then
	\begin{equation}
		\label{eqn_regular_measurable_exception_cost}
		\gamma(\mathcal E_{\eps\mid B})
		\le \sum_{k=1}^N
		\left(
			\gamma(\mathcal E_{X_k})
			+
			\gamma(\mathcal H_k)
		\right)
		\le \eps.
	\end{equation}

	We construct the comparison function piecewise.
	For each $k$, set
	\[
		A_k:=\bar{\mathcal M}_{X_k},
		\qquad
		Z_k:=\overline{\mathcal M_{X_k}^c}.
	\]
	By \Cref{eqn_representable_set_interior,eqn_representable_set_exterior}, $A_k\Subset X_k$ and $Z_k\Subset B_{X_k}\setminus X_k$.
	Since both $A_k$ and $Z_k$ are supports of finite multiblocks, they are uniformly apart.
	If $A_k$ is empty, use $\theta_k\equiv0$; if $Z_k$ is empty, use $\theta_k\equiv1$.
	Otherwise define the scalar cutoff
	\begin{equation}
		\label{eqn_regular_measurable_cutoff}
		\theta_k(x)
		:=
		\frac{d(x,Z_k)}
		{d(x,A_k)+d(x,Z_k)}.
	\end{equation}
	The denominator is positive everywhere because $A_k$ and $Z_k$ are uniformly apart.
	The function $\theta_k$ is continuous, satisfies $0\le\theta_k\le1$, and obeys
	\begin{equation}
		\label{eqn_regular_measurable_cutoff_values}
		\theta_k=1 \ \text{on } A_k,
		\qquad
		\theta_k=0 \ \text{on } Z_k.
	\end{equation}
	By \Cref{lem_finite_multiblock_interpolation}, applied to the branch $f_k^+:X_k\to\R^m$ restricted to $A_k$, choose a continuous function $H_{k,\eps}:B\to\R^m$ such that
	\begin{equation}
		\label{eqn_regular_measurable_branch_interpolation}
		\forall x\in A_k \spc
		\nrm{H_{k,\eps}(x)-f_k^+(x)}\le \eps/N.
	\end{equation}
	Now set
	\begin{equation}
		\label{eqn_regular_measurable_piece}
		g_{k,\eps}(x) := \theta_k(x) H_{k,\eps}(x),
		\qquad x\in B.
	\end{equation}
	This is continuous because $H_{k,\eps}$ and $\theta_k$ are continuous.
	Finally, define
	\begin{equation}
		\label{eqn_regular_measurable_sum}
		g_{\eps\mid B} := \sum_{k=1}^N g_{k,\eps}.
	\end{equation}
	Thus $g_{\eps\mid B}$ is continuous on $B$.

	It remains to compare $g_{\eps\mid B}$ and $f$ away from $\bar{\mathcal E}_{\eps\mid B}$.
	Fix $k$.
	Let $x\in B\cap\bar{\mathcal E}_{\eps\mid B}^c$.
	Since $x\notin\bar{\mathcal H}_k$, the construction of the trimmed coordinate refinement gives a cell $C\in\mathcal C_k$ with $x\in C$.
	For every constituent block of the $k$th witness, either $C$ is well contained in that block or it is disjoint from that block.
	If $C$ were disjoint from $A_k$, $\bar{\mathcal E}_{X_k}$, and $Z_k$, then $C$ would be a positive-volume part of $B$ outside the measure-contained support in \Cref{eqn_representable_set_decomposition}.
	Hence $C$ is well contained in one of these three supports.
	Because $x\notin\bar{\mathcal E}_{X_k}$, the exception case is impossible.
	Thus $x\in A_k\cup Z_k$.
	On $A_k$, \Cref{eqn_regular_measurable_cutoff_values,eqn_regular_measurable_branch_interpolation} give
	\[
		\nrm{g_{k,\eps}(x)-f_k^+(x)}\le\eps/N.
	\]
	On $Z_k$, the same cutoff equation gives $g_{k,\eps}(x)=0$, while the exterior branch of the bounded-support complemented function gives $f_k^-(x)=0$.
	Hence, for the fixed $x\in B\cap\bar{\mathcal E}_{\eps\mid B}^c$,
	\[
		\nrm{g_{k,\eps}(x)-f_k(x)}\le\eps/N.
	\]
	Summing over $k$, the restrictions of $g_{\eps\mid B}$ and $f$ satisfy
	\begin{equation}
		\label{eqn_regular_function_measurable}
		\nrm{f(x)-g_{\eps\mid B}(x)} \le \eps,
	\end{equation}
	for all $x\in B\cap\bar{\mathcal E}_{\eps\mid B}^c$.
	\Cref{eqn_regular_measurable_exception_cost} gives the required bound on the exception multiblock.
	Since $B$ was arbitrary, $f$ is locally measurable.
\end{proof}

\begin{corollary}
	\label{cor_full_regular_measurable}
	Every full regular function is measurable.
\end{corollary}

\begin{proof}
	Every full regular function is regular, so the claim follows from \Cref{lem_regular_functions_measurable}.
\end{proof}

\begin{corollary}
	\label{cor_locally_regular_locally_measurable}
	Every locally regular function is locally measurable.
\end{corollary}

\begin{proof}
	Let $f$ be locally regular and fix a block $B$.
	By \Cref{dfn_locally_regular_function}, the restriction of $f$ to $B$ is represented by a regular function.
	By \Cref{lem_regular_functions_measurable}, this restriction is measurable on $B$.
	Therefore $f$ is locally measurable.
\end{proof}

\subsection{Differentiation and Integration}

\subsubsection{Differentiation}
\label{subsubsec_differentiation}

\begin{definition}[Operator norm]
	\label{dfn_operator_norm}
	For a linear map $A:\R^n\to\R^m$, its operator norm is
	\[
	\nrm{A}_{\mathrm{op}}
	:=
	\sup_{\nrm{v}\le 1}\nrm{Av}.
	\]
	After choosing the standard bases, $A$ is an $m\times n$ matrix.
\end{definition}

\begin{remark}
	\label{rem_operator_norm_constructive}
	The operator norm of a rational matrix is constructively computable.
	Indeed, $\nrm{A}_{\mathrm{op}}^2$ is the supremum of the continuous function
	\[
	v\mapsto \nrm{Av}^2
	\]
	on the unit ball of $\R^n$, and the latter ball has explicit finite rational nets at every precision.
	Equivalently, for rational $q>0$ the estimate $\nrm{A}_{\mathrm{op}}\le q$ can be checked by the positive semidefiniteness of $q^2I-A^\top A$, where $I$ is the identity matrix of the appropriate dimension.
	One could use the Frobenius norm instead, since it is also computable and satisfies $\nrm{A}_{\mathrm{op}}\le\nrm{A}_{F}$.
	We keep the operator norm because it is the exact Lipschitz bound of the corresponding linear map, whereas the Frobenius norm would impose a stronger certificate than needed.
\end{remark}

\begin{definition}[Upper right derivative]
	\label{dfn_upper_right_derivative}
	Let $J\subset\R$ be an interval, let $\phi:J\to\R$, and let $t\in J$.
	The \emph{upper right derivative} of $\phi$ at $t$ is denoted by $\Dplus\phi(t)$ and is defined by
	\[
		\Dplus\phi(t)
		:=
		\limsup_{h\downarrow0}
		\frac{\phi(t+h)-\phi(t)}{h},
	\]
	whenever the displayed upper limit is available.
	In estimates, the assertion
	\[
		\Dplus\phi(t)\le a
	\]
	means that, for every $\eta>0$, all sufficiently small positive $h$ satisfy
	\[
		\frac{\phi(t+h)-\phi(t)}{h}\le a+\eta.
	\]
\end{definition}

\begin{remark}
	\label{rem_upper_right_derivative_viability}
	The upper right derivative is used below only as a scalar comparison tool.
	In \Cref{subsec_viable_solutions}, it records the one-sided growth estimate for the distance of a trajectory to a constraint set.
	No differentiability of that distance function is asserted.
	The statement $\Dplus\phi(t)\le a$ is just the quantified upper difference-quotient estimate in \Cref{dfn_upper_right_derivative}, and the subsequent scalar comparison is the usual finite Gronwall argument applied to these one-sided estimates.
\end{remark}

\begin{definition}[Differentiability at a point]
	\label{dfn_differentiability_point}
	Let $X\subseteq\R^n$ and let $f:X\to\R^m$.
	The function $f$ is called \emph{differentiable at $x\in X$} if there exists a linear map
	\[
		A_x:\R^n\to\R^m
	\]
	and a modulus $\omega_{\Diff f,x}:\Q_{>0}\to\Q_{>0}$ such that, for every $\eps>0$ and every $y\in X$,
	\[
		0<\nrm{y-x}\le \omega_{\Diff f,x}(\eps)
		\quad\Longrightarrow\quad
		\nrm{f(y)-f(x)-A_x(y-x)}
		\le
		\eps\nrm{y-x}.
	\]
	The derivative is denoted by $\Diff f(x):=A_x$.
\end{definition}

\begin{definition}[Differentiable and continuously differentiable functions]
	\label{dfn_differentiable_functions}
	Let $X\subseteq\R^n$.
	A function $f:X\to\R^m$ is called \emph{differentiable on $X$} if it is differentiable at every $x\in X$ and the derivative operation $x\mapsto\Diff f(x)$ is given as part of the data.
	It is called \emph{continuously differentiable on $X$} if, additionally, the derivative map is continuous on $X$ with respect to the operator norm.
\end{definition}

\begin{notation}[Gradient]
	\label{ntn_gradient}
	If $f:X\to\R$ is differentiable at $x$, its gradient is the vector $\nabla f(x)\in\R^n$ determined by
	\[
		\Diff f(x)v=\nabla f(x)^\top v,
		\qquad v\in\R^n.
	\]
\end{notation}

\subsubsection{Integration}
\label{subsubsec_integration}

\begin{definition}[Tagged multiblock and mesh]
	\label{dfn_tagged_partition_mesh}
	Let $B$ be a full block.
	A \emph{partitioning multiblock} of $B$ is a finite full multiblock
	\[
		\mathcal P=\{B_k\}_{k=1}^N
	\]
	such that $\bar{\mathcal P}=B$ and the interiors of $B_1,\ldots,B_N$ are pairwise disjoint.
	A \emph{tagged multiblock} is such a multiblock together with points $\xi_k\in B_k$.
	Its \textit{mesh size} is
	\[
		|\mathcal P|:=\max_{k=1,\ldots,N}\diam(B_k).
	\]
	For a bounded function $f:B\to\R^m$, the corresponding tagged sum is
	\[
		\Sigma(f,\mathcal P):=
		\sum_{k=1}^N f(\xi_k)\mu(B_k).
	\]
\end{definition}

\begin{definition}[Riemann integrable function on a block]
	\label{dfn_riemann_integrable_block}
	Let $B$ be a full block and let $f:B\to\R^m$ be bounded.
	The function $f$ is called \emph{Riemann integrable on $B$} if its tagged multiblock sums are Cauchy with a given modulus:
	for every $\eps>0$ there exists $\delta>0$ such that any two tagged multiblocks $\mathcal P,\mathcal Q$ partitioning $B$ with
	\[
		|\mathcal P|\le\delta,
		\qquad
		|\mathcal Q|\le\delta
	\]
	satisfy
	\[
		\nrm{\Sigma(f,\mathcal P)-\Sigma(f,\mathcal Q)}\le\eps.
	\]
\end{definition}

\begin{definition}[Riemann integral on a block]
	\label{dfn_riemann_integral_block}
	Let $B$ be a full block and let $f:B\to\R^m$ be Riemann integrable on $B$.
	The Riemann integral
	\[
		\int_B f(x)\diff x
	\]
	is the limit of the tagged multiblock sums $\Sigma(f,\mathcal P)$ as $|\mathcal P|\to0$.
\end{definition}

\begin{remark}
	\label{rem_riemann_integral_continuous_block}
	Every continuous function on a block is Riemann integrable in the sense of \Cref{dfn_riemann_integrable_block}.
	Indeed, let $V:=\mu(B)$ and let $\omega_f$ be a modulus of continuity on $B$.
	For $\eps>0$, choose
	\[
		\delta\le \omega_f\left(\frac{\eps}{2V}\right)
	\]
	when $V>0$; if $V=0$, the integral is $0$.
	After passing two tagged multiblocks to a common refinement, points lying in the same refined block have distance at most $\delta$.
	Hence the corresponding function values differ by at most $\eps/(2V)$, and summing over the refined blocks gives
	\[
		\nrm{\Sigma(f,\mathcal P)-\Sigma(f,\mathcal Q)}\le\eps.
	\]
\end{remark}

\begin{definition}[Integral over bounded support]
	\label{dfn_integral_bounded_support}
	Let $f:\R^n\to\R^m$ be a function with bounded support in the sense that there is a block $B$ such that $f(x)=0$ for $x\notin B$.
	If $f_{\mid B}$ is Riemann integrable, set
	\[
		\int_{\R^n} f(x)\diff x
		:=
		\int_B f(x)\diff x.
	\]
	This value is independent of the chosen witnessing block, since enlarging the block only adds regions on which the function is zero.
\end{definition}


\begin{definition}[Improper integral]
	\label{dfn_improper_integral}
	Let $f:\R^n\to\R^m$ and write
	\[
		B_R:=[-R,R]^n
	\]
	for rational $R>0$.
	The improper integral $\int_{\R^n} f(x)\diff x$ exists if $f_{\mid B_R}$ is Riemann integrable for every rational $R>0$ and the net of partial integrals is Cauchy:
	for every $\eps>0$ there exists $R_0>0$ such that
	\[
		R,S\ge R_0
		\quad\Longrightarrow\quad
		\nrm{\int_{B_R}f(x)\diff x-\int_{B_S}f(x)\diff x}\le\eps.
	\]
	The value of the improper integral is the corresponding limit.
\end{definition}

\begin{definition}[Regular functions on a block]
	\label{dfn_regular_functions_on_block}
	Let $B$ be a full block.
	Denote by $\Reg{B}$ the class of real-valued regular functions whose supports are well contained in $B$.
	Denote by $\BReg{B}$ the real vector space of effectively bounded members of $\Reg{B}$.
	Denote by $\SReg{B}$ the subspace of simple regular functions in $\BReg{B}$.
\end{definition}

\begin{definition}[Locally effectively bounded regular functions]
	\label{dfn_global_bounded_regular_functions}
	A real-valued function $f$ on $\R^n$ belongs to $\BReg{\R^n}$ if, for every full block $B\subset\R^n$, there are a full block $C$ and a function $f_C\in\BReg{C}$ such that
	\[
		B\Subset C
	\]
	and $f$ agrees with $f_C$ on the determined part of $B$.
	The domain is understood locally: on the window $B$, the determined domain of $f$ is inherited from
	\[
		B\cap\dom{f_C}.
	\]
	The function is called \emph{locally bounded away from zero} if, in addition, for every full block $B$, there are full blocks $X$ and $C$, a local representative $f_C\in\BReg{C}$, and a rational number $c_B>0$ such that
	\[
		B\Subset X\Subset C,
	\]
	the function $f$ agrees with $f_C$ on the determined part of $X$, and
	\[
		\abs{f_C(x)}\ge c_B
	\]
	on the determined part of $X$.
\end{definition}

\begin{remark}
	\label{rem_global_breg_localization}
	The notation $\BReg{\R^n}$ does not mean that a single finite support representation is given on all of $\R^n$.
	It means that every bounded window has finite regular data in a slightly larger block.
	In particular, the local supports, exterior branches, and possible undetermined boundary pieces are part of the local representative.
	This is the same localization principle as in \Cref{dfn_locally_regular_function}, with effective bound witnesses included locally.
\end{remark}

\begin{remark}[Locally representable domains]
	\label{rem_global_breg_locally_representable_domain}
	Let $D\subset\R^n$ be locally representable, and let $f:D\to\R$ be locally regular.
	If, for every full block $B$, the local representative on $B$ may be chosen effectively bounded, then the zero extension of $f$ outside $D$ is an element of $\BReg{\R^n}$ whenever the local exterior branch is supplied by the local representability witnesses for $D$.
	Thus a function naturally defined on a locally representable effective domain may still be integrated over the ambient Euclidean space after supplying the zero exterior data locally.

	The word ``determined'' does not introduce a second kind of domain.
	It records the usual complemented-set convention: a value is evaluated where the local data provide the relevant membership or apartness witness.
	Near boundary regions without such a witness, the local representative keeps the corresponding boundary layer in its support data and the integration arguments ignore it through the same exception estimates as before.
\end{remark}

\begin{lemma}[Local integrability of $\BReg{\R^n}$ functions]
	\label{lem_global_breg_local_integrability}
	Let $f\in\BReg{\R^n}$ and let $B\subset\R^n$ be a full block.
	Then $f_{\mid B}$ is Riemann integrable on $B$.
\end{lemma}

\begin{proof}
	Choose $C$ and $f_C\in\BReg{C}$ as in \Cref{dfn_global_bounded_regular_functions}, with $B\Subset C$.
	The indicator $\indic{B}$ is a simple regular function whose support is well contained in $C$.
	By the finite algebra of regular functions, \Cref{lem_regular_functions_finite_algebra}, the product $\indic{B}f_C$ belongs to $\BReg{C}$.
	Therefore it is Riemann integrable on $C$ by \Cref{lem_regular_functions_riemann_integrable}.
	The integral over $B$ is defined by
	\[
		\int_Bf(x)\diff x
		:=
		\int_C\indic{B}(x)f_C(x)\diff x.
	\]
	If another local representative is chosen, pass to a common well-containing block and use the same finite additivity argument for regular integrals; the two products agree on $B$ and vanish on the apartness complement of $B$ inside the common block.
\end{proof}

\begin{lemma}[Local algebra of $\BReg{\R^n}$ functions]
	\label{lem_global_breg_local_algebra}
	If $f,g\in\BReg{\R^n}$ and $a,b\in\R$, then $af+bg$ and $fg$ belong to $\BReg{\R^n}$.
	If $g$ is locally bounded away from zero, then $1/g\in\BReg{\R^n}$.
\end{lemma}

\begin{proof}
	Fix a full block $B$ and choose a common full block $C$ well containing $B$ on which local representatives for all functions under discussion are given.
	The linear and product constructions are exactly the blockwise constructions of \Cref{lem_regular_functions_finite_algebra} applied in $C$.
	For the reciprocal, choose the witness
	\[
		B\Subset X\Subset C,
		\qquad
		\abs{g_C(x)}\ge c_B>0
		\quad(x\in X\cap\dom g_C)
	\]
	from local boundedness away from zero.
	By \Cref{lem_regular_functions_finite_algebra},
	\[
		\indic{X}/g_C\in\BReg{C}.
	\]
	This quotient agrees with $1/g$ on the determined part of $B$.
	Since the construction is available for every full block $B$, the resulting function belongs to $\BReg{\R^n}$.
\end{proof}

\begin{notation}
	\label{ntn_wedges}
	For real-valued functions $f$ and $g$, write
	\[
	f\land g,\qquad f\lor g
	\]
	for the pointwise minimum and maximum, respectively, on the common domain of $f$ and $g$.
	If $c\in\R$, then $f\land c$ and $f\lor c$ mean $f\land(c\one)$ and $f\lor(c\one)$ where $\one$ is the function which is identically $1$.
\end{notation}

\begin{definition}[Integral of a simple regular function]
	\label{dfn_integral_simple_regular_function}
	Let $B$ be a full block, and let
	\[
		s=\sum_{k=1}^N c_k\indic{X_k}
	\]
	be a real-valued simple regular function whose supports satisfy $X_k\Subset B$.
	Its integral over $B$ is
	\[
		\int_B s(x)\diff x
		:=
		\sum_{k=1}^N c_k\mu(X_k),
	\]
	where $\mu(X_k)$ is the measure from \Cref{dfn_representable_set_measure}.
\end{definition}

\begin{lemma}[Effectively bounded regular functions are Riemann integrable]
	\label{lem_regular_functions_riemann_integrable}
	Let $B$ be a full block.
	Every effectively bounded regular function $f:B\to\R^m$ whose supports are well contained in $B$ is Riemann integrable on $B$.
\end{lemma}

\begin{proof}
	Fix $\eps>0$ and let $V:=\mu(B)$.
	If $V=0$, all tagged sums are zero.
	Assume $V>0$.
	Choose a representation of $f$ and a bound $M>0$ for $f$ on $B$.
	Enlarge $M$, if necessary, so that it also bounds the continuous comparison functions produced by the construction in \Cref{lem_regular_functions_measurable}.
	This is possible because those comparison functions are built from cutoffs between $0$ and $1$ and finite multiblock interpolants of the bounded branch data in the fixed representation.
	Apply \Cref{lem_regular_functions_measurable} with accuracy
	\[
		\beta:=
		\min\left\{
			\frac{\eps}{8V},
			\frac{\eps}{32(M+1)}
		\right\}
	\]
	and set
	\[
		\alpha:=\frac{\eps}{8V}.
	\]
	The construction in its proof gives a continuous function $g:B\to\R^m$ and an exception multiblock $\mathcal E$ such that $g$ is bounded on $B$ and
	\[
		\forall x\in B\cap\bar{\mathcal E}^c
		\quad
		\nrm{f(x)-g(x)}\le\alpha.
	\]
	Moreover,
	\[
		\gamma(\mathcal E)
		\le
		\frac{\eps}{32(M+1)}.
	\]
	Since $\mathcal E$ is finite, choose $\delta_0>0$ so small that
	\[
		\gamma(\mathcal E\oplus\delta_0)
		\le
		\gamma(\mathcal E)
		+
		\frac{\eps}{32(M+1)}.
	\]
	By \Cref{rem_riemann_integral_continuous_block}, choose $\delta\le\delta_0$ such that any two tagged multiblock sums of $g$ with mesh at most $\delta$ differ by at most $\eps/2$.
	Let $\mathcal P$ be such a tagged multiblock.
	The union of all blocks of $\mathcal P$ whose tags lie in $\bar{\mathcal E}$ is contained in $\bar{\mathcal E\oplus\delta_0}$.
	Therefore
	\[
	\begin{aligned}
		\nrm{\Sigma(f,\mathcal P)-\Sigma(g,\mathcal P)}
		&\le
		\alpha V
		+
		2M\gamma(\mathcal E\oplus\delta_0)
		\\
		&\le
		\frac{\eps}{8}
		+
		\frac{\eps}{8}
		\le
		\frac{\eps}{4}.
	\end{aligned}
	\]
	Hence, for any two tagged multiblocks $\mathcal P,\mathcal Q$ with mesh at most $\delta$,
	\[
		\nrm{\Sigma(f,\mathcal P)-\Sigma(f,\mathcal Q)}
		\le
		\frac{\eps}{4}
		+
		\frac{\eps}{2}
		+
		\frac{\eps}{4}
		=
		\eps.
	\]
	Thus the tagged sums of $f$ are Cauchy, so $f$ is Riemann integrable on $B$.
\end{proof}

\begin{definition}[Integral of an effectively bounded regular function]
	\label{dfn_integral_regular_function}
	Let $B$ be a full block, and let $f:B\to\R^m$ be an effectively bounded regular function whose supports are well contained in $B$.
	Define
	\[
		\int_B f(x)\diff x
	\]
	to be the Riemann integral given by \Cref{lem_regular_functions_riemann_integrable}.
	For real-valued simple regular functions this agrees with \Cref{dfn_integral_simple_regular_function}.
\end{definition}

\begin{remark}
	\label{rem_regular_integral_extends_simple}
	The agreement in \Cref{dfn_integral_regular_function} follows from the same exception-multiblock estimate used in \Cref{lem_regular_functions_riemann_integrable}.
	For a simple summand $c\indic{X}$, representability of $X$ gives a core multiblock and an exterior multiblock whose uncovered part has arbitrarily small cost.
	Tagged sums of $c\indic{X}$ therefore differ from $c\mu(X)$ only by the contribution of the exception layer, and finite additivity gives the stated formula for finite sums.
\end{remark}

\begin{example}[Bounded measurable data need not be Riemann integrable]
	\label{ex_bounded_measurable_not_riemann_integrable}
	In the classical reading, let $S\subset[0,1]$ be the Smith--Volterra--Cantor set from \Cref{ex_located_not_representable}.
	The indicator $\indic{S}$ is bounded and Lebesgue measurable.
	It is not Riemann integrable on $[0,1]$ because it is discontinuous at every point of $S$, and $S$ has positive Lebesgue measure.
	Thus bounded measurable functions are not automatically integrable in the Riemann sense used in \Cref{dfn_riemann_integrable_block}.
\end{example}

\subsection{Closure of Regular Functions}
\label{subsec_regular_function_closure}

\begin{remark}
	\label{rem_flattening_regular_functions}
	Flattening is the set operation of \Cref{dfn_representable_set_flattening}.
	When it is applied to supports appearing in a regular-function formula, it changes only the support data.
	The complemented summands already carry their positive and exterior parts, so no pointwise localization rule is added.
	Repeated, empty, degenerate, or zero-valued support pieces may be left in the finite data.
\end{remark}

\begin{lemma}[Finite algebra of regular functions]
	\label{lem_regular_functions_finite_algebra}
	Let $B$ be a full block.
	If $f,g\in\BReg{B}$ and $a,b\in\R$, then $af+bg\in\BReg{B}$ and
	\[
		\int_B(af+bg)
		=
		a\int_Bf+b\int_Bg.
	\]
	If $f_1,\ldots,f_r\in\BReg{B}$ are scalar regular functions, then every finite sum and finite product formed from them is again in $\BReg{B}$.
	If, in addition, after flattening the support data of a scalar regular function $g\in\BReg{B}$ one obtains formulas
	\[
		g=\sum_I g_I\indic{Y_I}
	\]
	such that
	\[
		\abs{g_I(x)}\ge c>0,
		\qquad
		x\in Y_I,
	\]
	on every support piece on which a numerator $f\in\BReg{B}$ may be nonzero, then $f/g\in\BReg{B}$.
	If the same bound holds on all support pieces of the flattened support data for $g$ on $B$, then the reciprocal $1/g$ belongs to $\BReg{B}$ on the represented support.
\end{lemma}

\begin{proof}
	For linearity, flatten the combined finite representation of the summands of $f$ and $g$ over $B$.
	The resulting supports are pairwise disjoint and representable by \Cref{lem_representable_set_flattening}.
	The positive part assigned to each support is the corresponding finite linear combination of the original positive parts, hence is continuous.
	Supports on which no original summand is present carry the zero formula and may remain only as redundant support data.
	The tagged-sum identity
	\[
		\Sigma(af+bg,\mathcal P)
		=
		a\Sigma(f,\mathcal P)+b\Sigma(g,\mathcal P)
	\]
	holds for every tagged multiblock $\mathcal P$, so the integral identity follows by passing to the Riemann limit.
	For products, write
	\[
		f=\sum_{i=1}^N f_i,
		\qquad
		g=\sum_{j=1}^M g_j
	\]
	with supports $X_i,Y_j\Subset B$.
	Flatten the finite family
	\[
		X_1,\ldots,X_N,Y_1,\ldots,Y_M
	\]
	over $B$.
	The resulting supports are pairwise disjoint and representable by \Cref{lem_representable_set_flattening}.
	For each flattened support, attach the product of the two finite sums of the original positive parts whose supports are present there.
	This formula is continuous, and it may be the zero formula.
	Thus $fg$ is regular.
	Finite products follow by induction, and finite sums by the linear part.
	For a polynomial expression, one may equivalently flatten the finitely many displayed supports and attach the polynomial formula to the corresponding support piece.
	If $\abs{g_I}\ge c>0$ on every support piece where the numerator formula may be nonzero, then $1/g_I$ is continuous there with the usual reciprocal modulus.
	Thus the formula $f_I/g_I$ represents $f/g$ on the corresponding support piece.
	The reciprocal is the special case with numerator $\one$ on the represented support.
\end{proof}

\begin{corollary}[Lattice closure]
	\label{cor_regular_functions_lattice_closure}
	Let $B$ be a full block and let $f,g\in\BReg{B}$ be scalar.
	Then
	\[
		f\land g,\qquad
		f\lor g,\qquad
		\abs{f}
	\]
	belong to $\BReg{B}$.
\end{corollary}

\begin{proof}
	Flatten fixed representations of $f$ and $g$ over $B$.
	The representing functions attached to the flattened supports are continuous.
	The pointwise operations $\min$, $\max$, and absolute value preserve continuity of scalar functions.
	Supports carrying the zero formula outside all supports may remain as refinement data and are not needed as summands.
	The finite representation is therefore regular by \Cref{lem_representable_set_flattening}.
\end{proof}

\begin{lemma}[Finite order and positivity]
	\label{lem_regular_integral_finite_order}
	Let $B$ be a full block and let $u,v\in\BReg{B}$.
	If $u\le v$ pointwise on $B$, then
	\[
		\int_Bu\le\int_Bv.
	\]
	In particular, if $0\le u\le c$, then
	\[
		0
		\le
		\int_Bu
		\le
		c\mu(B).
	\]
\end{lemma}

\begin{proof}
	For every tagged multiblock $\mathcal P$ partitioning $B$,
	\[
		\Sigma(u,\mathcal P)\le\Sigma(v,\mathcal P).
	\]
	Let $I_u$ and $I_v$ denote the two integrals.
	For an arbitrary $\eta>0$, choose a mesh bound $\delta>0$ such that every tagged multiblock $\mathcal P$ with mesh at most $\delta$ satisfies
	\[
		\abs{I_u-\Sigma(u,\mathcal P)}\le\eta/2,
		\qquad
		\abs{I_v-\Sigma(v,\mathcal P)}\le\eta/2.
	\]
	For such a common tagged multiblock $\mathcal P$ one has
	\[
		I_u
		\le
		\Sigma(u,\mathcal P)+\eta/2
		\le
		\Sigma(v,\mathcal P)+\eta/2
		\le
		I_v+\eta.
	\]
	Since $\eta>0$ is arbitrary, this is precisely the constructive order estimate $I_u\le I_v$.
	The second claim follows by applying the first claim to $0\le u\le c\one$.
\end{proof}

\begin{definition}[Rectangularly regular function]
	\label{dfn_rectangularly_regular_function}
	Let $A\subset\R^n$ and $C\subset\R^m$ be full blocks.
	A scalar function $h:A\times C\to\R$ is called \emph{rectangularly regular} if it has a finite representation
	\[
		h(x,y)=\sum_{q=1}^N h_q(x,y)\indic{X_q}(x)\indic{Y_q}(y),
	\]
	where $X_q\Subset A$ and $Y_q\Subset C$ are representable and each branch $h_q:X_q\times Y_q\to\R$ is continuous.
	By \Cref{lem_representable_product}, every support $X_q\times Y_q$ is representable, so a rectangularly regular function is a regular function on the product block $A\times C$.
\end{definition}

\begin{remark}
	\label{rem_rectangularly_regular_classical}
	Classically, a rectangularly regular function is analogous to a finite sum of continuous functions carried by measurable rectangles.
	The word ``rectangular'' records that each support splits as an $x$-support times a $y$-support.
	This is stronger than being merely regular on $A\times C$, but it is exactly the structure that makes the inner integrals in \Cref{thm_fubini_rectangular_regular} regular again.
\end{remark}

\begin{theorem}[Fubini for rectangularly regular functions]
	\label{thm_fubini_rectangular_regular}
	Let $A\subset\R^n$ and $C\subset\R^m$ be full blocks.
	Let $h:A\times C\to\R$ be effectively bounded and rectangularly regular.
	Then the functions
	\[
		H_A(x):=\int_C h(x,y)\diff y,
		\qquad
		H_C(y):=\int_A h(x,y)\diff x,
	\]
	are regular on $A$ and $C$, respectively, and
	\[
		\int_{A\times C}h(x,y)\diff(x,y)
		=
		\int_A H_A(x)\diff x
		=
		\int_C H_C(y)\diff y.
	\]
\end{theorem}

\begin{proof}
	By linearity it is enough to consider one summand
	\[
		h(x,y)=\phi(x,y)\indic{X}(x)\indic{Y}(y).
	\]
	For fixed $x\in X$, the function $y\mapsto\phi(x,y)\indic{Y}(y)$ is effectively bounded regular on $C$ and hence integrable by \Cref{lem_regular_functions_riemann_integrable}.
	Set
	\[
		\Phi(x):=\int_C \phi(x,y)\indic{Y}(y)\diff y,
		\qquad x\in X.
	\]
	The map $\Phi$ is continuous on $X$.
	Indeed, uniform continuity of $\phi$ on the bounded set $X\times Y$ gives, for nearby $x,x'\in X$,
	\[
		\abs{\phi(x,y)-\phi(x',y)}
	\]
	uniformly small on $Y$, and \Cref{lem_regular_integral_finite_order} bounds the resulting integral difference by this uniform bound times $\mu(Y)$.
	Thus
	\[
		H_A=\Phi\indic{X}
	\]
	is regular on $A$.
	The same argument gives regularity of $H_C$.
	To prove the equality of the integrals, first take representability witnesses for $X$ and $Y$.
	On the product of the corresponding core multiblocks, the branch $\phi$ is continuous on a finite union of full product blocks, where ordinary tagged-product Riemann sums satisfy Fubini.
	The cells involving an exception multiblock have total product cost bounded by
	\[
		\gamma(\mathcal E_X)\mu(C)+\gamma(\mathcal E_Y)\mu(A).
	\]
	Multiplying this cost by the bound witness for $\phi$ makes the contribution of those cells arbitrarily small.
	Letting the witness accuracies tend to zero gives the displayed identity for one summand.
	Finite additivity and linearity complete the proof.
\end{proof}

\begin{corollary}[Regular kernel action]
	\label{cor_regular_kernel_action}
	Let $A\subset\R^n$ and $C\subset\R^m$ be full blocks.
	Let $K:A\times C\to\R$ be effectively bounded and rectangularly regular, and let $\rho\in\BReg{A}$.
	Then
	\[
		y\mapsto\int_A K(x,y)\rho(x)\diff x
	\]
	is an effectively bounded regular function on $C$.
\end{corollary}

\begin{proof}
	Write
	\[
		K(x,y)
		=
		\sum_{q=1}^N K_q(x,y)\indic{X_q}(x)\indic{Y_q}(y),
		\qquad
		\rho(x)=\sum_{i=1}^M\rho_i(x).
	\]
	Let $R_i$ be the support of $\rho_i$.
	By \Cref{lem_representable_finite_intersection}, each $X_q\cap R_i$ is representable.
	Hence
	\[
		K(x,y)\rho(x)
		=
		\sum_{q=1}^N\sum_{i=1}^M
		K_q(x,y)\rho_i^+(x)\indic{X_q\cap R_i}(x)\indic{Y_q}(y)
	\]
	is rectangularly regular on $A\times C$.
	The claim follows from \Cref{thm_fubini_rectangular_regular}.
	A bound witness is obtained by multiplying bound witnesses for $K$ and $\rho$ and then by $\mu(A)$.
\end{proof}

\begin{remark}[Relation to Chan integration spaces]
	\label{rem_regular_chan_appendix}
	The finite calculus above is the part needed for finite algebraic manipulations of regular densities.
	For comparison, \Cref{app_chan_integration_space} records the stronger Chan integration-space statement of \cite{Chan2019FoundationsCon} and the Ye-style countable positivity lemma based on \cite[Lemma~6.1]{Ye2011StrictFinitism}.
	That material is useful for abstract Daniell-type extensions, but it is not needed for the finite closure and Fubini facts used here.
\end{remark}

\section{Optimal Functions}
\label{sec_optimal_functions}

\begin{remark}
	This section proves constructive extremum-value results for spaces of functions.
	The supremum-metric theorem uses finite interpolation and smoothing to obtain approximate minimizers for uniformly continuous functionals.
	The later $d_1$ and $d_2$ variants use the integration theory and impose finite support-geometry data to control the effect of moving discontinuity sets.
\end{remark}

\subsection{Function Spaces}

\begin{definition}[Supremum metric]
	\label{dfn_supremum_metric}
	Let $B$ be a block and let $f,g:B\to\R^m$ be continuous functions.
	Their supremum distance is
	\[
		d_\infty(f,g):=\sup_{x\in B}\nrm{f(x)-g(x)}.
	\]
\end{definition}

\begin{remark}
	\label{rem_supremum_metric_locally_regular}
	The supremum in \Cref{dfn_supremum_metric} exists constructively.
	Indeed,
	\[
		\psi(x):=\nrm{f(x)-g(x)}
	\]
	is continuous on the block $B$ and therefore has a modulus on $B$.
	A block has explicit finite rational nets at every precision, so the finite-net construction gives the supremum of $\psi$.
	If, in particular, $f$ and $g$ are $L$-Lipschitz, then
	\[
		\abs{\psi(x)-\psi(y)}
		\le \nrm{(f(x)-g(x))-(f(y)-g(y))}
		\le 2L\nrm{x-y}.
	\]
\end{remark}

\begin{definition}[Lipschitz-bounded function space]
	\label{dfn_lipschitz_bounded_function_space}
	Let $B$ be a block and let $L,K>0$ be rational.
	Denote by
	\[
		\mathcal F(B,L,K;m)
	\]
	the space of all Lipschitz-continuous functions $f:B\to\R^m$ such that
	\begin{equation}
		\label{eqn_lipschitz_space_lip}
		\Lip{B}(f)\le L,
	\end{equation}
	and
	\begin{equation}
		\label{eqn_lipschitz_space_bound}
		\forall x\in B\spc \nrm{f(x)}\le K.
	\end{equation}
	The metric on $\mathcal F(B,L,K;m)$ is $d_\infty$.
\end{definition}

\begin{definition}[Uniformly continuous functional]
	\label{dfn_uniformly_continuous_functional}
	Let $\mathcal F$ be a metric function space.
	A map $J:\mathcal F\to\R$ is called a \emph{uniformly continuous functional} if it is equipped with a modulus $\alpha_J:\Q_{>0}\to\Q_{>0}$ such that
	\[
		d_\infty(f,g)\le\alpha_J(\eps)
		\quad\Longrightarrow\quad
		\abs{J[f]-J[g]}\le\eps.
	\]
\end{definition}

\begin{definition}[Approximate minimizer]
	\label{dfn_approximate_minimizer}
	Let $J:\mathcal F\to\R$ and let $\eps>0$.
	A function $f_\eps\in\mathcal F$ is called an \emph{$\eps$-minimizer} of $J$ on $\mathcal F$ if
	\[
		\forall f\in\mathcal F\spc J[f_\eps]\le J[f]+\eps.
	\]
\end{definition}

\subsection{Smoothing}

\begin{definition}[Projection onto a block]
	\label{dfn_projection_block}
	Let $B=\prod_{k=1}^n[a_k,b_k]$ be a full block.
	The projection $\proj{\cdot}{B}:\R^n\to B$ is defined coordinatewise by
	\[
		(\proj{x}{B})_k:=\min\{b_k,\max\{a_k,x_k\}\}.
	\]
\end{definition}

\begin{lemma}
	\label{lem_projection_block_nonexpansive}
	The projection $\proj{\cdot}{B}$ is nonexpansive:
	\[
		\nrm{\proj{x}{B}-\proj{y}{B}}\le\nrm{x-y}.
	\]
\end{lemma}

\begin{proof}
	Each scalar clipping map $t\mapsto\min\{b_k,\max\{a_k,t\}\}$ is nonexpansive on $\R$.
	Summing the squared coordinate inequalities gives the claim.
\end{proof}

\begin{definition}[Mollifier]
	\label{dfn_mollifier}
	A \emph{mollifier} is a continuously differentiable function $\varphi:\R^n\to\R$ such that
	\[
		\varphi\ge0,\qquad
		\int_{\R^n}\varphi(z)\diff z=1,\qquad
		\supp{\varphi}\subseteq \{z:\nrm{z}\le1\},
	\]
	and the quantities
	\[
		\Lambda_\varphi:=\int_{\R^n}\nrm{\nabla\varphi(z)}\diff z,
		\qquad
		M_\varphi:=\sup_z\varphi(z)
	\]
	are computable.
	For $r>0$, set
	\[
		\varphi_r(z):=\frac{1}{r^n}\varphi\left(\frac{z}{r}\right).
	\]
	The integrals over $\R^n$ are understood in the bounded-support sense of \Cref{dfn_integral_bounded_support}; the support condition reduces them to integrals over any block containing the unit ball.
\end{definition}

\begin{example}[Mollifiers]
	\label{ex_mollifiers}
	The standard radial example is
	\[
		\varphi(z)=c_n
		\begin{cases}
			\exp\left(-\frac{1}{1-\nrm{z}^2}\right), & \nrm{z}<1,\\
			0, & \nrm{z}\ge 1,
		\end{cases}
	\]
	where $c_n$ is chosen so that $\int_{\R^n}\varphi(z)\diff z=1$.
	A product example is obtained from any nonnegative $C^1$ function $\theta:\R\to\R$ supported in $[-1/\sqrt n,1/\sqrt n]$ with $\int_{\R}\theta(t)\diff t=1$ by setting
	\[
		\varphi(z)=\prod_{k=1}^n\theta(z_k).
	\]
	In both cases the support is contained in the unit ball and the constants required in \Cref{dfn_mollifier} are computable from the defining data.
\end{example}

\begin{definition}[Smoothed contraction]
	\label{dfn_smoothed_contraction}
	Let $f\in\mathcal F(B,L,K;m)$ and let $0<\rho<1$.
	First extend and contract $f$ by
	\[
		\bar f_\rho(x):=(1-\rho)f(\proj{x}{B}).
	\]
	For $r>0$, define
	\[
		S_{\rho,r}[f](x)
		:=
		(\varphi_r*\bar f_\rho)(x)
		\triangleq
		\int_{\R^n}\varphi_r(z)\bar f_\rho(x-z)\diff z.
	\]
	The integrand is continuous in $z$ and has bounded support contained in $\{z:\nrm{z}\le r\}$, so the integral exists by \Cref{rem_riemann_integral_continuous_block,dfn_integral_bounded_support}.
\end{definition}

\begin{lemma}[Differentiating a compactly supported convolution]
	\label{lem_differentiate_compact_convolution}
	Let $\psi:\R^n\to\R$ be continuously differentiable with bounded support, and let $h:\R^n\to\R^m$ be continuous and bounded on blocks.
	Define
	\[
		U(x):=\int_{\R^n}\psi(z)h(x-z)\diff z.
	\]
	Then $U$ is continuously differentiable and
	\[
		\Diff U(x)
		=
		\int_{\R^n}h(x-z)\nabla\psi(z)^\top\diff z.
	\]
	Here $\nabla\psi$ is understood as in \Cref{ntn_gradient}.
\end{lemma}

\begin{proof}
	Using the change of variables $y=x-z$, write
	\[
		U(x)=\int_{\R^n}\psi(x-y)h(y)\diff y.
	\]
	Fix a block $B_0$ for the variable $x$.
	Since $\psi$ has bounded support, there exists a block $Q$ such that, for all $x\in B_0$, the functions $y\mapsto\psi(x-y)h(y)$ and $y\mapsto h(y)\nabla\psi(x-y)^\top$ vanish outside $Q$.
	The functions
	\[
		y\mapsto \psi(x-y)h(y),
		\qquad
		y\mapsto h(y)\nabla\psi(x-y)^\top
	\]
	are continuous on $Q$, hence Riemann integrable by \Cref{rem_riemann_integral_continuous_block}.
	Let
	\[
		M_Q:=\sup_{y\in Q}\nrm{h(y)},
		\qquad
		V_Q:=\mu(Q).
	\]
	If $M_QV_Q=0$ the claim is immediate, so assume $M_QV_Q>0$.
	Let $p\in\R^n$ be small.
	By differentiability of $\psi$ with continuous derivative, uniformly on the block containing all points $x-y$ and $x+p-y$ with $x\in B_0$ and $y\in Q$, we have
	\[
		\nrm{\psi(x+p-y)-\psi(x-y)-\nabla\psi(x-y)^\top p}
		\le
		\frac{\eps}{M_QV_Q}\nrm{p}
	\]
	whenever $\nrm{p}$ is below a suitable modulus.
	Therefore
	\[
		\nrm{
		U(x+p)-U(x)
		-
		\int_Q h(y)\nabla\psi(x-y)^\top p\diff y
		}
		\le
		\eps\nrm{p}.
	\]
	This is precisely differentiability on $B_0$, with
	\[
		\Diff U(x)p=
		\int_Q h(y)\nabla\psi(x-y)^\top p\diff y.
	\]
	Since the equality is linear in $p$, the displayed derivative operator is
	\[
		\Diff U(x)=\int_Q h(y)\nabla\psi(x-y)^\top\diff y.
	\]
	Changing variables back to $z=x-y$ gives the formula stated in the lemma.
	Finally, if $x,x'\in B_0$ are close, uniform continuity of $\nabla\psi$ on the same block gives
	\[
		\nrm{\nabla\psi(x-y)-\nabla\psi(x'-y)}
	\]
	uniformly small for $y\in Q$.
	Multiplication by the bounded function $h$ and integration over $Q$ therefore make $\nrm{\Diff U(x)-\Diff U(x')}_{\mathrm{op}}$ small.
	Hence $x\mapsto\Diff U(x)$ is continuous on $B_0$, and $B_0$ was arbitrary.
\end{proof}

\begin{lemma}[Smoothing estimates]
	\label{lem_smoothing_estimates}
	Let $f\in\mathcal F(B,L,K;m)$ and set $u:=S_{\rho,r}[f]$.
	Then
	\begin{align}
		\label{eqn_smoothing_lip}
		\nrm{u(x)-u(y)} &\le (1-\rho)L\nrm{x-y},
		\\
		\label{eqn_smoothing_error}
		\nrm{u(x)-(1-\rho)f(x)} &\le (1-\rho)Lr \quad (x\in B),
		\\
		\label{eqn_smoothing_bound}
		\nrm{u(x)} &\le (1-\rho)K.
	\end{align}
	Moreover, $u$ is continuously differentiable and its derivative satisfies the uniform modulus
	\begin{equation}
		\label{eqn_smoothing_derivative_modulus}
		\nrm{\Diff u(x)-\Diff u(y)}_{\mathrm{op}}
		\le
		(1-\rho)L\Lambda_\varphi \frac{1}{r}\nrm{x-y}.
	\end{equation}
\end{lemma}

\begin{proof}
	The map $\bar f_\rho$ is $(1-\rho)L$-Lipschitz by \Cref{lem_projection_block_nonexpansive}.
	Thus
	\[
		\nrm{u(x)-u(y)}
		\le
		\int\varphi_r(z)\nrm{\bar f_\rho(x-z)-\bar f_\rho(y-z)}\diff z
		\le
		(1-\rho)L\nrm{x-y}.
	\]
	For $x\in B$ and $z\in\supp{\varphi_r}$, one has $\nrm{x-(x-z)}\le r$ and hence
	\[
		\nrm{\bar f_\rho(x-z)-\bar f_\rho(x)}
		\le
		(1-\rho)Lr.
	\]
	Integrating gives \Cref{eqn_smoothing_error}.
	The bound \Cref{eqn_smoothing_bound} follows from convexity of the norm and from $\int\varphi_r=1$.
	By \Cref{lem_differentiate_compact_convolution}, differentiation under the integral gives
	\[
		\Diff u(x)
		=
		\int \bar f_\rho(x-z)\nabla\varphi_r(z)^\top\diff z.
	\]
	Since $\int\nabla\varphi_r(z)\diff z=0$, this may be written with $\bar f_\rho(x-z)-\bar f_\rho(x)$ in place of $\bar f_\rho(x-z)$.
	Consequently,
	\[
		\nrm{\Diff u(x)-\Diff u(y)}_{\mathrm{op}}
		\le
		(1-\rho)L\nrm{x-y}\int\nrm{\nabla\varphi_r(z)}\diff z
		=
		(1-\rho)L\Lambda_\varphi \frac{1}{r}\nrm{x-y}.
	\]
\end{proof}

\subsection{Piecewise Affine Approximation}

\begin{definition}[Simplex]
	\label{dfn_simplex}
	Let $x_0,\ldots,x_n\in\R^n$.
	The simplex generated by these points is
	\[
		S:=\conv{\{x_0,\ldots,x_n\}}.
	\]
	It is called \emph{full} if the vectors $x_1-x_0,\ldots,x_n-x_0$ are linearly independent.
	In that case every $x\in S$ has unique barycentric coordinates $\lambda_0,\ldots,\lambda_n\ge0$ satisfying
	\[
		\sum_{j=0}^n\lambda_j=1,
		\qquad
		x=\sum_{j=0}^n\lambda_jx_j.
	\]
	They are obtained from
	\[
		(\lambda_1,\ldots,\lambda_n)^T
		=
		[x_1-x_0,\ldots,x_n-x_0]^{-1}(x-x_0),
		\qquad
		\lambda_0=1-\sum_{j=1}^n\lambda_j.
	\]
	The simplex is called rational if its vertices are rational vectors.
\end{definition}

\begin{definition}[Rational triangulation]
	\label{dfn_rational_triangulation}
	A \emph{rational triangulation} $\mathcal T$ of a full block $B$ is a finite sequence of full simplices with rational vertices whose union is $B$ and whose pairwise intersections are common faces.
	Its mesh size is
	\[
		|\mathcal T|:=\max_{S\in\mathcal T}\diam(S).
	\]
\end{definition}

\begin{definition}[Regular rational triangulation scheme]
	\label{dfn_regular_rational_triangulation_scheme}
	A \emph{regular rational triangulation scheme} assigns to every full block $B$ and every rational mesh parameter $h>0$ a rational triangulation $\mathcal T_h$ of $B$ such that:
	\begin{enumerate}
		\item $|\mathcal T_h|\le C_{B,n}h$ for a constant $C_{B,n}$ depending only on $B$ and $n$,
		\item for each simplex $S=\conv{\{x_0,\ldots,x_n\}}$ in $\mathcal T_h$, with $V_S=[x_1-x_0,\ldots,x_n-x_0]$, one has
		\[
			\nrm{V_S^{-1}}_{\mathrm{op}}\le \frac{C'_{B,n}}{h}
		\]
		for a constant $C'_{B,n}$ depending only on $B$ and $n$,
		\item only finitely many simplex shapes occur after rescaling by $h$.
	\end{enumerate}
	Such schemes are obtained by subdividing $B$ into rational subblocks with side lengths comparable to $h$ and splitting each subblock by a fixed ordering of its coordinate directions.
\end{definition}

\begin{definition}[Piecewise affine map]
	\label{dfn_rational_pl_map}
	Let $\mathcal T$ be a rational triangulation of $B$.
	A map $P:B\to\R^m$ is called \emph{piecewise affine} on $\mathcal T$ if it is affine on every simplex of $\mathcal T$.
	It is called \emph{rational piecewise affine} if, additionally, all its vertex values are rational vectors.
	If $S=\conv{\{x_0,\ldots,x_n\}}$, define
	\[
		V_S:=[x_1-x_0,\ldots,x_n-x_0],
		\qquad
		W_S:=[P(x_1)-P(x_0),\ldots,P(x_n)-P(x_0)].
	\]
	The derivative of the affine restriction of $P$ to $S$ is
	\[
		A_S:=W_SV_S^{-1}.
	\]
\end{definition}

\begin{definition}[Affine interpolation on a triangulation]
	\label{dfn_affine_interpolation_triangulation}
	Let $\mathcal T$ be a rational triangulation of $B$ and let $u:B\to\R^m$.
	The affine interpolation $\mathcal I_{\mathcal T}[u]$ is the piecewise affine map on $\mathcal T$ which agrees with $u$ at every vertex of $\mathcal T$.
	For the triangulation $\mathcal T_h$ from \Cref{dfn_regular_rational_triangulation_scheme}, write $\mathcal I_h[u]:=\mathcal I_{\mathcal T_h}[u]$.
\end{definition}

\begin{definition}[Lipschitz certificate]
	\label{dfn_lipschitz_certificate}
	A piecewise affine map $P$ on $\mathcal T$ has an \emph{$L$-Lipschitz certificate} if
	\[
		\forall S\in\mathcal T\spc \nrm{A_S}_{\mathrm{op}}\le L.
	\]
\end{definition}

\begin{remark}
	\label{rem_operator_norm_certificate_decidable}
	For rational $A_S$ and rational $L$, the condition $\nrm{A_S}_{\mathrm{op}}\le L$ is a finite rational check.
	Indeed, it is equivalent to the positive semidefiniteness of the rational symmetric matrix
	\[
		L^2I-A_S^\top A_S,
	\]
	where $I$ is the identity matrix of the appropriate dimension.
	This is checked by the finitely many principal minors.
\end{remark}

\begin{lemma}
	\label{lem_pl_certificate_lipschitz}
	Let $\mathcal T$ be a rational triangulation of a full block $B$.
	If a piecewise affine map $P:B\to\R^m$ has an $L$-Lipschitz certificate on $\mathcal T$, then $P$ is $L$-Lipschitz on $B$.
\end{lemma}

\begin{proof}
	On a simplex $S$, the affine formula gives
	\[
		P(y)-P(x)=A_S(y-x)
	\]
	whenever $x,y\in S$.
	Thus the certificate gives
	\[
		\nrm{P(y)-P(x)}\le L\nrm{y-x}.
	\]
	We first prove the estimate for rational points $\tilde x,\tilde y\in B$.
	Parametrize the segment by
	\[
		\ell(t):=(1-t)\tilde x+t\tilde y,
		\qquad
		t\in[0,1].
	\]
	Let
	\[
		H_s=\{z\in\R^n:a_s\cdot z=b_s\},
		\qquad
		s=1,\ldots,R,
	\]
	be a finite list of rational affine hyperplanes containing all facets of all simplexes in $\mathcal T$.
	For each $s$, define
	\[
		g_s(t):=a_s\cdot\ell(t)-b_s
		=
		\alpha_s+\beta_s t,
	\]
	where
	\[
		\alpha_s:=a_s\cdot\tilde x-b_s,
		\qquad
		\beta_s:=a_s\cdot(\tilde y-\tilde x).
	\]
	All numbers $\alpha_s,\beta_s$ are rational.
	If $\beta_s\ne0$ and
	\[
		\tau_s:=-\frac{\alpha_s}{\beta_s}\in[0,1],
	\]
	include $\tau_s$ in the cutting list.
	The alternatives $\beta_s=0$, $\beta_s\ne0$, and $\tau_s\in[0,1]$ are decidable rational comparisons.
	If $\beta_s=0$, then $g_s$ is constant: either the whole segment lies in $H_s$, when $\alpha_s=0$, or it does not meet $H_s$, when $\alpha_s\ne0$.
	Thus no nonrational parameter is introduced.
	Collect these finitely many rational parameters, add $0$ and $1$, and sort the finite rational list as
	\[
		0=t_0<t_1<\cdots<t_N=1.
	\]
	For $t\in(t_{q-1},t_q)$, no affine function $g_s$ with nonzero slope vanishes.
	Therefore the segment portion $\ell((t_{q-1},t_q))$ crosses no facet transversally.
	Let $S,S'\in\mathcal T$ be two simplexes sharing a facet $F=S\cap S'$.
	Write the affine restrictions as
	\[
		P_{\mid S}(z)=A_Sz+c_S,
		\qquad
		P_{\mid S'}(z)=A_{S'}z+c_{S'}.
	\]
	Because $P$ is a map on the triangulation, the two formulas agree at every common vertex $v$ of the facet $F$:
	\[
		A_Sv+c_S=A_{S'}v+c_{S'}.
	\]
	If
	\[
		z=\sum_{r=1}^n\lambda_rv_r,
		\qquad
		\lambda_r\ge0,
		\qquad
		\sum_{r=1}^n\lambda_r=1,
	\]
	is a point of $F=\conv{\{v_1,\ldots,v_n\}}$, then affinity gives
	\[
		A_Sz+c_S
		=
		\sum_{r=1}^n\lambda_r(A_Sv_r+c_S)
		=
		\sum_{r=1}^n\lambda_r(A_{S'}v_r+c_{S'})
		=
		A_{S'}z+c_{S'}.
	\]
	Thus the restrictions of adjacent affine formulas agree on their common facet.
	Consequently, after assigning endpoints to adjacent closed simplexes, each closed subsegment
	\[
		[\ell(t_{q-1}),\ell(t_q)]
	\]
	is evaluated by one affine formula, or by the common restriction of finitely many adjacent affine formulas on a shared face.
	Applying the simplexwise estimate and summing,
	\[
		\nrm{P(\tilde y)-P(\tilde x)}
		\le
		\sum_{q=1}^N L\nrm{\ell(t_q)-\ell(t_{q-1})}.
	\]
	Since
	\[
		\ell(t_q)-\ell(t_{q-1})
		=
		(t_q-t_{q-1})(\tilde y-\tilde x),
	\]
	we get
	\[
		\sum_{q=1}^N\nrm{\ell(t_q)-\ell(t_{q-1})}
		=
		\left(\sum_{q=1}^N(t_q-t_{q-1})\right)
		\nrm{\tilde y-\tilde x}
		=
		\nrm{\tilde y-\tilde x}.
	\]
	Hence
	\[
		\nrm{P(\tilde y)-P(\tilde x)}
		\le
		L\nrm{\tilde y-\tilde x}.
	\]
	Now let $x,y\in B$ be arbitrary and let $\eta>0$.
	The finite piecewise affine map $P$ has a modulus of continuity on the block $B$.
	Choose $\delta>0$ such that
	\[
		\nrm{u-v}\le\delta
		\quad\Longrightarrow\quad
		\nrm{P(u)-P(v)}\le\eta/3,
	\]
	and also $2L\delta\le\eta/3$.
	By density of rational points in the rational block $B$, with coordinatewise clipping to $B$ if necessary, choose rational points $\tilde x,\tilde y\in B$ with
	\[
		\nrm{x-\tilde x}\le\delta,
		\qquad
		\nrm{y-\tilde y}\le\delta.
	\]
	Then
	\[
	\begin{aligned}
		\nrm{P(y)-P(x)}
		&\le
		\nrm{P(y)-P(\tilde y)}
		+
		\nrm{P(\tilde y)-P(\tilde x)}
		+
		\nrm{P(\tilde x)-P(x)}
		\\
		&\le
		\frac{\eta}{3}
		+
		L\nrm{\tilde y-\tilde x}
		+
		\frac{\eta}{3}
		\\
		&\le
		L\nrm{y-x}
		+
		\eta.
	\end{aligned}
	\]
	Since $\eta>0$ is arbitrary, $\nrm{P(y)-P(x)}\le L\nrm{y-x}$.
\end{proof}

\begin{lemma}[Fixed-mesh interpolation of smoothed functions]
	\label{lem_fixed_mesh_interpolation}
	For every full block $B$, there are constants $C_{B,n},D_{B,n}>0$ with the following property.
	Let $u:B\to\R^m$ be continuously differentiable, let
	\[
		\nrm{\Diff u(x)}_{\mathrm{op}}\le L_0,
		\qquad
		\nrm{\Diff u(x)-\Diff u(y)}_{\mathrm{op}}\le \Omega\nrm{x-y}
	\]
	for all $x,y\in B$.
	For every rational mesh parameter $h>0$, let $\mathcal T_h$ be given by a regular rational triangulation scheme and let $\mathcal I_h[u]$ be the affine interpolation of $u$ on $\mathcal T_h$.
	Then
	\begin{align}
		\label{eqn_fixed_mesh_interpolation_lip}
		\Lip{B}(\mathcal I_h[u])
		&\le L_0+D_{B,n}\Omega h,
		\\
		\label{eqn_fixed_mesh_interpolation_err}
		d_\infty(\mathcal I_h[u],u)
		&\le L_0C_{B,n}h.
	\end{align}
\end{lemma}

\begin{proof}
	Fix a simplex $S=\conv{\{x_0,\ldots,x_n\}}$ of $\mathcal T_h$.
	For $j=1,\ldots,n$,
	\[
		u(x_j)-u(x_0)=\Diff u(x_0)(x_j-x_0)+e_j,
	\]
	where
	\[
		\nrm{e_j}
		\le
		\Omega\diam(S)\nrm{x_j-x_0}.
	\]
	With $V_S,W_S$ as in \Cref{dfn_rational_pl_map}, this gives
	\[
		W_S=\Diff u(x_0)V_S+E_S.
	\]
	Since $|\mathcal T_h|\le C_{B,n}h$, the column estimates give
	\[
		\nrm{E_S}_{\mathrm{op}}\le C''_{B,n}\Omega h^2
	\]
	for a constant $C''_{B,n}$.
	Together with
	\[
		\nrm{V_S^{-1}}_{\mathrm{op}}\le C'_{B,n}/h
	\]
	this gives
	\[
		\nrm{E_SV_S^{-1}}_{\mathrm{op}}\le D_{B,n}\Omega h
	\]
	for a constant $D_{B,n}$.
	Therefore the derivative of the affine restriction of $\mathcal I_h[u]$ to $S$ has norm at most $L_0+D_{B,n}\Omega h$.
	\Cref{lem_pl_certificate_lipschitz} gives \Cref{eqn_fixed_mesh_interpolation_lip}.
	For the interpolation error, let $\lambda_0,\ldots,\lambda_n$ be the barycentric coordinates of $x\in S$, computed as in \Cref{dfn_simplex}.
	Then
	\[
		\nrm{\mathcal I_h[u](x)-u(x)}
		\le
		\sum_j\lambda_j\nrm{u(x_j)-u(x)}
		\le
		L_0\diam(S)
		\le
		L_0C_{B,n}h.
	\]
\end{proof}

\begin{lemma}[Rational rounding on a fixed triangulation]
	\label{lem_rational_rounding_triangulation}
	For every full block $B$, there is a constant $R_{B,n}>0$ with the following property.
	Let $\mathcal T_h$ be given by a regular rational triangulation scheme with mesh parameter $h$.
	Let $u:B\to\R^m$.
	For every vertex $v$ of $\mathcal T_h$, choose a rational vector $r_v\in\R^m$ such that
	\[
		\nrm{r_v-u(v)}\le q.
	\]
	Let $P$ be the rational piecewise affine map with vertex values $P(v):=r_v$.
	Then
	\begin{align}
		\label{eqn_rounding_error}
		d_\infty(P,\mathcal I_h[u])&\le q,
		\\
		\label{eqn_rounding_lip_error}
		\Lip{B}(P-\mathcal I_h[u])&\le R_{B,n}\frac{q}{h}.
	\end{align}
\end{lemma}

\begin{proof}
	Fix a simplex $S=\conv{\{x_0,\ldots,x_n\}}$ and a point $x\in S$.
	Let $\lambda_0,\ldots,\lambda_n$ be the barycentric coordinates of $x$, computed as in \Cref{dfn_simplex}.
	For each vertex $x_j$ put
	\[
		\rho_j:=P(x_j)-u(x_j).
	\]
	Then $\nrm{\rho_j}\le q$ and
	\[
		P(x)-\mathcal I_h[u](x)
		=
		\sum_j\lambda_j\rho_j.
	\]
	Since $\lambda_j\ge0$ and $\sum_j\lambda_j=1$, convexity of the norm gives
	\[
		\nrm{P(x)-\mathcal I_h[u](x)}
		\le
		\sum_j\lambda_jq
		=
		q.
	\]
	For the second estimate, consider one simplex $S=\conv{\{x_0,\ldots,x_n\}}$.
	The derivative of $P-\mathcal I_h[u]$ on $S$ is
	\[
		\Delta W_SV_S^{-1},
	\]
	where the $j$th column of $\Delta W_S$ is
	\[
		\rho_j-\rho_0.
	\]
	Therefore every column of $\Delta W_S$ has norm at most $2q$.
	The column bound gives
	\[
		\nrm{\Delta W_S}_{\mathrm{op}}\le C'''_{n,m}q
	\]
	for a constant $C'''_{n,m}$ depending only on the dimensions.
	Using again $\nrm{V_S^{-1}}_{\mathrm{op}}\le C'_{B,n}/h$, we obtain
	\[
		\nrm{\Delta W_SV_S^{-1}}_{\mathrm{op}}
		\le
		R_{B,n}\frac{q}{h}
	\]
	for a constant $R_{B,n}$.
\end{proof}

\subsection{Total Boundedness}

\begin{theorem}[Finite piecewise affine approximation of Lipschitz-bounded functions]
	\label{thm_lipschitz_function_space_totally_bounded}
	Let $B$ be a full block, and let $L,K>0$ be rational.
	Let $\mathcal F(B,L,K;m)$ be the space of all functions $f:B\to\R^m$ satisfying
	\[
		\Lip{B}(f)\le L,
		\qquad
		\nrm{f(x)}\le K
		\quad(x\in B),
	\]
	with the metric $d_\infty$.
	For every $\eps>0$, there exists a finite set $\mathcal P_\eps$ of rational piecewise affine maps $P:B\to\R^m$ such that:
	\begin{enumerate}
		\item every $P\in\mathcal P_\eps$ belongs to $\mathcal F(B,L,K;m)$,
		\item for every $f\in\mathcal F(B,L,K;m)$, there exists $P\in\mathcal P_\eps$ such that
		\[
			d_\infty(f,P)\le\eps.
		\]
	\end{enumerate}
\end{theorem}

\begin{proof}
	Fix $\eps>0$.
	Choose $0<\rho<1$ so that $\rho K\le\eps/4$.
	Choose $r>0$ so that $(1-\rho)Lr\le\eps/4$.
	Set
	\[
		\Omega:=(1-\rho)L\Lambda_\varphi\frac{1}{r}.
	\]
	Choose a rational $h>0$ so that
	\[
		D_{B,n}\Omega h\le \rho L/4,
		\qquad
		(1-\rho)LC_{B,n}h\le\eps/4.
	\]
	Finally choose a rational $q>0$ so that
	\[
		q\le\eps/4,\qquad
		q\le\rho K/2,\qquad
		R_{B,n}q/h\le\rho L/4.
	\]
	Let $\mathcal T_h$ be the rational triangulation of $B$ produced by the chosen regular triangulation scheme with mesh parameter $h$.
	Let $\mathcal Y_q$ be a finite rational $q$-net of the closed ball $\{y\in\R^m:\nrm{y}\le K\}$.
	Define $\mathcal P_\eps$ to be the finite set of all rational piecewise affine maps on $\mathcal T_h$ whose vertex values lie in $\mathcal Y_q$ and whose simplexwise matrices satisfy the certificate of \Cref{dfn_lipschitz_certificate}.
	This is a finite construction by \Cref{rem_operator_norm_certificate_decidable}.

	Let $f\in\mathcal F(B,L,K;m)$ and set $u=S_{\rho,r}[f]$.
	By \Cref{lem_smoothing_estimates}, $u$ is $(1-\rho)L$-Lipschitz, $\nrm{u}\le(1-\rho)K$, and $\Diff u$ has modulus $\Omega$.
	Let $\mathcal I_h[u]$ be its affine interpolant on $\mathcal T_h$.
	By \Cref{lem_fixed_mesh_interpolation},
	\[
		\Lip{B}(\mathcal I_h[u])
		\le
		(1-\rho)L+\rho L/4.
	\]
	At every vertex $v$ of $\mathcal T_h$, choose a rational vector $y_v\in\mathcal Y_q$ with $\nrm{y_v-u(v)}\le q$.
	This is possible because $\nrm{u(v)}\le(1-\rho)K$ and $q\le\rho K/2$.
	Let $P$ be the rational piecewise affine map with these vertex values.
	By \Cref{lem_rational_rounding_triangulation},
	\[
		\Lip{B}(P)
		\le
		(1-\rho)L+\rho L/4+\rho L/4
		\le L.
	\]
	Moreover, all vertex values of $P$ lie in the convex ball $\nrm{y}\le K$, hence $\nrm{P(x)}\le K$ for all $x\in B$.
	Thus $P\in\mathcal P_\eps$.

	Combining the contraction, smoothing, interpolation, and rounding estimates gives
	\[
		d_\infty(f,P)
		\le
		\rho K
		+
		(1-\rho)Lr
		+
		(1-\rho)LC_{B,n}h
		+
		q
		\le
		\eps.
	\]
\end{proof}

\begin{corollary}
	\label{cor_lipschitz_function_space_totally_bounded}
	Let $B$ be a full block, and let $L,K>0$ be rational.
	Let $\mathcal F(B,L,K;m)$ be the space of all functions $f:B\to\R^m$ satisfying
	\[
		\Lip{B}(f)\le L,
		\qquad
		\nrm{f(x)}\le K
		\quad(x\in B),
	\]
	with the metric $d_\infty$.
	The space $\mathcal F(B,L,K;m)$ is totally bounded.
\end{corollary}

\begin{proof}
	This is exactly \Cref{thm_lipschitz_function_space_totally_bounded} in the terminology of \Cref{dfn_finite_net}.
\end{proof}

\subsection{Functional Extremum Value Theorem}

\begin{theorem}[Functional extremum value theorem]
	\label{thm_functional_evt}
	Let $B$ be a full block and let $L,K>0$ be rational.
	Let $\mathcal F(B,L,K;m)$ be the space of all functions $f:B\to\R^m$ satisfying
	\[
		\Lip{B}(f)\le L,
		\qquad
		\nrm{f(x)}\le K
		\quad(x\in B),
	\]
	with the metric $d_\infty$.
	Let
	\[
		J:\mathcal F(B,L,K;m)\to\R
	\]
	be a uniformly continuous functional.
	Then, for every $\eps>0$, there exists an $\eps$-minimizer $f_\eps\in\mathcal F(B,L,K;m)$.
\end{theorem}

\begin{proof}
	Let $\alpha_J$ be the modulus of $J$.
	Choose a finite $\alpha_J(\eps/3)$-net
	\[
		P_1,\ldots,P_N
	\]
	of $\mathcal F(B,L,K;m)$ using \Cref{cor_lipschitz_function_space_totally_bounded}.
	Compute rational approximations $a_i$ of $J[P_i]$ such that
	\[
		\abs{a_i-J[P_i]}\le\eps/6
	\]
	for all $i$.
	Choose an index $i_\ast$ such that $a_{i_\ast}\le a_i$ for every $i$, taking the least such index if several occur, and set $f_\eps:=P_{i_\ast}$.
	Let $f\in\mathcal F(B,L,K;m)$.
	Choose $P_j$ in the net with $d_\infty(f,P_j)\le\alpha_J(\eps/3)$.
	Then
	\[
		J[P_j]\le J[f]+\eps/3.
	\]
	By the choice of $i_\ast$,
	\[
		J[P_{i_\ast}]
		\le
		a_{i_\ast}+\eps/6
		\le
		a_j+\eps/6
		\le
		J[P_j]+\eps/3.
	\]
	Thus
	\[
		J[f_\eps]\le J[f]+\frac{2\eps}{3}\le J[f]+\eps.
	\]
	Since $f$ was arbitrary, $f_\eps$ is an $\eps$-minimizer.
\end{proof}

\begin{corollary}
	\label{cor_functional_evt_max}
	Under the assumptions of \Cref{thm_functional_evt}, every uniformly continuous functional $J$ admits an $\eps$-maximizer for every $\eps>0$.
\end{corollary}

\begin{proof}
	Apply \Cref{thm_functional_evt} to $-J$.
\end{proof}

\subsection{Examples for the Supremum Metric}

\begin{example}[Integral cost with uniformly Lipschitz integrand]
	\label{ex_sup_metric_integral_cost}
	Let $B$ be a full block and let
	\[
		\ell:B\times\{y\in\R^m:\nrm{y}\le K\}\to\R
	\]
	be continuous and Lipschitz in the second argument with constant $C_\ell$.
	For $f\in\mathcal F(B,L,K;m)$ put
	\[
		J[f]:=\int_B\ell(x,f(x))\diff x.
	\]
	Then
	\[
		\abs{J[f]-J[g]}
		\le
		C_\ell\mu(B)d_\infty(f,g),
	\]
	so \Cref{thm_functional_evt} applies to $J$.
\end{example}

\begin{example}[Finite sampling and terminal costs]
	\label{ex_sup_metric_sampling_cost}
	Fix points $x_1,\ldots,x_s\in B$ and a uniformly continuous function
	\[
		\Phi:\left(\R^m\right)^s\to\R
	\]
	on the product of the closed ball $\nrm{y}\le K$.
	The functional
	\[
		J[f]:=\Phi(f(x_1),\ldots,f(x_s))
	\]
	is uniformly continuous with respect to $d_\infty$.
	Thus the functional extremum value theorem covers finite terminal and sampling costs.
\end{example}

\begin{remark}[Relation to earlier applications]
	\label{rem_fevt_previous_applications}
	The construction above is the multidimensional version of the constructive extremum-value mechanism used for function-space optimization in \cite{Osinenko2018constructiveve,Osinenko2018Analysisextrem}.
	In particular, the finite-horizon optimal-control and dynamic-programming examples there fit the supremum-metric theorem when the admissible policy or value-function class is represented by a bounded Lipschitz ball.
	The present statement allows functions $B\to\R^m$ on full blocks $B\subset\R^n$.
\end{remark}

\subsection{Integral Metrics with Cone-Regular Supports}

\begin{definition}[Integral distances]
	\label{dfn_integral_distances}
	Let $B$ be a full block and let $f,g$ be vector-valued regular functions whose scalar components belong to $\BReg{B}$.
	Define
	\[
		d_1(f,g):=\int_B\nrm{f(x)-g(x)}\diff x
	\]
	and
	\[
		d_2(f,g):=
		\left(\int_B\nrm{f(x)-g(x)}^2\diff x\right)^{1/2}.
	\]
	The integrands are regular after flattening the finitely many supports and applying the continuous functions $y\mapsto\nrm{y}$ and $y\mapsto\nrm{y}^2$ on the resulting bounded branch ranges.
\end{definition}

\begin{definition}[Symmetric difference]
	\label{dfn_symmetric_difference}
	For sets $X,Y\subseteq B$, define their \emph{symmetric difference} by
	\[
		X\triangle Y
		:=
		(X\setminus Y)\cup(Y\setminus X).
	\]
	It is the part where the two supports disagree.
\end{definition}

\begin{remark}[Representability of symmetric differences]
	\label{rem_symmetric_difference_representable}
	If $X,Y\Subset B$ are representable, then $X\triangle Y$ is representable relative to $B$.
	Indeed,
	\[
		X\triangle Y
		=
		\left(X\cap(B\setminus Y)\right)
		\cup
		\left(Y\cap(B\setminus X)\right),
	\]
	where the complements are relative apartness complements inside $B$.
	These relative complements are representable by \Cref{lem_representable_relative_complement}, finite intersections are representable by \Cref{lem_representable_finite_intersection}, and finite unions are representable by \Cref{lem_representable_finite_union}.
	Thus $\mu(X\triangle Y)$ is the measure of a representable set, not an additional primitive.
\end{remark}

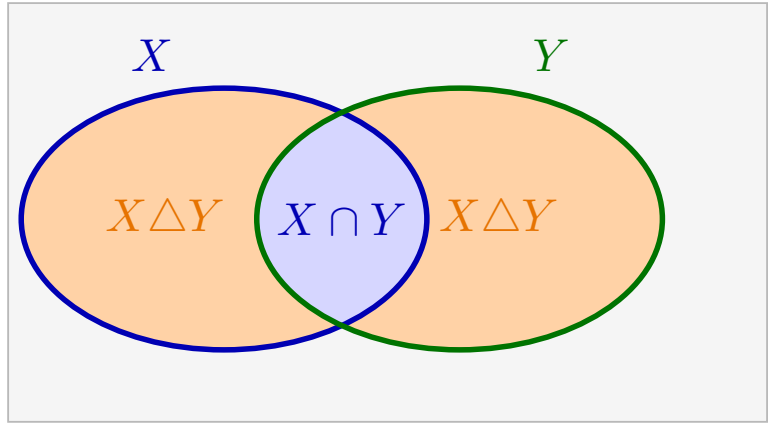
\begin{figure}[ht]
	\centering
	\resizebox{0.62\linewidth}{!}{\begin{tikzpicture}[x=1cm,y=1cm,font=\small]
	\fill[gray!8] (-0.4,-0.35) rectangle (5.4,2.85);
	\draw[gray!55] (-0.4,-0.35) rectangle (5.4,2.85);

	\begin{scope}
		\clip (1.25,1.2) ellipse (1.55 and 1.0);
		\fill[orange!35] (-0.4,-0.35) rectangle (5.4,2.85);
		\fill[white] (3.05,1.2) ellipse (1.55 and 1.0);
	\end{scope}
	\begin{scope}
		\clip (3.05,1.2) ellipse (1.55 and 1.0);
		\fill[orange!35] (-0.4,-0.35) rectangle (5.4,2.85);
		\fill[white] (1.25,1.2) ellipse (1.55 and 1.0);
	\end{scope}

	\begin{scope}
		\clip (1.25,1.2) ellipse (1.55 and 1.0);
		\clip (3.05,1.2) ellipse (1.55 and 1.0);
		\fill[blue!16] (-0.4,-0.35) rectangle (5.4,2.85);
	\end{scope}

	\draw[very thick,blue!70!black] (1.25,1.2) ellipse (1.55 and 1.0);
	\draw[very thick,green!45!black] (3.05,1.2) ellipse (1.55 and 1.0);

	\node[blue!70!black] at (0.70,2.45) {$X$};
	\node[green!45!black] at (3.75,2.45) {$Y$};
	\node[orange!90!black] at (0.8,1.20) {$X\triangle Y$};
	\node[orange!90!black] at (3.35,1.20) {$X\triangle Y$};	
	\node[blue!70!black] at (2.15,1.2) {$X\cap Y$};
\end{tikzpicture}}
	\caption{The symmetric difference $X\triangle Y$ is the union of the two one-sided discrepancies.}
	\label{fig_dfn_symmetric_difference}
\end{figure}

\begin{definition}[Finitely netted support sequence]
	\label{dfn_finitely_netted_support_sequence}
	Let $B$ be a full block.
	A support sequence
	\[
		\mathcal X=(X_i)_{i\in I}
	\]
	of representable subsets of $B$ is called \emph{finitely netted in measure} if, for every $\eta>0$, there are finitely many members
	\[
		Y_1^\eta,\ldots,Y_{M(\eta)}^\eta
	\]
	of the sequence such that, for every $i\in I$, some $Y_j^\eta$ satisfies
	\[
		\mu(X_i\triangle Y_j^\eta)\le\eta.
	\]
	Here $I$ is either a finite initial segment of $\N$ or $\N$ itself, and the sequence data include the routine which produces representability witnesses for $X_i$.
\end{definition}

\begin{figure}[ht]
	\centering
	\resizebox{0.78\linewidth}{!}{

\begin{tikzpicture}[x=1cm,y=1cm,font=\small]
	\fill[gray!8] (0,0) rectangle (8,4.3);
	\draw[gray!55] (0,0) rectangle (8,4.3);
	\node[gray!70!black] at (0.35,4.05) {$B$};
	
	\fill[blue!16] (0.7,0.8) rectangle (2.1,2.0);
	\draw[very thick,blue!70!black] (0.7,0.8) rectangle (2.1,2.0);
	
	\fill[blue!16]
	(3.0,0.55) -- (4.65,0.55) -- (4.65,1.95) -- (3.55,2.45) -- (3.0,1.75) -- cycle;
	\draw[very thick,blue!70!black]
	(3.0,0.55) -- (4.65,0.55) -- (4.65,1.95) -- (3.55,2.45) -- (3.0,1.75) -- cycle;
	
	\fill[blue!16] (6.1,1.25) ellipse (0.85 and 0.65);
	\draw[very thick,blue!70!black] (6.1,1.25) ellipse (0.85 and 0.65);
	
	\path[pattern=north east lines,pattern color=orange!45]
	(0.55,0.7) rectangle (2.25,2.13);
	\draw[orange!85!black,thick] (0.55,0.7) rectangle (2.25,2.13);
	
	\path[pattern=north east lines,pattern color=orange!45]
	(2.86,0.44) -- (4.76,0.44) -- (4.83,2.05) -- (3.48,2.63) -- (2.83,1.82) -- cycle;
	\draw[orange!85!black,thick]
	(2.86,0.44) -- (4.76,0.44) -- (4.83,2.05) -- (3.48,2.63) -- (2.83,1.82) -- cycle;
	
	\path[pattern=north east lines,pattern color=orange!45]
	(6.22,1.38) ellipse (0.93 and 0.73);
	\draw[orange!85!black,thick] (6.22,1.38) ellipse (0.93 and 0.73);
	
	\node[blue!70!black] at (1.4,1.25) {$Y_1^\eta$};
	\node[blue!70!black] at (3.82,1.45) {$Y_2^\eta$};
	\node[blue!70!black] at (6.1,1.25) {$Y_3^\eta$};		
	
	\node[orange!85!black,font=\scriptsize] at (2.35,2.33) {$X_{i_1}$};
	\node[orange!85!black,font=\scriptsize] at (5.10,2.22) {$X_{i_2}$};
	\node[orange!85!black,font=\scriptsize] at (7.5,1.55) {$X_{i_3}$};
	
	\node[align=center,font=\scriptsize] at (4.0,3.85)
	{finite list at scale $\eta$};
	\node[orange!85!black,align=center,font=\scriptsize] at (4.0,3.05)
	{sample supports $X_i$ are matched to\\ nearby finite representatives $Y_j^\eta$};
	
\end{tikzpicture}}
	\caption{A finitely netted support sequence need not itself be finite. At each precision $\eta$, finitely many representatives $Y_j^\eta$ approximate all members of the sequence in the measure of symmetric difference.}
	\label{fig_dfn_finitely_netted_support_sequence}
\end{figure}
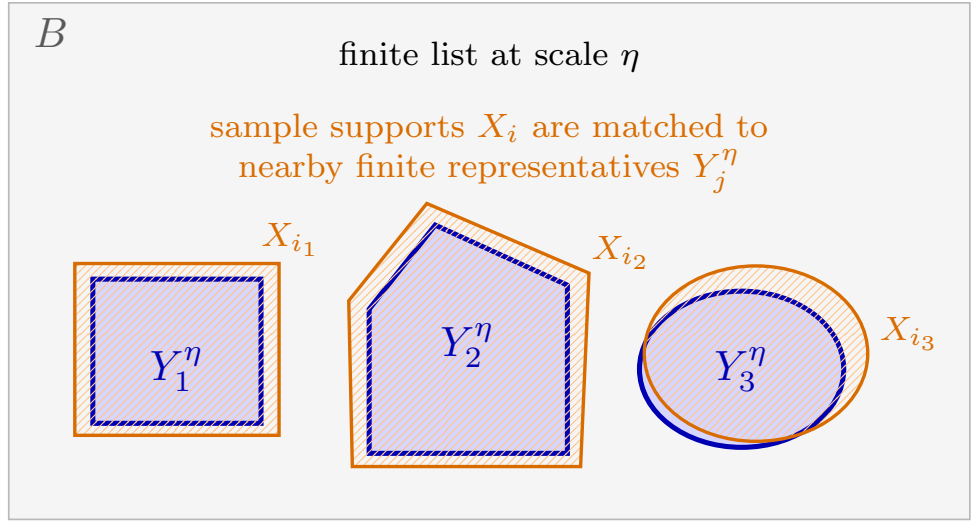

\begin{remark}[Meaning of finite support nets]
	\label{rem_finitely_netted_support_sequence_intuition}
	The finite list $Y_1^\eta,\ldots,Y_{M(\eta)}^\eta$ is a finite description at the chosen precision, not a decision procedure for arbitrary real points.
	For a given support $X_i$, the data specify a representative $Y_j^\eta$ and a representability witness for the small discrepancy $X_i\triangle Y_j^\eta$.
	Thus the condition is a measure-theoretic finite approximation of support shapes.
	The quantity $\mu(X_i\triangle Y_j^\eta)$ is the measure analogue of a distance between supports, not a Hausdorff or pointwise distance.
\end{remark}

\begin{definition}[Cone-collar support data]
	\label{dfn_cone_collar_support_data}
	Let
	\[
		\mathcal X=(X_i)_{i\in I}
	\]
	be a support sequence of representable subsets of a full block $B$.
	We say that $\mathcal X$ has \emph{cone-collar data} if there is a rational constant $C_{\mathcal X}>0$ such that, for every rational $\lambda>0$, there are finitely many members
	\[
		Y_1^\lambda,\ldots,Y_{M(\lambda)}^\lambda
	\]
	of the sequence with the following property.
	For every $i\in I$, there are an index $j$ and a finite multiblock $\mathcal E_{i,\lambda}$ such that
	\[
		X_i\triangle Y_j^\lambda\subseteq\bar{\mathcal E}_{i,\lambda}.
	\]
	The cost estimate is
	\[
		\gamma(\mathcal E_{i,\lambda})\le C_{\mathcal X}\lambda.
	\]
	In particular, $\mathcal X$ is finitely netted in measure.
\end{definition}

\begin{figure}[ht]
	\centering
	\resizebox{0.82\linewidth}{!}{\begin{tikzpicture}[x=1cm,y=1cm,font=\small]
	\fill[gray!8] (0,0) rectangle (8,4.8);
	\draw[gray!55] (0,0) rectangle (8,4.8);

	\fill[orange!28,opacity=0.85] (0.45,1.90) rectangle (1.90,2.78);
	\fill[orange!28,opacity=0.85] (1.90,2.32) rectangle (3.35,3.13);
	\fill[orange!28,opacity=0.85] (3.35,2.20) rectangle (4.85,3.02);
	\fill[orange!28,opacity=0.85] (4.85,2.05) rectangle (6.25,2.82);
	\fill[orange!28,opacity=0.85] (6.25,2.15) rectangle (7.55,2.70);

	\path[pattern=north east lines,pattern color=orange!65!black]
		(0.45,2.22)
		.. controls (0.78,2.05) and (1.25,2.35) .. (2.05,2.72)
		.. controls (3.0,3.15) and (3.85,2.82) .. (4.75,2.42)
		.. controls (5.65,2.05) and (6.65,2.75) .. (7.55,2.36)
		-- (7.55,2.50) -- (6.10,2.32) -- (4.70,2.55) -- (3.20,2.90) -- (1.80,2.55) -- (0.45,2.12)
		-- cycle;

	\draw[very thick,blue!70!black]
		(0.45,2.22)
		.. controls (0.78,2.05) and (1.25,2.35) .. (2.05,2.72)
		.. controls (3.0,3.15) and (3.85,2.82) .. (4.75,2.42)
		.. controls (5.65,2.05) and (6.65,2.75) .. (7.55,2.36);
	\draw[very thick,green!45!black,dashed]
		(0.45,2.12) -- (1.80,2.55) -- (3.20,2.90) -- (4.70,2.55) -- (6.10,2.32) -- (7.55,2.50);	
		
	\fill[green!18,opacity=0.4]
	(0.45,0.55) -- (7.55,0.55) -- (7.55,2.50) --
	(6.10,2.32) -- (4.70,2.55) -- (3.20,2.90) -- (1.80,2.55) -- (0.45,2.12) -- cycle;
	
	\fill[blue!18,opacity=0.3]
	(0.45,0.55) -- (7.55,0.55) -- (7.55,2.36)
	.. controls (6.65,2.75) and (5.65,2.05) .. (4.75,2.42)
	.. controls (3.85,2.82) and (3.0,3.15) .. (2.05,2.72)
	.. controls (1.25,2.35) and (0.78,2.05) .. (0.45,2.22)
	-- cycle;
	
	\node[orange!90!black] at (3.95,2.4) {$\mathcal E_{i,\lambda}$};
	\node[orange!90!black,right,font=\scriptsize] at (5.58,3.28) {$\gamma(\mathcal E_{i,\lambda})\le C_{\mathcal X}\lambda$};
	\node[align=center,font=\scriptsize] at (1.45,3.4)
	{$X_i\triangle Y_j^\lambda\subseteq\bar{\mathcal E}_{i,\lambda}$};
	
	\draw[->,>=stealth,thick,red!40!black,shorten <=2pt,shorten >=2pt] (0.9,3.2) -- (2.55,2.75);	

	\fill[blue!18,opacity=0.5] (6.5,4.35) rectangle (6.9,4.55);
	\draw[blue!70!black,thick] (6.5,4.35) rectangle (6.9,4.55);
	\node[right,blue!70!black,font=\footnotesize] at (7.0,4.45) {$X_i$};
	
	\fill[green!18,opacity=0.6] (6.5,3.95) rectangle (6.9,4.15);
	\draw[green!45!black,thick,dashed] (6.5,3.95) rectangle (6.9,4.15);
	\node[right,green!45!black,font=\footnotesize] at (7.0,4.05) {$Y_j^\lambda$};
	
\end{tikzpicture}}
	\caption{Cone-collar data strengthen finite measure nets by requiring the discrepancy $X_i\triangle Y_j^\lambda$ to be contained in an explicit collar multiblock of cost at most $C_{\mathcal X}\lambda$.}
	\label{fig_dfn_cone_collar_support_data}
\end{figure}
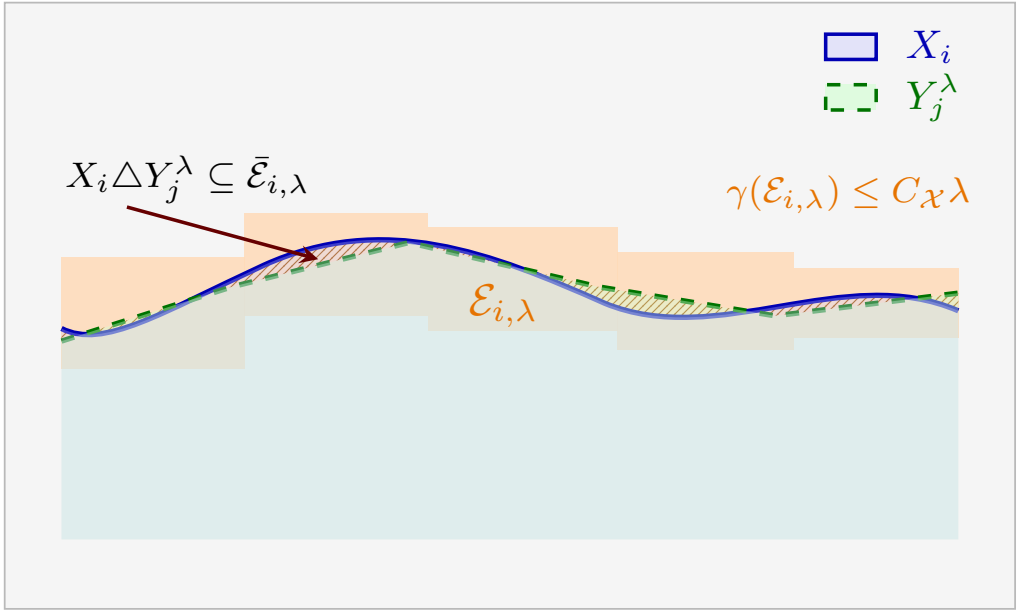

\begin{remark}[No point localization in support nets]
	\label{rem_cone_collar_no_localization}
	Neither \Cref{dfn_finitely_netted_support_sequence} nor \Cref{dfn_cone_collar_support_data} asks whether a real point belongs to $X_i$ or to $Y_j^\lambda$.
	The data are finite set-level certificates: a representative support, a representability witness for the discrepancy, and a multiblock cost estimate.
	This is the same measurement style used for representable sets.
\end{remark}

\begin{remark}[Why a perimeter budget enters]
	\label{rem_why_cone_collar_data}
	The supremum-metric theorem controls functions by moving their values on a fixed block.
	In integral metrics, support motion is harmless only when the moved region has small measure.
	A bounded cone-collar constant plays the role of a constructive perimeter budget: a support perturbation of size $\lambda$ changes the support only inside a multiblock of cost $\bigo{\lambda}$.
	Without such a uniform budget, supports can oscillate with arbitrarily large boundary cost, and no finite net for the metrics $d_1$ or $d_2$ should be expected.
\end{remark}

\begin{example}[Finite support sequences]
	\label{ex_finite_support_catalog_cone_collar}
	Every finite support sequence
	\[
		\mathcal X=(X_1,\ldots,X_S)
	\]
	of representable subsets of $B$ is finitely netted in measure.
	It also has cone-collar data: at every scale take the whole finite sequence as its list of representatives and use the empty multiblock for every discrepancy.
	This is the degenerate case in which the support sequence itself is finite.
\end{example}

\begin{example}[Cone-chart support sequences]
	\label{ex_cone_chart_support_sequence}
	Fix a natural number $S$ of charts, rational chart blocks $Q_l=A_l\times[c_l,d_l]\subset B$ for $l=1,\ldots,S$, side signs, a rational Lipschitz bound $L_\theta$, and a bound $H$.
	Let $\mathcal X$ be a support sequence whose members, in each chart, are described by a subgraph or supergraph of a function
	\[
		\varphi_l\in\mathcal F(A_l,L_\theta,H;1),
	\]
	with the remaining pieces fixed by rational blocks.
	By \Cref{thm_lipschitz_function_space_totally_bounded}, the boundary functions admit finite $d_\infty$-nets at every precision.
	Replacing every $\varphi_l$ by a net representative gives finitely many support representatives.
	If
	\[
		\nrm{\varphi_l-\psi_l}_\infty\le\lambda,
	\]
	then the symmetric difference of the corresponding subgraphs in $Q_l$ is contained in the vertical strip
	\[
		\{(u,t)\in Q_l:\abs{t-\psi_l(u)}\le\lambda\}.
	\]
	This strip is covered by a finite multiblock of cost bounded by a chart constant times $\lambda$.
	Summing over $l=1,\ldots,S$ gives cone-collar data for $\mathcal X$.
\end{example}

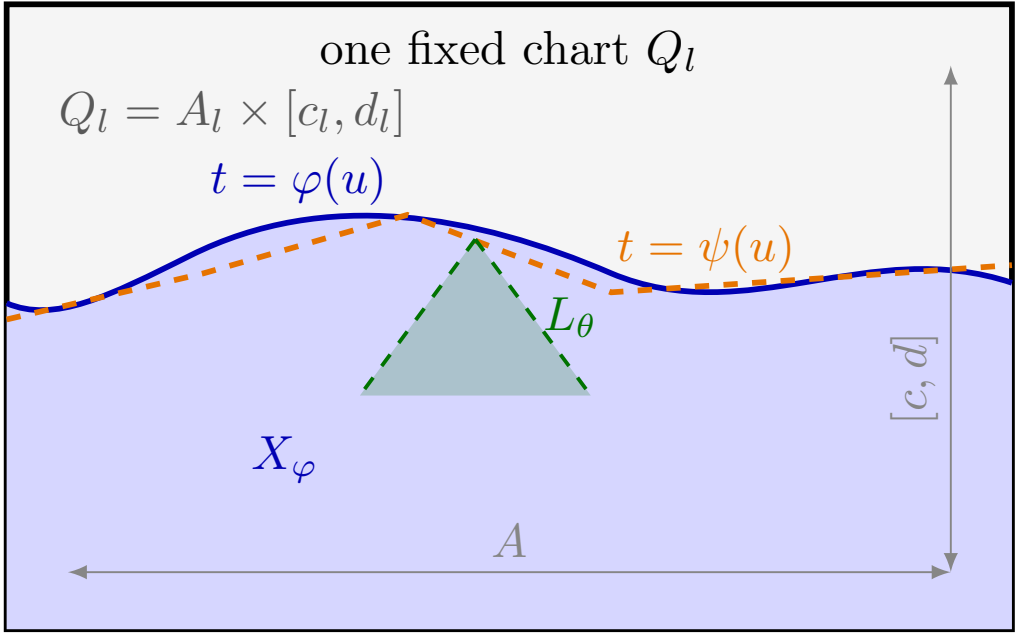
\begin{figure}[ht]
	\centering
	\resizebox{0.82\linewidth}{!}{\begin{tikzpicture}[x=1cm,y=1cm,font=\small]
	\fill[gray!8] (0,0) rectangle (7.4,4.6);
	\draw[very thick] (0,0) rectangle (7.4,4.6);
	\node[gray!70!black,anchor=west] at (0.25,3.8) {$Q_l=A_l\times[c_l,d_l]$};

	\fill[blue!16]
		(0,0) -- (7.4,0) -- (7.4,2.55)
		.. controls (6.35,2.9) and (5.45,2.2) .. (4.45,2.62)
		.. controls (3.45,3.05) and (2.35,3.25) .. (1.35,2.75)
		.. controls (0.75,2.45) and (0.35,2.25) .. (0,2.4)
		-- cycle;
	\draw[very thick,blue!70!black]
		(0,2.4)
		.. controls (0.35,2.25) and (0.75,2.45) .. (1.35,2.75)
		.. controls (2.35,3.25) and (3.45,3.05) .. (4.45,2.62)
		.. controls (5.45,2.2) and (6.35,2.9) .. (7.4,2.55);

	\draw[very thick,orange!90!black,dashed]
		(0,2.28) -- (1.45,2.62) -- (2.95,3.05) -- (4.45,2.48) -- (5.95,2.58) -- (7.4,2.68);

	\coordinate (p) at (3.45,2.87);
	\fill[green!45!black,opacity=0.20] (p) -- ++(-0.85,-1.15) -- ++(1.70,0) -- cycle;
	\draw[green!45!black,thick,dashed] (p) -- ++(-0.85,-1.15);
	\draw[green!45!black,thick,dashed] (p) -- ++(0.85,-1.15);
	\node[green!45!black] at (4.15,2.30) {$L_\theta$};

	\draw[latex-latex,gray!95] (0.45,0.42) -- (6.95,0.42);
	\node[gray!95] at (3.7,0.65) {$A$};
	\draw[latex-latex,gray!95] (6.95,0.42) -- (6.95,4.15);
	\node[gray!95,rotate=90] at (6.68,1.85) {$[c,d]$};

	\node[blue!70!black] at (2.05,1.25) {$X_\varphi$};
	\node[blue!70!black] at (2.15,3.3) {$t=\varphi(u)$};
	\node[orange!90!black] at (5.15,2.8) {$t=\psi(u)$};
	\node[align=center] at (3.7,4.25)
		{one fixed chart $Q_l$};
\end{tikzpicture}}
	\caption{A cone-chart support sequence. The graph $t=\varphi(u)$ is approximated by a finite representative $t=\psi(u)$, and the cone slope gives the uniform Lipschitz bound needed for finite nets and collar estimates.}
	\label{fig_ex_cone_chart_support_sequence}
\end{figure}

\begin{example}[Rolling-ball support sequences]
	\label{ex_rolling_ball_support_sequence}
	Fix a rational rolling radius $r>0$, a fixed number $S$ of charts, and rational ambient blocks.
	A support sequence with interior and exterior rolling-ball data of radius $r$ has, after shrinking charts if needed, a two-sided cone presentation with a cone slope depending only on the fixed chart geometry and on $r$.
	Therefore \Cref{ex_cone_chart_support_sequence} gives cone-collar data.
	The point is not that the sequence is finite, but that its boundary functions range over a finitely netted Lipschitz-bounded function space.
\end{example}

\begin{definition}[Cone-regular branch class]
	\label{dfn_cone_regular_branch_class}
	Let $B$ be a full block, let $L,K>0$, let $N\in\N$, and let $\mathcal X=(X_i)_{i\in I}$ be a support sequence of representable subsets of $B$.
	Denote by
	\[
		\mathcal F_{\mathcal X}^{N}(B,L,K;m)
	\]
	the class of regular functions obtained from a finite formal sum
	\[
		f=\sum_{k=1}^N u_k\indic{X_{i_k}}
	\]
	by applying the finite flattening construction to its support data, where each index $i_k$ belongs to $I$ and each ambient branch satisfies
	\[
		u_k\in\mathcal F(B,L,K;m).
	\]
	Empty and zero summands are allowed.
	The constant $K$ is a branch bound; after flattening, the resulting function has the evident bound witness $NK$.
\end{definition}

\begin{remark}[Branches are not restricted to constants]
	\label{rem_cone_regular_not_piecewise_constant}
	The class $\mathcal F_{\mathcal X}^{N}(B,L,K;m)$ is not a class of piecewise-constant functions.
	The supports may move along the support sequence $\mathcal X$, and each branch is an arbitrary member of the already studied Lipschitz-bounded space $\mathcal F(B,L,K;m)$.
	The flattening step only converts the finite formal sum into the disjoint-support representation required by regular functions.
\end{remark}

\begin{theorem}[Total boundedness in integral metrics]
	\label{thm_cone_regular_integral_total_bounded}
	Let $B$ be a full block and let $\mathcal X=(X_i)_{i\in I}$ be a support sequence of representable subsets of $B$ with cone-collar data.
	Let $L,K>0$ be rational and let $N\in\N$.
	Let $\mathcal F_{\mathcal X}^{N}(B,L,K;m)$ be the class of regular functions obtained from formal sums of at most $N$ terms $u_k\indic{X_{i_k}}$ after applying the finite flattening construction to the support data, where every $i_k$ belongs to $I$ and $u_k\in\mathcal F(B,L,K;m)$.
	Then
	\[
		\mathcal F_{\mathcal X}^{N}(B,L,K;m)
	\]
	is totally bounded with respect to both $d_1$ and $d_2$.
\end{theorem}

\begin{proof}
	Fix $\eps>0$.
	Let
	\[
		V:=\mu(B).
	\]
	Choose $\eta>0$ and $\rho>0$ so small that
	\[
		N(V\rho+2K\eta)\le\eps
	\]
	and
	\[
		N\left(V\rho^2+4K^2\eta\right)^{1/2}\le\eps.
	\]
	Choose rational $\lambda>0$ with
	\[
		C_{\mathcal X}\lambda\le\eta.
	\]
	Take the finite support representatives
	\[
		Y_1,\ldots,Y_M
	\]
	supplied by the cone-collar data at scale $\lambda$.
	Take a finite $\rho$-net
	\[
		P_1,\ldots,P_R
	\]
	of $\mathcal F(B,L,K;m)$ in $d_\infty$ using \Cref{cor_lipschitz_function_space_totally_bounded}.
	The constants $L$ and $K$ enter here: $L$ gives the finite branch net, while $K$ bounds the contribution of the support discrepancy.
	Form the finite family of all regular functions
	\[
		g=\sum_{k=1}^N P_{r_k}\indic{Y_{j_k}},
		\qquad
		r_k\in\{1,\ldots,R\},\quad j_k\in\{1,\ldots,M\}.
	\]
	After flattening, each such expression is regular.
	Let
	\[
		f=\sum_{k=1}^N u_k\indic{X_{i_k}}
	\]
	belong to $\mathcal F_{\mathcal X}^{N}(B,L,K;m)$.
	Choose $Y_{j_k}$ and $\mathcal E_{i_k,\lambda}$ with
	\[
		X_{i_k}\triangle Y_{j_k}\subseteq\bar{\mathcal E}_{i_k,\lambda}.
	\]
	This containment and the inequality $C_{\mathcal X}\lambda\le\eta$ imply
	\[
		\mu(X_{i_k}\triangle Y_{j_k})\le\eta.
	\]
	Choose $P_{r_k}$ with
	\[
		d_\infty(u_k,P_{r_k})\le\rho.
	\]
	Set
	\[
		g:=\sum_{k=1}^N P_{r_k}\indic{Y_{j_k}}.
	\]
	For each summand,
	\[
		\int_B\nrm{u_k\indic{X_{i_k}}-P_{r_k}\indic{Y_{j_k}}}\diff x
		\le
		V\rho+2K\eta.
	\]
	Summing over $k$ gives
	\[
		d_1(f,g)\le N(V\rho+2K\eta)\le\eps.
	\]
	For the quadratic estimate, use
	\[
		\nrm{\sum_{k=1}^N a_k}^2\le N\sum_{k=1}^N\nrm{a_k}^2.
	\]
	Each summand satisfies
	\[
		\int_B\nrm{u_k\indic{X_{i_k}}-P_{r_k}\indic{Y_{j_k}}}^2\diff x
		\le
		V\rho^2+4K^2\eta.
	\]
	Hence
	\[
		d_2(f,g)
		\le
		N\left(V\rho^2+4K^2\eta\right)^{1/2}
		\le\eps.
	\]
	The displayed finite family is therefore an $\eps$-net for both integral metrics.
\end{proof}

\begin{definition}[$L_p$-continuous functional]
	\label{dfn_Lp_continuous_functional}
	Let $p\in\{1,2\}$.
	A functional
	\[
		J:\mathcal F_{\mathcal X}^{N}(B,L,K;m)\to\R
	\]
	is called \emph{uniformly $L_p$-continuous} if it has a modulus $\alpha_J$ such that
	\[
		d_p(f,g)\le\alpha_J(\eps)
		\quad\Longrightarrow\quad
		\abs{J[f]-J[g]}\le\eps.
	\]
\end{definition}

\begin{theorem}[Functional extremum value theorem in integral metrics]
	\label{thm_integral_metric_functional_evt}
	Let $B$ be a full block and let $\mathcal X=(X_i)_{i\in I}$ be a support sequence of representable subsets of $B$ with cone-collar data.
	Let $L,K>0$ be rational and let $N\in\N$.
	Let $\mathcal F_{\mathcal X}^{N}(B,L,K;m)$ be the class of regular functions obtained from formal sums of at most $N$ terms $u_k\indic{X_{i_k}}$ after applying the finite flattening construction to the support data, where every $i_k$ belongs to $I$ and $u_k\in\mathcal F(B,L,K;m)$.
	Let $p\in\{1,2\}$.
	If
	\[
		J:\mathcal F_{\mathcal X}^{N}(B,L,K;m)\to\R
	\]
	is uniformly $L_p$-continuous, then, for every $\eps>0$, there exists an $\eps$-minimizer of $J$ on
	\[
		\mathcal F_{\mathcal X}^{N}(B,L,K;m).
	\]
\end{theorem}

\begin{proof}
	Repeat the finite-net proof of \Cref{thm_functional_evt}, using the finite $d_p$-net supplied by \Cref{thm_cone_regular_integral_total_bounded}.
\end{proof}

\subsection{Examples for Integral Metrics}

\begin{example}[Least-squares tracking]
	\label{ex_L2_least_squares_tracking}
	Let $r\in\BReg{B}$ be effectively bounded.
	The functional
	\[
		J[f]:=\int_B\nrm{f(x)-r(x)}^2\diff x
	\]
	is uniformly $L_2$-continuous on every class $\mathcal F_{\mathcal X}^{N}(B,L,K;m)$.
	Indeed, if $\nrm{f},\nrm{g},\nrm{r}\le M$ on $B$, then
	\[
		\abs{J[f]-J[g]}
		\le
		4M\,d_1(f,g)
		\le
		4M\mu(B)^{1/2}d_2(f,g).
	\]
	Thus \Cref{thm_integral_metric_functional_evt} gives approximate least-squares minimizers.
\end{example}

\begin{example}[Integral stage cost]
	\label{ex_L1_stage_cost}
	Let
	\[
		\ell:B\times\{y\in\R^m:\nrm{y}\le K\}\to\R
	\]
	be bounded and Lipschitz in $y$ with constant $C_\ell$.
	For
	\[
		J[f]:=\int_B\ell(x,f(x))\diff x
	\]
	one has
	\[
		\abs{J[f]-J[g]}
		\le
		C_\ell d_1(f,g).
	\]
	Hence $J$ is uniformly $L_1$-continuous.
	On a finite-volume block it is also uniformly $L_2$-continuous, since
	\[
		d_1(f,g)\le\mu(B)^{1/2}d_2(f,g).
	\]
\end{example}

\begin{example}[Switching regions with cone geometry]
	\label{ex_switching_regions_cone_geometry}
	Fix rational chart blocks
	\[
		Q_l=A_l\times[c_l,d_l]\subset B,
		\qquad
		l=1,\ldots,S,
	\]
	side signs $\sigma_l\in\{-1,1\}$, a rational cone slope $L_\theta$, and a bound $H$.
	For each precision index $q\in\N$, take the dyadic grid of mesh $2^{-q}$ on every $A_l$, with its standard rational triangulation.
	A finite rational code consists of $q$ and rational vertex values
	\[
		a_{l,v}\in[c_l,d_l]\cap[-H,H],
		\qquad
		v\in V_{l,q},
	\]
	satisfying the finite cone inequalities
	\[
		\abs{a_{l,v}-a_{l,w}}
		\le
		L_\theta\nrm{v-w}
		\qquad
		v,w\in V_{l,q}.
	\]
	Let $P_{l,q,a}:A_l\to[c_l,d_l]$ be the piecewise affine interpolation of these vertex values on that triangulation.
	Enumerate all such finite rational codes as $e_1,e_2,\ldots$.
	The code $e_i$ produces a switching support $X_i\subseteq B$ by the chart formulas
	\[
		X_i\cap Q_l
		=
		\{(u,t)\in Q_l:\sigma_l(t-P_{l,e_i}(u))\le0\},
		\qquad
		l=1,\ldots,S,
	\]
	together with the fixed rational-block pieces outside the chart interiors.
	Thus
	\[
		\mathcal X=(X_i)_{i\in\N}
	\]
	is a support sequence: the index $i$ is decoded into finite rational data, and those data mechanically produce $X_i$ and its representability witnesses.
	The fixed cone slope and fixed number of charts put this sequence under \Cref{ex_cone_chart_support_sequence}, so $\mathcal X$ has cone-collar data.
	Regular functions whose branches are bounded Lipschitz and whose active supports are chosen from this sequence belong to a class covered by \Cref{thm_integral_metric_functional_evt}.
	Such finitely enumerable switching regions are the form one obtains in a finite search over piecewise feedback laws, including sliding-mode-style candidates.
\end{example}

\section{Multifunctions}
\label{sec_multifunctions}

\begin{remark}
	This section develops set-valued maps in a form compatible with the regular-function machinery.
	Domain conditions and value conditions are kept separate, so that total, domain-regular, marginalized, block-full, convex, and Hausdorff-Lipschitz assumptions can be combined only when needed.
	The main output is a collection of exact, approximate, regular, and Lipschitz selector mechanisms.
\end{remark}

\subsection{Common Definitions}
\label{subsec_multifunctions}

\begin{definition}[Multifunction]
	\label{dfn_multifunction}
	Let $X\subset\R^n$.
	A \emph{multifunction} from $X$ to $\R^m$ is an assignment
	\[
		F:X\rightrightarrows\R^m,
		\qquad
		x\mapsto F(x)\subset\R^m.
	\]
	The set $X$ is denoted by $\dom F$.
	For $A\subset\R^n$, its graph over $A$ is
	\[
		\Gamma_A(F)
		:=
		\{(x,y)\in (A\cap\dom F)\times\R^m:y\in F(x)\}.
	\]
\end{definition}

\begin{definition}[Domain conditions for multifunctions]
	\label{dfn_multifunction_domain_conditions}
	Let $B\subset\R^n$ be a full block.
	A multifunction $F$ is called \emph{$B$-total} if
	\[
		B\subseteq\dom F.
	\]
	It is called \emph{total} if $\dom F=\R^n$.
	It is called \emph{domain-regular on $B$} if it is given with pairwise disjoint representable supports
	\[
		X_1,\ldots,X_N\Subset B
	\]
	and
	\[
		\supp F:=\bigcup_{k=1}^N X_k,
		\qquad
		\dom F:=\supp F\cup(\supp F)^c.
	\]
	For $x\in\supp F$, the value $F(x)$ is a subset of $\R^m$.
	For $x\in(\supp F)^c$, the value is defined to be $\{0\}$.
\end{definition}

\begin{remark}
	\label{rem_domain_regular_disjoint_not_apart}
	The pairwise disjointness in \Cref{dfn_multifunction_domain_conditions} is the ordinary no-overlap condition of \Cref{dfn_disjoint_sets}.
	It is sufficient for the domain bookkeeping and for exception-set estimates.
	Exact pointwise selector constructions that assemble separate complemented summands may require the stronger mutual-apartness datum of \Cref{dfn_mutual_apartness}; this additional datum is stated where it is used.
\end{remark}

\begin{definition}[Locally domain-regular multifunction]
	\label{dfn_locally_domain_regular_multifunction}
	A multifunction $F$ is called \emph{locally domain-regular} if, for every full block $B$, there is a domain-regular multifunction $\hat F_B$ on $B$ such that $F$ and $\hat F_B$ agree on their common determined part in $B$.
	The local domain on $B$ is taken to be
	\[
		B\cap\dom{\hat F_B}.
	\]
\end{definition}

\begin{remark}
	\label{rem_multifunction_domain_value_separation}
	The conditions in \Cref{dfn_multifunction_domain_conditions} concern only where the multifunction is determined.
	They impose no regularity, convexity, locatedness, or thickness condition on its values.
	Value conditions are recorded separately below and may be combined with the domain conditions independently.
\end{remark}

\subsection{Value Conditions}
\label{subsec_multifunction_value_conditions}

\begin{definition}[Marginalized multifunction]
	\label{dfn_marginalized_multifunction}
	Let $B\subset\R^n$ be a full block.
	A multifunction $F$ is called \emph{marginalized on $B$} if there are scalar regular functions
	\[
		\ell_j,u_j\in\BReg{B},
		\qquad
		j=1,\ldots,m,
	\]
	such that
	\[
		\ell_j(x)\le u_j(x),
		\qquad
		j=1,\ldots,m,
	\]
	and
	\[
		F(x)
		=
		\prod_{j=1}^m[\ell_j(x),u_j(x)]
	\]
	whenever the displayed marginal values are determined.
	The functions $\ell_j,u_j$ are called \emph{margins} of $F$.
\end{definition}

\begin{definition}[Block-full values]
	\label{dfn_block_full_multifunction}
	Let $B\subset\R^n$ be a full block.
	The admissible base blocks for a $B$-total multifunction are the full blocks $Q\subseteq B$.
	The admissible base blocks for a domain-regular multifunction on $B$ with supports $X_1,\ldots,X_N$ are the full blocks $Q\subseteq B$ satisfying
	\[
		Q\Subset X_k
	\]
	for some $k$.
	The multifunction $F$ is said to have \emph{block-full values on $B$} if, for every admissible base block $Q$, there exists a full block $D_Q\subset\R^m$ such that
	\[
		Q\times D_Q\subseteq\Gamma_B(F).
	\]
\end{definition}

\begin{remark}
	\label{rem_block_full_not_image_full}
	The condition in \Cref{dfn_block_full_multifunction} is stronger than saying that the union
	\[
		F(Q):=\bigcup_{x\in Q}F(x)
	\]
	contains a full block.
	The useful datum is the uniform graph inclusion
	\[
		x\in Q,\quad y\in D_Q
		\quad\Longrightarrow\quad
		y\in F(x).
	\]
	In this sense the block $D_Q$ project-covers $Q$ through the graph of $F$.
\end{remark}

\begin{definition}[Hausdorff distance]
	\label{dfn_hausdorff_distance}
	Let $A,C\subset\R^m$ be inhabited totally bounded sets.
	Their Hausdorff distance is
	\[
		d_H(A,C)
		:=
		\max\left\{
			\sup_{a\in A} d(a,C),
			\sup_{c\in C} d(c,A)
		\right\}.
	\]
	The distance operations for $A$ and $C$ are supplied by \Cref{lem_totally_bounded_located}, and the displayed suprema exist by finite-net approximation.
\end{definition}

\begin{definition}[Hausdorff-Lipschitz values]
	\label{dfn_hausdorff_lipschitz_multifunction}
	Let $B\subset\R^n$ be a full block.
	A multifunction $F$ is called \emph{$L$-Hausdorff-Lipschitz on $B$} if every value $F(x)$, $x\in B\cap\dom F$, is inhabited and totally bounded, and if the Hausdorff distance between the values satisfies
	\[
		d_H(F(x),F(y))
		\le
		L\nrm{x-y},
		\qquad
		x,y\in B\cap\dom F.
	\]
\end{definition}

\begin{definition}[Hausdorff-continuous values]
	\label{dfn_hausdorff_continuous_multifunction}
	Let $B\subset\R^n$ be a full block and let $\omega:\Q_{>0}\to\Q_{>0}$ satisfy $\omega(r)\to0$ as $r\to0$.
	A multifunction $F$ is called \emph{Hausdorff-continuous on $B$ with modulus $\omega$} if every value $F(x)$, $x\in B\cap\dom F$, is inhabited and totally bounded, and
	\[
		\nrm{x-y}\le r
		\quad\Longrightarrow\quad
		d_H(F(x),F(y))
		\le
		\omega(r),
		\qquad
		x,y\in B\cap\dom F,\quad r\in\Q_{>0}.
	\]
\end{definition}

\begin{remark}
	\label{rem_lipschitz_multifunction_options}
	The definition is meaningful when the values are equipped with the finite-net data needed for \Cref{dfn_hausdorff_distance}.
	It compares the fibers $F(x)$ directly, rather than comparing finite enclosures of the graph.
	The estimate supplies the finite oscillation control used in \Cref{thm_steiner_hausdorff_lipschitz_selector,thm_regular_metrical_selector_convex_hausdorff_lipschitz}.
\end{remark}

\begin{definition}[Convex values]
	\label{dfn_convex_value_multifunction}
	Let $B\subset\R^n$ be a full block.
	A multifunction $F$ is said to have \emph{convex values on $B$} if, for every $x\in B\cap\dom F$ and all $y_0,y_1\in F(x)$,
	\[
		\lambda y_0+(1-\lambda)y_1\in F(x),
		\qquad
		0\le\lambda\le1.
	\]
\end{definition}

\begin{definition}[Support function]
	\label{dfn_support_function}
	Let $C\subset\R^m$ be inhabited and totally bounded.
	For $u\in\R^m$, its \emph{support function} is
	\[
		h_C(u):=\sup_{y\in C}\scal{u,y}.
	\]
	The supremum exists by finite-net approximation of the totally bounded set $C$.
\end{definition}

\begin{figure}[ht]
	\centering
	\begin{tikzpicture}[>=Stealth,x=1cm,y=1cm,font=\small]
	\draw[->,black!70] (-0.9,0) -- (4.25,0) node[right] {$x_1$};
	\draw[->,black!70] (0,-0.7) -- (0,3.25) node[above] {$x_2$};
	\node[below left] at (0,0) {$0$};

	\def\ang{30}
	\coordinate (H) at (\ang:3.30);
	\coordinate (U) at (\ang:0.72);
	\coordinate (Ooff) at ({\ang-90}:0.28);
	\coordinate (Hoff) at ($(H)+({\ang-90}:0.28)$);

	\filldraw[fill=blue!12,draw=blue!65,very thick,smooth cycle]
		plot coordinates {(0.95,0.62) (1.18,2.15) (2.12,2.32) (H) (2.42,0.50)};
	\node[blue!70!black] at (1.72,1.42) {$C$};

	\draw[dashed,black!65] ({\ang+180}:0.75) -- (\ang:4.05);
	\draw[->,red!80!black,very thick] (0,0) -- (U) node[above left] {$u$};

	\draw[orange!85!black,very thick]
		($(H)+({\ang+90}:1.45)$) -- ($(H)+({\ang-90}:1.45)$)
		node[pos=0.88,above right] {$\scal{u,z}=h_C(u)$};

	\filldraw[black] (H) circle (1.5pt) node[above right] {$y$};
	\draw[dotted,magenta!70!black] (0,0) -- (Ooff);
	\draw[dotted,magenta!70!black] (H) -- (Hoff);
	\draw[<->,magenta!80!black,thick] (Ooff) -- (Hoff)
		node[midway,below,sloped] {$h_C(u)$};

	\draw[black!70] ($(H)+({\ang+90}:0.20)$) --
		($(H)+({\ang+90}:0.20)+({\ang+180}:0.20)$) --
		($(H)+({\ang+180}:0.20)$);
\end{tikzpicture}
	\caption{For unit $u$, the support value $h_C(u)$ is the signed distance from the origin to the supporting hyperplane perpendicular to $u$.}
	\label{fig_support_function}
\end{figure}
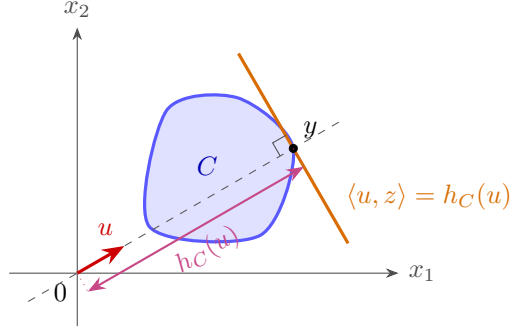

\subsection{Steiner Point}
\label{subsec_steiner_point}

\begin{definition}[Normalized spherical area and integral]
	\label{dfn_normalized_spherical_area}
	Let
	\[
		\mathbb S^{m-1}:=\{u\in\R^m:\nrm{u}=1\}
	\]
	be the unit sphere.
	The \emph{normalized spherical area} on $\mathbb S^{m-1}$ is denoted by $\sigma$ and satisfies
	\[
		\sigma(\mathbb S^{m-1})=1.
	\]
	Here $\phi$ below denotes an arbitrary uniformly continuous integrand on the sphere.
	For $m=1$, put
	\[
		\int_{\mathbb S^0}\phi(u)\diff\sigma(u)
		:=
		\frac{\phi(-1)+\phi(1)}{2}.
	\]
	For $m\ge2$, let
	\[
		D_m:=[0,\pi]^{m-2}\times[0,2\pi].
	\]
	For $\theta=(\theta_1,\ldots,\theta_{m-1})\in D_m$, define
	\[
		(U_m(\theta))_j
		:=
		\left(\prod_{i=1}^{j-1}\sin\theta_i\right)\cos\theta_j,
		\qquad
		j=1,\ldots,m-1,
	\]
	and
	\[
		(U_m(\theta))_m
		:=
		\prod_{i=1}^{m-2}\sin\theta_i\,\sin\theta_{m-1},
	\]
	with the usual convention that an empty product is $1$.
	Define the surface Jacobian
	\[
		J_m(\theta)
		:=
		\prod_{i=1}^{m-2}(\sin\theta_i)^{m-1-i}.
	\]
	Let
	\[
		\omega_m:=\int_{D_m}J_m(\theta)\diff\theta.
	\]
	Then $\omega_m>0$; for example, on
	\[
		Q_m:=[\pi/3,2\pi/3]^{m-2}\times[0,1]\subset D_m
	\]
	one has
	\[
		J_m(\theta)
		\ge
		2^{-(m-1)(m-2)/2},
	\]
	so the integral over $Q_m$ is strictly positive.
	If $\phi:\mathbb S^{m-1}\to\R^r$ is uniformly continuous, its spherical integral
	\[
		\int_{\mathbb S^{m-1}}\phi(u)\diff\sigma(u)
	\]
	is defined, for $m\ge2$, by the ordinary block integral
	\[
		\int_{\mathbb S^{m-1}}\phi(u)\diff\sigma(u)
		:=
		\frac{1}{\omega_m}
		\int_{D_m}\phi(U_m(\theta))J_m(\theta)\diff\theta .
	\]
\end{definition}

\begin{remark}
	\label{rem_normalized_spherical_area_details}
	The functional $\sigma$ is not the block volume $\mu$ from \Cref{sec_sets}.
	It is the usual surface area on the sphere, normalized by the total surface area $\omega_m$.
	The parametrization $U_m$ may represent a boundary point of the angular block more than once, but the repeated part is contained in finitely many faces of $D_m$ and has block volume zero.
	No point of the sphere is therefore assigned to a chart or to a preferred representative.
	In low dimensions the preceding formula gives
	\[
		\int_{\mathbb S^1}\phi(u)\diff\sigma(u)
		=
		\frac{1}{2\pi}\int_0^{2\pi}\phi(\cos\theta,\sin\theta)\diff\theta
	\]
	and
	\[
		\int_{\mathbb S^2}\phi(u)\diff\sigma(u)
		=
		\frac{1}{4\pi}
		\int_0^\pi\int_0^{2\pi}
		\phi(\cos\theta,\sin\theta\cos\varphi,\sin\theta\sin\varphi)
		\sin\theta\diff\varphi\diff\theta .
	\]
	If a spherical cell $A\subseteq\mathbb S^{m-1}$ is separately given by finitely many inequalities whose pullbacks under $U_m$ are representable subsets of $D_m$, its normalized area is
	\[
		\sigma(A)
		:=
		\frac{1}{\omega_m}
		\int_{D_m}\indic{U_m^{-1}(A)}(\theta)J_m(\theta)\diff\theta.
	\]
	More generally, for such a cell and a uniformly continuous $\phi$,
	\[
		\int_A\phi(u)\diff\sigma(u)
		:=
		\frac{1}{\omega_m}
		\int_{D_m}\indic{U_m^{-1}(A)}(\theta)\phi(U_m(\theta))J_m(\theta)\diff\theta.
	\]
	All integrals are ordinary Riemann integrals on the angular block $D_m$.
	These cell formulas are not part of the definition of the Steiner point; they only explain how the same notation applies to explicitly described spherical regions.
\end{remark}

\begin{definition}[Steiner point]
	\label{dfn_steiner_point}
	Let $C\subset\R^m$ be inhabited, closed, convex, and totally bounded.
	Its \emph{Steiner point} is
	\[
		\operatorname{St}(C)
		:=
		m\int_{\mathbb S^{m-1}}u\,h_C(u)\diff\sigma(u).
	\]
	The integral is the spherical integral of \Cref{dfn_normalized_spherical_area}.
	Indeed, total boundedness gives a bound $\nrm{y}\le R$ on $C$, and hence
	\[
		\abs{h_C(u)-h_C(v)}
		\le
		R\nrm{u-v},
		\qquad
		u,v\in\mathbb S^{m-1}.
	\]
	Thus $u\mapsto u\,h_C(u)$ is uniformly continuous on the sphere, so the defining integral exists.
\end{definition}

\begin{theorem}[Finite-dimensional Steiner point theorem]
	\label{thm_steiner_point_finite_dimensional}
	Let $A,C\subset\R^m$ be inhabited, closed, convex, and totally bounded.
	Then
	\[
		\operatorname{St}(C)\in C
	\]
	and
	\[
		\nrm{\operatorname{St}(A)-\operatorname{St}(C)}
		\le
		m d_H(A,C).
	\]
\end{theorem}

\begin{proof}
	The containment $\operatorname{St}(C)\in C$ is the finite-dimensional Steiner point theorem for compact convex sets; see, for example, \cite[Section~1.7]{Schneider2014ConvexBodies}.
	In the present hypotheses, closed and totally bounded convex sets are the compact convex bodies to which that theorem applies.
	No measurable-selection principle is involved: the point is defined by the explicit support-function integral in \Cref{dfn_steiner_point}, and the standard proof is a finite-dimensional convex-geometric argument using support functions, finite polytope approximation, and separation.
	We record the Lipschitz estimate in the notation of this paper, because it is the estimate used below.
	Total boundedness gives locatedness by \Cref{lem_totally_bounded_located}, so the support values and Hausdorff distances appearing in the estimate are available.

	For every $u\in\mathbb S^{m-1}$ and every $\eta>0$, the definition of Hausdorff distance gives, for each $a\in A$, a point $c_\eta\in C$ such that
	\[
		\nrm{a-c_\eta}\le d_H(A,C)+\eta.
	\]
	Thus
	\[
		\scal{u,a}
		\le
		\scal{u,c_\eta}+\nrm{a-c_\eta}
		\le
		h_C(u)+d_H(A,C)+\eta.
	\]
	Taking the supremum over $a\in A$ and using that $\eta>0$ is arbitrary gives
	\[
		h_A(u)\le h_C(u)+d_H(A,C).
	\]
	The reverse inequality is identical, hence
	\[
		\abs{h_A(u)-h_C(u)}
		\le
		d_H(A,C).
	\]
	For $m=1$, the definition on $\mathbb S^0=\{-1,1\}$ gives the same estimate by averaging over the two points.
	Assume now that $m\ge2$.
	Therefore
	\[
	\begin{aligned}
		\nrm{\operatorname{St}(A)-\operatorname{St}(C)}
		&=
		\nrm{
			\frac{m}{\omega_m}
			\int_{D_m}
			U_m(\theta)
			\bigl(h_A(U_m(\theta))-h_C(U_m(\theta))\bigr)
			J_m(\theta)\diff\theta
		}
		\\
		&\le
		\frac{m}{\omega_m}
		\int_{D_m}
		\nrm{U_m(\theta)}
		\abs{h_A(U_m(\theta))-h_C(U_m(\theta))}
		J_m(\theta)\diff\theta
		\\
		&\le
		\frac{m}{\omega_m}
		d_H(A,C)
		\int_{D_m}J_m(\theta)\diff\theta
		=
		m d_H(A,C).
	\end{aligned}
	\]
\end{proof}

\subsection{Selection}
\label{subsec_multifunction_selection}

\begin{definition}[Selector]
	\label{dfn_selector}
	Let $F$ be a multifunction and let $A\subset\dom F$.
	A function $s$ is called an \emph{exact selector} of $F$ on $A$ if $A\subset\dom s$ and
	\[
		s(x)\in F(x),
		\qquad
		x\in A.
	\]
	Without specifying a set, an exact selector means an exact selector on $\dom F$.
	If $s$ is regular, it is called a \emph{regular exact selector}.
	If $s$ is measurable, it is called a \emph{measurable exact selector}.
\end{definition}

\begin{remark}
	\label{rem_regular_selector_integrability}
	The distinction between measurable and regular selectors matters for integration.
	An effectively bounded regular selector on a block, with supports well contained in that block, is Riemann integrable by \Cref{lem_regular_functions_riemann_integrable}.
	A merely measurable selector need not be integrable in this Riemann sense, even when it is bounded.
\end{remark}

\begin{definition}[Full and metrical selectors]
	\label{dfn_approximate_selector}
	Let $F$ be a multifunction on a full block $B$, let $s$ be a regular function, and let $\eps>0$.
	The function $s$ is called an \emph{$\eps$-full selector} of $F$ on $B$ if there is a finite full multiblock $\mathcal E\Subset B$ with $\gamma(\mathcal E)\le\eps$ such that $s$ is an exact selector on
	\[
		B\cap\dom F\cap\bar{\mathcal E}^c.
	\]
	It is called an \emph{$\eps$-metrical selector} of $F$ on $B$ if, for every $x\in B\cap\dom F\cap\dom s$,
	\[
		(\{s(x)\}\oplus\eps)\cap F(x)
	\]
	is inhabited.
\end{definition}

\begin{definition}[Lipschitz selector]
	\label{dfn_lipschitz_selector}
	A regular selector $s$ is called \emph{Lipschitz} on a set $A$ if the function $s:A\cap\dom s\to\R^m$ is Lipschitz in the sense of \Cref{dfn_lipschitz_constant}.
\end{definition}

\begin{remark}
	\label{rem_total_block_full_exact_selector}
	If $F$ is $B$-total and has block-full values on $B$, then applying \Cref{dfn_block_full_multifunction} to $Q=B$ gives a full block $D_B\subset\R^m$ with
	\[
		B\times D_B\subseteq\Gamma_B(F).
	\]
	Any rational point of $D_B$ therefore defines a constant exact, hence Lipschitz, selector on $B$.
\end{remark}

\begin{proposition}[Exact Lipschitz selectors from block-full data]
	\label{prop_block_full_lipschitz_selector}
	Let $F$ be $B$-total and have block-full values on a full block $B\subset\R^n$.
	Let $D_B$ be a block-full witness with
	\[
		B\times D_B\subseteq\Gamma_B(F).
	\]
	If $s:B\to D_B$ is $L$-Lipschitz, then $s$ is an exact $L$-Lipschitz selector of $F$ on $B$.
	If, moreover, $s$ is regular, then it is an exact Lipschitz-regular selector.
\end{proposition}

\begin{proof}
	For every $x\in B$, the inclusion $s(x)\in D_B$ and the graph containment give
	\[
		(x,s(x))\in\Gamma_B(F).
	\]
	Hence $s(x)\in F(x)$.
	The Lipschitz and regularity assertions are exactly the assumptions on $s$.
\end{proof}

\begin{corollary}[Steiner point is a Lipschitz selector]
	\label{crl_steiner_point_lipschitz_selector}
	Let $A,C\subset\R^m$ be inhabited, closed, convex, and totally bounded.
	Then
	\[
		\operatorname{St}(C)\in C
	\]
	and
	\[
		\nrm{\operatorname{St}(A)-\operatorname{St}(C)}
		\le
		m d_H(A,C).
	\]
\end{corollary}

\begin{proof}
	This is \Cref{thm_steiner_point_finite_dimensional}.
\end{proof}

\begin{theorem}[Steiner selector for convex Hausdorff-Lipschitz values]
	\label{thm_steiner_hausdorff_lipschitz_selector}
	Let $F$ be $B$-total and $L$-Hausdorff-Lipschitz on a full block $B\subset\R^n$.
	Assume that every value $F(x)$, $x\in B$, is closed and convex.
	Then
	\[
		s(x):=\operatorname{St}(F(x))
	\]
	is an exact selector of $F$ on $B$ and satisfies
	\[
		\nrm{s(x)-s(y)}
		\le
		m L\nrm{x-y},
		\qquad
		x,y\in B.
	\]
\end{theorem}

\begin{proof}
	For each $x\in B$, \Cref{crl_steiner_point_lipschitz_selector} gives
	\[
		s(x)=\operatorname{St}(F(x))\in F(x).
	\]
	For $x,y\in B$, the same lemma and the Hausdorff-Lipschitz estimate give
	\[
		\nrm{s(x)-s(y)}
		\le
		m d_H(F(x),F(y))
		\le
		m L\nrm{x-y}.
	\]
\end{proof}

\begin{theorem}[Regular $\eps$-metrical selector for convex Hausdorff-Lipschitz values]
	\label{thm_regular_metrical_selector_convex_hausdorff_lipschitz}
	Let $F$ satisfy the assumptions of \Cref{thm_steiner_hausdorff_lipschitz_selector}.
	Then, for every $\eps>0$, $F$ admits a Lipschitz-regular $\eps$-metrical selector on $B$.
\end{theorem}

\begin{proof}
	Let
	\[
		s(x):=\operatorname{St}(F(x)).
	\]
	By \Cref{thm_steiner_hausdorff_lipschitz_selector}, the function $s$ is an exact $mL$-Lipschitz selector.
	Choose a rational $K>0$ such that
	\[
		\nrm{s(x)}\le K,
		\qquad
		x\in B.
	\]
	Such a $K$ is obtained from a bound on one totally bounded value $F(x_0)$ and the estimate
	\[
		d_H(F(x),F(x_0))\le L\diam(B).
	\]
	The role of \Cref{thm_lipschitz_function_space_totally_bounded} is to approximate a bounded Lipschitz function on a block by rational piecewise affine functions while keeping the prescribed Lipschitz bound.
	Applied to the exact Steiner selector, this gives a regular function uniformly close to an exact selector.
	Apply \Cref{thm_lipschitz_function_space_totally_bounded} to the function space $\mathcal F(B,mL,K;m)$ with accuracy $\eps$.
	This gives a rational piecewise affine map $P:B\to\R^m$ such that
	\[
		d_\infty(P,s)\le\eps.
	\]
	The map $P$, with support $B$ and exterior value zero, is a Lipschitz-regular function.
	For every $x\in B$, exactness of $s$ gives $s(x)\in F(x)$ and the last estimate gives
	\[
		\nrm{P(x)-s(x)}\le\eps.
	\]
	Thus
	\[
		(\{P(x)\}\oplus\eps)\cap F(x)
	\]
	is inhabited.
\end{proof}

\begin{theorem}[$\eps$-full selector for domain-regular block-full values]
	\label{thm_domain_regular_block_full_selector}
	Let $F$ be domain-regular on a full block $B\subset\R^n$ and have block-full values on $B$.
	Then, for every $\eps>0$, $F$ admits a simple regular $\eps$-full selector on $B$.
\end{theorem}

\begin{proof}
	Let $X_1,\ldots,X_N$ be the representable supports of $F$.
	If $N=0$, take the empty-sum regular function, which is identically zero on $B$.
	This is an exact selector because $F(x)=\{0\}$ on $B$.
	Assume $N>0$.
	Choose representability witnesses for the supports $X_k$ with exception costs at most $\eps/(3N)$.
	Apply one common coordinate refinement to $B$ and to the finite multiblock obtained by concatenating all
	\[
		\mathcal M_{X_k},\qquad
		\mathcal E_{X_k},\qquad
		\mathcal M^c_{X_k},
		\qquad k=1,\ldots,N.
	\]
	Choose the coordinate-boundary multiblock $\mathcal H$ with cost at most $\eps/3$.
	Let $\mathcal C$ be the corresponding trimmed refinement.
	Let $\mathcal A$ be the concatenation of $\mathcal H$ and all original exception multiblocks $\mathcal E_{X_k}$.
	Then
	\[
		\gamma(\mathcal A)\le\eps.
	\]
	For every cell $C\in\mathcal C$ outside $\bar{\mathcal A}$ and every $k$, the cell lies either in the core multiblock $\mathcal M_{X_k}$ or in the exterior multiblock $\mathcal M^c_{X_k}$.
	Since the supports are pairwise disjoint, at most one index $k(C)$ satisfies
	\[
		C\Subset X_{k(C)}.
	\]
	If such an index exists, block-fullness gives a full block $D_C\subset\R^m$ with
	\[
		C\times D_C\subseteq\Gamma_B(F).
	\]
	Choose a rational tag $y_C\in D_C$.
	If no such index exists, put $y_C:=0$.
	Call the cells of the first kind support cells.
	The finite sum of the simple complemented functions with supports equal to the support cells is a simple regular function $s$.
	For $x\in B\cap\bar{\mathcal A}^c$, the point $x$ lies in one of the retained cells.
	On a core cell, the inclusion $C\times D_C\subseteq\Gamma_B(F)$ gives
	\[
		s(x)=y_C\in F(x).
	\]
	On a cell exterior to all supports, the definition of domain-regular multifunction gives
	\[
		F(x)=\{0\},
	\]
	and $s(x)=0$.
	Thus $s$ is exact on $B\cap\dom F\cap\bar{\mathcal A}^c$, and hence is an $\eps$-full selector on $B$.
\end{proof}

\begin{theorem}[Midpoint selector for marginalized multifunctions]
	\label{thm_marginalized_midpoint_selector}
	Let $F$ be marginalized on a full block $B$ with margins $\ell_j,u_j\in\BReg{B}$.
	Define
	\[
		s_j:=\frac{\ell_j+u_j}{2},
		\qquad
		j=1,\ldots,m.
	\]
	Then $s=(s_1,\ldots,s_m)$ is a regular function and is an exact selector of $F$ wherever the marginal values are determined.
\end{theorem}

\begin{proof}
	By \Cref{lem_regular_functions_finite_algebra}, each $s_j$ is regular.
	Since
	\[
		\ell_j\le s_j\le u_j,
		\qquad
		j=1,\ldots,m,
	\]
	the vector $s(x)$ belongs to
	\[
		\prod_{j=1}^m[\ell_j(x),u_j(x)]
		=
		F(x)
	\]
	whenever the displayed values are determined.
\end{proof}

\begin{corollary}[Lipschitz midpoint selector for marginalized multifunctions]
	\label{cor_lipschitz_marginalized_midpoint_selector}
	Let $F$ be marginalized on a full block $B$ with margins $\ell_j,u_j\in\BReg{B}$.
	Assume that every margin is $L$-Lipschitz-regular.
	Then the midpoint selector from \Cref{thm_marginalized_midpoint_selector} is $\sqrt m L$-Lipschitz-regular.
\end{corollary}

\begin{proof}
	After finite flattening, every scalar branch of
	\[
		s_j:=\frac{\ell_j+u_j}{2}
	\]
	is $L$-Lipschitz on its support.
	Hence, for every common support cell and all $x,y$ in that cell,
	\[
		\abs{s_j(x)-s_j(y)}\le L\nrm{x-y},
		\qquad
		j=1,\ldots,m.
	\]
	Summing the coordinate estimates gives
	\[
		\nrm{s(x)-s(y)}\le \sqrt m L\nrm{x-y}.
	\]
	Thus the vector-valued midpoint selector is $\sqrt m L$-Lipschitz-regular.
\end{proof}

\begin{corollary}[Domain-regular support-wise Steiner selectors]
	\label{cor_domain_regular_supportwise_steiner_selector}
	Let $F$ be domain-regular on a full block $B$ with supports $X_1,\ldots,X_N$.
	Thus the supports are representable by \Cref{dfn_multifunction_domain_conditions}.
	Assume, in addition, that these supports are mutually apart in the sense of \Cref{dfn_mutual_apartness}.
	Assume that, on every support piece $X_k$, the values are inhabited, closed, convex, and totally bounded, and satisfy
	\[
		d_H(F(x),F(y))
		\le
		L\nrm{x-y},
		\qquad
		x,y\in X_k.
	\]
	For $x\in X_k$ put
	\[
		s_k^+(x):=\operatorname{St}(F(x))
	\]
	and give the corresponding complemented summand the zero exterior branch on $X_k^c$.
	Then
	\[
		s:=\sum_{k=1}^N s_k
	\]
	is an $mL$-Lipschitz-regular function.
	It is an exact selector of $F$ on $\dom F$.
\end{corollary}

\begin{proof}
	For each $k$, \Cref{crl_steiner_point_lipschitz_selector} gives, for all $x,y\in X_k$,
	\[
		s_k^+(x)=\operatorname{St}(F(x))\in F(x)
	\]
	and
	\[
		\nrm{s_k^+(x)-s_k^+(y)}
		\le
		mL\nrm{x-y},
		\qquad
		x,y\in X_k.
	\]
	Thus each positive branch is $mL$-Lipschitz on its support.
	Because the supports are representable and mutually apart, the support representation determines the exterior branches whenever a point lies in one of the other supports.
	Thus the displayed finite sum is a regular function with these supports and zero exterior branches.
	Let $x\in\dom F$.
	If $x\in X_l$ for some $l$, then the mutual-apartness data give $x\in X_k^c$ for every $k\ne l$, and hence
	\[
		s(x)
		=
		s_l^+(x)+\sum_{k\ne l}0
		=
		\operatorname{St}(F(x))
		\in F(x).
	\]
	If $x\in(\supp F)^c$, then $x\in X_k^c$ for every $k$, and therefore
	\[
		s(x)=\sum_{k=1}^N0=0.
	\]
	The domain-regular convention gives $F(x)=\{0\}$ on $(\supp F)^c$, so again $s(x)\in F(x)$.
	Thus $s$ is exact on $\dom F$.
\end{proof}

\subsection{Selector Limits and Coherence}
\label{subsec_inside_selector_limits}

\begin{remark}[A weak convergence-transfer principle]
	\label{rem_inside_selectors_converge_in_measure}
	Let $A\subseteq B\subset\R^n$ with $B$ a block, and let $F:A\rightrightarrows\R^m$ be a multifunction.
	Suppose that $(s_i)_i$ is a sequence of measurable functions $s_i:A\to\R^m$ converging in measure to a measurable function $s:A\to\R^m$ on $A\cap B$.
	Assume that there is a sequence $(\alpha_i)_i$ of positive rationals with $\alpha_i\to0$ such that, for every $i$, there exists a finite full multiblock $\mathcal E_i\Subset B$ with
	\[
		\gamma(\mathcal E_i)\le\alpha_i
	\]
	and
	\[
		s_i(x)\in F(x),
		\qquad
		x\in A\cap B\cap\bar{\mathcal E}_i^c.
	\]
	Then, for every $\eps>0$, there is a finite full multiblock $\mathcal E\Subset B$ with
	\[
		\gamma(\mathcal E)\le\eps
	\]
	such that
	\[
		(\{s(x)\}\oplus\eps)\cap F(x)
	\]
	is inhabited for every $x\in A\cap B\cap\bar{\mathcal E}^c$.
	Indeed, choose $i$ so large that
	\[
		\alpha_i\le\eps/2.
	\]
	By convergence in measure, possibly increasing $i$, there exists a finite full multiblock $\mathcal H_i\Subset B$ such that
	\[
		\gamma(\mathcal H_i)\le\eps/2
	\]
	and
	\[
		x\in A\cap B\cap\bar{\mathcal H}_i^c
		\quad\Longrightarrow\quad
		\nrm{s_i(x)-s(x)}\le\eps.
	\]
	Let $\mathcal E$ be the concatenation of $\mathcal E_i$ and $\mathcal H_i$.
	Then $\gamma(\mathcal E)\le\eps$.
	If $x\in A\cap B\cap\bar{\mathcal E}^c$, then
	\[
		s_i(x)\in F(x)
	\]
	and
	\[
		\nrm{s(x)-s_i(x)}\le\eps.
	\]
	Hence $s_i(x)$ witnesses that
	\[
		(\{s(x)\}\oplus\eps)\cap F(x)
	\]
	is inhabited.
\end{remark}

\begin{remark}[Why coherence is extra data]
	\label{rem_constructing_inside_selector_sequences}
	\Cref{rem_inside_selectors_converge_in_measure} is only a transfer principle.
	It does not construct a convergent sequence.
	For arbitrary $\eps$-full simple selectors obtained by unrelated tags, no Cauchy property follows merely from $\eps\to0$.
	At step $i+1$, the new tag may be close to the fiber but far from the tag chosen at step $i$.
	Nesting the domain exception multiblocks does not fix this by itself, because it controls where the selector is certified, not which nearby value is selected there.
	A genuine nested construction would have to impose a second condition, for instance
	\[
		\nrm{s_{i+1}(x)-s_i(x)}\le 2^{-i}
	\]
	on the retained full cells, in addition to metrical membership in the fiber.
	In the tractable inverse-image situation of \cite[Theorem~2]{Osinenko2021constructiveex}, such coherence is obtained from representable slices of the form
	\[
		\{x\in B:\ d(r,F(x))<\delta\}
	\]
	together with the additional predecessor-closeness restrictions.
	The present multifunction hypotheses are intentionally looser and do not include that inverse-slice structure.
	Thus \Cref{thm_domain_regular_block_full_selector} supplies finite $\eps$-full selectors, but does not by itself supply a Cauchy-in-measure sequence of such selectors.
	In the Picard--Steiner construction of \Cref{thm_picard_steiner_certified_approximants}, this role is played by the explicit factorial estimate.
\end{remark}

\begin{remark}
	\label{rem_inside_limit_eps}
	The $\eps$ in \Cref{rem_inside_selectors_converge_in_measure} is unavoidable at this level of data.
	Convergence in measure gives closeness to the limiting function only outside a small multiblock and only at a prescribed numerical accuracy.
	Therefore exactness of $s_i$ outside a small multiblock transfers to $\eps$-metrical membership of the limit outside a small multiblock.
	Exact membership of the limit requires additional inside data, such as a fixed block-full witness as in \Cref{thm_block_full_cauchy_limit_selector}, or a pointwise almost-uniform limiting argument with closed fibers.
	Moreover, a limit obtained by convergence in measure is generally measurable rather than Lipschitz-regular.
\end{remark}

\begin{remark}
	\label{rem_block_full_cauchy_sequence_not_constructed}
	The block-full Cauchy limit statement below is a transfer principle.
	It does not construct the Cauchy-in-measure sequence of selectors.
	It records what follows once such a coherent sequence has been supplied by some additional construction.
\end{remark}

\begin{theorem}[Block-full Cauchy limit selector]
	\label{thm_block_full_cauchy_limit_selector}
	Let $F$ have block-full values on a full block $B\subset\R^n$.
	Let $Q$ be an admissible base block for $F$, and let $D_Q$ be a block-full witness, so that
	\[
		Q\times D_Q\subseteq\Gamma_B(F).
	\]
	Let $(s_k)_{k\in\N}$ be simple regular functions such that
	\[
		s_k(x)\in D_Q,
		\qquad
		x\in Q,\quad k\in\N.
	\]
	Assume that $(s_k)_k$ is Cauchy in measure on $Q$ in the sense of \Cref{dfn_cauchy_in_measure}.
	Then there exists a measurable function $s:Q\to\R^m$, a subsequence $(s_{i_k})_k$, and a subset $Q_\ast\subseteq Q$ full in $Q$ relative to $B$ such that $(s_k)_k$ converges to $s$ in measure on $Q$, $(s_{i_k})_k$ converges almost uniformly to $s$ on $Q$, and $s$ is an exact selector of $F$ on $Q_\ast$.
\end{theorem}

\begin{proof}
	By \Cref{thm_cauchy_in_measure_completion}, applied to the measurable functions $s_k:Q\to\R^m$, there exist a measurable limit $s$, a subsequence $(s_{i_k})_k$, and a subset $Q_\ast\subseteq Q$ full in $Q$ relative to $B$ such that the stated convergence properties hold and $(s_{i_k})_k$ converges pointwise to $s$ on $Q_\ast$.
	Fix $x\in Q_\ast$.
	Write the rational block $D_Q$ as
	\[
		D_Q=\prod_{r=1}^m[a_r,b_r].
	\]
	Since $s_{i_k}(x)\in D_Q$ for every $k$, each coordinate satisfies
	\[
		a_r\le (s_{i_k}(x))_r\le b_r.
	\]
	Pointwise convergence gives
	\[
		a_r\le s_r(x)\le b_r,
		\qquad
		r=1,\ldots,m.
	\]
	Indeed, for every $\eta>0$ there is $K$ such that
	\[
		k\ge K
		\quad\Longrightarrow\quad
		\abs{(s_{i_k}(x))_r-s_r(x)}\le\eta.
	\]
	Hence
	\[
		a_r-\eta
		\le
		s_r(x)
		\le
		b_r+\eta.
	\]
	Since $\eta>0$ is arbitrary, the displayed coordinate bounds follow.
	Thus $s(x)\in D_Q$.
	The graph inclusion
	\[
		Q\times D_Q\subseteq\Gamma_B(F)
	\]
	now yields $s(x)\in F(x)$.
\end{proof}

\begin{remark}
	\label{rem_block_full_selector_no_limit}
	\Cref{thm_block_full_cauchy_limit_selector} is an inside limit statement.
	The approximants take their values in the fixed rational block $D_Q$ before the limit is formed.
	The closedness used in the proof is therefore only the elementary closedness of $D_Q$, not a separate sequential closedness assumption on the fibers of $F$.
\end{remark}

\begin{remark}
	\label{rem_metrical_selector_limits_not_treated}
	A limit theorem for $\eps$-metrical selectors is not used here.
	Passing from
	\[
		(\{s_k(x)\}\oplus\eps_k)\cap F(x)
	\]
	being inhabited with $\eps_k\to0$ to exact membership of a limit requires a closedness principle for the fiber or graph.
	Without such data, $\eps$-metrical selectors are finite approximation objects rather than a route to an exact limiting selector.
\end{remark}

\begin{table}[ht]
	\centering
	{\setlength{\arrayrulewidth}{0.8pt}
	\begin{tabular}{@{}ccc@{}}
		\hline
		\begin{minipage}[t]{0.22\linewidth}\centering\textbf{Domain}\end{minipage}
		&
		\begin{minipage}[t]{0.30\linewidth}\centering\textbf{Value}\end{minipage}
		&
		\begin{minipage}[t]{0.36\linewidth}\centering\textbf{Selector type}\end{minipage}
		\\[0.8ex]
		\hline
		\begin{minipage}[t]{0.22\linewidth}\raggedright $B$-total\end{minipage}
		&
		\begin{minipage}[t]{0.30\linewidth}\raggedright block-full values\end{minipage}
		&
		\begin{minipage}[t]{0.36\linewidth}\raggedright exact Lipschitz; regular for constant tags\end{minipage}
		\\[1.6ex]
		\noalign{\hrule height 0.25pt}
		\noalign{\vskip 1.35ex}
		\begin{minipage}[t]{0.22\linewidth}\raggedright $B$-total\end{minipage}
		&
		\begin{minipage}[t]{0.30\linewidth}\raggedright closed convex Hausdorff-Lipschitz values\end{minipage}
		&
		\begin{minipage}[t]{0.36\linewidth}\raggedright exact Lipschitz Steiner\end{minipage}
		\\[1.6ex]
		\noalign{\hrule height 0.25pt}
		\noalign{\vskip 1.35ex}
		\begin{minipage}[t]{0.22\linewidth}\raggedright $B$-total\end{minipage}
		&
		\begin{minipage}[t]{0.30\linewidth}\raggedright closed convex Hausdorff-Lipschitz values\end{minipage}
		&
		\begin{minipage}[t]{0.36\linewidth}\raggedright Lipschitz $\eps$-metrical regular\end{minipage}
		\\[1.6ex]
		\noalign{\hrule height 0.25pt}
		\noalign{\vskip 1.35ex}
		\begin{minipage}[t]{0.22\linewidth}\raggedright domain-regular\end{minipage}
		&
		\begin{minipage}[t]{0.30\linewidth}\raggedright block-full values\end{minipage}
		&
		\begin{minipage}[t]{0.36\linewidth}\raggedright $\eps$-full simple regular\end{minipage}
		\\[1.6ex]
		\noalign{\hrule height 0.25pt}
		\noalign{\vskip 1.35ex}
		\begin{minipage}[t]{0.22\linewidth}\raggedright domain-regular\end{minipage}
		&
		\begin{minipage}[t]{0.30\linewidth}\raggedright closed convex support-wise Hausdorff-Lipschitz values\end{minipage}
		&
		\begin{minipage}[t]{0.36\linewidth}\raggedright $\eps$-full Lipschitz-regular Steiner\end{minipage}
		\\[1.6ex]
		\noalign{\hrule height 0.25pt}
		\noalign{\vskip 1.35ex}
		\begin{minipage}[t]{0.22\linewidth}\raggedright any domain condition\end{minipage}
		&
		\begin{minipage}[t]{0.30\linewidth}\raggedright marginalized regular values\end{minipage}
		&
		\begin{minipage}[t]{0.36\linewidth}\raggedright exact regular midpoint\end{minipage}
		\\[1.6ex]
		\noalign{\hrule height 0.25pt}
		\noalign{\vskip 1.35ex}
		\begin{minipage}[t]{0.22\linewidth}\raggedright any domain condition\end{minipage}
		&
		\begin{minipage}[t]{0.30\linewidth}\raggedright marginalized Lipschitz values\end{minipage}
		&
		\begin{minipage}[t]{0.36\linewidth}\raggedright Lipschitz-regular midpoint\end{minipage}
		\\
		\hline
	\end{tabular}
	}
	\par\bigskip
	{\setlength{\arrayrulewidth}{0.8pt}
	\begin{tabular}{@{}ccc@{}}
		\hline
		\begin{minipage}[t]{0.22\linewidth}\centering\textbf{Domain}\end{minipage}
		&
		\begin{minipage}[t]{0.30\linewidth}\centering\textbf{Value}\end{minipage}
		&
		\begin{minipage}[t]{0.36\linewidth}\centering\textbf{Limit selector}\end{minipage}
		\\[0.8ex]
		\hline
		\begin{minipage}[t]{0.22\linewidth}\raggedright $B$-total\end{minipage}
		&
		\begin{minipage}[t]{0.30\linewidth}\raggedright block-full values\end{minipage}
		&
		\begin{minipage}[t]{0.36\linewidth}\raggedright exact on a full subset for Cauchy limits staying in one witness block\end{minipage}
		\\[1.6ex]
		\noalign{\hrule height 0.25pt}
		\noalign{\vskip 1.35ex}
		\begin{minipage}[t]{0.22\linewidth}\raggedright $B$-total\end{minipage}
		&
		\begin{minipage}[t]{0.30\linewidth}\raggedright closed convex Hausdorff-Lipschitz values\end{minipage}
		&
		\begin{minipage}[t]{0.36\linewidth}\raggedright not needed before regular approximation\end{minipage}
		\\[1.6ex]
		\noalign{\hrule height 0.25pt}
		\noalign{\vskip 1.35ex}
		\begin{minipage}[t]{0.22\linewidth}\raggedright $B$-total\end{minipage}
		&
		\begin{minipage}[t]{0.30\linewidth}\raggedright closed convex Hausdorff-Lipschitz values\end{minipage}
		&
		\begin{minipage}[t]{0.36\linewidth}\raggedright no exact limit without additional closed graph data\end{minipage}
		\\[1.6ex]
		\noalign{\hrule height 0.25pt}
		\noalign{\vskip 1.35ex}
		\begin{minipage}[t]{0.22\linewidth}\raggedright domain-regular\end{minipage}
		&
		\begin{minipage}[t]{0.30\linewidth}\raggedright block-full values\end{minipage}
		&
		\begin{minipage}[t]{0.36\linewidth}\raggedright convergence in measure gives $\eps$-metrical membership outside a small multiblock\end{minipage}
		\\[1.6ex]
		\noalign{\hrule height 0.25pt}
		\noalign{\vskip 1.35ex}
		\begin{minipage}[t]{0.22\linewidth}\raggedright domain-regular\end{minipage}
		&
		\begin{minipage}[t]{0.30\linewidth}\raggedright closed convex support-wise Hausdorff-Lipschitz values\end{minipage}
		&
		\begin{minipage}[t]{0.36\linewidth}\raggedright no limit needed for the support-wise exact construction\end{minipage}
		\\[1.6ex]
		\noalign{\hrule height 0.25pt}
		\noalign{\vskip 1.35ex}
		\begin{minipage}[t]{0.22\linewidth}\raggedright any domain condition\end{minipage}
		&
		\begin{minipage}[t]{0.30\linewidth}\raggedright marginalized regular values\end{minipage}
		&
		\begin{minipage}[t]{0.36\linewidth}\raggedright exact wherever margins are determined\end{minipage}
		\\[1.6ex]
		\noalign{\hrule height 0.25pt}
		\noalign{\vskip 1.35ex}
		\begin{minipage}[t]{0.22\linewidth}\raggedright any domain condition\end{minipage}
		&
		\begin{minipage}[t]{0.30\linewidth}\raggedright marginalized Lipschitz values\end{minipage}
		&
		\begin{minipage}[t]{0.36\linewidth}\raggedright finite-stage Lipschitz selector\end{minipage}
		\\
		\hline
	\end{tabular}
	}
	\caption{Selector mechanisms used in \Cref{sec_multifunctions}. The last column records what the selector-limit principles of \Cref{subsec_inside_selector_limits} add after passing to measure limits.}
	\label{tab_selector_mechanisms}
\end{table}

\section{Filippov Solutions}
\label{sec_filippov_solutions}

\begin{remark}
	This section applies the selector machinery to differential inclusions.
	The main objects are exact and certified approximate Filippov solutions, sample-and-hold certificates, Picard-type approximants, and viability certificates.
	The statements are arranged so that the finite data certifying the inclusion, the time exceptions, and the integral trajectory are explicit.
\end{remark}

\subsection{Differential Inclusions}
\label{subsec_differential_inclusions}

\begin{definition}[Autonomous differential inclusion]
	\label{dfn_autonomous_differential_inclusion}
	Let $B\subset\R^n$ be a full block, let $I:=[0,T]\subset\R$ be a rational time interval with $T>0$, and let
	\[
		F:B\rightrightarrows\R^n
	\]
	be a multifunction.
	The corresponding autonomous differential inclusion is written
	\[
		\dot x(t)\in F(x(t)),
		\qquad
		t\in I.
	\]
	The block $B$ is called the state block and $F$ is called the velocity multifunction.
\end{definition}

\begin{definition}[Path]
	\label{dfn_path}
	Let $I=[0,T]\subset\R$ be a rational interval.
	A \emph{path} in $\R^n$ on $I$ is a continuous function
	\[
		x:I\to\R^n.
	\]
	If $x(I)\subseteq B$, it is called a path in the block $B$.
\end{definition}

\begin{definition}[Pullback velocity multifunction]
	\label{dfn_pullback_velocity_multifunction}
	Let $F:B\rightrightarrows\R^n$ be a multifunction and let $x:I\to B$ be a path.
	The pullback of $F$ along $x$ is the multifunction
	\[
		F_x:I\rightrightarrows\R^n,
		\qquad
		F_x(t):=F(x(t)).
	\]
	Thus a velocity function $v:I\to\R^n$ is an exact selector of $F_x$ on $A\subseteq I$ precisely when
	\[
		v(t)\in F(x(t)),
		\qquad
		t\in A.
	\]
\end{definition}

\begin{definition}[Integral trajectory]
	\label{dfn_integral_trajectory}
	Let $v:I\to\R^n$ be Riemann integrable on every rational subinterval $[a,b]\subseteq I$.
	A continuous path $x:I\to\R^n$ is called an \emph{integral trajectory} of $v$ with initial value $x_0$ if
	\[
		x(0)=x_0
	\]
	and, for every rational $0\le a\le b\le T$,
	\[
		x(b)-x(a)
		=
		\int_a^b v(t)\diff t.
	\]
\end{definition}

\begin{remark}
	\label{rem_integral_trajectory_absolute_continuity}
	Classically, Filippov solutions are absolutely continuous paths.
	For the present constructive development the integral identity in \Cref{dfn_integral_trajectory} is the useful datum.
	When read classically, a continuous path satisfying that identity with an integrable velocity is absolutely continuous, and its derivative agrees with the velocity almost everywhere.
\end{remark}

\begin{definition}[Exact Filippov solution]
	\label{dfn_exact_filippov_solution}
	Let $F:B\rightrightarrows\R^n$ be a velocity multifunction and let $x_0\in B$.
	A pair $(x,v)$ is called an \emph{exact Filippov solution} on $I$ if:
	\begin{enumerate}
		\item $x:I\to B$ is an integral trajectory of $v$ with initial value $x_0$,
		\item $v$ is an exact selector of the pullback multifunction $F_x$ on $I\cap\dom v$.
	\end{enumerate}
\end{definition}

\begin{remark}
	\label{rem_exact_solution_not_general_existence}
	\Cref{dfn_exact_filippov_solution} is a verification definition.
	It does not assert that an exact selector of the pullback multifunction can be extracted in general.
	The selector results in \Cref{sec_multifunctions} identify concrete cases where such selector data are available.
\end{remark}

\subsection{Certified Approximate Solutions}
\label{subsec_certified_approximate_solutions}

\begin{definition}[$\eps$-full Filippov certificate]
	\label{dfn_eps_full_filippov_certificate}
	Let $F:B\rightrightarrows\R^n$ be a velocity multifunction, let $x_0\in B$, and let $\eps>0$.
	An \emph{$\eps$-full Filippov certificate} on $I$ is a tuple
	\[
		(x,v,\mathcal E,M)
	\]
	such that:
	\begin{enumerate}
		\item $x:I\to B$ is an integral trajectory of $v$ with initial value $x_0$,
		\item $\mathcal E\Subset I$ is a finite full multiblock with $\gamma(\mathcal E)\le\eps$,
		\item $v$ is an exact selector of $F_x$ on
		\[
			I\cap\dom v\cap\bar{\mathcal E}^c,
		\]
		\item $\nrm{v(t)}\le M$ for every $t\in I\cap\dom v$.
	\end{enumerate}
	The multiblock $\mathcal E$ is called the time-exception multiblock of the certificate.
\end{definition}

\begin{definition}[$(\eps,\eta)$-metrical Filippov certificate]
	\label{dfn_metrical_filippov_certificate}
	Let $F:B\rightrightarrows\R^n$ be a velocity multifunction, let $x_0\in B$, and let $\eps,\eta\ge0$.
	An \emph{$(\eps,\eta)$-metrical Filippov certificate} on $I$ is a tuple
	\[
		(x,v,\mathcal E,M)
	\]
	such that:
	\begin{enumerate}
		\item $x:I\to B$ is an integral trajectory of $v$ with initial value $x_0$,
		\item $\mathcal E\Subset I$ is a finite full multiblock with $\gamma(\mathcal E)\le\eps$,
		\item for every $t\in I\cap\dom v\cap\bar{\mathcal E}^c$, the set
		\[
			(\{v(t)\}\oplus\eta)\cap F(x(t))
		\]
		is inhabited,
		\item $\nrm{v(t)}\le M$ for every $t\in I\cap\dom v$.
	\end{enumerate}
\end{definition}

\begin{remark}
	\label{rem_certificate_meaning}
	The certificate in \Cref{dfn_eps_full_filippov_certificate} is not a comparison with an unavailable ideal solution.
	It records the finite statement that the integral path follows velocities from $F(x(t))$ outside a time set of cost at most $\eps$, while the velocity is bounded everywhere it is determined.
	Reducing $\eps$ refines this statement in the same sense as the exception multiblocks used for representable sets and measurable functions.
\end{remark}

\begin{proposition}[Integral effect of a time exception]
	\label{prop_time_exception_integral_effect}
	Let $v,w:I\to\R^n$ be Riemann integrable on rational subintervals and let $\mathcal E\Subset I$ be a finite full multiblock.
	Assume that, for some $M,\eta>0$,
	\[
		\nrm{v(t)}\le M,
		\qquad
		\nrm{w(t)}\le M,
		\qquad
		t\in I,
	\]
	and
	\[
		t\in I\cap\bar{\mathcal E}^c
		\quad\Longrightarrow\quad
		\nrm{v(t)-w(t)}\le\eta.
	\]
	If $x$ and $z$ are integral trajectories of $v$ and $w$ with the same initial value, then, for every rational $t\in I$,
	\[
		\nrm{x(t)-z(t)}
		\le
		\eta T+2M\gamma(\mathcal E).
	\]
\end{proposition}

\begin{proof}
	Fix a rational $t\in I$.
	By the integral-trajectory identities,
	\[
		x(t)-z(t)
		=
		\int_0^t (v(\tau)-w(\tau))\diff\tau.
	\]
	Split the interval $[0,t]$ into the finite multiblock pieces covered by $\mathcal E$ and their complement in the one-dimensional partition induced by the endpoints of $\mathcal E$.
	On the complement of $\bar{\mathcal E}$ the integrand has norm at most $\eta$.
	On the pieces covered by $\mathcal E$ it has norm at most $2M$.
	Therefore
	\[
		\nrm{x(t)-z(t)}
		\le
		\eta\mu([0,t])+2M\gamma(\mathcal E)
		\le
		\eta T+2M\gamma(\mathcal E).
	\]
\end{proof}

\begin{corollary}[Certificate stability]
	\label{cor_filippov_certificate_stability}
	Let $(x,v,\mathcal E,M)$ be an $\eps$-full Filippov certificate.
	Let $z$ be an integral trajectory of a velocity $w$ with the same initial value.
	If $\nrm{w(t)}\le M$ on $I$ and
	\[
		t\in I\cap\bar{\mathcal E}^c
		\quad\Longrightarrow\quad
		\nrm{v(t)-w(t)}\le\eta,
	\]
	then, for every rational $t\in I$,
	\[
		\nrm{x(t)-z(t)}
		\le
		\eta T+2M\eps.
	\]
\end{corollary}

\begin{proof}
	This is \Cref{prop_time_exception_integral_effect} and the estimate $\gamma(\mathcal E)\le\eps$.
\end{proof}

\begin{proposition}[Practical stability transfer]
	\label{prop_practical_stability_transfer}
	Let $(x,v,\mathcal E,M)$ be an $\eps$-full Filippov certificate on $I=[0,T]$.
	Let $z$ be an integral trajectory of a velocity $w$ with the same initial value and with $\nrm{w(t)}\le M$ on $I$.
	Assume that
	\[
		t\in I\cap\bar{\mathcal E}^c
		\quad\Longrightarrow\quad
		\nrm{v(t)-w(t)}\le\eta.
	\]
	Set
	\[
		\rho:=\eta T+2M\eps.
	\]
	If, for some rational $\tau\in I$ and $r>0$,
	\[
		\nrm{x(t)}\le r,
		\qquad
		t\in[\tau,T],
	\]
	then
	\[
		\nrm{z(t)}\le r+\rho,
		\qquad
		t\in[\tau,T].
	\]
\end{proposition}

\begin{proof}
	By \Cref{cor_filippov_certificate_stability},
	\[
		\nrm{x(t)-z(t)}\le\rho,
		\qquad
		t\in I.
	\]
	For $t\in[\tau,T]$, the triangle inequality gives
	\[
		\nrm{z(t)}
		\le
		\nrm{x(t)}+\nrm{x(t)-z(t)}
		\le
		r+\rho.
	\]
\end{proof}

\begin{remark}
	\label{rem_practical_stability_asymptotic_reading}
	If a family of certificates for increasing horizons has the property that, for every enlarged target radius $r+\delta$, there is a finite reaching time after which the certified path remains in the ball of radius $r+\delta$, then \Cref{prop_practical_stability_transfer} transfers the same statement to every trajectory whose velocity is close in the displayed certificate sense.
	The radius is enlarged by the explicit quantity $\rho=\eta T+2M\eps$ on the considered finite horizon.
\end{remark}

\begin{proposition}[Stability through a Lipschitz selection mechanism]
	\label{prop_lipschitz_selection_mechanism_stability}
	Let $s:B\to\R^n$ be an $L_s$-Lipschitz function with
	\[
		\nrm{s(y)}\le M,
		\qquad
		y\in B.
	\]
	Let $x,z:I\to B$ be integral trajectories of velocities $v,w$ with the same initial value, and assume
	\[
		\nrm{v(t)}\le M,
		\qquad
		\nrm{w(t)}\le M.
	\]
	Let $\mathcal E\Subset I$ be a finite full multiblock with $\gamma(\mathcal E)\le\eps$.
	Assume that, for every $t\in I\cap\bar{\mathcal E}^c$,
	\[
		\nrm{v(t)-s(x(t))}\le\alpha,
		\qquad
		\nrm{w(t)-s(z(t))}\le\beta.
	\]
	Then, for every rational $t\in I$,
	\[
		\nrm{x(t)-z(t)}
		\le
		\bigl((\alpha+\beta)T+2M\eps\bigr)\exp(L_sT).
	\]
\end{proposition}

\begin{proof}
	Put
	\[
		u(t):=\nrm{x(t)-z(t)}.
	\]
	Outside $\bar{\mathcal E}$,
	\[
	\begin{aligned}
		\nrm{v(t)-w(t)}
		&\le
		\nrm{v(t)-s(x(t))}
		+
		\nrm{s(x(t))-s(z(t))}
		+
		\nrm{s(z(t))-w(t)}
		\\
		&\le
		\alpha+\beta+L_su(t).
	\end{aligned}
	\]
	On $\bar{\mathcal E}$ we use the crude bound $\nrm{v(t)-w(t)}\le2M$.
	Hence
	\[
		u(t)
		\le
		(\alpha+\beta)T+2M\gamma(\mathcal E)+L_s\int_0^t u(\tau)\diff\tau.
	\]
	The elementary integral Gronwall estimate gives the displayed bound.
\end{proof}

\begin{remark}
	\label{rem_inclusion_membership_not_stability}
	Plain membership $v(t)\in F(x(t))$ and $w(t)\in F(z(t))$ does not by itself imply that the velocities are close.
	A single fiber may contain two far apart velocities.
	The preceding proposition records the useful finite situation: both trajectories follow the same Lipschitz selection mechanism, up to the displayed tolerances and a small time-exception multiblock.
\end{remark}

\begin{remark}
	\label{rem_stability_estimates_relation}
	\Cref{cor_filippov_certificate_stability} is an a posteriori comparison: it applies once the two velocity functions are already known to be close outside a small time-exception multiblock.
	\Cref{prop_lipschitz_selection_mechanism_stability} is a sufficient mechanism for obtaining such closeness.
	It says that if both velocities are generated, up to small errors, by the same Lipschitz selection rule evaluated along their own paths, then trajectory closeness follows from the Lipschitz estimate and Gronwall's inequality.
	Without a common selection mechanism or a direct velocity-closeness certificate, arbitrary choices from a wide fiber can produce trajectories that separate.
\end{remark}

\subsection{Sample-Hold Solutions}
\label{subsec_sample_hold_solutions}

\begin{definition}[Sample-and-hold data]
	\label{dfn_sample_hold_data}
	Let $I=[0,T]$ and let $x_0\in\Q^n$.
	A \emph{sample-and-hold datum} consists of rational times
	\[
		0=t_0<t_1<\cdots<t_N=T
	\]
	and rational velocity tags
	\[
		y_j\in\Q^n,
		\qquad
		j=0,\ldots,N-1.
	\]
	The associated polygonal path $x_h:I\to\R^n$ and held velocity $v_h:I\to\R^n$ are defined by
	\[
		x_h(t):=x_j+(t-t_j)y_j,
		\qquad
		v_h(t):=y_j,
		\qquad
		t\in[t_j,t_{j+1}),
	\]
	where
	\[
		x_{j+1}:=x_j+(t_{j+1}-t_j)y_j,
		\qquad
		j=0,\ldots,N-1.
	\]
	At $t=T$ one sets
	\[
		x_h(T):=x_N,
		\qquad
		v_h(T):=y_{N-1}.
	\]
\end{definition}

\begin{definition}[Switching set and switching layer]
	\label{dfn_switching_set_layer}
	For a sample-and-hold datum, the \emph{switching set} is the finite set
	\[
		\mathcal S_h:=\{t_j:j=1,\ldots,N-1\}.
	\]
	For $\eps>0$, an \emph{$\eps$-switching layer} is a finite full multiblock $\mathcal E_\eps\Subset I$ such that
	\[
		\mathcal S_h\subseteq\bar{\mathcal E}_\eps,
		\qquad
		\gamma(\mathcal E_\eps)\le\eps.
	\]
\end{definition}

\begin{definition}[Certified sample-and-hold datum]
	\label{dfn_certified_sample_hold_data}
	Let $F:B\rightrightarrows\R^n$ be a velocity multifunction.
	A sample-and-hold datum is called \emph{certified for $F$ on $B$} if, for every $j=0,\ldots,N-1$, there are full blocks
	\[
		Q_j\subseteq B,
		\qquad
		D_j\subset\R^n
	\]
	such that:
	\begin{enumerate}
		\item $Q_j\times D_j\subseteq\Gamma_B(F)$,
		\item $y_j\in D_j$,
		\item $x_j,x_{j+1}\in Q_j$.
	\end{enumerate}
	The third condition implies
	\[
		x_h([t_j,t_{j+1}])\subseteq Q_j
	\]
	since every block is convex.
\end{definition}

\begin{theorem}[Sample-and-hold Filippov certificate]
	\label{thm_sample_hold_filippov_certificate}
	Let $F:B\rightrightarrows\R^n$ be a velocity multifunction and let a sample-and-hold datum be certified for $F$ on $B$.
	Let
	\[
		M\ge\max_{j=0,\ldots,N-1}\nrm{y_j}.
	\]
	Then, for every $\eps>0$, the associated path and held velocity give an $\eps$-full Filippov certificate
	\[
		(x_h,v_h,\mathcal E_\eps,M)
	\]
	for a finite time-exception multiblock $\mathcal E_\eps\Subset I$ with $\gamma(\mathcal E_\eps)\le\eps$.
\end{theorem}

\begin{proof}
	Choose an $\eps$-switching layer $\mathcal E_\eps$ in the sense of \Cref{dfn_switching_set_layer}.
	On each interval component of
	\[
		I\cap\bar{\mathcal E}_\eps^c
	\]
	the held velocity is constant and equal to some $y_j$.
	For such a time $t\in[t_j,t_{j+1})$, \Cref{dfn_certified_sample_hold_data} gives
	\[
		x_h(t)\in Q_j,
		\qquad
		y_j\in D_j,
		\qquad
		Q_j\times D_j\subseteq\Gamma_B(F).
	\]
	Therefore
	\[
		v_h(t)=y_j\in F(x_h(t)).
	\]
	The polygonal identity
	\[
		x_h(b)-x_h(a)=\int_a^b v_h(t)\diff t
	\]
	for rational $0\le a\le b\le T$ follows by summing the constant-velocity contributions over the finite partition.
	The bound by $M$ is immediate from the definition of $M$.
\end{proof}

\begin{corollary}[Exact sample-and-hold solutions away from switches]
	\label{cor_exact_sample_hold_away_switches}
	Under the assumptions of \Cref{thm_sample_hold_filippov_certificate}, the sample-and-hold pair $(x_h,v_h)$ satisfies
	\[
		v_h(t)\in F(x_h(t))
	\]
	for every $j=0,\ldots,N-1$ and every $t\in(t_j,t_{j+1})$.
	Equivalently,
	\[
		t\in I\setminus\mathcal S_h
		\quad\Longrightarrow\quad
		v_h(t)\in F(x_h(t)).
	\]
	For every $\eps>0$ and every $\eps$-switching layer $\mathcal E_\eps$ satisfying
	\[
		\mathcal S_h\subseteq\bar{\mathcal E}_\eps,
		\qquad
		\gamma(\mathcal E_\eps)\le\eps,
	\]
	the tuple $(x_h,v_h,\mathcal E_\eps,M)$ is an $\eps$-full Filippov certificate.
\end{corollary}

\begin{proof}
	This is the membership calculation in the proof of \Cref{thm_sample_hold_filippov_certificate}.
\end{proof}

\begin{remark}
	\label{rem_exact_sample_hold_engineering}
	The corollary is the finite verification statement behind sample-and-hold simulation.
	An engineer checks finitely many blocks $Q_j$ and velocity tags $y_j$.
	The differential inclusion is then certified on each open holding interval, and the only times not certified exactly are the finitely many switch times, which are covered by an arbitrarily small switching layer.
\end{remark}

\begin{remark}
	\label{rem_sample_hold_relation_classical}
	Sample-and-hold arcs are the finite objects behind Euler-type constructions for differential inclusions; compare the classical treatments of differential inclusions and viability theory in \cite{Aubin2012DifferentialIn,Filippov1962certainquestio}.
	In the present setting the finite certificate is the main object.
	One records the blocks $Q_j$, the velocity blocks $D_j$, the tags $y_j$, and the confinement checks $x_j,x_{j+1}\in Q_j$.
	No measurable-selection theorem is used in this verification step.
\end{remark}

\begin{theorem}[Steiner sample-and-hold metrical certificate]
	\label{thm_steiner_sample_hold_metrical_certificate}
	Let $F:B\rightrightarrows\R^n$ be $B$-total and Hausdorff-continuous on a full block $B$ with modulus $\omega$.
	Assume that every value is closed and convex.
	Let $x_0\in B$ and assume that there is $M>0$ such that
	\[
		\nrm{\operatorname{St}(F(x))}\le M,
		\qquad
		x\in B.
	\]
	Let
	\[
		0=t_0<t_1<\cdots<t_N=T
	\]
	be rational times, put
	\[
		h:=\max_{j=0,\ldots,N-1}(t_{j+1}-t_j),
	\]
	let $r_h\in\Q_{>0}$ satisfy
	\[
		Mh\le r_h,
	\]
	and define
	\[
		y_j:=\operatorname{St}(F(x_j)),
		\qquad
		x_{j+1}:=x_j+(t_{j+1}-t_j)y_j.
	\]
	Assume the confinement condition
	\[
		\{z\in\R^n:\nrm{z-x_0}\le MT\}\subseteq B.
	\]
	Then the associated sample-and-hold path and velocity give, for every $\eps>0$ and every $\rho>0$, an
	\[
		(\eps,\omega(r_h)+\rho)
	\]
	-metrical Filippov certificate.
\end{theorem}

\begin{proof}
	The bound on the Steiner values gives $\nrm{y_j}\le M$.
	Hence the polygonal path satisfies
	\[
		\nrm{x_h(t)-x_0}\le Mt\le MT,
	\]
	and the confinement condition gives $x_h(t)\in B$.
	Choose an $\eps$-switching layer for the finite switching set.
	For $t\in[t_j,t_{j+1})$ outside this layer,
	\[
		v_h(t)=y_j\in F(x_j)
	\]
	by \Cref{crl_steiner_point_lipschitz_selector}, applied with $A=C=F(x_j)$.
	Moreover,
	\[
		\nrm{x_h(t)-x_j}\le M(t-t_j)\le Mh.
	\]
	Hausdorff-continuity gives
	\[
		d(y_j,F(x_h(t)))\le\omega(r_h).
	\]
	Since
	\[
		\omega(r_h)<\omega(r_h)+\rho,
	\]
	locatedness of the value $F(x_h(t))$ supplies a point
	\[
		w_t\in F(x_h(t))
	\]
	with
	\[
		\nrm{w_t-y_j}<\omega(r_h)+\rho.
	\]
	This point witnesses the metrical membership condition.
	The integral-trajectory identity and the velocity bound are the same as in \Cref{thm_sample_hold_filippov_certificate}.
\end{proof}

\begin{remark}
	\label{rem_constructive_filippov_counterpart}
	\Cref{thm_steiner_sample_hold_metrical_certificate} is the closest finite counterpart here to the Euler part of the classical Filippov theorem.
	Classically, one would pass from such polygonal arcs to an exact solution by compactness and a closed-graph or upper-semicontinuity argument.
	In the present constructive setting the finite output is retained as a metrical certificate.
	The exact limit is recovered in \Cref{thm_picard_steiner_certified_approximants} only when Hausdorff-continuity is strengthened to Hausdorff-Lipschitz continuity, so that the Steiner-selected field is a Lipschitz single-valued field.
\end{remark}

\begin{table}[ht]
	\centering
	\begin{tabular}{@{}lll@{}}
		\hline
		\textbf{Classical step} & \textbf{Constructive replacement} & \textbf{Output here}
		\\
		\hline
		Euler polygonal arcs & sample-and-hold data & finite certificate
		\\
		compactness subsequence & avoided at finite stage & no choice principle
		\\
		upper-semicontinuity closure & Hausdorff modulus estimate & metrical membership
		\\
		measurable selector & regular or Steiner selector when available & integrable velocity
		\\
		Lipschitz selected field & Picard--Steiner iteration & exact solution
		\\
		\hline
	\end{tabular}
	\caption{Classical Filippov existence mechanisms and the corresponding constructive objects used in this section.}
	\label{tab_filippov_constructive_comparison}
\end{table}

\subsection{Picard-Type Certified Approximants}
\label{subsec_picard_type_certified_approximants}

\begin{theorem}[Picard--Lindelof theorem for Steiner-selected inclusions]
	\label{thm_picard_steiner_certified_approximants}
	Let $F:B\rightrightarrows\R^n$ be $B$-total and $L$-Hausdorff-Lipschitz on a full block $B\subset\R^n$.
	Assume that every value $F(x)$, $x\in B$, is closed and convex.
	Thus every value is inhabited and totally bounded by \Cref{dfn_hausdorff_lipschitz_multifunction}, and hence located by \Cref{lem_totally_bounded_located}.
	Write $C_{\mathrm{St},n}$ for the Steiner Lipschitz constant in dimension $n$.
	Let $x_0\in B$ and assume that there is $M>0$ such that
	\[
		\nrm{\operatorname{St}(F(x))}\le M,
		\qquad
		x\in B.
	\]
	Assume also the confinement condition
	\[
		\{z\in\R^n:\nrm{z-x_0}\le MT\}\subseteq B.
	\]
	Define paths and velocities inductively by
	\[
		x^{(0)}(t):=x_0,
		\qquad
		v^{(k)}(t):=\operatorname{St}(F(x^{(k)}(t))),
	\]
	and
	\[
		x^{(k+1)}(t)
		:=
		x_0+\int_0^t v^{(k)}(\tau)\diff\tau.
	\]
	Then each $x^{(k)}$ is a path in $B$, each $v^{(k)}$ is integrable on rational subintervals, and, for every $k\in\N$, the tuple
	\[
		(x^{(k+1)},v^{(k)},\emptyset,M)
	\]
	is a $(0,\eta_k)$-metrical Filippov certificate with
	\[
		\eta_k
		:=
		L M\frac{(C_{\mathrm{St},n} L)^kT^{k+1}}{(k+1)!}
		+
		2^{-(k+1)}.
	\]
	Consequently, for every $\eta>0$, there exists $k$ such that $(x^{(k+1)},v^{(k)},\emptyset,M)$ is a $(0,\eta)$-metrical Filippov certificate.
	Moreover, the paths $x^{(k)}$ are uniformly Cauchy.
	By uniform completeness of the path space, their uniform limit $x$ satisfies
	\[
		x(t)=x_0+\int_0^t s(x(\tau))\diff\tau,
		\qquad
		s(y):=\operatorname{St}(F(y)).
	\]
	Thus $(x,s\circ x)$ is an exact Filippov solution.
\end{theorem}

\begin{proof}
	By \Cref{thm_steiner_hausdorff_lipschitz_selector}, the map
	\[
		s(x):=\operatorname{St}(F(x))
	\]
	is an exact $C_{\mathrm{St},n}L$-Lipschitz selector of $F$ on $B$.
	The assumed bound gives
	\[
		\nrm{s(x)}\le M,
		\qquad
		x\in B.
	\]
	Inductively, if $x^{(k)}$ is a path in $B$, then
	\[
		v^{(k)}=s\circ x^{(k)}
	\]
	is continuous and hence Riemann integrable on rational subintervals.
	Moreover,
	\[
		\nrm{x^{(k+1)}(t)-x_0}
		\le
		\int_0^t\nrm{v^{(k)}(\tau)}\diff\tau
		\le
		Mt
		\le
		MT,
	\]
	so the confinement condition gives $x^{(k+1)}(t)\in B$.
	This proves the path and integrability assertions.

	Put
	\[
		\Delta_k(t):=\nrm{x^{(k+1)}(t)-x^{(k)}(t)}.
	\]
	For $k=0$,
	\[
		\Delta_0(t)
		\le
		Mt.
	\]
	If
	\[
		\Delta_k(t)\le M\frac{(C_{\mathrm{St},n}L)^kt^{k+1}}{(k+1)!},
	\]
	then
	\[
	\begin{aligned}
		\Delta_{k+1}(t)
		&\le
		\int_0^t\nrm{v^{(k+1)}(\tau)-v^{(k)}(\tau)}\diff\tau
		\\
		&\le
		C_{\mathrm{St},n}L\int_0^t\Delta_k(\tau)\diff\tau
		\\
		&\le
		M\frac{(C_{\mathrm{St},n}L)^{k+1}t^{k+2}}{(k+2)!}.
	\end{aligned}
	\]
	Thus, for every $t\in I$,
	\[
		\nrm{x^{(k+1)}(t)-x^{(k)}(t)}
		\le
		M\frac{(C_{\mathrm{St},n}L)^kT^{k+1}}{(k+1)!}.
	\]
	Put
	\[
		\theta_k:=
		LM\frac{(C_{\mathrm{St},n}L)^kT^{k+1}}{(k+1)!}.
	\]
	Since $v^{(k)}(t)\in F(x^{(k)}(t))$ and $F$ is $L$-Hausdorff-Lipschitz,
	\[
		d\bigl(v^{(k)}(t),F(x^{(k+1)}(t))\bigr)
		\le
		L\nrm{x^{(k+1)}(t)-x^{(k)}(t)}
		\le
		\theta_k.
	\]
	Since
	\[
		\theta_k<\eta_k,
	\]
	the upper-estimate part of locatedness supplies a point
	\[
		w_t\in F(x^{(k+1)}(t))
	\]
	with
	\[
		\nrm{w_t-v^{(k)}(t)}<\eta_k.
	\]
	This point $w_t$ witnesses that
	\[
		(\{v^{(k)}(t)\}\oplus\eta_k)\cap F(x^{(k+1)}(t))
	\]
	is inhabited for every $t\in I$.
	The path $x^{(k+1)}$ is an integral trajectory of $v^{(k)}$ by definition:
	\[
		x^{(k+1)}(b)-x^{(k+1)}(a)
		=
		\int_a^b v^{(k)}(\tau)\diff\tau
	\]
	for rational $0\le a\le b\le T$.
	The bound $\nrm{v^{(k)}(t)}\le M$ was shown above.
	Thus all conditions of \Cref{dfn_metrical_filippov_certificate} hold with empty exception multiblock.
	Finally, put
	\[
		q:=C_{\mathrm{St},n}LT.
	\]
	Then
	\[
		\theta_k=LMT\frac{q^k}{(k+1)!}
	\]
	If $q=0$, then $\theta_k=0$ for every $k$.
	If $q>0$, then
	\[
		\frac{\theta_{k+1}}{\theta_k}
		=
		\frac{q}{k+2}.
	\]
	In the second case, for all sufficiently large $k$, this ratio is at most $1/2$, so the tail of $(\theta_k)_k$ is bounded by a geometric tail.
	The added term $2^{-(k+1)}$ tends to zero as well.
	Hence $\eta_k\to0$.
	The estimates in the proof show that $(x^{(k)})_k$ is uniformly Cauchy, because
	\[
		\nrm{x^{(k+1)}(t)-x^{(k)}(t)}
		\le
		M\frac{(C_{\mathrm{St},n}L)^kT^{k+1}}{(k+1)!}
	\]
	uniformly in $t\in I$.
	If $x$ denotes the uniform limit and $s(y):=\operatorname{St}(F(y))$, then the Lipschitz estimate for $s$ gives uniform convergence
	\[
		s(x^{(k)}(\cdot))\to s(x(\cdot)).
	\]
	The limiting velocity is continuous, hence Riemann integrable, and satisfies
	\[
		s(x(t))\in F(x(t)).
	\]
	Uniform convergence of the velocities allows the integral identity
	\[
		x(t)=x_0+\int_0^t s(x(\tau))\diff\tau
	\]
	to be obtained by passing to the limit in the integral identities for $x^{(k+1)}$.
	This proves the exact solution claim.
\end{proof}

\begin{remark}
	\label{rem_picard_steiner_scope}
	Thus, under the stronger Steiner--Hausdorff assumptions of \Cref{thm_picard_steiner_certified_approximants}, there is no integrability obstruction: the limiting velocity is continuous.
	This is a special uniform-convergence route, essentially the constructive ordinary differential equation proof for the Lipschitz single-valued field
	\[
		x\mapsto\operatorname{St}(F(x)).
	\]
	The more general selector-limit statements in \Cref{subsec_inside_selector_limits} are weaker.
	There the limiting object may be merely measurable, and exact fiber membership still requires closed inside data.
	Classical Filippov existence theorems cover broader semicontinuous and measurable settings where no canonical continuous selector is available; those are the situations in which measurable-selection machinery becomes relevant.
\end{remark}

\begin{corollary}[Stability of Picard--Steiner approximants]
	\label{cor_picard_steiner_approximant_stability}
	Under the assumptions of \Cref{thm_picard_steiner_certified_approximants}, let $x$ be the exact solution obtained there.
	Then, for every $k$ and every rational $t\in I$,
	\[
		\nrm{x(t)-x^{(k)}(t)}
		\le
		\sum_{r=k}^{\infty}
		M\frac{(C_{\mathrm{St},n}L)^rT^{r+1}}{(r+1)!}.
	\]
	In particular, every sufficiently high Picard--Steiner certificate remains uniformly close to the exact solution on the whole interval.
\end{corollary}

\begin{proof}
	The proof of \Cref{thm_picard_steiner_certified_approximants} gives
	\[
		\nrm{x^{(r+1)}(t)-x^{(r)}(t)}
		\le
		M\frac{(C_{\mathrm{St},n}L)^rT^{r+1}}{(r+1)!}
	\]
	uniformly in $t$.
	Summing the telescopic tail and passing to the uniform limit gives the estimate.
\end{proof}

\subsection{Exact Solutions from Exact Selectors}
\label{subsec_exact_solutions_from_exact_selectors}

\begin{remark}[Exact selector gives an exact inclusion trajectory]
	\label{rem_exact_selector_exact_inclusion_trajectory}
	Let $F:B\rightrightarrows\R^n$ be a velocity multifunction and let $x:I\to B$ be a path.
	Assume that a velocity $v:I\to\R^n$ is Riemann integrable on rational subintervals, that $x$ is an integral trajectory of $v$ with initial value $x_0$, and that $v$ is an exact selector of $F_x$ on $I\cap\dom v$.
	Then $(x,v)$ is an exact Filippov solution in the sense of \Cref{dfn_exact_filippov_solution}, because the two displayed assumptions are precisely the two conditions in that definition.
\end{remark}

\begin{corollary}[Block-full exact sample-and-hold case]
	\label{cor_block_full_exact_sample_hold_case}
	Let $F$ be $B$-total and have block-full values on $B$.
	Assume that a sample-and-hold datum satisfies
	\[
		x_h(I)\subseteq B
	\]
	and that all velocity tags belong to one block-full witness $D_B$ with
	\[
		B\times D_B\subseteq\Gamma_B(F).
	\]
	Then the associated sample-and-hold pair is exact away from its finite switching set and gives an $\eps$-full Filippov certificate for every $\eps>0$.
\end{corollary}

\begin{proof}
	Take $Q_j:=B$ and $D_j:=D_B$ in \Cref{dfn_certified_sample_hold_data}.
	Then \Cref{thm_sample_hold_filippov_certificate,cor_exact_sample_hold_away_switches} apply.
\end{proof}

\begin{remark}
	\label{rem_exact_solution_rarity}
	Exact regular selectors are useful when they are supplied by the value structure, for instance by block-full data with a constant tag or by the midpoint construction for marginalized values.
	For general multifunctions the constructive output expected from this section is the finite certificate of \Cref{dfn_eps_full_filippov_certificate} or the sample-and-hold certificate of \Cref{thm_sample_hold_filippov_certificate}, rather than an unrestricted exact existence theorem.
\end{remark}

\subsection{Viable Solutions}
\label{subsec_viable_solutions}

\begin{definition}[Viable and approximately viable paths]
	\label{dfn_viable_path}
	Let $K\subseteq B$ be located.
	A path $x:I\to B$ is called \emph{viable in $K$} if
	\[
		x(t)\in K,
		\qquad
		t\in I.
	\]
	It is called \emph{$\rho$-viable in $K$} if
	\[
		d(x(t),K)\le\rho,
		\qquad
		t\in I.
	\]
	A Filippov solution or certificate is called viable, respectively $\rho$-viable, if its path has the corresponding property.
\end{definition}

\begin{definition}[Projection witness]
	\label{dfn_projection_witness_constraint}
	Let $K\subseteq B$ be located and closed.
	A \emph{projection witness} for $K$ on $B$ is a map
	\[
		\proj{\cdot}{K}:B\to K
	\]
	such that, for every $z\in B$,
	\[
		d(z,K)=\nrm{z-\proj{z}{K}}
	\]
	and, for every $y\in K$,
	\[
		\scal{z-\proj{z}{K},y-\proj{z}{K}}\le0.
	\]
\end{definition}

\begin{figure}[ht]
	\centering
	\begin{tikzpicture}[>=Stealth, scale=1.0]
	\coordinate (p) at (3.1,1.4);
	\begin{scope}[shift={(0.15,0.15)}]
		\fill[blue!8] (-0.2,-0.08) .. controls (0.0,1.28) and (1.62,2.0) .. (3.0,1.22)
			.. controls (3.92,0.64) and (3.34,-0.56) .. (1.8,-0.78)
			.. controls (0.68,-0.9) and (-0.42,-0.68) .. (-0.2,-0.08);
		\draw[blue!70!black, thick] (-0.2,-0.08) .. controls (0.0,1.28) and (1.62,2.0) .. (3.0,1.22)
			.. controls (3.92,0.64) and (3.34,-0.56) .. (1.8,-0.78)
			.. controls (0.68,-0.9) and (-0.42,-0.68) .. (-0.2,-0.08);
	\end{scope}
	\node[blue!70!black] at (1.65,-0.3) {$K$};

	\coordinate (z) at (4.39,2.40);
	\coordinate (y) at (1.02,0.8);

	\draw[orange!80!black, thick] ($(p)+(-0.78,0.96)$) -- ($(p)+(0.98,-1.2)$);
	\draw[gray!60, dashed] ($(p)+(-1.02,-0.82)$) -- ($(p)+(0.01,0.01)$);
	\fill[black] (p) circle (1.5pt);
	\fill[black] (z) circle (1.5pt) node[above right] {$z$};
	\fill[black] (y) circle (1.3pt) node[above left] {$y$};

	\draw[->, red!75!black, thick] (p) -- (z) node[midway, above left, xshift=-0.08cm, yshift=0.25cm] {$z-\proj{z}{K}$};
	\draw[->, blue!70!black, thick] (p) -- (y) node[midway, below left, xshift=0.9cm, yshift=-0.2cm] {$y-\proj{z}{K}$};
	\node[black, anchor=west] at ($(p)+(0.17,0.05)$) {$\proj{z}{K}$};

	\draw[gray!80] ($(p)+(-0.10,0.12)$) -- ($(p)+(0.03,0.23)$) -- ($(p)+(0.14,0.10)$);
	\node[orange!80!black, anchor=west] at (3.8,0.5) {$\scal{z-\proj{z}{K},\,y-\proj{z}{K}}\le0$};
\end{tikzpicture}
	\caption{Projection witness for a convex constraint. The residual $z-\proj{z}{K}$ is outward normal at $\proj{z}{K}$, hence its scalar product with every feasible displacement $y-\proj{z}{K}$ is non-positive.}
	\label{fig_projection_witness}
\end{figure}
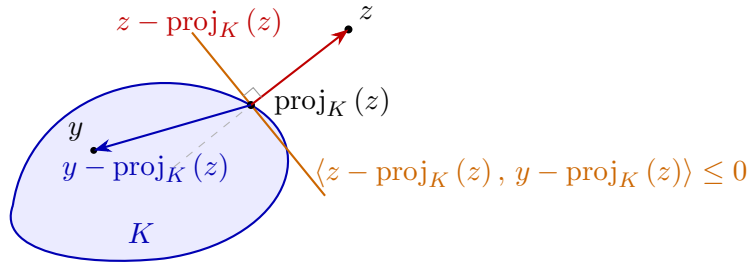

\begin{remark}
	\label{rem_projection_witness_convex_constraints}
	For a closed convex located set in finite-dimensional Euclidean space, \Cref{dfn_projection_witness_constraint} is the usual metric projection and normal-cone inequality.
	We state it as data because the projection is the constructive content used in the proof below.
	The point $\proj{z}{K}$ is an exact point of $K$ realizing the exact distance from $z$ to $K$.
	Approximate projections can be useful for numerical certificates, but the exact viability theorem below uses the exact witness.
\end{remark}

\begin{example}[Block constraint]
	\label{ex_projection_witness_block_constraint}
	Let
	\[
		K=\prod_{r=1}^n[a_r,b_r]
	\]
	be a full block.
	Define $\proj{z}{K}$ coordinatewise by clipping:
	\[
		(\proj{z}{K})_r:=\min\{b_r,\max\{a_r,z_r\}\}.
	\]
	Then
	\[
		d(z,K)=\nrm{z-\proj{z}{K}}
	\]
	and, for every $y\in K$,
	\[
		(z_r-(\proj{z}{K})_r)(y_r-(\proj{z}{K})_r)\le0
	\]
	in each coordinate.
	Summing over $r$ gives the projection-witness inequality.
\end{example}

\begin{example}[Ball constraint]
	\label{ex_projection_witness_ball_constraint}
	Let
	\[
		K:=\{y\in\R^n:\nrm{y-c}\le R\}
	\]
	with $R>0$.
	If $\nrm{z-c}\le R$, put $\proj{z}{K}:=z$.
	If $\nrm{z-c}>R$, put
	\[
		\proj{z}{K}:=c+R\frac{z-c}{\nrm{z-c}}.
	\]
	Then $d(z,K)=\nrm{z-\proj{z}{K}}$.
	For $y\in K$, the Cauchy--Schwarz inequality gives
	\[
		\scal{z-\proj{z}{K},y-\proj{z}{K}}
		=
		\frac{\nrm{z-c}-R}{\nrm{z-c}}\scal{z-c,y-\proj{z}{K}}
		\le0.
	\]
\end{example}

\begin{example}[Rational convex polytope]
	\label{ex_projection_witness_rational_polytope_constraint}
	Let
	\[
		K:=\{y\in\R^n:Ay\le b\}
	\]
	be a bounded rational convex polytope.
	A projection witness is obtained by solving the finite convex quadratic problem
	\[
		\text{minimize }\nrm{z-y}^2
		\quad\text{subject to }Ay\le b.
	\]
	The minimizer $p=\proj{z}{K}$ is unique because the objective is strictly convex.
	The finite Karush--Kuhn--Tucker conditions give multipliers $\lambda_i\ge0$ such that
	\[
		z-p=\sum_i\lambda_i a_i,
		\qquad
		\lambda_i(a_i^\top p-b_i)=0,
	\]
	where $a_i^\top$ are the rows of $A$.
	For $y\in K$,
	\[
		\scal{z-p,y-p}
		=
		\sum_i\lambda_i a_i^\top(y-p)
		\le0,
	\]
	which is the projection-witness inequality.
\end{example}

\begin{definition}[Projection-inward field]
	\label{dfn_projection_inward_field}
	Let $K\subseteq B$ have a projection witness, and let $s:B\to\R^n$.
	For $\delta\ge0$, the field $s$ is called \emph{$\delta$-inward for $K$} if, for every $z\in B$,
	\[
		\scal{z-\proj{z}{K},s(\proj{z}{K})}
		\le
		\delta\nrm{z-\proj{z}{K}}.
	\]
	The case $\delta=0$ is called \emph{inward}.
\end{definition}

\begin{remark}
	\label{rem_viability_upper_derivative_recall}
	We use the upper right derivative from \Cref{dfn_upper_right_derivative}.
	For a scalar function $\phi$ on an interval,
	\[
		\Dplus\phi(t)
		=
		\limsup_{h\downarrow0}
		\frac{\phi(t+h)-\phi(t)}{h}.
	\]
	In the theorem below, it is applied to the squared distance from a trajectory to the constraint.
\end{remark}

\begin{theorem}[Steiner viability for projective constraints]
	\label{thm_steiner_viability_projective_constraint}
	Let the assumptions of \Cref{thm_picard_steiner_certified_approximants} hold, and let
	\[
		s(x):=\operatorname{St}(F(x)).
	\]
	Let $K\subseteq B$ be located and closed with a projection witness, and let $x_0\in K$.
	Assume that $s$ is $\delta$-inward for $K$.
	Let $x$ be the exact Picard--Steiner solution from \Cref{thm_picard_steiner_certified_approximants}.
	Put
	\[
		L_s:=C_{\mathrm{St},n}L.
	\]
	Then, for every $t\in I$,
	\[
		d(x(t),K)
		\le
		\begin{cases}
			\delta t, & L_s=0,\\[0.4ex]
			\displaystyle \frac{\delta}{L_s}\bigl(\exp(L_st)-1\bigr), & L_s>0.
		\end{cases}
	\]
	In particular, if $s$ is inward, then $(x,s\circ x)$ is an exact Filippov solution viable in $K$.
\end{theorem}

\begin{proof}
	Let
	\[
		u(t):=d(x(t),K).
	\]
	Fix $t\in I$ and write
	\[
		p:=\proj{x(t)}{K}.
	\]
	For small $h>0$ with $t+h\in I$, the integral identity gives
	\[
		x(t+h)
		=
		x(t)+h s(x(t))+r_h,
	\]
	where
	\[
		\nrm{r_h}
		\le
		\int_t^{t+h}\nrm{s(x(\tau))-s(x(t))}\diff\tau
		\le
		\frac12 L_sMh^2.
	\]
	Since $p\in K$,
	\[
		u(t+h)^2
		\le
		\nrm{x(t)+hs(x(t))+r_h-p}^2.
	\]
	Because $p$ is an exact projection witness at $x(t)$,
	\[
		u(t)^2=\nrm{x(t)-p}^2.
	\]
	Put $a:=x(t)-p$.
	Subtracting the last display gives
	\[
	\begin{aligned}
		u(t+h)^2-u(t)^2
		&\le
		\nrm{a+hs(x(t))+r_h}^2-\nrm{a}^2
		\\
		&=
		2h\scal{a,s(x(t))}
		+
		2\scal{a,r_h}
		+
		h^2\nrm{s(x(t))}^2
		\\
		&\quad+
		2h\scal{s(x(t)),r_h}
		+
		\nrm{r_h}^2.
	\end{aligned}
	\]
	Here $\nrm{r_h}\le L_sMh^2/2$, the vector $a$ is bounded because $x(t),p\in B$, and $\nrm{s(x(t))}\le M$.
	After division by $2h$, all terms except $\scal{a,s(x(t))}$ tend to zero as $h\downarrow0$.
	Thus
	\[
		\Dplus\frac{u(t)^2}{2}
		\le
		\scal{x(t)-p,s(x(t))}.
	\]
	Now
	\[
	\begin{aligned}
		\scal{x(t)-p,s(x(t))}
		&=
		\scal{x(t)-p,s(p)}
		+
		\scal{x(t)-p,s(x(t))-s(p)}
		\\
		&\le
		\delta u(t)+L_su(t)^2.
	\end{aligned}
	\]
	Hence
	\[
		\Dplus\frac{u(t)^2}{2}
		\le
		\delta u(t)+L_su(t)^2.
	\]
	Fix $\rho>0$ and define
	\[
		u_\rho(t):=\sqrt{u(t)^2+\rho^2}.
	\]
	Since $u_\rho(t)\ge\rho$, the identity
	\[
		u_\rho(t+h)-u_\rho(t)
		=
		\frac{u(t+h)^2-u(t)^2}{u_\rho(t+h)+u_\rho(t)}
	\]
	and the continuity of $u_\rho$ give
	\[
		\Dplus u_\rho(t)
		\le
		\frac{\delta u(t)+L_su(t)^2}{u_\rho(t)}
		\le
		\delta+L_su_\rho(t).
	\]
	Since $x_0\in K$, one has $u(0)=0$.
	Thus
	\[
		u_\rho(0)=\rho.
	\]
	The elementary Gronwall comparison gives
	\[
		u_\rho(t)
		\le
		\begin{cases}
			\rho+\delta t, & L_s=0,\\[0.4ex]
			\displaystyle \rho\exp(L_st)+\frac{\delta}{L_s}\bigl(\exp(L_st)-1\bigr), & L_s>0.
		\end{cases}
	\]
	Because $u(t)\le u_\rho(t)$ for every $\rho>0$, letting $\rho\downarrow0$ yields the displayed bound.
	If $\delta=0$, then $u(t)=0$ for all rational $t\in I$.
	By continuity of $u$, this holds for all $t\in I$.
	Closedness of $K$ gives $x(t)\in K$ for all $t\in I$.
\end{proof}

\begin{proposition}[Viability from approximate projections]
	\label{prop_viability_approximate_projection_witnesses}
	Let $K\subseteq B$ be located and closed.
	Let $s:B\to\R^n$ be $L_s$-Lipschitz and bounded by $M$ on $B$.
	For every $z\in B$ and every $\eps>0$, let
	\[
		\aproj{z}{K}{\eps}\in K
	\]
	be a point supplied by locatedness such that
	\[
		\nrm{z-\aproj{z}{K}{\eps}}\le d(z,K)+\eps.
	\]
	Assume that these approximate projections may be chosen so that
	\[
		\scal{z-\aproj{z}{K}{\eps},s(\aproj{z}{K}{\eps})}
		\le
		\delta\nrm{z-\aproj{z}{K}{\eps}}+\eps
	\]
	for every $z\in B$ and every $\eps>0$.
	Let $x:I\to B$ be an integral trajectory of $s\circ x$ with $x(0)\in K$.
	Then, for every $t\in I$,
	\[
		d(x(t),K)
		\le
		\begin{cases}
			\delta t, & L_s=0,\\[0.4ex]
			\displaystyle \frac{\delta}{L_s}\bigl(\exp(L_st)-1\bigr), & L_s>0.
		\end{cases}
	\]
	In particular, if $\delta=0$, then $x$ is viable in $K$.
\end{proposition}

\begin{proof}
	Let
	\[
		u(t):=d(x(t),K).
	\]
	Fix $t\in I$ and let $h>0$ be small enough that $t+h\in I$.
	Choose
	\[
		\hat p_h:=\aproj{x(t)}{K}{h^2}.
	\]
	As in the proof of \Cref{thm_steiner_viability_projective_constraint},
	\[
		x(t+h)=x(t)+hs(x(t))+r_h,
		\qquad
		\nrm{r_h}\le\frac12L_sMh^2.
	\]
	Since $\hat p_h\in K$,
	\[
		u(t+h)^2\le\nrm{x(t)+hs(x(t))+r_h-\hat p_h}^2.
	\]
	By the approximate projection estimate,
	\[
		\nrm{x(t)-\hat p_h}^2-u(t)^2
		\le
		(u(t)+h^2)^2-u(t)^2
		=
		2h^2u(t)+h^4.
	\]
	Expanding the remaining square terms gives
	\[
		\Dplus\frac{u(t)^2}{2}
		\le
		\limsup_{h\downarrow0}
		\scal{x(t)-\hat p_h,s(x(t))}.
	\]
	Furthermore,
	\[
	\begin{aligned}
		\scal{x(t)-\hat p_h,s(x(t))}
		&=
		\scal{x(t)-\hat p_h,s(\hat p_h)}
		+
		\scal{x(t)-\hat p_h,s(x(t))-s(\hat p_h)}
		\\
		&\le
		\delta\nrm{x(t)-\hat p_h}
		+
		h^2
		+
		L_s\nrm{x(t)-\hat p_h}^2.
	\end{aligned}
	\]
	Since $\nrm{x(t)-\hat p_h}\le u(t)+h^2$, letting $h\downarrow0$ yields
	\[
		\Dplus\frac{u(t)^2}{2}
		\le
		\delta u(t)+L_su(t)^2.
	\]
	Fix $\rho>0$ and put
	\[
		u_\rho(t):=\sqrt{u(t)^2+\rho^2}.
	\]
	As in the proof of \Cref{thm_steiner_viability_projective_constraint},
	\[
		\Dplus u_\rho(t)\le\delta+L_su_\rho(t).
	\]
	Since $u_\rho(0)=\rho$, Gronwall gives
	\[
		u_\rho(t)
		\le
		\begin{cases}
			\rho+\delta t, & L_s=0,\\[0.4ex]
			\displaystyle \rho\exp(L_st)+\frac{\delta}{L_s}\bigl(\exp(L_st)-1\bigr), & L_s>0.
		\end{cases}
	\]
	Letting $\rho\downarrow0$ gives the displayed bound.
	If $\delta=0$, then $d(x(t),K)=0$ for all $t\in I$, and closedness of $K$ gives $x(t)\in K$.
\end{proof}

\begin{corollary}[Viability of Picard--Steiner approximants]
	\label{cor_picard_steiner_approximant_viability}
	Under the assumptions of \Cref{thm_steiner_viability_projective_constraint}, put
	\[
		L_s:=C_{\mathrm{St},n}L.
	\]
	Then the Picard--Steiner approximants satisfy, for every $k$ and every $t\in I$,
	\[
		d(x^{(k)}(t),K)
		\le
		R_\delta(t)
		+
		\sum_{r=k}^{\infty}
		M\frac{(C_{\mathrm{St},n}L)^rT^{r+1}}{(r+1)!},
	\]
	where
	\[
		R_\delta(t)
		:=
		\begin{cases}
			\delta t, & L_s=0,\\[0.4ex]
			\displaystyle \frac{\delta}{L_s}\bigl(\exp(L_st)-1\bigr), & L_s>0.
		\end{cases}
	\]
\end{corollary}

\begin{proof}
	First let $t$ be rational.
	By \Cref{thm_steiner_viability_projective_constraint},
	\[
		d(x(t),K)\le R_\delta(t).
	\]
	By \Cref{cor_picard_steiner_approximant_stability},
	\[
		\nrm{x(t)-x^{(k)}(t)}
		\le
		\sum_{r=k}^{\infty}
		M\frac{(C_{\mathrm{St},n}L)^rT^{r+1}}{(r+1)!}.
	\]
	The distance function is $1$-Lipschitz, so the claim follows.
	The estimate for arbitrary $t\in I$ follows by continuity of both sides and density of rational times.
\end{proof}

\begin{proposition}[Viable sample-and-hold certificate]
	\label{prop_viable_sample_hold_certificate}
	Let $K\subseteq B$ be convex.
	Let a sample-and-hold datum be certified for $F$ on $B$ in the sense of \Cref{dfn_certified_sample_hold_data}.
	Assume that every node belongs to $K$:
	\[
		x_j\in K,
		\qquad
		j=0,\ldots,N.
	\]
	Then the associated sample-and-hold path satisfies
	\[
		x_h(t)\in K,
		\qquad
		t\in I.
	\]
	Consequently, for every $\eps>0$, the tuple from \Cref{thm_sample_hold_filippov_certificate} is a viable $\eps$-full Filippov certificate.
\end{proposition}

\begin{proof}
	For $t\in[t_j,t_{j+1}]$ there is $\lambda\in[0,1]$ such that
	\[
		x_h(t)=(1-\lambda)x_j+\lambda x_{j+1}.
	\]
	Since $x_j,x_{j+1}\in K$ and $K$ is convex, this point belongs to $K$.
	The certificate statement is exactly \Cref{thm_sample_hold_filippov_certificate} with the added viability of the path.
\end{proof}

\begin{definition}[Robust Euler tangency]
	\label{dfn_robust_euler_tangency}
	Let $K\subseteq B$ be located, let $s:B\to\R^n$, let $\delta\ge0$, and let $\omega:\Q_{>0}\to\Q_{\ge0}$ satisfy $\omega(h)\to0$ as $h\to0$.
	The field $s$ is called \emph{$(\delta,\omega)$-Euler tangent to $K$ on $B$} if, for all $z\in B$ and all rational $h>0$ with $z+hs(z)\in B$,
	\[
		d(z+hs(z),K)
		\le
		d(z,K)+h\delta+h\omega(h).
	\]
\end{definition}

\begin{figure}[ht]
	\centering
	\begin{tikzpicture}[>=Stealth, scale=1.0]
	\fill[blue!8] (-0.3,0.0) .. controls (0.0,1.3) and (1.5,2.1) .. (2.8,1.45)
		.. controls (4.0,0.8) and (3.5,-0.6) .. (2.1,-0.75)
		.. controls (0.8,-0.9) and (-0.6,-0.55) .. (-0.3,0.0);
	\draw[blue!70!black, thick] (-0.3,0.0) .. controls (0.0,1.3) and (1.5,2.1) .. (2.8,1.45)
		.. controls (4.0,0.8) and (3.5,-0.6) .. (2.1,-0.75)
		.. controls (0.8,-0.9) and (-0.6,-0.55) .. (-0.3,0.0);
	\node[blue!70!black] at (1.4,0.32) {$K$};

	\coordinate (z) at (2.38,1.18);
	\coordinate (e) at (3.62,1.60);
	\coordinate (p) at (3.23,1.07);

	\fill[black] (z) circle (1.5pt) node[below left] {$z$};
	\fill[black] (e) circle (1.5pt) node[above right, xshift=0.06cm] {$z+h s(z)$};
	\fill[black] (p) circle (1.2pt);

	\draw[->, red!75!black, thick] (z) -- (e) node[midway, above, yshift=0.07cm] {$h s(z)$};
	\draw[magenta!80!black, thick] (e) -- (p);
	\draw[dashed, gray!85] (e) circle (0.66);
	\node[magenta!80!black, anchor=west] at (3.65,0.8) {$d(z+hs(z),K)\le d(z,K)+h\delta+h\omega(h)$};
\end{tikzpicture}
	\caption{Robust Euler tangency. A forward Euler point $z+hs(z)$ may leave the constraint. The dashed circle is centered at $z+hs(z)$ and represents the distance radius $d(z+hs(z),K)$; in the picture it is tangent to $K$ at a nearest point. The condition bounds this radius by the previous distance, the slack $h\delta$, and the geometric modulus $h\omega(h)$.}
	\label{fig_robust_euler_tangency}
\end{figure}
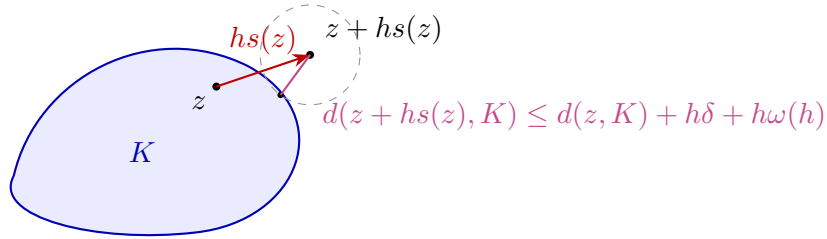

\begin{proposition}[Viability from robust Euler tangency]
	\label{prop_viability_robust_euler_tangency}
	Let $K\subseteq B$ be located and closed.
	Let $s:B\to\R^n$ be $L_s$-Lipschitz and bounded by $M$ on $B$.
	Assume that $s$ is $(\delta,\omega)$-Euler tangent to $K$ on $B$.
	Let $x:I\to B$ be an integral trajectory of $s\circ x$ with $x(0)\in K$.
	Then, for every $t\in I$,
	\[
		d(x(t),K)\le\delta t.
	\]
	In particular, if $\delta=0$, then $x$ is viable in $K$.
\end{proposition}

\begin{proof}
	First let $t$ be rational.
	Let $0=t_0<\cdots<t_N=t$ be a rational partition with mesh $h_\ast$.
	Set
	\[
		h_j:=t_{j+1}-t_j.
	\]
	Fix $\eta>0$.
	Refine the partition so that
	\[
		\omega(h_j)\le\eta
		\quad\text{and}\quad
		h_j\le\eta
	\]
	for all $j$.
	For each $j$, the integral identity gives
	\[
		x(t_{j+1})
		=
		x(t_j)+h_js(x(t_j))+r_j
	\]
	with
	\[
		\nrm{r_j}
		\le
		\frac12 L_sMh_j^2.
	\]
	Using the $1$-Lipschitz property of the distance to $K$ and the Euler-tangency condition,
	\[
		d(x(t_{j+1}),K)
		\le
		d(x(t_j),K)+h_j\delta+h_j\omega(h_j)+\frac12 L_sMh_j^2.
	\]
	Summing over $j$ gives
	\[
		d(x(t),K)
		\le
		\delta t
		+
		t\eta
		+
		\frac12 L_sMt\eta.
	\]
	Since $\eta>0$ is arbitrary, the displayed estimate follows.
	The estimate for arbitrary $t\in I$ follows by continuity of both sides and density of rational times.
	If $\delta=0$, then $d(x(t),K)=0$, and closedness of $K$ gives $x(t)\in K$.
\end{proof}

\begin{remark}
	\label{rem_viability_relation_classical}
	The bare contingent cone condition from classical viability theory is not used in the preceding results.
	The projection condition in \Cref{thm_steiner_viability_projective_constraint} is the constructive Nagumo condition for convex projective constraints and the canonical Steiner field.
	The Euler-tangency condition in \Cref{dfn_robust_euler_tangency} is stronger than the classical contingent condition, but it is quantitative and directly checkable from geometric data.
	For block constraints or rational polytope constraints it reduces to finitely many face inequalities on the relevant cells; for smooth constraints the modulus $\omega$ records the local curvature error of a straight Euler step.
\end{remark}

\section{Probability Densities}
\label{sec_probability_densities}

\begin{remark}
	This section introduces probability without abstract sample spaces.
	Regular densities are effectively bounded locally regular functions normalized by the improper integral, and events, expectations, conditionals, kernels, and finite-horizon Markov chains are expressed through the integration tools developed earlier.
	The local formulation covers the usual unbounded Euclidean state spaces, while bounded-support models appear as a special case.
\end{remark}

\subsection{Densities and Events}
\label{subsec_densities_events}

\begin{definition}[Regular densities]
	\label{dfn_regular_density}
	\label{dfn_global_regular_density}
	A scalar function $p\in\BReg{\R^n}$ is called a \emph{regular density} on $\R^n$ if
	\[
		p(x)\ge0
	\]
	on its determined domain, the improper integral exists, and
	\[
		\int_{\R^n}p(x)\diff x=1.
	\]
	A scalar function $p\in\BReg{\R^n}$ is called a \emph{regular subdensity} on $\R^n$ if $p\ge0$, the improper integral exists, and
	\[
		\int_{\R^n}p(x)\diff x\le1.
	\]
	Let $B$ be a full block.
	A scalar function $p\in\BReg{B}$ is called a \emph{block regular density}, or simply a \emph{regular density on $B$}, if
	\[
		p(x)\ge0,
		\qquad
		x\in B\cap\dom p,
	\]
	and
	\[
		\int_B p(x)\diff x=1.
	\]
	A scalar function $p\in\BReg{B}$ is called a \emph{block regular subdensity}, or a \emph{regular subdensity on $B$}, if $p\ge0$ and
	\[
		\int_B p(x)\diff x\le1.
	\]
\end{definition}

\begin{remark}
	\label{rem_density_no_sample_space}
	A regular density is not introduced as a measure on an abstract sample space.
	It is an effectively bounded locally regular function normalized by the improper integral developed above.
	All probabilistic quantities below are therefore ordinary integrals of regular functions over Euclidean spaces or over finite localization blocks.
	Unless a full block is explicitly named, ``regular density'' means a density on the ambient Euclidean space in the first part of \Cref{dfn_regular_density}.
	A block regular density is the finite localization used inside proofs and is the special case in which the support is contained in one full block and the tail is empty.
\end{remark}

\begin{definition}[Global event probability and expectation]
	\label{dfn_global_event_probability_expectation}
	Let $p$ be a regular density on $\R^n$.
	If $A\subset\R^n$ is a representable event and $B$ is any full block with $A\Subset B$, define
	\[
		\PP{A}_p
		:=
		\int_B\indic{A}(x)p(x)\diff x.
	\]
	The value is independent of the chosen block $B$, because enlarging $B$ only adds a region where the indicator $\indic{A}$ vanishes.
	If $g\in\BReg{\R^n}$ and the improper integral of $gp$ exists, define
	\[
		\E[p]{g}
		:=
		\int_{\R^n}g(x)p(x)\diff x.
	\]
\end{definition}

\begin{example}[Gaussian density]
	\label{ex_gaussian_global_regular_density}
	The standard Gaussian density
	\[
		\varphi_n(x)
		:=
		(2\pi)^{-n/2}\exp\left(-\frac{\nrm{x}^2}{2}\right),
		\qquad
		x\in\R^n,
	\]
	is a regular density on $\R^n$.
\end{example}

\begin{proof}
	On every full block $B$, choose full blocks $X$ and $C$ with
	\[
		B\Subset X\Subset C.
	\]
	The function
	\[
		x\mapsto(2\pi)^{-n/2}\exp\left(-\frac{\nrm{x}^2}{2}\right)
	\]
	is continuous on $X$ with a modulus inherited from the local modulus on $C$ and is bounded by $(2\pi)^{-n/2}$.
	If $X\subseteq[-R,R]^n$, then a positive rational number below
	\[
		(2\pi)^{-n/2}\exp\left(-\frac{nR^2}{2}\right)
	\]
	is a lower bound on $X$.
	Taking $X$ as the support and the zero exterior branch gives a member of $\BReg{C}$ which agrees with $\varphi_n$ on $B$.
	Hence $\varphi_n\in\BReg{\R^n}$.
	The same local lower bounds show that $\varphi_n$ is locally bounded away from zero.
	It is nonnegative.
	The improper integral exists because the tails satisfy the standard estimate
	\[
		\int_{\R^n\setminus[-R,R]^n}\varphi_n(x)\diff x
		\le
		2n\int_R^\infty(2\pi)^{-1/2}\exp\left(-\frac{s^2}{2}\right)\diff s
		\le
		\frac{2n}{R}(2\pi)^{-1/2}\exp\left(-\frac{R^2}{2}\right)
	\]
	for $R>0$.
	The normalization
	\[
		\int_{\R^n}\varphi_n(x)\diff x=1
	\]
	follows from the one-dimensional Gaussian integral and finite products.
\end{proof}

\begin{definition}[Block probability of a representable event]
	\label{dfn_event_probability_regular_density}
	Let $p$ be a block regular density on a full block $B$.
	Let $A\Subset B$ be representable.
	The probability of $A$ under $p$ is
	\[
		\PP{A}_p
		:=
		\int_B \indic{A}(x)p(x)\diff x.
	\]
\end{definition}

\begin{proposition}[Elementary event calculus]
	\label{prop_event_probability_calculus}
	Let $p$ be a regular density on $\R^n$.
	For every representable event $A\subset\R^n$, the number $\PP{A}_p$ is well-defined and satisfies
	\[
		0\le\PP{A}_p\le1.
	\]
	If $A_1,\ldots,A_N\subset\R^n$ are pairwise disjoint representable events, then
	\[
		\PP{\bigcup_{i=1}^N A_i}_p
		=
		\sum_{i=1}^N\PP{A_i}_p.
	\]
\end{proposition}

\begin{proof}
	Choose a full block $B$ with $A\Subset B$ and a local representative of $p$ on a well-containing block.
	The indicator $\indic{A}$ is a simple regular function on this localization block, and $\indic{A}p\in\BReg{B}$ by \Cref{lem_regular_functions_finite_algebra}.
	Moreover,
	\[
		0\le \indic{A}p\le p.
	\]
	Hence \Cref{lem_regular_integral_finite_order}, applied on a localization block containing $A$, gives
	\[
		0\le\PP{A}_p\le\int_{\R^n}p(x)\diff x=1.
	\]
	For pairwise disjoint events, the indicators satisfy
	\[
		\indic{\bigcup_{i=1}^N A_i}
		=
		\sum_{i=1}^N\indic{A_i}
	\]
	on the represented support, after finite flattening of the event supports.
	Linearity of the regular integral gives the displayed additivity.
\end{proof}

\begin{definition}[Block expectation under a regular density]
	\label{dfn_expectation_regular_density}
	Let $p$ be a block regular density on a full block $B$.
	Let $g\in\BReg{B}$.
	The expectation of $g$ under $p$ is
	\[
		\E[p]{g}
		:=
		\int_B g(x)p(x)\diff x.
	\]
\end{definition}

\begin{proposition}[Finite expectation rules]
	\label{prop_expectation_rules_regular_density}
	Let $p$ be a regular density on $\R^n$.
	Let $g,h\in\BReg{\R^n}$, let $a,b\in\R$, and assume that the improper integrals of $gp$, $hp$, and $(ag+bh)p$ exist.
	Then
	\[
		\E[p]{ag+bh}
		=
		a\E[p]{g}+b\E[p]{h}.
	\]
	If $g\le h$ and the displayed expectations exist, then
	\[
		\E[p]{g}\le\E[p]{h}.
	\]
\end{proposition}

\begin{proof}
	The products $gp$ and $hp$ belong locally to $\BReg{\R^n}$ by \Cref{lem_global_breg_local_algebra}.
	The first identity follows from linearity of the improper integral, using the supplied Cauchy witnesses.
	If $g\le h$, then
	\[
		gp\le hp
	\]
	because $p\ge0$.
	The order claim follows by applying \Cref{lem_regular_integral_finite_order} on localization blocks and passing to the improper limits.
\end{proof}

\subsection{Joint, Marginal, and Conditional Densities}
\label{subsec_joint_marginal_conditional_densities}

\begin{proposition}[Product densities]
	\label{prop_product_densities}
	Let $p$ be a regular density on $\R^n$ and let $q$ be a regular density on $\R^m$.
	Then
	\[
		r(x,y):=p(x)q(y)
	\]
	is a regular density on $\R^{n+m}$.
\end{proposition}

\begin{proof}
	On every product block, choose local representatives of $p$ and $q$.
	Their product has local supports of the form $X_k\times Y_l$, hence is locally rectangularly regular by \Cref{lem_representable_product}.
	It is locally effectively bounded by the product of the local bound witnesses for $p$ and $q$.
	It is nonnegative because both factors are nonnegative.
	For the improper integral, write $B_R=[-R,R]^n$ and $C_R=[-R,R]^m$.
	The partial product integrals satisfy
	\[
		\int_{B_R\times C_R}p(x)q(y)\diff(x,y)
		=
		\left(\int_{B_R}p(x)\diff x\right)
		\left(\int_{C_R}q(y)\diff y\right)
	\]
	by \Cref{thm_fubini_rectangular_regular} on $B_R\times C_R$.
	The two factors converge to $1$, so the product partial integrals converge to $1$.
	The tails are Cauchy because, for $S\ge R$,
	\[
	\begin{aligned}
		0\le
		\int_{B_S\times C_S}pq-\int_{B_R\times C_R}pq
		&\le
		\left(\int_{B_S}p(x)\diff x\right)
		\left(\int_{C_S\setminus C_R}q(y)\diff y\right)
		\\
		&\quad+
		\left(\int_{B_S\setminus B_R}p(x)\diff x\right)
		\left(\int_{C_R}q(y)\diff y\right),
	\end{aligned}
	\]
	and the right-hand side is bounded by the tail witnesses of $p$ and $q$.
\end{proof}

\begin{remark}
	\label{rem_block_product_density_special_case}
	If the two densities are block regular densities on full blocks $A$ and $C$, \Cref{prop_product_densities} reduces to the finite statement that $p(x)q(y)$ is a block regular density on $A\times C$.
	In that case all improper-tail clauses are empty.
\end{remark}

\begin{definition}[Joint and marginal regular densities]
	\label{dfn_joint_marginal_regular_density}
	A \emph{joint regular density} of variables in $\R^n$ and $\R^m$ is a regular density $p_{AC}$ on $\R^{n+m}$.
	It is called \emph{locally rectangularly regular} if, on every product of full blocks $A\subset\R^n$ and $C\subset\R^m$, it has an effectively bounded rectangularly regular local representative on a well-containing product block.
	A \emph{marginal tail--Fubini certificate} for such a joint density is the following data.
	First, the improper inner integrals
	\[
		p_A(x)
		:=
		\int_{\R^m}p_{AC}(x,y)\diff y,
		\qquad
		p_C(y)
		:=
		\int_{\R^n}p_{AC}(x,y)\diff x
	\]
	define functions $p_A\in\BReg{\R^n}$ and $p_C\in\BReg{\R^m}$: for every full block in the remaining variable, the inner integrals have local effectively bounded regular representatives on a well-containing block.
	Second, the two successive improper integrations satisfy
	\[
		\int_{\R^n}
		\left(\int_{\R^m}p_{AC}(x,y)\diff y\right)
		\diff x
		=
		\int_{\R^{n+m}}p_{AC}(x,y)\diff(x,y)
	\]
	and
	\[
		\int_{\R^m}
		\left(\int_{\R^n}p_{AC}(x,y)\diff x\right)
		\diff y
		=
		\int_{\R^{n+m}}p_{AC}(x,y)\diff(x,y).
	\]
	If $A$ and $C$ are full blocks and $p_{AC}$ is effectively bounded and rectangularly regular on $A\times C$, the same notation denotes the block marginals obtained by replacing $\R^m$ and $\R^n$ with $C$ and $A$.
\end{definition}

\begin{remark}[Why the marginal tail certificate is separate]
	\label{rem_marginal_tail_fubini_certificate}
	In classical measure theory, Tonelli--Fubini theorems are commonly invoked in a single step for nonnegative or integrable functions on product measure spaces.
	Here the global integral over $\R^{n+m}$ is an improper limit of finite block integrals, and the marginal functions are expected to be regular functions in the remaining variable.
	The certificate records the extra effective information needed to pass from finite block Fubini, \Cref{thm_fubini_rectangular_regular}, to these improper marginal identities.

	The certificate is automatic for block-supported joint densities: if the support is contained in $A\times C$ for full blocks $A$ and $C$, all tails outside $A\times C$ are empty and \Cref{thm_fubini_rectangular_regular} applies directly.
	It is also supplied by dominated tail estimates.
	For example, if
	\[
		0\le p_{AC}(x,y)\le r(x)s(y),
	\]
	where $r$ and $s$ are regular densities with known improper-tail witnesses, then the tails of $p_{AC}$ outside $B_R^n\times B_R^m$ are bounded by the tails of $r$ and $s$.
	Gaussian bounds are one instance of this pattern, but bounded supports, exponential tails, and compactly supported mixtures are equally finite certificates when the corresponding tail estimates are part of the data.
\end{remark}

\begin{example}[Block-supported joint density]
	\label{ex_block_supported_marginal_tail_certificate}
	Let $A\subset\R^n$ and $C\subset\R^m$ be full blocks.
	If $p_{AC}$ is a block regular density on $A\times C$ and is extended by zero outside $A\times C$, then its marginal tail--Fubini certificate is finite:
	\[
		p_A(x)=\int_Cp_{AC}(x,y)\diff y,
		\qquad
		p_C(y)=\int_Ap_{AC}(x,y)\diff x,
	\]
	and all integrals outside $A\times C$ vanish.
	Thus bounded-support densities are not exceptional; they are the case where the tail part of the certificate has zero cost.
\end{example}

\begin{example}[Dominated product tails]
	\label{ex_dominated_product_tail_certificate}
	Suppose $p_{AC}$ is locally rectangularly regular and
	\[
		0\le p_{AC}(x,y)\le r(x)s(y)
	\]
	for regular subdensities $r$ on $\R^n$ and $s$ on $\R^m$ with supplied improper-tail witnesses.
	If $B_R\subset\R^n$ and $C_R\subset\R^m$ are increasing full blocks, then
	\[
		\int_{(B_R\times C_R)^c}p_{AC}(x,y)\diff(x,y)
		\le
		\int_{\R^n\setminus B_R}r(x)\diff x
		+
		\int_{\R^m\setminus C_R}s(y)\diff y.
	\]
	Hence the product tail witnesses for $r$ and $s$ give a marginal tail--Fubini certificate for $p_{AC}$.
\end{example}

\begin{proposition}[Marginals are densities]
	\label{prop_marginals_regular_densities}
	Let $p_{AC}$ be a locally rectangularly regular joint density on $\R^{n+m}$ with a marginal tail--Fubini certificate.
	Then the marginals $p_A$ and $p_C$ are regular densities on $\R^n$ and $\R^m$, respectively.
\end{proposition}

\begin{proof}
	Local regularity of $p_A$ and $p_C$ follows by applying \Cref{thm_fubini_rectangular_regular} on each product of full blocks and then using the local representatives supplied by the certificate.
	Nonnegativity follows from \Cref{lem_regular_integral_finite_order} applied to the nonnegative inner integrands.
	The certified successive integration gives
	\[
		\int_{\R^n}p_A(x)\diff x
		=
		\int_{\R^{n+m}}p_{AC}(x,y)\diff(x,y)
		=
		1,
	\]
	and analogously for $p_C$.
\end{proof}

\begin{remark}
	\label{rem_block_marginals_special_case}
	For a block joint density on $A\times C$, \Cref{prop_marginals_regular_densities} reduces to the earlier finite statement: the marginal functions are block regular densities on $A$ and $C$.
	The marginal tail--Fubini certificate is automatic from \Cref{thm_fubini_rectangular_regular}.
\end{remark}

\begin{definition}[Conditional regular density]
	\label{dfn_conditional_regular_density}
	Let $p_{AC}$ be a locally rectangularly regular joint density on $\R^{n+m}$ with marginal $p_A$.
	Assume that $p_A$ is locally bounded away from zero in the sense of \Cref{dfn_global_bounded_regular_functions}.
	The conditional density of $y\in\R^m$ given $x\in\R^n$ is
	\[
		p_{C\mid A}(y\mid x)
		:=
		\frac{p_{AC}(x,y)}{p_A(x)}.
	\]
	In the block case $A\times C$, this is the same definition with a supplied rational lower bound for $p_A$ on $A$.
\end{definition}

\begin{proposition}[Conditional density normalization]
	\label{prop_conditional_density_normalization}
	Under the assumptions of \Cref{dfn_conditional_regular_density}, the function $p_{C\mid A}$ is nonnegative and locally rectangularly regular.
	Moreover,
	\[
		\int_{\R^m}p_{C\mid A}(y\mid x)\diff y=1
	\]
	for every $x$ for which the local data determine the value of $p_A(x)$.
\end{proposition}

\begin{proof}
	By the reciprocal part of \Cref{lem_global_breg_local_algebra}, $1/p_A\in\BReg{\R^n}$ locally.
	Multiplying the local rectangular representation of $p_{AC}$ by this $x$-dependent local regular factor and flattening the $x$-supports gives a local rectangular representation of $p_{C\mid A}$.
	Nonnegativity follows from $p_{AC}\ge0$ and $p_A>0$.
	For a determined value of $x$ in the local domain of $p_A$,
	\[
		\int_{\R^m}p_{C\mid A}(y\mid x)\diff y
		=
		\frac{1}{p_A(x)}
		\int_{\R^m}p_{AC}(x,y)\diff y
		=
		1.
	\]
\end{proof}

\begin{definition}[Conditional expectation]
	\label{dfn_conditional_expectation_regular_density}
	Under the assumptions of \Cref{dfn_conditional_regular_density}, let $g\in\BReg{\R^m}$ and assume that the improper integral below exists locally in $x$.
	The conditional expectation of $g$ given $x\in\R^n$ is
	\[
		\E[p_{C\mid A}]{g\mid x}
		:=
		\int_{\R^m}g(y)p_{C\mid A}(y\mid x)\diff y.
	\]
\end{definition}

\begin{proposition}[Conditional expectations are regular]
	\label{prop_conditional_expectation_regular}
	Under the assumptions of \Cref{dfn_conditional_expectation_regular_density}, the map
	\[
		x\mapsto\E[p_{C\mid A}]{g\mid x}
	\]
	belongs to $\BReg{\R^n}$.
\end{proposition}

\begin{proof}
	The product
	\[
		(x,y)\mapsto g(y)p_{C\mid A}(y\mid x)
	\]
	is locally effectively bounded and locally rectangularly regular.
	The claim follows window by window from \Cref{thm_fubini_rectangular_regular} and from the improper-integral witness in \Cref{dfn_conditional_expectation_regular_density}.
\end{proof}

\subsection{Regular Markov Kernels}
\label{subsec_regular_markov_kernels}

\begin{definition}[Local regular policies and controlled kernels]
	\label{dfn_local_regular_policy_kernel}
	Let $\mathbb X\subset\R^n$, $\mathbb U\subset\R^m$, and $\mathbb Y\subset\R^n$ be locally representable.
	A function
	\[
		\pi:\mathbb X\times\mathbb U\to\R
	\]
	is called a \emph{local regular policy density} from $\mathbb X$ to $\mathbb U$ if the following data are supplied:
	\begin{enumerate}
		\item for every pair of full blocks $B_{\mathbb X}\subset\R^n$ and $B_{\mathbb U}\subset\R^m$, the restriction of $\pi$ to
		\[
			(\mathbb X\cap B_{\mathbb X})\times(\mathbb U\cap B_{\mathbb U})
		\]
		has an effectively bounded rectangularly regular local representative on a well-containing product block;
		\item for every determined $x\in\mathbb X$, the function $u\mapsto\pi(x,u)$, extended by zero outside $\mathbb U$ through the local representability witnesses, is a regular density on $\R^m$.
	\end{enumerate}
	A function
	\[
		K:\mathbb X\times\mathbb U\times\mathbb Y\to\R
	\]
	is called a \emph{local controlled regular Markov kernel} from $\mathbb X\times\mathbb U$ to $\mathbb Y$ if the analogous finite-window rectangular regularity holds on every product of full blocks and, for every determined $(x,u)\in\mathbb X\times\mathbb U$, the function $y\mapsto K(x,u,y)$, extended by zero outside $\mathbb Y$, is a regular density on $\R^n$.
\end{definition}

\begin{remark}
	\label{rem_local_markov_kernel_windows}
	The local representatives in \Cref{dfn_local_regular_policy_kernel} carry their own supports and determined domains, exactly as in \Cref{dfn_global_bounded_regular_functions}.
	Thus no decision is made about whether an arbitrary real point lies in $\mathbb X$, $\mathbb U$, or $\mathbb Y$.
	In the local Markov-chain statements, an integral over a locally representable set means the improper integral over the ambient Euclidean space of the function extended by zero outside that set through its local witnesses.
	When the state and action sets are full blocks, these local definitions reduce to the finite block definitions below.
\end{remark}

\begin{remark}
	\label{rem_markov_kernel_global_default}
	The finite block kernels below are localization tools.
	They are used to prove window-wise regularity and the finite identities of \Cref{thm_fubini_rectangular_regular}, while tail estimates are supplied separately when the state or action spaces are unbounded.
\end{remark}

\begin{definition}[Block regular Markov kernel]
	\label{dfn_regular_markov_kernel}
	Let $A\subset\R^n$ and $C\subset\R^m$ be full blocks.
	A function $K:A\times C\to\R$ is called a \emph{block regular Markov kernel from $A$ to $C$} if it is effectively bounded, nonnegative, and rectangularly regular, and if
	\[
		\int_CK(x,y)\diff y=1
	\]
	for every determined $x\in A$.
\end{definition}

\begin{proposition}[Kernel action on densities]
	\label{prop_kernel_action_density}
	Let $K$ be a block regular Markov kernel from $A$ to $C$.
	Let $\rho$ be a regular density on $A$.
	Define
	\[
		\mathcal K_K[\rho](y)
		:=
		\int_AK(x,y)\rho(x)\diff x.
	\]
	Then $\mathcal K_K[\rho]$ is a regular density on $C$.
\end{proposition}

\begin{proof}
	Regularity and effective boundedness follow from \Cref{cor_regular_kernel_action}.
	Nonnegativity follows from \Cref{lem_regular_integral_finite_order}.
	By \Cref{thm_fubini_rectangular_regular} and the kernel normalization,
	\[
		\int_C\mathcal K_K[\rho](y)\diff y
		=
		\int_A\rho(x)\left(\int_CK(x,y)\diff y\right)\diff x
		=
		\int_A\rho(x)\diff x
		=
		1.
	\]
\end{proof}

\begin{definition}[Local controlled regular Markov kernel]
	\label{dfn_controlled_regular_markov_kernel}
	Let $\mathbb X\subset\R^n$, $\mathbb U\subset\R^m$, and $\mathbb Y\subset\R^n$ be locally representable.
	A function
	\[
		K:\mathbb X\times\mathbb U\times\mathbb Y\to\R
	\]
	is called a \emph{local controlled regular Markov kernel} if it satisfies the kernel part of \Cref{dfn_local_regular_policy_kernel}.
\end{definition}

\begin{definition}[Local regular policy density]
	\label{dfn_regular_policy_density}
	Let $\mathbb X\subset\R^n$ and $\mathbb U\subset\R^m$ be locally representable.
	A function
	\[
		\pi:\mathbb X\times\mathbb U\to\R
	\]
	is called a \emph{local regular policy density} from $\mathbb X$ to $\mathbb U$ if it satisfies the policy part of \Cref{dfn_local_regular_policy_kernel}.
\end{definition}

\begin{proposition}[One-step controlled density update]
	\label{prop_controlled_density_update}
	Let $\mathbb X\subset\R^n$, $\mathbb U\subset\R^m$, and $\mathbb Y\subset\R^n$ be locally representable.
	Let $\rho$ be a regular density on $\R^n$ supported on $\mathbb X$.
	Let $\pi$ be a local regular policy density from $\mathbb X$ to $\mathbb U$.
	Let $K$ be a local controlled regular Markov kernel from $\mathbb X\times\mathbb U$ to $\mathbb Y$.
	Assume that the improper integral below has a one-step tail--Fubini certificate.
	Define
	\[
		\rho^+(y)
		:=
		\int_{\mathbb X\times\mathbb U}K(x,u,y)\pi(x,u)\rho(x)\diff(x,u).
	\]
	Then $\rho^+$, extended by zero outside $\mathbb Y$, is a regular density on $\R^n$ supported on $\mathbb Y$.
\end{proposition}

\begin{proof}
	The product
	\[
		(x,u,y)\mapsto K(x,u,y)\pi(x,u)\rho(x)
	\]
	is effectively bounded and rectangularly regular as a function of $(x,u)$ and $y$ on every product of full localization blocks, by the local representatives of $K$, $\pi$, and $\rho$.
	Hence $\rho^+$ is locally regular by \Cref{thm_fubini_rectangular_regular}.
	It is nonnegative by \Cref{lem_regular_integral_finite_order}.
	The one-step tail--Fubini certificate supplies the improper Cauchy witness and permits the following successive integration:
	\[
	\begin{aligned}
		\int_{\mathbb Y}\rho^+(y)\diff y
		&=
		\int_{\mathbb X}\rho(x)
		\left(
			\int_{\mathbb U}\pi(x,u)
			\left(
				\int_{\mathbb Y}K(x,u,y)\diff y
			\right)
			\diff u
		\right)
		\diff x
		\\
		&=
		\int_{\mathbb X}\rho(x)\diff x
		=
		1.
	\end{aligned}
	\]
\end{proof}

\subsection{Finite-Horizon Density Calculus}
\label{subsec_finite_horizon_density_calculus}

\begin{definition}[Local finite-horizon controlled density model]
	\label{dfn_local_finite_horizon_density_model}
	A local finite-horizon controlled density model consists of:
	\begin{enumerate}
		\item locally representable state sets $\mathbb X_0,\ldots,\mathbb X_T$,
		\item locally representable action sets $\mathbb U_0,\ldots,\mathbb U_{T-1}$,
		\item an initial regular density $\rho_0$ on $\R^n$ supported on $\mathbb X_0$,
		\item local regular policy densities $\pi_t:\mathbb X_t\times\mathbb U_t\to\R$ for $t=0,\ldots,T-1$,
		\item local controlled regular Markov kernels $K_t:\mathbb X_t\times\mathbb U_t\times\mathbb X_{t+1}\to\R$ for $t=0,\ldots,T-1$.
	\end{enumerate}
	The corresponding trajectory set is
	\[
		\mathcal Z_T
		:=
		\mathbb X_0\times \mathbb U_0\times \mathbb X_1\times \mathbb U_1\times\cdots\times \mathbb U_{T-1}\times \mathbb X_T.
	\]
	A \emph{finite-horizon tail--Fubini certificate} for the model is a witness that the trajectory product
	\[
		p_T(z)
		:=
		\rho_0(x_0)
		\prod_{t=0}^{T-1}
		\pi_t(x_t,u_t)K_t(x_t,u_t,x_{t+1})
	\]
	has an improper integral over $\mathcal Z_T$, that the tails outside sufficiently large product blocks have arbitrarily small integral, and that the successive integrations in the order
	\[
		x_T,\ u_{T-1},\ x_{T-1},\ldots,u_0,\ x_0
	\]
	are the corresponding improper limits.
\end{definition}

\begin{remark}
	\label{rem_engineer_tail_certificate}
	The engineer-friendly way to provide the certificate in \Cref{dfn_local_finite_horizon_density_model} is not to reason about an abstract trajectory space.
	It is to give finite truncation blocks and explicit tail bounds for the initial density, policies, and kernels, uniformly on the previously retained truncation blocks.
	\Cref{prop_sufficient_finite_horizon_tail_certificate} records this finite check.
\end{remark}

\begin{definition}[Block finite-horizon controlled density model]
	\label{dfn_finite_horizon_density_model}
	A block finite-horizon controlled density model is the special case of \Cref{dfn_local_finite_horizon_density_model} in which the data are carried by:
	\begin{enumerate}
		\item full state blocks $\mathbb X_0,\ldots,\mathbb X_T$,
		\item full action blocks $\mathbb U_0,\ldots,\mathbb U_{T-1}$,
		\item an initial regular density $\rho_0$ on $\mathbb X_0$,
		\item block regular policy densities $\pi_t:\mathbb X_t\times \mathbb U_t\to\R$ for $t=0,\ldots,T-1$, meaning block regular Markov kernels from $\mathbb X_t$ to $\mathbb U_t$,
		\item block controlled regular Markov kernels $K_t:\mathbb X_t\times \mathbb U_t\times \mathbb X_{t+1}\to\R$ for $t=0,\ldots,T-1$, meaning block regular Markov kernels from $\mathbb X_t\times\mathbb U_t$ to $\mathbb X_{t+1}$.
	\end{enumerate}
\end{definition}

\begin{definition}[Trajectory density]
	\label{dfn_trajectory_density}
	For a block finite-horizon controlled density model, define the trajectory block
	\[
		\mathcal Z_T
		:=
		\mathbb X_0\times \mathbb U_0\times \mathbb X_1\times \mathbb U_1\times\cdots\times \mathbb U_{T-1}\times \mathbb X_T.
	\]
	For
	\[
		z=(x_0,u_0,x_1,u_1,\ldots,u_{T-1},x_T)\in\mathcal Z_T,
	\]
	set
	\[
		p_T(z)
		:=
		\rho_0(x_0)
		\prod_{t=0}^{T-1}
		\pi_t(x_t,u_t)K_t(x_t,u_t,x_{t+1}).
	\]
\end{definition}

\begin{theorem}[Block finite-horizon trajectory density]
	\label{thm_finite_horizon_trajectory_density}
	For every block finite-horizon controlled density model, the function $p_T$ is a regular density on $\mathcal Z_T$.
\end{theorem}

\begin{proof}
	Regularity and effective boundedness follow from finite product closure of regular functions and representability of finite products of supports.
	Nonnegativity follows from nonnegativity of the factors.
	For normalization, integrate successively from the last state and action variables backwards:
	\[
		\int_{\mathbb X_T}K_{T-1}(x_{T-1},u_{T-1},x_T)\diff x_T=1,
	\]
	then
	\[
		\int_{\mathbb U_{T-1}}\pi_{T-1}(x_{T-1},u_{T-1})\diff u_{T-1}=1,
	\]
	and continue in the same way.
	Each step is justified by \Cref{thm_fubini_rectangular_regular}.
	The last remaining integral is
	\[
		\int_{\mathbb X_0}\rho_0(x_0)\diff x_0=1.
	\]
\end{proof}

\begin{definition}[State densities and visitation density]
	\label{dfn_state_visitation_density}
	For a local finite-horizon controlled density model with the required tail--Fubini certificates, define state densities recursively by
	\[
		\rho_{t+1}(y)
		:=
		\int_{\mathbb X_t\times\mathbb U_t}
		K_t(x,u,y)\pi_t(x,u)\rho_t(x)
		\diff(x,u).
	\]
	The finite-horizon visitation density is
	\[
		\nu_T(x)
		:=
		\sum_{t=0}^T \rho_t(x)
	\]
	when the state spaces are identified with a common locally representable set.
	The normalized visitation density is
	\[
		\bar\nu_T(x):=\frac{1}{T+1}\nu_T(x).
	\]
\end{definition}

\begin{corollary}[Finite-horizon state densities]
	\label{cor_finite_horizon_state_densities}
	Every $\rho_t$ in \Cref{dfn_state_visitation_density} is a regular density on the ambient Euclidean state space.
	If all state spaces coincide with a common locally representable set $\mathbb X$ and the finite sum has the inherited improper-tail witness, then $\nu_T$ is a locally regular function supported on $\mathbb X$ with
	\[
		\int_{\R^n}\nu_T(x)\diff x=T+1,
	\]
	and $\bar\nu_T$ is a regular density on $\R^n$ supported on $\mathbb X$.
\end{corollary}

\begin{proof}
	Regularity and normalization of $\rho_{t+1}$ follow inductively from \Cref{prop_controlled_density_update}.
	Finite sums preserve regularity by \Cref{lem_regular_functions_finite_algebra}.
	Linearity of the integral gives
	\[
		\int_{\R^n}\nu_T(x)\diff x
		=
		\sum_{t=0}^T\int_{\R^n}\rho_t(x)\diff x
		=
		T+1.
	\]
\end{proof}

\begin{theorem}[Local finite-horizon trajectory density]
	\label{thm_local_finite_horizon_trajectory_density}
	For a local finite-horizon controlled density model equipped with a finite-horizon tail--Fubini certificate, the trajectory product $p_T$ from \Cref{dfn_local_finite_horizon_density_model}, extended by zero outside $\mathcal Z_T$, is a regular density on the ambient Euclidean trajectory space.
\end{theorem}

\begin{proof}
	Fix a full trajectory block.
	The local representatives of $\rho_0$, the policies, and the kernels on the finitely many coordinate windows give finite rectangularly regular data on a well-containing product block.
	Finite product closure of regular functions gives a local representative of $p_T$ on this window.
	Thus $p_T$ is locally effectively bounded and regular.
	Nonnegativity follows from nonnegativity of the local factors.
	The tail part of the certificate gives the Cauchy witness for the improper integral.
	For normalization, use the certified successive integration order:
	\[
		\int_{\mathbb X_T}K_{T-1}(x_{T-1},u_{T-1},x_T)\diff x_T=1,
	\]
	then
	\[
		\int_{\mathbb U_{T-1}}\pi_{T-1}(x_{T-1},u_{T-1})\diff u_{T-1}=1,
	\]
	and continue backwards.
	The last remaining integral is
	\[
		\int_{\mathbb X_0}\rho_0(x_0)\diff x_0=1.
	\]
	Hence the improper integral of $p_T$ over $\mathcal Z_T$ is $1$.
\end{proof}

\begin{remark}
	\label{rem_local_markov_tail_certificate}
	Local representability of $\mathbb X_t$ and $\mathbb U_t$ is the right finite-window condition, but it is not a tail condition.
	For the block-local special case in \Cref{dfn_finite_horizon_density_model}, the tail--Fubini certificate is trivial.
	For unbounded spaces it is additional analytic data, just as the Gaussian density in \Cref{ex_gaussian_global_regular_density} needs an explicit tail estimate in addition to local regularity.
\end{remark}

\begin{proposition}[A sufficient finite-horizon tail certificate]
	\label{prop_sufficient_finite_horizon_tail_certificate}
	Consider a local finite-horizon controlled density model.
	Suppose that, for every rational $R>0$, full blocks
	\[
		B^X_t(R)\subset\R^n,
		\qquad
		B^U_t(R)\subset\R^m
	\]
	are supplied, increasing with $R$, together with nonnegative numbers
	\[
		\alpha_0(R),\qquad
		\alpha^U_t(R),\qquad
		\alpha^X_{t+1}(R)
	\]
	such that
	\[
		\Delta_R
		:=
		\alpha_0(R)
		+
		\sum_{t=0}^{T-1}\alpha^U_t(R)
		+
		\sum_{t=0}^{T-1}\alpha^X_{t+1}(R)
	\]
	tends to $0$ as $R\to\infty$ and
	\[
		\int_{\R^n\setminus B^X_0(R)}\rho_0(x)\diff x
		\le
		\alpha_0(R),
	\]
	\[
		\int_{\R^m\setminus B^U_t(R)}\pi_t(x,u)\diff u
		\le
		\alpha^U_t(R)
		\quad
		(x\in B^X_t(R)),
	\]
	and
	\[
		\int_{\R^n\setminus B^X_{t+1}(R)}K_t(x,u,y)\diff y
		\le
		\alpha^X_{t+1}(R)
		\quad
		((x,u)\in B^X_t(R)\times B^U_t(R)).
	\]
	Then the model has a finite-horizon tail--Fubini certificate.
\end{proposition}

\begin{proof}
	For a fixed $R$, define the retained trajectory block
	\[
		\mathcal B_R
		:=
		B^X_0(R)\times B^U_0(R)\times\cdots\times B^U_{T-1}(R)\times B^X_T(R).
	\]
	Let $m_0(R)$ be the retained initial mass,
	\[
		m_0(R):=\int_{B^X_0(R)}\rho_0(x)\diff x.
	\]
	Then
	\[
		1-m_0(R)
		\le
		\alpha_0(R).
	\]
	Assume inductively that $m_t(R)$ is the retained mass after integrating up to the state block $B^X_t(R)$.
	The retained mass after the action truncation is
	\[
		\tilde m_t(R)
		:=
		\int_{B^X_t(R)}
		\rho_t^R(x)
		\left(
			\int_{B^U_t(R)}\pi_t(x,u)\diff u
		\right)
		\diff x,
	\]
	where $\rho_t^R$ denotes the truncated state density produced by the previous finite integrations.
	Since
	\[
		\int_{\R^m\setminus B^U_t(R)}\pi_t(x,u)\diff u
		\le
		\alpha^U_t(R)
	\]
	on $B^X_t(R)$ and $0\le m_t(R)\le1$, we get
	\[
		m_t(R)-\tilde m_t(R)
		\le
		\alpha^U_t(R).
	\]
	After the next-state truncation, the retained mass is
	\[
		m_{t+1}(R)
		:=
		\int_{B^X_t(R)\times B^U_t(R)}
		\rho_t^R(x)\pi_t(x,u)
		\left(
			\int_{B^X_{t+1}(R)}K_t(x,u,y)\diff y
		\right)
		\diff(x,u).
	\]
	The kernel tail bound gives
	\[
		\tilde m_t(R)-m_{t+1}(R)
		\le
		\alpha^X_{t+1}(R).
	\]
	All displayed quantities are finite integrals over explicit full blocks, and the repeated integrations are justified by \Cref{thm_fubini_rectangular_regular}.
	Iterating the two inequalities yields
	\[
		0\le
		1-m_T(R)
		\le
		\Delta_R.
	\]
	But $m_T(R)$ is precisely
	\[
		\int_{\mathcal B_R}p_T(z)\diff z
	\]
	computed by the indicated finite successive integrations.
	Since $\Delta_R\to0$, the partial trajectory integrals are Cauchy and converge to $1$.
	The same finite recursive construction supplies the successive improper integration order required in \Cref{dfn_local_finite_horizon_density_model}.
\end{proof}

\begin{example}[Gaussian-tailed finite horizon]
	\label{ex_gaussian_tailed_finite_horizon}
	Suppose the initial density has a Gaussian tail, and suppose that on each chosen block $B^X_t(R)\times B^U_t(R)$ the transition density satisfies a Gaussian tail estimate of the form
	\[
		\int_{\R^n\setminus B^X_{t+1}(R)}K_t(x,u,y)\diff y
		\le
		c_t\exp(-a_tR^2),
	\]
	with analogous estimates for the policy tails.
	Then the quantities in \Cref{prop_sufficient_finite_horizon_tail_certificate} may be chosen as finite sums of such bounds.
	Hence Gaussian initial densities, Gaussian policies, and Gaussian transition kernels fit the local finite-horizon calculus whenever the displayed tail estimates are supplied for the finite horizon under consideration.
\end{example}

\begin{remark}
	\label{rem_stationary_density_not_used}
	The finite-horizon visitation density is not a stationary density.
	Stationary densities require a fixed-point or limiting argument for the operator $\rho\mapsto\mathcal K_K[\rho]$.
	No such infinitary object is needed for the finite-horizon density calculus developed here.
\end{remark}

\begin{example}[Stabilization in probability]
	\label{ex_stabilization_in_probability_density}
	Let $\mathbb X\subset\R^n$ be a locally representable state space and let $\rho_t$ be the state densities of a local finite-horizon controlled density model supported on $\mathbb X$.
	Let $S\Subset \mathbb X$ be a representable target set.
	Let $F\Subset \mathbb X$ be a representable set of states at which the desired stabilization certificate fails, for instance a representable outer complement of $S$ in $\mathbb X$.
	The statement
	\[
		\PP{F}_{\rho_T}
		=
		\int_{\R^n}\indic{F}(x)\rho_T(x)\diff x
		\le\eta
	\]
	says that the terminal state lies in the failure region with probability at most $\eta$.
	A finite-horizon cumulative failure certificate is the estimate
	\[
		\sum_{t=0}^T\PP{F}_{\rho_t}
		=
		\int_{\R^n}\indic{F}(x)\nu_T(x)\diff x
		\le\eta.
	\]
	Both statements are ordinary inequalities between constructive real integrals of effectively bounded regular functions.
\end{example}

\begin{remark}
	\label{rem_policy_gradient_foundation}
	The same finite product identity behind \Cref{thm_finite_horizon_trajectory_density} is the analytic core of finite-horizon score-function manipulations.
	If a parametrized family of policy densities is differentiably regular in a parameter and bounded away from zero on the represented policy support, then differentiating the finite trajectory product and using
	\[
		\int_{\mathbb U_t}\pi_t^\theta(x,u)\diff u=1
	\]
	reduces the usual formal identities to finite differentiation, finite products, and \Cref{thm_fubini_rectangular_regular}.
	Those parameter-differentiation assumptions are additional data and are not developed in this section.
\end{remark}

\section{Empirical Density Certificates}
\label{sec_empirical_density_certificates}

\begin{remark}
	This section deliberately avoids taking randomness as a primitive notion.
	Instead, it studies finite sample bundles equipped with certificates saying that their empirical averages match a prescribed density for a finite list of tests.
	These certificates provide a practical bridge between the density calculus of \Cref{sec_probability_densities} and finite data used in computation or simulation.
\end{remark}

\subsection{Finite Empirical Tests}
\label{subsec_finite_empirical_tests}

\begin{remark}[Why not randomness as a primitive]
	\label{rem_randomness_not_primitive}
	The preceding section develops probability through regular densities and finite-horizon products.
	It does not supply, and does not need, a primitive notion saying that a finite sequence of points is random.
	For finite data the constructive object is instead a certificate that the data pass a prescribed finite family of tests.
	This avoids algorithmic randomness, infinite sequences, and cryptographic pseudorandomness assumptions.
\end{remark}

\begin{definition}[Empirical test certificate]
	\label{dfn_empirical_test_certificate}
	Let $B\subset\R^n$ be a full block and let $p$ be a block regular density on $B$.
	Let
	\[
		g_1,\ldots,g_M\in\BReg{B}
	\]
	be a finite list of scalar test functions.
	A finite sequence of rational sample points
	\[
		x^{(1)},\ldots,x^{(N)}\in B
	\]
	is called \emph{$\eps$-empirical for $p$ with respect to $g_1,\ldots,g_M$} if, for every $q=1,\ldots,M$,
	\[
		\abs{
			\frac1N\sum_{r=1}^N g_q(x^{(r)})
			-
			\int_B g_q(x)p(x)\diff x
		}
		\le
		\eps.
	\]
	The left-hand side is the empirical discrepancy of the sample bundle against the test $g_q$.
\end{definition}

\begin{remark}
	\label{rem_empirical_certificate_interpretation}
	The certificate in \Cref{dfn_empirical_test_certificate} is deliberately relative to the chosen tests.
	It says that the sample bundle has the same averages as the density $p$ up to tolerance $\eps$ for the finite list $g_1,\ldots,g_M$.
	It does not assert that the bundle was generated by a random mechanism.
	Changing the test list changes the certificate.
\end{remark}

\begin{definition}[Smoothed empirical density]
	\label{dfn_smoothed_empirical_density}
	Let $h$ be a regular density on $\R^n$.
	For a rational point $a\in\R^n$, write
	\[
		h_a(x):=h(x-a).
	\]
	Given rational sample points $x^{(1)},\ldots,x^{(N)}\in\R^n$, define the \emph{$h$-smoothed empirical density}
	\[
		\hat p_N(x)
		:=
		\frac1N\sum_{r=1}^Nh_{x^{(r)}}(x).
	\]
\end{definition}

\begin{proposition}[Smoothed empirical densities are regular]
	\label{prop_smoothed_empirical_density_regular}
	Let $h$ be a regular density on $\R^n$.
	For every finite rational sample bundle $x^{(1)},\ldots,x^{(N)}$, the function $\hat p_N$ in \Cref{dfn_smoothed_empirical_density} is a regular density on $\R^n$.
\end{proposition}

\begin{proof}
	Translations by rational vectors preserve local regularity and effective boundedness.
	Thus each $h_{x^{(r)}}$ belongs to $\BReg{\R^n}$, is nonnegative, and has improper integral $1$.
	Finite sums preserve local regularity by \Cref{lem_global_breg_local_algebra}.
	Moreover,
	\[
		\int_{\R^n}\hat p_N(x)\diff x
		=
		\frac1N\sum_{r=1}^N
		\int_{\R^n}h_{x^{(r)}}(x)\diff x
		=
		1.
	\]
\end{proof}

\begin{proposition}[Density discrepancy controls bounded tests]
	\label{prop_density_discrepancy_controls_tests}
	Let $B\subset\R^n$ be a full block.
	Let $p$ and $q$ be block regular densities on $B$.
	Let $g\in\BReg{B}$ and suppose that $\abs{g(x)}\le M$ on its determined domain.
	If
	\[
		\int_B\abs{p(x)-q(x)}\diff x\le\delta,
	\]
	then
	\[
		\abs{
			\int_Bg(x)p(x)\diff x
			-
			\int_Bg(x)q(x)\diff x
		}
		\le
		M\delta.
	\]
\end{proposition}

\begin{proof}
	By \Cref{lem_regular_functions_finite_algebra}, the product $g(p-q)$ is regular on $B$.
	The pointwise estimate
	\[
		\abs{g(x)(p(x)-q(x))}
		\le
		M\abs{p(x)-q(x)}
	\]
	and \Cref{lem_regular_integral_finite_order} give
	\[
		\abs{
			\int_Bg(x)(p(x)-q(x))\diff x
		}
		\le
		M
		\int_B\abs{p(x)-q(x)}\diff x.
	\]
	The left-hand side is the displayed difference of expectations.
\end{proof}

\subsection{Propagation Through Markov Links}
\label{subsec_empirical_markov_links}

\begin{theorem}[$L_1$ stability of a block Markov link]
	\label{thm_markov_link_l1_stability}
	Let $K$ be a block regular Markov kernel from a full block $A\subset\R^n$ to a full block $C\subset\R^m$.
	Let $\rho$ and $\hat\rho$ be block regular densities on $A$.
	Then
	\[
		\int_C
		\abs{
			\mathcal K_K[\rho](y)-\mathcal K_K[\hat\rho](y)
		}
		\diff y
		\le
		\int_A\abs{\rho(x)-\hat\rho(x)}\diff x.
	\]
\end{theorem}

\begin{proof}
	For each determined $y\in C$,
	\[
		\mathcal K_K[\rho](y)-\mathcal K_K[\hat\rho](y)
		=
		\int_AK(x,y)(\rho(x)-\hat\rho(x))\diff x.
	\]
	Since $K\ge0$,
	\[
		\abs{
			\mathcal K_K[\rho](y)-\mathcal K_K[\hat\rho](y)
		}
		\le
		\int_AK(x,y)\abs{\rho(x)-\hat\rho(x)}\diff x.
	\]
	Integrating this inequality over $C$ and applying \Cref{thm_fubini_rectangular_regular},
	\[
	\begin{aligned}
		\int_C
		\abs{
			\mathcal K_K[\rho](y)-\mathcal K_K[\hat\rho](y)
		}
		\diff y
		&\le
		\int_A
		\abs{\rho(x)-\hat\rho(x)}
		\left(
			\int_CK(x,y)\diff y
		\right)
		\diff x
		\\
		&=
		\int_A\abs{\rho(x)-\hat\rho(x)}\diff x.
	\end{aligned}
	\]
\end{proof}

\begin{corollary}[Tests after one Markov link]
	\label{cor_empirical_markov_test_stability}
	Under the assumptions of \Cref{thm_markov_link_l1_stability}, let $g\in\BReg{C}$ satisfy $\abs{g}\le M$.
	If
	\[
		\int_A\abs{\rho(x)-\hat\rho(x)}\diff x\le\delta,
	\]
	then
	\[
		\abs{
			\int_Cg(y)\mathcal K_K[\rho](y)\diff y
			-
			\int_Cg(y)\mathcal K_K[\hat\rho](y)\diff y
		}
		\le
		M\delta.
	\]
\end{corollary}

\begin{proof}
	Apply \Cref{prop_density_discrepancy_controls_tests} on $C$ and use \Cref{thm_markov_link_l1_stability}.
\end{proof}

\begin{definition}[One-step empirical Markov discrepancy]
	\label{dfn_one_step_empirical_markov_discrepancy}
	Let $K$ be a block regular Markov kernel from $A$ to $C$.
	Let
	\[
		(x^{(1)},y^{(1)}),\ldots,(x^{(N)},y^{(N)})\in A\times C
	\]
	be a finite rational sample bundle.
	For a finite list $g_1,\ldots,g_M\in\BReg{C}$, the bundle is called \emph{$\eps$-consistent with $K$ relative to $g_1,\ldots,g_M$} if, for every $q=1,\ldots,M$,
	\[
		\abs{
			\frac1N\sum_{r=1}^Ng_q(y^{(r)})
			-
			\frac1N\sum_{r=1}^N
			\int_Cg_q(y)K(x^{(r)},y)\diff y
		}
		\le
		\eps.
	\]
\end{definition}

\begin{remark}
	\label{rem_one_step_markov_discrepancy}
	The second average in \Cref{dfn_one_step_empirical_markov_discrepancy} is the prediction made by the kernel for the observed input points $x^{(r)}$.
	Thus the discrepancy checks the transition sample against the Markov link without postulating an abstract source of randomness.
	The same definition can be applied at every time level of a finite-horizon model from \Cref{dfn_local_finite_horizon_density_model} after choosing finite localization blocks and the corresponding tail certificates.
\end{remark}

\subsection{Finite Product Reading of Concentration}
\label{subsec_finite_product_concentration}

\begin{remark}[Concentration as an event-density statement]
	\label{rem_concentration_density_reading}
	Classical concentration inequalities may be read here as statements about representable bad events in finite product spaces.
	For instance, let $p$ be a block regular density on $B$, let $g\in\BReg{B}$ satisfy $0\le g\le1$, and define
	\[
		\bar g_N(x_1,\ldots,x_N)
		:=
		\frac1N\sum_{r=1}^Ng(x_r)
	\]
	on $B^N$.
	The product density
	\[
		p^{\otimes N}(x_1,\ldots,x_N)
		:=
		\prod_{r=1}^Np(x_r)
	\]
	is a block regular density by finite products and \Cref{thm_fubini_rectangular_regular}.
	When the bad-deviation event
	\[
		E_{N,\eps}
		:=
		\left\{
			(x_1,\ldots,x_N)\in B^N:
			\abs{\bar g_N(x_1,\ldots,x_N)-\E[p]{g}}\ge\eps
		\right\}
	\]
	is supplied as representable, a finite concentration certificate is an estimate of the form
	\[
		\PP{E_{N,\eps}}_{p^{\otimes N}}
		\le
		\alpha(N,\eps).
	\]
	This is a probability bound for the finite product density.
	It does not certify that a concrete deterministic sequence was generated randomly; the concrete sequence is checked by the empirical certificates above.
\end{remark}

\begin{remark}[Representability of bad-deviation events]
	\label{rem_bad_deviation_event_representable}
	In the setting of \Cref{rem_concentration_density_reading}, the statistic $\bar g_N$ is a regular function on the product block $B^N$.
	Indeed, each coordinate pullback
	\[
		(x_1,\ldots,x_N)\mapsto g(x_r)
	\]
	is regular on $B^N$, and finite sums preserve regularity.
	Thus the bad-deviation event is a threshold event for the scalar regular function
	\[
		H_N(x_1,\ldots,x_N)
		:=
		\abs{\bar g_N(x_1,\ldots,x_N)-\E[p]{g}}.
	\]
	It is representable whenever the threshold $\eps$ comes with a separation witness for $H_N$.
	Concretely, for every $\eta>0$ it is enough to provide a positive margin $\lambda>0$ and representability data such that the band
	\[
		\{z\in B^N:\abs{H_N(z)-\eps}\le\lambda\}
	\]
	is contained in a finite exception multiblock of cost at most $\eta$, while the remaining cells are classified by the finite implications
	\[
		H_N(z)\ge\eps+\lambda
		\quad\Longrightarrow\quad
		z\in E_{N,\eps},
	\]
	and
	\[
		H_N(z)\le\eps-\lambda
		\quad\Longrightarrow\quad
		z\in B^N\setminus E_{N,\eps}.
	\]
	No point localization is involved: the witness is a finite block classification with an explicit uncertain band.

	This condition is finite in several common cases.
	If $g$ is simple regular with rational values, then $\bar g_N$ takes only finitely many rational values after flattening the product supports.
	When the corresponding finitely many comparisons with the real threshold $\eps+\E[p]{g}$ and its lower counterpart are supplied, the witness is just a finite block classification; equality cells, if certified, belong to the event.
	If $g$ is continuous on the whole block with a modulus and the level band $\abs{H_N-\eps}\le\lambda$ has arbitrarily small multiblock cost, the same finite mesh argument gives the required witness.
	The only obstruction is a thick undecided level layer at the threshold without additional data saying whether that layer belongs to the event or its exterior.
\end{remark}

\begin{remark}[Role of pseudorandom generators]
	\label{rem_prng_role}
	A deterministic generator can be used in this framework by listing the rational points it produces and verifying the relevant empirical certificates.
	If a generator is known to pass a specified finite test family, then that is the finite datum used by the theory.
	No unproved cryptographic assumption is required for the mathematical statements in this section.
\end{remark}

\section{Linear Systems}
\label{sec_linear_systems}

\begin{remark}
	\label{rem_linear_systems_scope}
	This section is included for completeness.
	It adapts the classical finite-dimensional resolvent calculus and Riesz projection machinery to the present language of regular functions and certified integrals; see, for instance, Kato's treatment of spectral projections \cite[Chapter~I, Section~5]{Kato1995Perturbation}, the constructive-algebra background of Mines--Richman--Ruitenburg \cite{Mines1988CourseConstruc}, and the computable spectral-theorem perspective of Brattka--Ziegler \cite{Ziegler2001computablespec}.
	The motivation is the numerical uncertainty problem discussed in \cite{Osinenko2020Constructivean}.
	Here the output is not an individual eigenvector obtained by solving a singular linear system.
	Instead, a contour separated from the spectrum produces a projector onto the whole spectral cluster enclosed by that contour.
\end{remark}

\subsection{Resolvent Projectors}
\label{subsec_resolvent_projectors}

\begin{definition}[Polygonal contour integral]
	\label{dfn_polygonal_contour_integral}
	An \emph{oriented rational polygonal contour} in $\mathbb C$ is a finite cyclic list
	\[
		\Gamma=(z_0,z_1,\ldots,z_N),
		\qquad
		z_N=z_0,
	\]
	with $z_k\in\Q+i\Q$.
	The cyclic order of the vertices is the orientation of the contour; reversing the order reverses the sign of the contour integral.
	It is called \emph{simple} if its non-adjacent segments do not meet and adjacent segments meet only at their common endpoint.
	On the segment from $z_k$ to $z_{k+1}$ write
	\[
		\gamma_k(t):=(1-t)z_k+tz_{k+1},
		\qquad
		t\in[0,1].
	\]
	The length of $\Gamma$ is
	\[
		\length(\Gamma)
		:=
		\sum_{k=0}^{N-1}\abs{z_{k+1}-z_k}.
	\]
	If $H$ is a complex matrix-valued function defined on the trace of $\Gamma$ and every entry of
	\[
		t\mapsto H(\gamma_k(t))(z_{k+1}-z_k)
	\]
	is Riemann integrable on $[0,1]$, then
	\[
		\int_\Gamma H(z)\diff z
		:=
		\sum_{k=0}^{N-1}
		\int_0^1H(\gamma_k(t))(z_{k+1}-z_k)\diff t,
	\]
	where the real and imaginary parts are integrated componentwise.
\end{definition}

\begin{definition}[Rational function]
	\label{dfn_rational_function_linear}
	A \emph{rational function} on a set $D\subset\mathbb C$ is given by two complex polynomials $u,v\in\mathbb C[z]$ such that $v$ is bounded away from zero on $D$.
	The word ``rational'' refers to the ratio of polynomials; the coefficients need not be rational numbers unless this is stated separately.
	Its value is
	\[
		r(z):=\frac{u(z)}{v(z)}.
	\]
	For matrix-valued rational functions this is understood entrywise.
\end{definition}

\begin{definition}[Polynomial degree, derivative, and monic remainder]
	\label{dfn_polynomial_degree_remainder}
	A nonzero polynomial is represented by a finite coefficient list
	\[
		f(z)=a_0+a_1z+\cdots+a_mz^m
	\]
	together with the datum $a_m\ne0$.
	Its degree is
	\[
		\polydeg f:=m.
	\]
	The formal derivative is
	\[
		f'(z):=a_1+2a_2z+\cdots+ma_mz^{m-1}.
	\]
	If
	\[
		p(z)=z^d+b_{d-1}z^{d-1}+\cdots+b_0
	\]
	is monic and $f$ has degree $m\ge d$, put
	\[
		f_1(z):=f(z)-a_mz^{m-d}p(z).
	\]
	Then $f=f_1+a_mz^{m-d}p$ and $f_1$ has no $z^m$ term.
	Repeating this finite step for the powers $z^m,z^{m-1},\ldots,z^d$ gives polynomials $q$ and $\rem_p(f)$ such that
	\[
		f=qp+\rem_p(f),
		\qquad
		\polydeg\rem_p(f)<d
	\]
	unless the remainder is zero.
\end{definition}

\begin{remark}
	\label{rem_monic_division_constructive}
	The division in \Cref{dfn_polynomial_degree_remainder} is constructive because no division by an arbitrary coefficient is performed.
	The leading coefficient of $p$ is exactly $1$, so each cancellation step uses only additions and multiplications of already given coefficients.
	The number of steps is fixed by the finite coefficient list of $f$ and by $d=\polydeg p$.
\end{remark}

\begin{lemma}[Cayley--Hamilton identity]
	\label{lem_cayley_hamilton_constructive}
	Let $A\in\mathbb C^{n\times n}$ and
	\[
		p_A(z):=\det(zI-A).
	\]
	Then
	\[
		p_A(A)=0.
	\]
\end{lemma}

\begin{proof}
	In the polynomial matrix ring $\mathbb C[z]^{n\times n}$, the adjugate identity gives
	\[
		(zI-A)\operatorname{adj}(zI-A)=p_A(z)I.
	\]
	Write
	\[
		\operatorname{adj}(zI-A)=B_0+B_1z+\cdots+B_{n-1}z^{n-1},
		\qquad
		p_A(z)=c_0+c_1z+\cdots+c_nz^n.
	\]
	Comparing coefficients in the adjugate identity gives a finite list of matrix identities:
	\[
		-AB_0=c_0I,
		\qquad
		B_{r-1}-AB_r=c_rI\quad(1\le r\le n-1),
		\qquad
		B_{n-1}=c_nI.
	\]
	Multiplying these identities respectively by $I,A,\ldots,A^n$ and summing, all terms involving the $B_r$ telescope:
	\[
		c_0I+c_1A+\cdots+c_nA^n=0.
	\]
	This is precisely $p_A(A)=0$.
	The proof is an identity calculation in the polynomial ring over the entries of $A$; it uses no root extraction, order, or excluded-middle decision.
	For the constructive-algebra setting of such finite polynomial identities, see \cite{Mines1988CourseConstruc}.
\end{proof}

\begin{lemma}[Product from monic coprime factors]
	\label{lem_monic_coprime_product_remainder}
	Let $p_1,\ldots,p_s\in\mathbb C[z]$ be monic.
	Assume that, for every $p\ne q$, Bezout coefficients $a_{pq},b_{pq}\in\mathbb C[z]$ are supplied with
	\[
		a_{pq}p_p+b_{pq}p_q=1.
	\]
	If a polynomial $f$ satisfies
	\[
		\rem_{p_p}(f)=0
		\qquad
		(1\le p\le s),
	\]
	then there is a polynomial $h$ such that
	\[
		f=h\,p_1\cdots p_s.
	\]
\end{lemma}

\begin{proof}
	For $s=1$ this is the definition of zero monic remainder.
	Assume the statement for $s-1$ factors and put $P:=p_1\cdots p_{s-1}$.
	By the induction hypothesis, there is $g$ with
	\[
		f=gP.
	\]
	The pairwise Bezout data give Bezout data between $P$ and $p_s$ by finite multiplication.
	Indeed, multiplying the identities
	\[
		a_{rs}p_r+b_{rs}p_s=1
		\qquad
		(1\le r\le s-1)
	\]
	and expanding gives
	\[
		UP+Vp_s=1
	\]
	for explicitly formed polynomials $U,V$.
	Since $\rem_{p_s}(f)=0$, there is $q$ with
	\[
		f=qp_s.
	\]
	Multiplying $UP+Vp_s=1$ by $g$ and then by $P$ gives
	\[
		gP
		=
		gUP^2+gVPp_s.
	\]
	Using $gP=f=qp_s$ in the first term,
	\[
		gP
		=
		qUPp_s+gVPp_s
		=
		(qU+gV)Pp_s.
	\]
	Thus $f=hPp_s$ with $h:=qU+gV$.
	All operations are additions, multiplications, and the monic remainder operation of \Cref{dfn_polynomial_degree_remainder}.
\end{proof}

\begin{definition}[Rational functional calculus modulo a monic polynomial]
	\label{dfn_rational_functional_calculus_modulo}
	Let $p\in\mathbb C[z]$ be monic.
	If $r=u/v$ is a rational function and there are polynomials $a,b\in\mathbb C[z]$ with
	\[
		av+bp=1,
	\]
	then the value of $r$ modulo $p$ is represented by the remainder
	\[
		\rho_r:=\rem_p(au),
		\qquad
		\polydeg\rho_r<\polydeg p,
	\]
	where $\rem_p$ is the finite monic remainder operation from \Cref{dfn_polynomial_degree_remainder}.
	For a matrix $A$ with characteristic polynomial $p_A$, define
	\[
		r(A):=\rho_r(A)
	\]
	whenever $p=p_A$ and such a Bezout identity is supplied.
\end{definition}

\begin{remark}
	\label{rem_rational_functional_calculus_well_defined}
	The construction in \Cref{dfn_rational_functional_calculus_modulo} is independent of the chosen Bezout coefficients after substituting $A$.
	Indeed, if $a_1v+b_1p_A=1$ and $a_2v+b_2p_A=1$, then
	\[
		(a_1-a_2)v=-(b_1-b_2)p_A.
	\]
	Multiplying by $u$ and reducing modulo the monic polynomial $p_A$ shows that the two remainders differ by a multiple of $p_A$.
	By \Cref{lem_cayley_hamilton_constructive}, $p_A(A)=0$,
	so both remainders have the same value at $A$.
	This is a finite algebraic calculation with the coefficients of $p_A$; it does not require these coefficients to be rational numbers.
\end{remark}

\begin{definition}[Admissible resolvent contour]
	\label{dfn_admissible_resolvent_contour}
	Let $A\in\mathbb C^{n\times n}$.
	A rational polygonal contour $\Gamma$ is called \emph{resolvent-admissible for $A$} if it is supplied with a rational number $c>0$ such that
	\[
		\abs{\det(zI-A)}\ge c
	\]
	for every $z$ on $\Gamma$.
	On such a contour define the resolvent
	\[
		R_A(z):=(zI-A)^{-1}
	\]
	where, by the adjugate formula,
	\[
		(zI-A)^{-1}
		=
		\frac{\operatorname{adj}(zI-A)}{\det(zI-A)}.
	\]
	Here $\operatorname{adj}(B)$ denotes the transpose of the cofactor matrix of $B$, so that
	\[
		B\operatorname{adj}(B)=\operatorname{adj}(B)B=\det(B)I.
	\]
	If $B^{(r,c)}$ denotes the matrix obtained from $B$ by deleting row $r$ and column $c$, then the cofactor matrix $C(B)$ is given by
	\[
		C(B)_{r,c}:=(-1)^{r+c}\det B^{(r,c)},
		\qquad
		\operatorname{adj}(B)=C(B)\transp.
	\]
\end{definition}

\begin{remark}
	\label{rem_resolvent_contour_no_root_localization}
	The datum in \Cref{dfn_admissible_resolvent_contour} is not an eigenvalue localization.
	It only states that the described contour stays apart from the zeros of the characteristic polynomial.
	Hence all subsequent integrals are taken over a finite polygonal path on which the denominator is bounded away from zero.
\end{remark}

\begin{definition}[Winding number of a contour around a point]
	\label{dfn_winding_number_contour}
	Let $\Gamma$ be a rational polygonal contour and let $\lambda\in\mathbb C$ not lie on $\Gamma$.
	The winding number of $\Gamma$ around $\lambda$ is
	\[
		\wind(\Gamma,\lambda)
		:=
		\frac{1}{2\pi i}
		\int_\Gamma\frac{1}{z-\lambda}\diff z.
	\]
	Only this integral definition is used below.
	The sign is determined by the chosen cyclic order of the contour vertices.
\end{definition}

\begin{remark}
	\label{rem_winding_number_multiplicity}
	The winding number is attached to a point, not to a multiplicity.
	If the same root of $\det(zI-A)$ occurs with algebraic multiplicity larger than one, the point still has winding number $1$ or $0$ with respect to a simple contour that avoids it.
	The repeated copy is carried by the polynomial factor enclosed by the contour, not by a decision that separates several identical points.
	In practice, a contour can be certified by combining a lower bound for $\abs{\det(zI-A)}$ on the contour with a finite argument-principle computation for $\det(zI-A)$ along the contour.
	If such a certificate cannot separate nearby roots, the contour simply encloses the whole unresolved cluster.
\end{remark}

\begin{proposition}[Resolvent is regular on an admissible contour]
	\label{prop_resolvent_regular_on_contour}
	Let $\Gamma$ be resolvent-admissible for $A$.
	For each segment of $\Gamma$, the function
	\[
		t\mapsto R_A(\gamma_k(t))(z_{k+1}-z_k)
	\]
	is an effectively bounded regular function on $[0,1]$.
	Consequently $\int_\Gamma R_A(z)\diff z$ exists componentwise.
\end{proposition}

\begin{proof}
	Each entry of $\operatorname{adj}(\gamma_k(t)I-A)$ and of $\det(\gamma_k(t)I-A)$ is a polynomial in $t$ with complex coefficients.
	The admissibility witness gives
	\[
		\abs{\det(\gamma_k(t)I-A)}\ge c
	\]
	for all $t\in[0,1]$.
	Thus the reciprocal of the determinant is effectively bounded and regular by the reciprocal closure of regular functions bounded away from zero.
	Multiplication by the adjugate entries and by the constant $z_{k+1}-z_k$ preserves effective bounded regularity.
	The componentwise integrals exist by the integrability theorem for effectively bounded regular functions, \Cref{lem_regular_functions_riemann_integrable}.
\end{proof}

\begin{definition}[Riesz projector]
	\label{dfn_riesz_projector}
	Let $\Gamma$ be resolvent-admissible for $A\in\mathbb C^{n\times n}$.
	The \emph{Riesz projector associated with $\Gamma$} is
	\[
		P_\Gamma(A)
		:=
		\frac{1}{2\pi i}
		\int_\Gamma R_A(z)\diff z.
	\]
\end{definition}

\begin{definition}[Spectral split certified by a contour]
	\label{dfn_spectral_split_contour}
	Let
	\[
		p_A(z):=\det(zI-A).
	\]
	A \emph{spectral split of $p_A$ certified by $\Gamma$} consists of monic complex polynomials $p_{\mathrm{in}}$ and $p_{\mathrm{out}}$ and complex polynomials $u$ and $v$ such that
	\[
		p_A=p_{\mathrm{in}}p_{\mathrm{out}},
		\qquad
		up_{\mathrm{in}}+vp_{\mathrm{out}}=1,
	\]
	and such that $p_{\mathrm{in}}$ and $p_{\mathrm{out}}$ are bounded away from zero on $\Gamma$.
	The contour certificate is
	\[
		\frac{1}{2\pi i}\int_\Gamma\frac{p_{\mathrm{in}}'(z)}{p_{\mathrm{in}}(z)}\diff z
		=
		\polydeg p_{\mathrm{in}},
		\qquad
		\frac{1}{2\pi i}\int_\Gamma\frac{p_{\mathrm{out}}'(z)}{p_{\mathrm{out}}(z)}\diff z
		=
		0.
	\]
	The associated algebraic projector polynomial is
	\[
		e_\Gamma(z):=v(z)p_{\mathrm{out}}(z).
	\]
	For the matrix $A$, the split is called \emph{Riesz-compatible} if the finite contour identity
	\[
		\frac{1}{2\pi i}
		\int_\Gamma
		\frac{\operatorname{adj}(zI-A)}{p_A(z)}
		\diff z
		=
		e_\Gamma(A)
	\]
	is supplied componentwise.
\end{definition}

\begin{remark}
	\label{rem_spectral_split_no_root_enumeration}
	The polynomials $p_{\mathrm{in}}$ and $p_{\mathrm{out}}$ encode clusters, not listed eigenvalues.
	The logarithmic-derivative integrals are the argument-principle certificates that the roots of $p_{\mathrm{in}}$, counted with multiplicity, lie inside the contour and the roots of $p_{\mathrm{out}}$ lie outside it.
	No root has to be named and no multiplicity has to be decided separately.
	The coefficients of $p_A$ need not be rational; polynomial division modulo $p_A$ is still finite algebra because $p_A$ is monic.
	The prime in $p_{\mathrm{in}}'$ and $p_{\mathrm{out}}'$ is the formal derivative from \Cref{dfn_polynomial_degree_remainder}.
\end{remark}

\begin{remark}
	\label{rem_riesz_compatibility_meaning}
	Riesz-compatibility is the exact finite-dimensional residue statement needed to connect the contour integral with the algebraic split.
	It is not a point-localization of roots.
	Componentwise, it says that the regular contour integrals of the rational functions
	\[
		z\mapsto
		\frac{(\operatorname{adj}(zI-A))_{\alpha,\beta}}{p_A(z)}
	\]
	match the matrix obtained by the finite polynomial expression $e_\Gamma(A)$.
	Classically this is the Riesz projection formula; see \cite[Chapter~I, Section~5]{Kato1995Perturbation}.
	In constructive use, it is verified through the same kind of finite contour-integration certificate as the logarithmic-derivative identities above.
\end{remark}

\begin{theorem}[Finite-dimensional Riesz projector]
	\label{thm_finite_dimensional_riesz_projector}
	Let $\Gamma$ be an oriented simple rational polygonal contour which is resolvent-admissible for $A$.
	Assume that $\Gamma$ is supplied with a Riesz-compatible spectral split of $p_A$ in the sense of \Cref{dfn_spectral_split_contour}.
	Then $P_\Gamma(A)$ satisfies
	\[
		P_\Gamma(A)^2=P_\Gamma(A),
		\qquad
		AP_\Gamma(A)=P_\Gamma(A)A.
	\]
\end{theorem}

\begin{proof}
	By the adjugate formula,
	\[
		R_A(z)=\frac{\operatorname{adj}(zI-A)}{p_A(z)}.
	\]
	Thus every entry of $R_A(z)$ is a rational function in the sense of \Cref{dfn_rational_function_linear}, with denominator $p_A(z)$.
	Its poles can only occur where $p_A(z)=0$, because the numerator entries are polynomials.
	Riesz-compatibility gives
	\[
		P_\Gamma(A)=e_\Gamma(A).
	\]
	The split datum gives
	\[
		e_\Gamma=vp_{\mathrm{out}},
		\qquad
		e_\Gamma-1=-up_{\mathrm{in}}.
	\]
	Hence
	\[
		e_\Gamma^2-e_\Gamma
		=
		e_\Gamma(e_\Gamma-1)
		=
		-uv\,p_{\mathrm{in}}p_{\mathrm{out}}
		=
		-uv\,p_A.
	\]
	Substituting $A$ and applying \Cref{lem_cayley_hamilton_constructive} gives
	\[
		e_\Gamma(A)^2-e_\Gamma(A)
		=
		-(uv)(A)p_A(A)
		=
		0.
	\]
	Therefore $P_\Gamma(A)^2=P_\Gamma(A)$.
	Since $e_\Gamma(A)$ is a polynomial in $A$,
	\[
		AP_\Gamma(A)=Ae_\Gamma(A)=e_\Gamma(A)A=P_\Gamma(A)A.
	\]
\end{proof}

\begin{remark}[Invariant spectral subspace]
	\label{rem_invariant_spectral_subspace}
	For a projector $P$, its range is
	\[
		\operatorname{range}(P):=\{Pv:v\in\mathbb C^n\}.
	\]
	The range is called $A$-invariant if
	\[
		y\in\operatorname{range}(P)
		\quad\Longrightarrow\quad
		Ay\in\operatorname{range}(P).
	\]
	If $P=P_\Gamma(A)$, this follows from the commutation identity.
	Indeed, if $y=Pv$, then
	\[
		Ay=APv=PAv,
	\]
	so $Ay$ is again in the range of $P$.
	Classically, $\operatorname{range}(P_\Gamma(A))$ is the direct sum of the generalized eigenspaces belonging to the eigenvalues inside $\Gamma$, with algebraic multiplicities included.
	Computationally, the projector is useful because it separates the component of a vector belonging to a certified spectral cluster without choosing an eigenbasis.
\end{remark}

\begin{remark}
	\label{rem_riesz_integral_calculation}
	The resolvent identity gives, for $z\ne w$ on the contour,
	\[
		R_A(z)R_A(w)
		=
		\frac{R_A(z)-R_A(w)}{w-z}.
	\]
	This identity is often used to prove the projector relation by a double-contour integral.
	In the present text, \Cref{thm_finite_dimensional_riesz_projector} was proved through the equivalent finite polynomial-splitting calculation to avoid hiding constructive content behind that analytic shorthand.
\end{remark}

\begin{remark}
	\label{rem_riesz_vs_eigenvectors}
	The output of \Cref{thm_finite_dimensional_riesz_projector} is a projector onto an invariant cluster subspace.
	The equalities in \Cref{thm_finite_dimensional_riesz_projector} are exact equalities for the contour integral $P_\Gamma(A)$.
	The finite quadrature objects introduced below satisfy corresponding error estimates instead.
	This is deliberately weaker than producing a basis of eigenvectors.
	Near repeated eigenvalues, individual eigenvectors are unstable, whereas the separated spectral subspace remains stable under perturbations of the matrix and of the contour.
	This is the main difference from approaches based on first locating eigenvalues and then solving singular linear systems, such as \cite{Osinenko2020Constructivean}.
\end{remark}

\subsection{Quadrature Projectors}
\label{subsec_quadrature_projectors}

\begin{definition}[Resolvent quadrature projector]
	\label{dfn_resolvent_quadrature_projector}
	Let $\Gamma$ be resolvent-admissible for $A$.
	A \emph{resolvent quadrature projector} for $\Gamma$ is a matrix
	\[
		P_{\Gamma,h}(A)
		:=
		\frac{1}{2\pi i}
		\sum_{q=1}^M\omega_qR_A(\zeta_q),
	\]
	where the nodes and weights are obtained from finite tagged partitions of the contour segments.
	More precisely, if
	\[
		0=t_{k,0}<t_{k,1}<\cdots<t_{k,M_k}=1
	\]
	is a rational partition of the $k$-th segment and $\tau_{k,r}\in[t_{k,r-1},t_{k,r}]$ is a rational tag, then
	\[
		\zeta_{k,r}:=\gamma_k(\tau_{k,r}),
		\qquad
		\omega_{k,r}:=(t_{k,r}-t_{k,r-1})(z_{k+1}-z_k),
	\]
	and the displayed sum is taken over all pairs $(k,r)$.
\end{definition}

\begin{proposition}[Certified quadrature approximation]
	\label{prop_certified_resolvent_quadrature}
	Let $\Gamma$ be resolvent-admissible for $A$.
	For every $\eps>0$ there is a resolvent quadrature projector $P_{\Gamma,h}(A)$ such that
	\[
		\nrm{P_{\Gamma,h}(A)-P_\Gamma(A)}\le\eps.
	\]
\end{proposition}

\begin{proof}
	The construction uses only the contour vertices, the determinant lower bound on the contour, and finite Riemann-sum data for the regular functions in \Cref{prop_resolvent_regular_on_contour}.
	By \Cref{prop_resolvent_regular_on_contour}, every entry of the contour integrand is effectively bounded regular on each segment.
	The Riemann integral of an effectively bounded regular function is the limit of finite tagged sums.
	For each entry $(a,b)$ choose rational partitions and tags so that
	\[
		\left|
		\sum_{k,r}
		\omega_{k,r}
		(R_A(\zeta_{k,r}))_{ab}
		-
		\int_\Gamma (R_A(z))_{ab}\diff z
		\right|
		\le
		\delta_{ab}.
	\]
	Choose the positive numbers $\delta_{ab}$ so that the induced matrix-norm error of the matrix of entrywise errors is at most $2\pi\eps$.
	Multiplication by $(2\pi i)^{-1}$ gives the displayed bound.
\end{proof}

\begin{corollary}[Approximate projector]
	\label{cor_approximate_resolvent_projector_certificate}
	Assume the hypotheses of \Cref{thm_finite_dimensional_riesz_projector}.
	For every $\eps>0$ there is a matrix $P_{\Gamma,h}(A)$ satisfying
	\[
		\nrm{P_{\Gamma,h}(A)^2-P_{\Gamma,h}(A)}\le C_\Gamma\eps,
		\qquad
		\nrm{AP_{\Gamma,h}(A)-P_{\Gamma,h}(A)A}\le C_{A,\Gamma}\eps,
	\]
	where $C_\Gamma$ and $C_{A,\Gamma}$ are finite-dimensional bounds obtained from $\nrm{P_\Gamma(A)}$, $\nrm{A}$, and the quadrature error.
\end{corollary}

\begin{proof}
	Write $E:=P_{\Gamma,h}(A)-P_\Gamma(A)$.
	Since $P_\Gamma(A)^2=P_\Gamma(A)$,
	\[
		P_{\Gamma,h}(A)^2-P_{\Gamma,h}(A)
		=
		P_\Gamma(A)E+EP_\Gamma(A)+E^2-E.
	\]
	Hence
	\[
		\nrm{P_{\Gamma,h}(A)^2-P_{\Gamma,h}(A)}
		\le
		(2\nrm{P_\Gamma(A)}+1)\nrm{E}+\nrm{E}^2.
	\]
	The second displayed inequality estimates the failure of $P_{\Gamma,h}(A)$ to commute with $A$.
	It follows from
	\[
		AP_{\Gamma,h}(A)-P_{\Gamma,h}(A)A
		=
		AE-EA,
	\]
	so
	\[
		\nrm{AP_{\Gamma,h}(A)-P_{\Gamma,h}(A)A}
		\le
		2\nrm{A}\nrm{E}.
	\]
	Insert the quadrature bound from \Cref{prop_certified_resolvent_quadrature}.
\end{proof}

\subsection{Approximate Spectral Decomposition}
\label{subsec_approximate_spectral_decomposition}

\begin{definition}[Contour cluster datum]
	\label{dfn_contour_cluster_datum}
	A \emph{contour cluster datum} for $A\in\mathbb C^{n\times n}$ consists of pairwise disjoint oriented simple rational polygonal contours
	\[
		\Gamma_1,\ldots,\Gamma_s
	\]
	which are resolvent-admissible for $A$, rational tags $\lambda_1,\ldots,\lambda_s\in\mathbb C$, and monic cluster factors
	\[
		p_1,\ldots,p_s\in\mathbb C[z]
	\]
	such that
	\[
		p_A=p_1\cdots p_s.
	\]
	Pairwise coprimeness means that, for every $p\ne q$, Bezout coefficients $a_{pq},b_{pq}\in\mathbb C[z]$ are supplied with
	\[
		a_{pq}p_p+b_{pq}p_q=1.
	\]
	For every $p\in\{1,\ldots,s\}$, the contour $\Gamma_p$ is supplied with the spectral split
	\[
		p_{\mathrm{in}}=p_p,
		\qquad
		p_{\mathrm{out}}=\prod_{q\ne p}p_q
	\]
	in the sense of \Cref{dfn_spectral_split_contour}, and this split is Riesz-compatible for $A$.
\end{definition}

\begin{remark}
	\label{rem_contour_cluster_datum_restriction}
	Completeness does not require splitting repeated eigenvalues or deciding their multiplicities.
	A cluster factor $p_p$ may contain several roots, including roots that are too close to separate at the available precision.
	The datum says that the selected factors cover the characteristic polynomial and that the contours certify where those factors lie.
	It can be certified at a finite precision by checking the displayed polynomial identity $p_A=p_1\cdots p_s$, the Bezout identities for the factors, lower bounds on the relevant factors along each contour, and the contour integrals in \Cref{dfn_spectral_split_contour}.
\end{remark}

\begin{remark}
	\label{rem_cluster_datum_valid_spectral_split}
	For each fixed contour $\Gamma_p$, the prescription
	\[
		p_{\mathrm{in}}=p_p,
		\qquad
		p_{\mathrm{out}}=\prod_{q\ne p}p_q
	\]
	is a valid spectral split because
	\[
		p_A=p_{\mathrm{in}}p_{\mathrm{out}}
	\]
	holds by the factorization datum.
	The Bezout relation
	\[
		up_{\mathrm{in}}+vp_{\mathrm{out}}=1
	\]
	is obtained from the pairwise Bezout data by the same finite multiplication used in \Cref{lem_monic_coprime_product_remainder}.
	Thus no root separation is performed at this step: the input is a finite factorization and finite contour certificates for those factors.
\end{remark}

\begin{theorem}[Resolvent cluster decomposition]
	\label{thm_resolvent_cluster_decomposition}
	Let
	\[
		(\Gamma_p,\lambda_p)_{p=1}^s
	\]
	be a contour cluster datum for $A$.
	Put
	\[
		P_p:=P_{\Gamma_p}(A).
	\]
	Then
	\[
		P_pP_q=0
		\quad (p\ne q),
		\qquad
		\sum_{p=1}^sP_p=I,
		\qquad
		AP_p=P_pA.
	\]
	Moreover,
	\[
		A
		=
		\sum_{p=1}^sAP_p
	\]
	and the tagged cluster approximation
	\[
		A_\Lambda
		:=
		\sum_{p=1}^s\lambda_pP_p
	\]
	satisfies the contour estimate
	\[
		\nrm{A-A_\Lambda}
		\le
		\sum_{p=1}^s
		\frac{\length(\Gamma_p)}{2\pi}
		\sup_{z\in\Gamma_p}\abs{z-\lambda_p}
		\sup_{z\in\Gamma_p}\nrm{R_A(z)}.
	\]
\end{theorem}

\begin{proof}
	For each contour $\Gamma_p$, let
	\[
		e_p:=v_p\prod_{q\ne p}p_q
	\]
	be the algebraic projector polynomial supplied by the spectral split for $\Gamma_p$.
	The Bezout identity for that split gives
	\[
		e_p\equiv1\quad\text{mod }p_p,
		\qquad
		e_p\equiv0\quad\text{mod }p_q\quad(q\ne p).
	\]
	In remainder notation this says
	\[
		\rem_{p_p}(e_p-1)=0,
		\qquad
		\rem_{p_q}(e_p)=0\quad(q\ne p).
	\]
	Fix $p\ne q$.
	For the factor $p_p$, the second relation for $e_q$ gives
	\[
		\rem_{p_p}(e_q)=0,
	\]
	hence $\rem_{p_p}(e_pe_q)=0$.
	For the factor $p_q$, the second relation for $e_p$ gives
	\[
		\rem_{p_q}(e_p)=0,
	\]
	hence $\rem_{p_q}(e_pe_q)=0$.
	For every remaining factor $p_r$, either relation gives
	\[
		\rem_{p_r}(e_pe_q)=0.
	\]
	Therefore $e_pe_q$ has zero monic remainder modulo every factor.
	Similarly, for a fixed factor $p_r$,
	\[
		\rem_{p_r}\left(\sum_{p=1}^se_p-1\right)
		=
		\rem_{p_r}\left((e_r-1)+\sum_{p\ne r}e_p\right)
		=
		0.
	\]
	By \Cref{lem_monic_coprime_product_remainder}, applied to the factor system $p_1,\ldots,p_s$, there are polynomials $h_{pq}$ and $h$ such that
	\[
		e_pe_q=h_{pq}p_A
		\quad(p\ne q),
		\qquad
		\sum_{p=1}^se_p-1=hp_A.
	\]
	Applying \Cref{lem_cayley_hamilton_constructive} gives
	\[
		P_pP_q=0
		\quad(p\ne q),
		\qquad
		\sum_{p=1}^sP_p=I.
	\]
	The idempotence $P_p^2=P_p$ for each $p$ is \Cref{thm_finite_dimensional_riesz_projector} applied to the contour $\Gamma_p$.
	The commutation identity $AP_p=P_pA$ follows from \Cref{thm_finite_dimensional_riesz_projector}.
	Therefore
	\[
		A=A\left(\sum_{p=1}^sP_p\right)=\sum_{p=1}^sAP_p.
	\]
	For the error estimate, use
	\[
		AP_p
		=
		\frac{1}{2\pi i}
		\int_{\Gamma_p}AR_A(z)\diff z.
	\]
	Since $AR_A(z)=zR_A(z)-I$ and $\int_{\Gamma_p}I\diff z=0$,
	\[
		AP_p-\lambda_pP_p
		=
		\frac{1}{2\pi i}
		\int_{\Gamma_p}(z-\lambda_p)R_A(z)\diff z.
	\]
	Taking norms and summing over $p$ gives the displayed bound.
	The quantities on the right are finite contour quantities: the lengths and tag distances are read from the rational vertices, and the resolvent bound follows from the determinant lower bound and a finite bound on the adjugate entries along each contour.
\end{proof}

\begin{corollary}[Finite approximate eigendecomposition]
	\label{cor_finite_approximate_eigendecomposition}
	Let a contour cluster datum for $A$ be given.
	For every $\eps>0$ one can construct matrices
	\[
		P_{p,h}
	\]
	with error budgets $\eta_p>0$ satisfying
	\[
		\nrm{P_{p,h}-P_{\Gamma_p}(A)}\le \eta_p,
		\qquad
		\sum_{p=1}^s\abs{\lambda_p}\eta_p\le\eps,
	\]
	and
	\[
		A_{\Lambda,h}:=\sum_{p=1}^s\lambda_pP_{p,h}
	\]
	such that
	\[
		\nrm{A-A_{\Lambda,h}}
	\]
	is bounded by the contour estimate in \Cref{thm_resolvent_cluster_decomposition} plus $\eps$.
	The matrices $P_{p,h}$ are finite quadrature sums of resolvents and satisfy the approximate projector and approximate commutation estimates of \Cref{cor_approximate_resolvent_projector_certificate}.
\end{corollary}

\begin{proof}
	Apply \Cref{prop_certified_resolvent_quadrature} to each contour with error budget small enough that
	\[
		\nrm{\sum_{p=1}^s\lambda_p(P_{p,h}-P_p)}\le\eps.
	\]
	Combine this with \Cref{thm_resolvent_cluster_decomposition}.
\end{proof}

\begin{remark}
	\label{rem_resolvent_scope}
	This section is not a replacement for a full constructive spectral theorem.
	The constructive advantage is that the certificates are contour separation bounds and finite quadrature data.
	The price is also clear: when no contour can be certified apart from the spectrum, the method returns a coarser cluster or no cluster at that scale.
	This is the expected behavior for matrices with repeated or nearly repeated eigenvalues.
\end{remark}

\subsection{Linear Stability}
\label{subsec_linear_stability}

\begin{definition}[Matrix exponential and linear system]
	\label{dfn_matrix_exponential_linear_system}
	For $A\in\mathbb C^{n\times n}$ define
	\[
		\exp(tA):=
		\sum_{m=0}^{\infty}\frac{t^mA^m}{m!}.
	\]
	The series is absolutely and uniformly convergent on every bounded interval of $t$.
	The solution of the autonomous linear system
	\[
		\dot x(t)=Ax(t),
		\qquad
		x(0)=x_0,
	\]
	is
	\[
		x(t):=\exp(tA)x_0.
	\]
\end{definition}

\begin{definition}[Left pole-separating contour]
	\label{dfn_left_pole_separating_contour}
	Let $A\in\mathbb C^{n\times n}$ and $\alpha>0$.
	An oriented simple rational polygonal contour $\Gamma$ is called a \emph{left pole-separating contour with margin $\alpha$} if it is resolvent-admissible for $A$, is supplied with the Riesz-compatible spectral split
	\[
		p_{\mathrm{in}}=p_A,
		\qquad
		p_{\mathrm{out}}=1,
		\qquad
		u=0,
		\qquad
		v=1,
	\]
	and satisfies
	\[
		\Re{z}\le-\alpha
	\]
	for every $z$ on $\Gamma$.
	It is supplied with a resolvent bound $M_\Gamma$ if
	\[
		\nrm{R_A(z)}\le M_\Gamma
	\]
	for every $z$ on $\Gamma$.
\end{definition}

\begin{remark}
	\label{rem_left_pole_certificate_reading}
	Classically, \Cref{dfn_left_pole_separating_contour} says that all poles of the linear system lie to the left of the line $\Re{z}=-\alpha$.
	The constructive datum is stronger and finite: a contour, a determinant lower bound on the contour, the full-factor spectral split, and the displayed left-half-plane inequality along the contour.
\end{remark}

\begin{lemma}[Exponential contour formula]
	\label{lem_exponential_contour_formula}
	Let $\Gamma$ be a left pole-separating contour for $A$.
	Then, for every $t\ge0$,
	\[
		\exp(tA)
		=
		\frac{1}{2\pi i}
		\int_\Gamma \exp(tz)R_A(z)\diff z.
	\]
\end{lemma}

\begin{proof}
	The full spectral split has $e_\Gamma=1$.
	Riesz-compatibility therefore gives
	\[
		\frac{1}{2\pi i}\int_\Gamma R_A(z)\diff z=I.
	\]
	The resolvent identity $(zI-A)R_A(z)=I$ gives
	\[
		zR_A(z)=I+AR_A(z).
	\]
	Inductively, for every $m\ge0$,
	\[
		z^mR_A(z)
		=
		\sum_{r=0}^{m-1}z^{m-1-r}A^r+A^mR_A(z),
	\]
	where the finite sum is empty when $m=0$.
	The contour integral of each polynomial term is zero, since on each segment the integral is an endpoint difference of a polynomial primitive and the endpoint contributions cancel around the cyclic contour.
	Hence
	\[
		\frac{1}{2\pi i}
		\int_\Gamma z^mR_A(z)\diff z
		=
		A^m.
	\]
	For fixed $t$ the scalar exponential series is uniformly convergent on the finite contour:
	\[
		\exp(tz)=\sum_{m=0}^{\infty}\frac{t^mz^m}{m!}.
	\]
	The resolvent is bounded on $\Gamma$ by admissibility and \Cref{prop_resolvent_regular_on_contour}, so integration may be passed through this uniformly convergent series.
	Therefore
	\[
		\frac{1}{2\pi i}
		\int_\Gamma \exp(tz)R_A(z)\diff z
		=
		\sum_{m=0}^{\infty}\frac{t^m}{m!}
		\left(
		\frac{1}{2\pi i}
		\int_\Gamma z^mR_A(z)\diff z
		\right)
		=
		\sum_{m=0}^{\infty}\frac{t^mA^m}{m!}.
	\]
\end{proof}

\begin{theorem}[Pole separation implies exponential stability]
	\label{thm_left_poles_exponential_stability}
	Consider the linear system
	\[
		\dot x(t)=Ax(t).
	\]
	If $A$ admits a left pole-separating contour $\Gamma$ with margin $\alpha>0$ and resolvent bound $M_\Gamma$, then the origin is exponentially stable:
	\[
		\nrm{x(t)}
		\le
		C_\Gamma e^{-\alpha t}\nrm{x_0}
		\qquad
		(t\ge0),
	\]
	where
	\[
		C_\Gamma:=
		\frac{\length(\Gamma)}{2\pi}M_\Gamma.
	\]
\end{theorem}

\begin{proof}
	By \Cref{lem_exponential_contour_formula},
	\[
		x(t)
		=
		\frac{1}{2\pi i}
		\int_\Gamma \exp(tz)R_A(z)x_0\diff z.
	\]
	For every $z$ on $\Gamma$,
	\[
		\abs{\exp(tz)}
		=
		\exp(t\,\Re{z})
		\le
		e^{-\alpha t}.
	\]
	Thus
	\[
		\nrm{x(t)}
		\le
		\frac{\length(\Gamma)}{2\pi}
		\sup_{z\in\Gamma}\abs{\exp(tz)}
		\sup_{z\in\Gamma}\nrm{R_A(z)}
		\nrm{x_0}
		\le
		\frac{\length(\Gamma)}{2\pi}
		M_\Gamma e^{-\alpha t}\nrm{x_0}.
	\]
	This is the claimed estimate.
\end{proof}

\begin{theorem}[Certified pole-box test]
	\label{thm_certified_pole_box_exponential_stability}
	Consider the linear system
	\[
		\dot x(t)=Ax(t).
	\]
	Let
	\[
		(\Gamma_p,\lambda_p)_{p=1}^s
	\]
	be a contour cluster datum for $A$.
	Assume that each contour is supplied with a resolvent bound $M_p$ and that, for some $\alpha>0$,
	\[
		\Re{z}\le-\alpha
	\]
	for every $z$ on every $\Gamma_p$.
	Then the origin is exponentially stable:
	\[
		\nrm{x(t)}
		\le
		C e^{-\alpha t}\nrm{x_0}
		\qquad
		(t\ge0),
	\]
	where
	\[
		C:=
		\sum_{p=1}^s\frac{\length(\Gamma_p)}{2\pi}M_p.
	\]
\end{theorem}

\begin{proof}
	For each $p$, Riesz-compatibility gives
	\[
		\frac{1}{2\pi i}\int_{\Gamma_p}R_A(z)\diff z=P_p.
	\]
	As in the proof of \Cref{lem_exponential_contour_formula}, the resolvent identity gives, for every $m\ge0$,
	\[
		\frac{1}{2\pi i}
		\int_{\Gamma_p}z^mR_A(z)\diff z
		=
		A^mP_p.
	\]
	Passing the uniformly convergent exponential series through the finite contour integral gives
	\[
		\frac{1}{2\pi i}
		\int_{\Gamma_p}\exp(tz)R_A(z)\diff z
		=
		\exp(tA)P_p.
	\]
	Summing over $p$ and using $\sum_{p=1}^sP_p=I$ from \Cref{thm_resolvent_cluster_decomposition},
	\[
		x(t)
		=
		\exp(tA)x_0
		=
		\sum_{p=1}^s
		\frac{1}{2\pi i}
		\int_{\Gamma_p}\exp(tz)R_A(z)x_0\diff z.
	\]
	Hence
	\[
		\nrm{x(t)}
		\le
		\sum_{p=1}^s
		\frac{\length(\Gamma_p)}{2\pi}
		\sup_{z\in\Gamma_p}\abs{\exp(tz)}
		\sup_{z\in\Gamma_p}\nrm{R_A(z)}
		\nrm{x_0}.
	\]
	The left-half-plane condition gives
	\[
		\sup_{z\in\Gamma_p}\abs{\exp(tz)}
		\le
		e^{-\alpha t},
	\]
	and the displayed estimate follows.
\end{proof}

\begin{remark}
	\label{rem_engineering_pole_test}
	\Cref{thm_certified_pole_box_exponential_stability} is the usual engineering pole test in certified form.
	Numerically, one computes pole approximations, draws small rational polygonal boxes around the computed pole clusters, checks that the right edge of every box lies left of $-\alpha$, and verifies that no zero of $p_A$ lies on any box boundary by a determinant lower bound or an argument-principle computation.
	The theorem then gives the exponential estimate without constructing eigenvectors or a Jordan form.
\end{remark}

\section{Statement of AI Use}
\label{sec_ai}
AI was used in this work for clean English, LaTeX formatting and lesser technical parts of some geometric constructions, followed by a thorough line-for-line validation.
The manuscript is fully conceptualized by the author and based on a series of preceding works by him throughout a decade.
Every section builds on an earlier, original material.

\appendix
\section{Chan Integration-Space Comparison}
\label{app_chan_integration_space}

\begin{definition}[Integration space]
	\label{dfn_integrable_space}
	An \emph{integration space} is a triple
	\[
		\left(\Omega,\mathcal L,\int_\Omega\right)
	\]
	where $\Omega$ is a nonempty set, $\mathcal L$ is a set of real-valued functions whose domains are subsets of $\Omega$, and $\int_\Omega:\mathcal L\to\R$ is a nonzero functional with domain $\mathcal L$, satisfying the following conditions.
	Pointwise operations are understood on their natural domains; for example, $af+bg$ is defined on $\dom{f}\cap\dom{g}$.
	\begin{enumerate}
		\item If $f,g\in\mathcal L$ and $a,b\in\R$, then
		\[
			af+bg,\qquad |f|,\qquad f\land1
		\]
		belong to $\mathcal L$, and
		\[
			\int_\Omega(af+bg)
			=
			a\int_\Omega f+b\int_\Omega g.
		\]
		\item If $f_0,f_1,f_2,\ldots\in\mathcal L$, if $f_i\ge0$ for each $i\ge1$, and if
		\[
			\sum_{i=1}^\infty \int_\Omega f_i
			<
			\int_\Omega f_0,
		\]
		then there exists
		\[
			\omega\in\bigcap_{i=0}^{\infty}\dom{f_i}
		\]
		such that
		\[
			\sum_{i=1}^\infty f_i(\omega)<f_0(\omega).
		\]
		\item For each $f\in\mathcal L$,
		\[
			\int_\Omega (f\land n)\to\int_\Omega f,
			\qquad
			\int_\Omega (|f|\land n^{-1})\to0
		\]
		as $n\to\infty$ through positive integers.
	\end{enumerate}
	A function in $\mathcal L$ is called \emph{integrable} relative to $\int_\Omega$.
\end{definition}

\begin{remark}
	\label{rem_integrable_space_chan}
	This is Chan's scalar integration-space definition, written with the symbol $\int$ instead of a separate letter for the functional \cite{Chan2019FoundationsCon}.
	The definition is scalar: the order, absolute value, truncation, and positivity condition are all statements about real-valued functions.
\end{remark}

\begin{definition}[Vector-valued integration module]
	\label{dfn_vector_valued_integration_module}
	Let $\left(\Omega,\mathcal L,\int_\Omega\right)$ be an integration space and let $m\in\N$ be positive.
	Set
	\[
		\mathcal L^m:=
		\left\{
			F=(F_1,\ldots,F_m):F_j\in\mathcal L,\ j=1,\ldots,m
		\right\}.
	\]
	An element $F=(F_1,\ldots,F_m)$ is regarded as an $\R^m$-valued function on the common domain $\bigcap_{j=1}^m\dom{F_j}$.
	For $F\in\mathcal L^m$, define
	\[
		\int_\Omega F
		:=
		\left(
			\int_\Omega F_1,\ldots,\int_\Omega F_m
		\right)
		\in\R^m.
	\]
	The pair $\left(\mathcal L^m,\int_\Omega\right)$, with the integral symbol understood componentwise, is called the $m$-dimensional integration module over $\left(\Omega,\mathcal L,\int_\Omega\right)$.
\end{definition}

\begin{remark}
	\label{rem_vector_integration_componentwise}
	For $m>1$ we do not modify \Cref{dfn_integrable_space}.
	There is no canonical lattice order on $\R^m$, so Chan's positivity axiom remains scalar.
	Vector-valued regular functions are integrated componentwise; equivalently, every scalar linear observable of a vector-valued integrable function is integrable.
\end{remark}

\begin{lemma}[Finite positivity]
	\label{lem_regular_finite_sum_positivity}
	Let $B$ be a full block.
	Let $f_0,f_1,\ldots,f_N\in\BReg{B}$, and let $f_i\ge0$ for $i=1,\ldots,N$.
	If
	\[
		\sum_{i=1}^N\int_B f_i
		<
		\int_B f_0,
	\]
	then there exists $x\in B$ such that
	\[
		\sum_{i=1}^N f_i(x)<f_0(x).
	\]
\end{lemma}

\begin{proof}
	Set
	\[
		h:=f_0-\sum_{i=1}^N f_i.
	\]
	By \Cref{lem_regular_functions_finite_algebra}, $h$ is regular and
	\[
		\int_Bh>0.
	\]
	Choose a rational $\alpha>0$ with $\alpha<\int_Bh$.
	Since $h$ is Riemann integrable, choose a tagged multiblock $\mathcal P$ with
	\[
		\abs{\Sigma(h,\mathcal P)-\int_Bh}<\alpha/2.
	\]
	Then $\Sigma(h,\mathcal P)>\alpha/2$.
	The sum defining $\Sigma(h,\mathcal P)$ is finite.
	If every tag $\xi$ of $\mathcal P$ satisfied $h(\xi)\le0$, then $\Sigma(h,\mathcal P)\le0$, a contradiction.
	Therefore one tag satisfies $h(\xi)>0$, which is the required point.
\end{proof}

\begin{lemma}[Ye bisection positivity for regular functions]
	\label{lem_ye_bisection_regular_positivity}
	Let $B$ be a full block.
	Let $f_0,f_1,f_2,\ldots\in\BReg{B}$, and let $f_i\ge0$ for $i\ge1$.
	If
	\[
		\sum_{i=1}^\infty\int_B f_i
		<
		\int_B f_0,
	\]
	then there exists $\omega\in B$ such that
	\[
		\sum_{i=1}^\infty f_i(\omega)<f_0(\omega).
	\]
\end{lemma}

\begin{proof}
	The proof is the dyadic bisection and doubling construction of Ye's \cite[Lemma~6.1]{Ye2011StrictFinitism}, with trimmed coordinate refinements inserted to handle the finite support boundaries of regular functions.
	Set
	\[
		F:=f_0\lor0.
	\]
	The function $F$ is regular by \Cref{cor_regular_functions_lattice_closure}.
	Since $F\ge f_0$ and $f_i\ge0$, the number
	\[
		\eta
		:=
		\int_B F
		-
		\sum_{i=1}^\infty\int_B f_i
	\]
	is positive.
	Choose $\eps>0$ so small that
	\[
		\eps\mu(B)<\eta/2,
	\]
	and let $h_\eps$ be the constant function $\eps$ on $B$.
	It is enough to construct $\omega\in B$ with
	\[
		h_\eps(\omega)
		+
		\sum_{i=1}^\infty f_i(\omega)
		\le
		F(\omega).
	\]
	Indeed, then $F(\omega)>0$, hence $F(\omega)=f_0(\omega)$, and the desired strict inequality follows.
	For a full subblock $Q\subseteq B$ write
	\[
		\Delta_Q(g;(g_i)_{i\ge1})
		:=
		\int_Q g
		-
		\sum_{i=1}^\infty\int_Q g_i.
	\]
	The restrictions to such $Q$ are integrated by the same tagged-multiblock Riemann construction as on $B$.
	Starting from the positive gap
	\[
		\Delta_B(F;(h_\eps,f_1,f_2,\ldots))>0,
	\]
	construct nested full blocks
	\[
		B\supset Q_1\supset Q_2\supset\cdots
	\]
	with $\diam(Q_k)\to0$ and
	\[
		\Delta_{Q_k}(F;(h_\eps,f_1,f_2,\ldots))>0.
	\]
	At the $k$th step, take representability witnesses for the finitely many supports occurring in
	\[
		F,\quad h_\eps,\quad f_1,\ldots,f_{k+1}
	\]
	with total exception cost small compared with the current positive gap.
	Apply the trimmed coordinate refinement inside $Q_k$.
	The exception and coordinate-boundary cells remove less than half of the gap, so the finite sum of the gaps over the remaining trimmed cells is still positive.
	A finite search through those cells gives $Q_{k+1}$ with positive gap.
	The refinement is chosen thin enough that $Q_{k+1}$ also fixes the branch data of $F,h_\eps,f_1,\ldots,f_{k+1}$.
	The nested blocks determine a point $\omega\in B$.
	For every fixed $N$, the functions $F,h_\eps,f_1,\ldots,f_N$ are continuous on all sufficiently late $Q_k$ with fixed branches.
	The positive gap on such a $Q_k$ gives a tag $\xi_k\in Q_k$ with
	\[
		h_\eps(\xi_k)
		+
		\sum_{i=1}^N f_i(\xi_k)
		<
		F(\xi_k).
	\]
	Letting $k\to\infty$ gives
	\[
		h_\eps(\omega)
		+
		\sum_{i=1}^N f_i(\omega)
		\le
		F(\omega)
	\]
	for every $N$.
	Choose indices $N_j$ such that
	\[
		T_j
		:=
		\sum_{i=N_j}^\infty\int_B f_i
		<
		4^{-j}\eta/4.
	\]
	Apply the same construction to Ye's rearranged sequence
	\[
		h_\eps,\quad
		f_1,\quad
		2\sum_{i=N_1}^{N_2-1}f_i,\quad
		f_2,\quad
		2^2\sum_{i=N_2}^{N_3-1}f_i,\quad
		f_3,\ldots .
	\]
	The extra amplified-tail integral is bounded by
	\[
	\begin{aligned}
		\sum_{j=1}^{J}2^j
		\sum_{i=N_j}^{N_{j+1}-1}\int_B f_i
		&=
		\sum_{j=1}^{J}2^j(T_j-T_{j+1})
		\\
		&=
		2T_1+\sum_{j=2}^{J}2^{j-1}T_j-2^JT_{J+1}
		\\
		&\le
		2T_1+\sum_{j=2}^{J}2^{j-1}T_j
		\\
		&<
		\frac{\eta}{8}
		+
		\frac{\eta}{4}
		\sum_{j=2}^{J}2^{-j-1}
		\\
		&<
		\frac{\eta}{4}.
	\end{aligned}
	\]
	Hence the integral sum of the rearranged sequence is strictly below $\int_BF$ because the original gap is $\eta$ and $\int_Bh_\eps<\eta/2$.
	The finite-partial-sum conclusion for this rearranged sequence gives
	\[
		\sum_{i=N_j}^{N_{j+1}-1}f_i(\omega)
		\le
		2^{-j}F(\omega).
	\]
	Thus the value series $\sum_i f_i(\omega)$ converges, and the finite-partial-sum inequalities pass to
	\[
		h_\eps(\omega)
		+
		\sum_{i=1}^\infty f_i(\omega)
		\le
		F(\omega).
	\]
\end{proof}

\begin{remark}[Tails in the Ye construction]
	\label{rem_ye_tail_control_limit_point}
	The amplified-tail step in \Cref{lem_ye_bisection_regular_positivity} controls the tail only at the final limit point $\omega$.
	It does not provide a finite stage $x_k$ at which the full infinite tail is already certified.
	Such a finite stopping statement requires extra pointwise tail data, for instance a summable uniform bound $0\le f_i(x)\le M_i$ with a known convergence modulus for $\sum_iM_i$.
\end{remark}

\begin{theorem}[Regular functions as an integration space]
	\label{thm_regular_functions_integration_space}
	Let $B$ be a full block.
	The triple
	\[
		\left(B,\BReg{B},\int_B\right)
	\]
	is an integration space in the sense of \Cref{dfn_integrable_space}.
\end{theorem}

\begin{proof}
	The functional is nonzero because every full block contains a smaller full block $A\Subset B$, and the representable indicator $\indic{A}$ satisfies $\int_B\indic{A}=\mu(A)>0$.
	Linearity and finite order are \Cref{lem_regular_functions_finite_algebra,lem_regular_integral_finite_order}.
	The lattice operations $|f|$ and $f\land1$ are obtained from \Cref{cor_regular_functions_lattice_closure}.
	Chan's countable positivity axiom is \Cref{lem_ye_bisection_regular_positivity}.
	The truncation axiom follows from boundedness.
	Indeed, for each $f\in\BReg{B}$ and for $n$ larger than a bound witness for $f$ one has $f\land n=f$, while
	\[
		0
		\le
		\int_B(|f|\land n^{-1})
		\le
		n^{-1}\mu(B)
		\to0.
	\]
	Since every member of $\BReg{B}$ is Riemann integrable by \Cref{lem_regular_functions_riemann_integrable}, these axioms restrict to $\BReg{B}$.
\end{proof}

\begin{corollary}[Vector-valued regular integration module]
	\label{crl_vector_regular_integration_module}
	Let $B$ be a full block.
	For every positive $m\in\N$, the $\R^m$-valued effectively bounded regular functions on $B$ whose supports are well contained in $B$ form the componentwise integration module over the scalar integration space in \Cref{thm_regular_functions_integration_space}.
\end{corollary}

\begin{proof}
	This is \Cref{dfn_vector_valued_integration_module} applied to \Cref{thm_regular_functions_integration_space}.
\end{proof}

\printbibliography

\end{document}